\documentclass{amsart}

\usepackage{amsmath, amssymb, amsthm}
\usepackage{fullpage}
\usepackage{hyperref}
\usepackage{upgreek}
\usepackage{slashed}
\usepackage[retainorgcmds]{IEEEtrantools}
\usepackage{graphicx}
\usepackage{bbm}

\newtheorem{theorem}{Theorem}[section]
\newtheorem{lemma}[theorem]{Lemma}

\newtheorem{proposition}[theorem]{Proposition}
\newtheorem{definition}[theorem]{Definition}
\newtheorem{example}[theorem]{Example}
\newtheorem{remark}[theorem]{Remark}

\newcommand{\bc}{\varepsilon_0} 
\newcommand{\data}{\varepsilon_0}
\newcommand{\gslash}{\slashed{g}}
\newcommand{\Tslash}{\slashed{T}}
\newcommand{\Gammaslash}{\slashed{\Gamma}}

\newcommand{\nablaslash}{\slashed{\nabla}}
\newcommand{\Dslash}{\slashed{\mathcal{D}}}
\newcommand{\epsslash}{\slashed{\epsilon}}
\newcommand{\divslash}{\slashed{\mathrm{div}}}
\newcommand{\curlslash}{\slashed{\mathrm{curl}}}
\newcommand{\eslash}{\slashed{e}}
\newcommand{\Rslash}{\slashed{R}}
\newcommand{\chibar}{\underline{\chi}}
\newcommand{\etabar}{\underline{\eta}}
\newcommand{\mubar}{\underline{\mu}}
\newcommand{\omegadag}{\omega^{\dagger}}
\newcommand{\omegabar}{\underline{\omega}}
\newcommand{\alphabar}{\underline{\alpha}}
\newcommand{\betabar}{\underline{\beta}}
\newcommand{\tr}{\mathrm{tr}}
\newcommand{\area}{\mathrm{Area}}
\newcommand{\Riccit}{\overset{(3)}{\Gamma}_p}
\newcommand{\Riccif}{\overset{(4)}{\Gamma}_p}
\newcommand{\Thetat}{\overset{(3)}{\Theta}}
\newcommand{\Thetaf}{\overset{(4)}{\Theta}}
\newcommand{\TM}{T\mathcal{M}}
\newcommand{\gbar}{\overline{g}}
\newcommand{\ver}{\mathrm{Ver}}
\newcommand{\hor}{\mathrm{Hor}}
\newcommand{\nablabar}{\overline{\nabla}}
\newcommand{\supp}{\mathrm{supp}}
\newcommand{\partrans}{\mathrm{Par}}
\newcommand{\rp}{h}

\begin{document}
\title[Stability of Minkowski Space for Massless Einstein--Vlasov]{The Global Nonlinear Stability of Minkowski Space for the Massless Einstein--Vlasov System}
\author{Martin Taylor}
\address{University of Cambridge, Department of Pure Mathematics and Mathematical Statistics, Wilberforce Road, Cambridge, CB3 0WB, United Kingdom}
\email{m.taylor@maths.cam.ac.uk}
\date{\today}

\begin{abstract}
	Minkowski space is shown to be globally stable as a solution to the Einstein--Vlasov system in the case when all particles have zero mass.  The proof proceeds by showing that the matter must be supported in the ``wave zone'', and then proving a small data semi-global existence result for the characteristic initial value problem for the massless Einstein--Vlasov system in this region.  This relies on weighted estimates for the solution which, for the Vlasov part, are obtained by introducing the Sasaki metric on the mass shell and estimating Jacobi fields with respect to this metric by geometric quantities on the spacetime.  The stability of Minkowski space result for the vacuum Einstein equations is then appealed to for the remaining regions.
\end{abstract}
\maketitle

\tableofcontents

\section{Introduction}
It is of wide interest to understand the global dynamics of isolated self-gravitating systems in general relativity.  Without symmetry assumptions, problems of this form present a great challenge even for systems arising from small data.  In the vacuum, where no matter is present, the global properties of small data solutions were first understood in the monumental work of Christodoulou--Klainerman \cite{ChKl}.  They show that Minkowski space is globally stable to small perturbations of initial data, i.\@e.\@ the maximal development of an asymptotically flat initial data set for the vacuum Einstein equations which is sufficiently close to that of Minkowski space is geodesically complete, possesses a complete future null infinity and asymptotically approaches Minkowski space in every direction (see also Lindblad--Rodnianski \cite{LiRo}, Bieri \cite{BiZi}).  

In the presence of matter, progress has been confined to models described by wave equations.\footnote{The analogue of the Christodoulou--Klainerman theorem has been shown by Zipser \cite{BiZi} to hold for electromagnetic matter described by the Maxwell equations.  See also Loizelet \cite{Lo}.  Speck \cite{Sp} shows a similar result for the Einstein equations coupled to a large family of nonlinear electromagnetic equations.  Lindblad--Rodnianski \cite{LiRo} in their treatment also consider matter described by a scalar field.  There has also been related work on the Einstein--Klein--Gordon system \cite{Wa}, \cite{LeMa}.  There are more global stability results for the Einstein equations with a positive cosmological constant, for example the works of Friedrich \cite{Fr}, Ringstr\"{o}m \cite{Ri} and Rodnianski--Speck \cite{RoSp}.  A more comprehensive list can be found in the introduction to the work of Had\u{z}i\'{c}--Speck \cite{HaSp}.}  Here collisionless matter, described by the Einstein--Vlasov system, is considered.  This is a model which has been widely studied in both the physics and mathematics communities; see the review paper of Andr\'{e}asson \cite{An} for a summary of mathematical work on the system.  New mathematical difficulties are present since the governing equations for the matter are now transport equations, though in the case considered here, where the particles have zero mass and hence travel through spacetime along null curves, the decay properties of the function describing the matter are compatible in a nice way with those of the spacetime metric.

The Einstein--Vlasov system takes the form
\begin{equation}
	 Ric_{\mu\nu} - \frac{1}{2}Rg_{\mu \nu} = T_{\mu \nu}, \qquad T_{\mu \nu}(x) = \int_{P_x} f p_{\mu}p_{\nu} , \label{eq:EV}
\end{equation}
\begin{equation}
	 X(f) = 0. \label{eq:vlas}
\end{equation}
The unknown is a Lorentzian manifold $(\mathcal{M},g)$ together with a \emph{particle density function} $f\colon P \to [0,\infty)$, defined on a subset $P\subset T\mathcal{M}$ of the tangent bundle of $\mathcal{M}$ called the \emph{mass shell}.  The function $f(x,p)$ describes the density of the matter at $x\in \mathcal{M}$ with velocity $p\in P_x \subset T_x \mathcal{M}$.  Here $(x^{\mu},p^{\mu})$ denote coordinates on the tangent bundle $T\mathcal{M}$ with $p^{\mu}$ conjugate to $x^{\mu}$, so that $(x,p)$ denotes the point $p^{\mu} \partial_{x^{\mu}} \vert_x$ in $T\mathcal{M}$.  The Ricci curvature and scalar curvature of $(\mathcal{M},g)$ are denoted $Ric, R$ respectively.  The integral in \eqref{eq:EV} is taken with respect to a natural volume form, defined later in Section \ref{subsec:massshell}.  The vector field $X \in \Gamma(TT\mathcal{M})$ is the \emph{geodesic spray}, i.\@e.\@ the generator of the geodesic flow, of $(\mathcal{M},g)$.  The Vlasov equation \eqref{eq:vlas} therefore says that, given $(x,p) \in T\mathcal{M}$, if $\gamma_{x,p}$ denotes the unique geodesic in $\mathcal{M}$ such that $\gamma_{x,p}(0) = x, \dot{\gamma}_{x,p}(0) = p$, then $f$ is constant along $(\gamma_{x,p}(s), \dot{\gamma}_{x,p}(s))$, i.\@e.\@ $f$ is preserved by the \emph{geodesic flow} of $(\mathcal{M},g)$.  Equation \eqref{eq:vlas} is therefore equivalent to
\begin{equation} \label{eq:vlasexp}
	f(x,p) = f(\exp_s(x,p)),
\end{equation}
for all $s \in \mathbb{R}$ such that the above expression is defined, where $\exp_s : \TM \to \TM$ is the exponential map defined by $\exp_s(x,p) = (\gamma_{x,p}(s), \dot{\gamma}_{x,p}(s))$.

In the case considered here, where the collisionless matter has zero mass, $f$ is supported on the mass shell 
\[
	P := \{ (x,p) \in T\mathcal{M} \mid p \text{ is null and future directed}\},
\]
a hypersurface in $T\mathcal{M}$.  In this case one sees, by taking the trace of \eqref{eq:EV}, that the scalar curvature $R$ must vanish for any solution of \eqref{eq:EV}--\eqref{eq:vlas} and the Einstein equations reduce to
\begin{equation} \label{eq:Einstein}
	Ric_{\mu\nu} = T_{\mu \nu}.
\end{equation}
The main result is the following.

\begin{theorem} \label{thm:main}
	Given a smooth asymptotically flat initial data set for the massless Einstein--Vlasov system suitably close to that of Minkowski Space such that the initial particle density function is compactly supported on the mass shell, the resulting maximal development is geodesically complete and possesses a complete future null infinity.  Moreover the support of the matter is confined to the region between two outgoing null hypersurfaces, and each of the Ricci coefficients, curvature components and components of the energy momentum tensor with respect to a double null frame decay towards null infinity with quantitative rates.
\end{theorem}

The proof of Theorem \ref{thm:main}, after appealing to the corresponding result for the vacuum Einstein equations, quickly reduces to a semi-global problem.  This reduction is outlined below and the semi-global problem treated here is stated in Theorem \ref{thm:main2}.

Theorem \ref{thm:main} extends a result of Dafermos \cite{Da} which establishes the above under the additional restricted assumption of spherical symmetry.  Note also the result of Rein--Rendall \cite{ReRe} which treats the \emph{massive case} in spherical symmetry, where all of the particles have mass $m>0$ (i.\@e.\@ $f$ is supported on the set of future pointing timelike vectors $p$ in $T\mathcal{M}$ such that $g(p,p) = -m^2$).  The main idea in \cite{Da} was to show, using a bootstrap argument, that, for sufficiently late times, the matter is supported away from the centre of spherical symmetry.  By Birkhoff's Theorem the centre is therefore locally isometric to Minkowski space at these late times and the extension principle of Dafermos--Rendall \cite{DaRe} (see also \cite{DaRe07}) then guarantees that the spacetime will be geodesically complete.  

In these broad terms, a similar strategy is adopted here.  The absence of good quantities satisfying monotonicity properties which are available in spherical symmetry, however, makes the process of controlling the support of the matter, and proving the semi-global existence result for the region where it is supported, considerably more involved.  The use of Birkhoff's Theorem and the Dafermos--Rendall extension principle also have to be replaced by the much deeper result of the stability of Minkowski space for the vacuum Einstein equations.  The use of the vacuum stability result, which is in fact appealed to in two separate places, is outlined below.

\subsection{The Uncoupled Problem} \label{subsec:uncoupled}
It is useful to first recall what happens in the uncoupled problem of the Vlasov equation on a fixed Minkowski background.  Let $v = \frac{1}{2}(t+r), u = \frac{1}{2}(t-r)$ denote standard null coordinates on Minkowski space $\mathbb{R}^{3+1}$ (these form a well defined coordinate system on the quotient manifold $\mathbb{R}^{3+1}/SO(3)$ away from the centre of spherical symmetry $\{r=0\}$) and suppose $f$ is a solution of the Vlasov equation \eqref{eq:vlas} with respect to this fixed background arising from initial data with compact support in space.  From the geometry of null geodesics in Minkowski space it is clear that the projection of the support of $f$ to the spacetime is related to the projection of the initial support of $f$ as depicted in the Penrose diagram in Figure \ref{fig:uncoupled}.

\begin{figure} 
  \centering
	\includegraphics[scale=0.25]{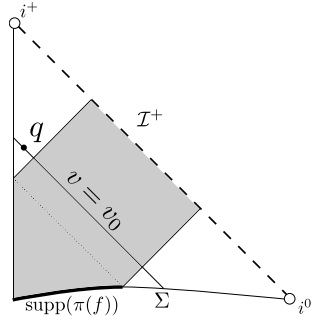}
	\caption{The projection of the support of $f$ in the uncoupled problem.} \label{fig:uncoupled}
\end{figure}

In particular, for sufficiently late advanced time $v_0$ the matter will be supported away from the centre $\{r=0\}$, and there exists a point $q \in \mathbb{R}^{3+1}/SO(3)$, lifting to a (round) 2-sphere $S \subset \mathbb{R}^{3+1}$, with $r(q) >0$ such that
\[
	\pi(\mathrm{supp}(f)) \cap \{ v \leq v_0\} \subset J^-(S),
\]
where $J^-(S)$ denotes the causal past of $S$ and $\pi \colon P \to \mathcal{M}$ denotes the natural projection.

\subsection{Initial Data and First Appeal to the Vacuum Result} \label{subsec:vacuumappeal}
Recall that initial data for the Einstein--Vlasov system \eqref{eq:EV}--\eqref{eq:vlas} consists of a 3-manifold $\Sigma$ with a Riemannian metric $g_0$, a symmetric $(0,2)$ tensor $K$ and an initial particle density function $f_0$ satisfying the constraint equations,
\begin{equation} \label{eq:constraints}
	\mathrm{div}_0 K_j - (d \tr_0 K)_j = T_{0j}, \qquad R_0 + (\tr_0 K)^2 - \vert K\vert^2_{g_0} = 2 T_{00},
\end{equation}
for $j=1,2,3$, where $\mathrm{div}_0, \tr_0,R_0$ denote the divergence, trace and scalar curvature of $g_0$ respectively, and $T_{00}, T_{0j}$ denote (what will become) the $00$ and $0j$ components of the energy momentum tensor.  See \cite{Ri} for a discussion of initial data for the Einstein--Vlasov system.  The topology of $\Sigma$ will here be assumed to be that of $\mathbb{R}^3$.  The issue of constructing solutions to the constraint equations \eqref{eq:constraints} will not be treated here.  A theorem of Choquet-Bruhat \cite{ChBr71} guarantees that, given such an initial data set, a solution to \eqref{eq:EV}--\eqref{eq:vlas} will exist locally in time.

The initial density function $f_0$ is assumed to have compact support.  It will moreover be assumed that $f_0$ and a finite number of its derivatives will be small initially.  The precise condition is given in Section \ref{section:ba}.  Note the assumption of compact support for $f_0$ is in both the spatial variable $x$, and in the momentum variable $p$.  As will become evident, the compact support in space is used in a crucial way.  The assumption of compact support in momentum is made for simplicity and can likely be weakened.\footnote{It is used in an important way in Section \ref{prop:emmain2} but this can in fact easily be avoided.}

Let $B \subset \Sigma$ be a simply connected compact set such that $\pi(\supp(f \vert_{P_{\Sigma}})) \subset B$, where $P_{\Sigma}$ denotes the mass shell over $\Sigma$.  By the domain of dependence property of the Einstein--Vlasov system the development of the complement of $B$ in $\Sigma$, $D^+(\Sigma \smallsetminus B)$, will solve the vacuum Einstein equations,
\begin{equation} \label{eq:Einsteinvacuum}
	Ric_{\mu \nu} = 0.
\end{equation}
The stability of Minkowski space theorem for the vacuum Einstein equations then guarantees the stability of this region.  See Klainerman--Nicol\`{o} \cite{KlNi} where exactly this situation is treated.  In particular, provided $g_0, K$ satisfy a smallness condition\footnote{The precise condition will not be discussed here.  See \cite{KlNi}.} in $\Sigma \smallsetminus B$ (i.\@e.\@ they are suitably close to the $g_0,K$ of Minkowski space), there exists a future complete, outgoing null hypersurface $\mathcal{N}$ in this region which can be foliated by a family of 2-spheres, $\{S_{u_0,v} \}$ parameterised by $v$, approaching the round 2-sphere as $v \to \infty$.  Moreover the Ricci coefficients and curvature components of the spacetime will decay to their corresponding Minkowski values and, by taking $g_0, K$ suitably small, certain weighted quantities involving them can be made arbitrarily small on $\mathcal{N}$.  It will be assumed that $g_0,K$ are sufficiently small so that the precise conditions stated in Theorem \ref{thm:main3} are satisfied on $\mathcal{N}$.  A second appeal to a form of the stability of Minkowski space result in the vacuum (which can be shown to also follow from the Christodoulou--Klainerman Theorem \cite{ChKl} using upcoming work) will be made in Section \ref{subsec:secondmain} below.

\subsection{Cauchy Stability} \label{subsec:Cauchystability}
By Cauchy stability for the Einstein--Vlasov system (see Choquet-Bruhat \cite{ChBr71} or Ringstr\"{o}m \cite{Ri}), Cauchy stability for the geodesic equations and the considerations of Section \ref{subsec:uncoupled}, provided the initial data on $\Sigma$ are taken sufficiently small, there exists a 2-sphere $S\subset \mathcal{M}$ and an incoming null hypersurface $\underline{\mathcal{N}}$ such that $S\subset \underline{\mathcal{N}}$, $\mathrm{Area}(S) >0$, $\pi(\text{supp}(f)) \cap S = \emptyset$, and
\[
	\pi(\text{supp}(f)) \cap J^-(\underline{\mathcal{N}}) \subset J^-(S).
\]
In other words, the existence of the point $q$ in the Penrose diagram of Figure \ref{fig:uncoupled} is stable.  It can moreover be assumed that the $\mathcal{N}$ above and $\underline{\mathcal{N}}$ intersect in one of the 2-spheres of the foliation of $\mathcal{N}$,
\[
	\mathcal{N} \cap \underline{\mathcal{N}} = S_{u_0, v_0},
\]
where $v_0$ can be chosen arbitrarily large.  The induced data on $\underline{\mathcal{N}}$ can be taken to be arbitrarily small, provided they are sufficiently small on $\Sigma$.

\subsection{A Second Version of the Main Theorem and Second Appeal to the Vacuum Result} \label{subsec:secondmain}
A more precise version of the main result can now be stated.  A final version, Theorem \ref{thm:main3}, is stated in Section \ref{section:ba}.

\begin{theorem} \label{thm:main2}
	Given characteristic initial data for the massless Einstein--Vlasov system \eqref{eq:EV}--\eqref{eq:vlas} on an outgoing null hypersurface $\mathcal{N}$ and an incoming null hypersurface $\underline{\mathcal{N}}$ as above\footnote{So that, in particular, the particle density function $f = 0$ on the mass shell over $\mathcal{N}$, and there exists a 2-sphere $S \subset \underline{\mathcal{N}}$ such that $\text{supp}(f)$ in the mass shell over $\underline{\mathcal{N}}$ is contained in the causal past, $J^-(S)$, of $S$.}, intersecting in a 2-sphere $S_{u_0,v_0}$ of the foliation of $\mathcal{N}$, then, if $v_0$ is sufficiently large and the characteristic initial data are sufficiently small\footnote{i.\@e.\@ certain weighted integrals of derivatives of metric components, Ricci coefficients and curvature components, along with pointwise bounds on certain derivatives of $f$ are small.  The precise smallness assumptions are given later in Section \ref{section:ba}.}, then there exists a unique spacetime $(\mathcal{M},g)$ endowed with a double null foliation $(u,v)$ solving the characteristic initial value problem for \eqref{eq:EV}--\eqref{eq:vlas} in the region $v_0\leq v < \infty$, $u_0 \leq u \leq u_f$, where $\mathcal{N} = \{u=u_0\}$,  $\underline{\mathcal{N}} = \{ v = v_0\}$, and $u_f$ can be chosen large so that $f=0$ on the mass shell over any point $x\in \mathcal{M}$ such that $u(x) \geq u_f -1$, i.\@e.\@ $\pi(\supp(f)) \subset J^-(\{u=u_f - 1\})$.  Moreover each of the Ricci coefficients, curvature components and components of the energy momentum tensor (with respect to a double null frame) decay towards null infinity with quantitative rates.
\end{theorem}

This is depicted in Figure \ref{fig:coupled}.

\begin{figure}
  \centering
	\includegraphics[scale=0.25]{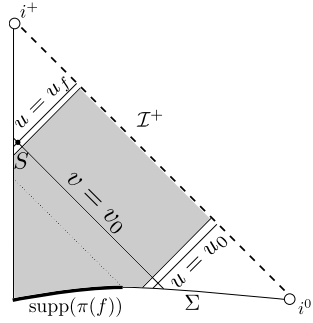}
	\caption{The matter is supported in the shaded region and hence the spacetime will solve the vacuum Einstein equations in the unshaded regions.} \label{fig:coupled}
\end{figure}

Theorem \ref{thm:main} follows from Theorem \ref{thm:main2} by the considerations of Section \ref{subsec:vacuumappeal}, Section \ref{subsec:Cauchystability}, and by another application of the vacuum stability of Minkowski space result with the induced data on a hyperboloid contained between the null hypersurfaces $\{u = u_f\}$ and $\{u = u_f-1\}$.  The problem of stability of Minkowski space for the vacuum Einstein equations \eqref{eq:Einsteinvacuum} with hyperboloidal initial data was treated by Friedrich \cite{Fr}, though his result requires the initial data to be \emph{asymptotically simple}.  This is, in general, inconsistent with the induced data arising from Theorem \ref{thm:main2}.\footnote{One could impose faster decay of the data on $\{u=u_0\}$ in Theorem \ref{thm:main2} and hope to propagate this decay so that the induced data on the hyperboloid is indeed sufficient for \cite{Fr} to apply directly.  We have chosen not to do so here in view of \cite{Ch00}, where Christodoulou shows that generic physically interesting spacetimes are never asymptotically simple.}  Whilst a proof of the hyperboloidal stability of Minkowski space problem with initial data compatible with Theorem \ref{thm:main2} can most likely be distilled from the work \cite{ChKl}, there is currently no precise statement to appeal to.  In future work it will be shown how one can alternatively appeal directly to \cite{ChKl} by extending the induced scattering data at null infinity and solving backwards, in the style of \cite{DHR}.

A precise formulation of Theorem \ref{thm:main}, including an explicit statement of the norms used in the first appeal to the vacuum result in Section \ref{subsec:vacuumappeal} and the Cauchy stability argument of Section \ref{subsec:Cauchystability}, will not be made here.  The assumptions made in Theorem \ref{thm:main3}, the final version of Theorem \ref{thm:main2}, will be given some justification at various places in the introduction however.  The remainder of the paper will concern Theorem \ref{thm:main2}, and in the remainder of the introduction its proof will be outlined.  The greatest new difficulty is in obtaining a priori control over derivatives of $f$.  The approach taken involves introducing the induced Sasaki metric on the mass shell $P$ and estimating certain Jacobi fields on $P$ in terms of geometric quantities on the spacetime $(\mathcal{M},g)$.  This approach is outlined in Section \ref{subsec:introf} below.

Note that the analogue of Theorem \ref{thm:main2} for the vacuum Einstein equations \eqref{eq:Einsteinvacuum} follows from a recent result of Li--Zhu \cite{LiZh}.

\subsection{The Bootstrap Argument} \label{subsec:introbootstrap}
The main step in the proof of Theorem \ref{thm:main2} is in obtaining global a priori estimates for all of the relevant quantities.  Once they have been established there is a standard procedure for obtaining global existence, which is outlined in Section \ref{section:lastslice}.  The remainder of the discussion is therefore focused on obtaining the estimates.

Moreover, using a bootstrap argument, it suffices to show that if the estimates already hold in a given \emph{bootstrap region} of the form $\{ u_0 \leq u \leq u' \} \cap \{ v_0 \leq v \leq v'\}$, depicted in Figure \ref{fig:br}, then they can be recovered in this region with better constants independently of $u',v'$.  This is extremely useful given the strongly coupled nature of the equations.

\begin{figure}
  \centering
	\includegraphics[scale=0.25]{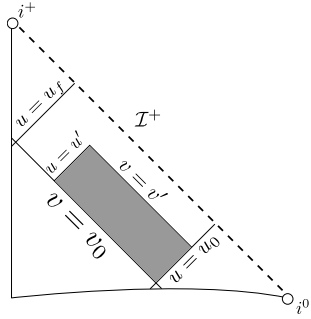}
	\caption{The bootstrap region.} \label{fig:br}
\end{figure}

The better constants in the bootstrap argument arise from either estimating the quantities by the initial data on $\{v=v_0\}$ and $\{ u = u_0\}$ or by $\frac{1}{v_0}$, and using the smallness of the initial data and the largeness of $v_0$.  Recall that, in the setting of Theorem \ref{thm:main}, both the largeness of $v_0$ and the smallness of the induced data on $\mathcal{N} = \{ u = u_0\}$, $\underline{\mathcal{N}} = \{ v = v_0\}$ arise by taking the asymptotically flat Cauchy data on $\Sigma$ to be suitably small.

\subsection{The Double Null Gauge}
The content of the Einstein equations is captured here through the structure equations and the null Bianchi equations associated to the double null foliation $(u,v)$.  The constant $u$ and constant $v$ hypersurfaces are outgoing and incoming null hypersurfaces respectively, and intersect in spacelike 2-spheres which are denoted $S_{u,v}$.  This choice of gauge is made due to its success in problems which require some form of the null condition\footnote{See Section \ref{subsec:pindex}.} to be satisfied.\footnote{It is well known that the Einstein equations in the harmonic gauge do not satisfy the classical null condition of \cite{Kl}.  Despite this fact, it has been shown by Lindblad--Rodnianski \cite{LiRo} that one can still prove stability of Minkowski space for the vacuum Einstein equations in this gauge.  One could therefore imagine adopting a similar strategy to approach the current problem in the harmonic gauge.  Note the recent work of Fajman--Joudioux--Smulevici \cite{FaJoSm} on the development of a vector field method for relativistic transport equations, which could play an important role in such an approach.}  See, for example, \cite{Ch1}, \cite{Ch}, \cite{KlNi}, \cite{Lu}, \cite{DHR}.

The foliation defines a double null frame (see Section \ref{subsec:coords}) in which one can decompose the Ricci coefficients, which satisfy so called \emph{null structure equations}, the Weyl (or conformal) curvature tensor, whose null decomposed components satisfy the \emph{null Bianchi equations}, and the energy momentum tensor (which, by the Einstein equations \eqref{eq:Einstein}, is equal to the Ricci curvature tensor).

It is the null structure and Bianchi equations which will be used, together with the Vlasov equation \eqref{eq:vlas}, to estimate the solution.  Following the notation of \cite{DHR}, \cite{Lu}, the null decomposed Ricci coefficients will be schematically denoted $\Gamma$.  Two examples are the outgoing shear $\hat{\chi}$, which is a $(0,2)$ tensor on the spheres $S_{u,v}$, and the renormalised outgoing expansion $\tr \chi - \frac{2}{r}$, which is a function on the spacetime, renormalised using the function $r$ so that the corresponding quantity in Minkowski space will vanish.

The null decomposed components of the Weyl curvature tensor will be schematically denoted $\psi$ and the null decomposed components of the energy momentum tensor will be schematically denoted $\mathcal{T}$.  This schematic notation, together with the $p$-index notation described in Section \ref{subsec:pindex} below, will be used to convey structural properties of the equations which are heavily exploited later.

\subsection{The Schematic Form of the Equations} \label{subsec:schematic}
The null structure equations for the Ricci coefficients $\Gamma$, which are stated in Section \ref{subsec:nsbianchi}, take the following schematic form,
\begin{equation} \label{eq:intronullstructure}
	\nablaslash_3 \Gamma = \frac{1}{r} \Gamma + \Gamma \cdot \Gamma + \psi + \mathcal{T}, \qquad \nablaslash_4 \Gamma = \frac{1}{r} \Gamma + \Gamma \cdot \Gamma + \psi + \mathcal{T}.
\end{equation}
Here $\nablaslash_3$ and $\nablaslash_4$ denote the projections of the covariant derivatives in the incoming and outgoing null directions respectively to the spheres $S_{u,v}$.  The $\frac{1}{r}\Gamma$ terms appear in the equations for the outgoing and incoming expansions $\tr \chi - \frac{2}{r}, \tr \chibar + \frac{2}{r}$, which are renormalised using the function $r$.  Each $\Gamma$ satisfies exactly one of the two form of equations \eqref{eq:intronullstructure} and hence are further decomposed as $\overset{(3)}{\Gamma}$ or $\overset{(4)}{\Gamma}$ depending on whether they satisfy an equation in the $\nablaslash_3$ or $\nablaslash_4$ direction respectively.  It should be noted that there are further null structure equations satisfied by the Ricci coefficients which take different forms to \eqref{eq:intronullstructure}, some of which will make an appearance later.

The Weyl curvature components $\psi$ can be further decomposed into \emph{Bianchi pairs}, defined in Section \ref{subsec:schemnot}, which are denoted $(\uppsi,\uppsi')$ (examples are $(\uppsi,\uppsi') = (\alpha,\beta)$ or $(\beta,(\rho,\sigma))$).  This notation is used to emphasise a special structure in the Bianchi equations, which take the form,
\begin{equation} \label{eq:introBianchi}
	\nablaslash_3 \uppsi = \Dslash_{\uppsi} \uppsi' + \frac{1}{r} \Gamma + \Gamma \cdot \psi + \Gamma \cdot \mathcal{T} + \nabla \mathcal{T}, \quad
	\nablaslash_4 \uppsi' = \Dslash_{\uppsi'} \uppsi + \frac{1}{r} \Gamma + \Gamma \cdot \psi + \Gamma \cdot \mathcal{T} + \nabla \mathcal{T}.
\end{equation}
Here $\Dslash$ denote certain angular derivative operators on the spheres of intersection of the double null foliation, and $\nabla \mathcal{T}$ schematically denote projected covariant derivatives of $\mathcal{T}$ in either the 3, 4 or angular directions.

The Ricci coefficients can be estimated using \emph{transport estimates} for the null structure equations \eqref{eq:intronullstructure} since derivatives of $\Gamma$ do not appear explicitly on the right hand sides of the equations.  The transport estimates are outlined below in Section \ref{subsec:introRicci} and carried out in detail in Section \ref{section:Ricci}.  Note that using such estimates does, however, come with a loss, namely the expected fact that angular derivatives of $\Gamma$ live at the same level of differentiability as curvature is not recovered.  This fact can be recovered through a well known elliptic procedure, which is outlined below in Section \ref{subsec:introelliptic} and treated in detail in Section \ref{section:Riccitop}.  One cannot do the same for the curvature components and the Bianchi equations \eqref{eq:introBianchi} due to the presence of the $\Dslash \psi$ terms on the right hand sides.  In order to obtain ``good'' estimates for the Bianchi equations one must exploit the special structure which, if $S$ denotes one of the spheres of intersection of the null foliation, takes the following form,
\[
	\int_S \Dslash_{\uppsi} \uppsi' \cdot \uppsi = - \int_S \uppsi ' \cdot \Dslash_{\uppsi '} \uppsi,
\]
i.\@e.\@ the adjoint of the operator $\Dslash_{\uppsi}$ is $-\Dslash_{\uppsi'}$.  Using this structure, if one contracts the $\nablaslash_3 \uppsi$ equation with $\uppsi$ and adds the $\nablaslash_4 \uppsi'$ equation contracted with $\uppsi'$, the terms involving the angular derivatives will cancel upon integration and an integration by parts yields \emph{energy estimates} for the Weyl curvature components.  It is through this procedure that the hyperbolicity of the Einstein equations manifests itself in the double null gauge.  These energy estimates form the content of Section \ref{section:curvature} and, again, are outlined below in Section \ref{subsec:introcurvature}.

We are therefore forced (at least at the highest order) to estimate the curvature components in $L^2$.  All of the estimates for the Ricci coefficients here will also be $L^2$ based.  In order to deal with the nonlinearities in the error terms of the equations, the same $L^2$ estimates are obtained for higher order derivatives of the quantities and Sobolev inequalities are used to obtain pointwise control over lower order terms.\footnote{One could actually be sharper and, with slightly more effort, close the estimates with fewer derivatives than are taken here.}  To do this, a set of differential operators $\mathfrak{D}$ is introduced which satisfy the \emph{commutation principle} of \cite{DHR}.  This says that the ``null condition'' satisfied by the equations (which is outlined below and crucial for the estimates) and the structure discussed above are preserved when the equations are commuted by $\mathfrak{D}$, i.\@e.\@ $\mathfrak{D} \Gamma$ and $\mathfrak{D} \psi$ satisfy similar equations to $\Gamma$ and $\psi$.  The set of operators $\mathfrak{D}$ is introduced in Section \ref{subsec:commutation}.

As they appear on the right hand side of the equations for $\psi,\Gamma$, the energy momentum tensor components $\mathcal{T}$ are also, at the highest order, estimated in $L^2$.  These estimates are obtained by first estimating $f$ using the Vlasov equation.  It is important that the components of the energy momentum tensor, and hence also $f$, are estimated at one degree of differentiability greater than the Weyl curvature components $\psi$.  The main difficulty in this work is in obtaining such estimates for the derivatives of $f$.  See Section \ref{subsec:introf} for an outline of the argument and Section \ref{section:emtensor} for the details.

\subsection{The $p$-index Notation and the Null Condition} \label{subsec:pindex}
The discussion in the previous section outlines how one can hope to close the estimates for $\Gamma$ and $\psi$ from the point of view of \emph{regularity}.  Since global estimates are required, it is also crucial that all of the error terms in the equations decay sufficiently fast in $v$ (or equivalently, since everything takes place in the ``wave zone'' where $r:= v-u+r_0$ is comparable to $v$, sufficiently fast in $r$) so that, when they appear in the estimates, they are globally integrable.  For quasilinear wave equations there is an algebraic condition on the nonlinearity, known as the \emph{null condition}, which guarantees this \cite{Kl}.  By analogy, we say the null structure and Bianchi equations ``satisfy the null condition'' to mean that, on the right hand sides of the equations, certain ``bad'' combinations of the terms do not appear.  There is an excellent discussion of this in the introduction of \cite{DHR}.  As they are highly relevant, the main points are recalled here.

Following \cite{DHR}, the correct hierarchy of asymptotics in $r$ for $\Gamma$, $\psi$ and $\mathcal{T}$ is first guessed.  This guess is encoded in the $p$-index notation.  Each $\Gamma,\psi,\mathcal{T}$ is labelled with a subscript $p$ to reflect the fact that $r^p\vert \Gamma_p \vert, r^p\vert \psi_p \vert, r^p\vert \mathcal{T}_p \vert$ are expected to be uniformly bounded.\footnote{It in fact may be the case that $r^p \vert \Gamma_p \vert$ etc.\@ can converge to 0 as $r\to \infty$.  The weights may therefore be weaker than the actual decay rates in the solutions which are finally constructed.}  Here $\vert \cdot \vert$ denotes the norm with respect to the induced metric on the 2-spheres $\gslash$.  The weighted $L^2$ quantities which will be shown to be uniformly bounded will imply, via Sobolev inequalities, that this will be the case at lower orders.

In Theorem \ref{thm:main3}, the precise formulation of Theorem \ref{thm:main2}, it is asymptotics consistent with the $p$-index notation which will be assumed to hold on the initial outgoing hypersurface $\mathcal{N} = \{ u=u_0\}$.  In the context of Theorem \ref{thm:main}, recall the use of the Klainerman--Nicol\`{o} \cite{KlNi} result in Section \ref{subsec:vacuumappeal}. The result of Klainerman--Nicol\`{o} guarantees that, provided the asymptotically flat Cauchy data on $\Sigma$ has sufficient decay, there indeed exists an outgoing null hypersurface in the development of the data on which asymptotics consistent with the $p$-index notation hold.

\subsection{Geometry of Null Geodesics and the Support of $f$} \label{subsec:introsuppf}
If the Ricci coefficients are assumed to have the asymptotics described in the previous section then it is straightforward to show that $u_f$ can be chosen to have the desired property that $f=0$ on the mass shell over any point $x\in \mathcal{M}$ with $u(x) \geq u_f - 1$.  In fact, it can also be seen that the size of the support of $f$ in $P_x$, the mass shell over the point $x\in \mathcal{M}$, will decay as $v(x) \to \infty$.  This decay is important as it is used to obtain the decay of the components of the energy momentum tensor.  The argument for obtaining the decay properties of $\supp(f)$ is outlined here and presented in detail in Section \ref{section:suppf}.

The decay of the size of the support of $f$ in $P_x$ can be seen by considering the decay of components of certain null geodesics.  Suppose first that $\gamma$ is a future directed null geodesic in Minkowski space emanating from a compact set in the hypersurface $\{t=0\}$ such that the initial tangent vector $\dot{\gamma}(0)$ is contained in a compact set in the mass shell over $\{t=0\}$.  One can show that, if
\[
	\dot{\gamma}(s) = p^4(s) e_4 + p^3(s) e_3 + p^A(s) e_A,
\]
where $e_1 = \partial_{\theta^1},e_2 = \partial_{\theta^2},e_3 = \partial_u,e_4 = \partial_v$ is the standard double null frame in Minkowski space, then the bounds,
\begin{equation} \label{eq:intromomentum}
	 p^4\leq C, \qquad r^2 p^3 \leq C p^4, \qquad r^2 \vert p^A \vert \leq C p^4, \text{ for } A=1,2,
\end{equation}
hold uniformly along $\gamma$ for some constant $C$.\footnote{One sees these are the correct asymptotics for $p^1(s),p^2(s)$ by using the three angular momentum Killing vector fields of Minkowski space $\Omega_1,\Omega_2,\Omega_3$ and the fact that, if $K$ is a Killing vector, $g(\dot{\gamma},K)$ is constant along $\gamma$, to see that the angular momentum of the geodesic,
\[
	r^2 \gslash_{AB} p^Ap^B = \sum_{i=1}^3 \left( g(\dot{\gamma},\Omega_i) \right)^2,
\]
is conserved along $\gamma$.  This fact together with the mass shell relation $4 p^3 p^4 = \gslash_{AB} p^A p^B$ and the fact that $p^4$ does not decay along a null geodesic (which can be seen in Minkowski space by looking at the geodesic equations and noting that $\dot{p}^4(s) \geq 0$) gives the required asymptotics for $p^3$.}

The bounds \eqref{eq:intromomentum} will be assumed to hold in $\supp(f)$ in the mass shell over the initial hypersurface $\{v = v_0\}$ in Theorem \ref{thm:main3}, the precise formulation of Theorem \ref{thm:main2}.  In the setting of Theorem \ref{thm:main}, the bounds \eqref{eq:intromomentum} can be taken to hold on the hypersurface $\underline{\mathcal{N}} = \{ v = v_0\}$ in view of the Cauchy stability argument of Section \ref{subsec:Cauchystability} and the fact that they hold globally in $\supp(f)$ for the uncoupled problem of the Vlasov equation on a fixed Minkowski background.

The idea is now to propagate the bounds \eqref{eq:intromomentum} from the initial hypersurface $\{v=v_0\}$ into the rest of the spacetime.  If $e_1,\ldots,e_4$ now denotes the double null frame of $(\mathcal{M},g)$ (defined in Section \ref{subsec:coords}), one then uses the geodesic equations,
\[
	\dot{p}^{\mu}(s) + p^{\alpha}(s)p^{\beta}(s) \Gamma_{\alpha \beta}^{\mu}(s) = 0,
\]
for a null geodesic $\gamma$ with $\dot{\gamma}(s) = p^{\mu}(s) e_{\mu}\vert_{\gamma(s)}$, a bootstrap argument and the pointwise bounds $r^p\vert \Gamma_p \vert \leq C$ to see that
\[
	 \dot{p}^4(s) = \mathcal{O}\left( \frac{p^4(0)^2}{r(s)^2} \right), \qquad \frac{d}{ds} \left( r(s)^2p^3(s) \right) = \mathcal{O}\left( \frac{p^4(0)^2}{r(s)^2} \right),
\]
\[
	\frac{d}{ds} \left( r(s)^2p^A(s) \right) = \mathcal{O}\left( \frac{p^4(0)^2}{r(s)^2} \right), \text{ for } A=1,2.
\]
The estimates \eqref{eq:intromomentum} follow by integrating along $\gamma$ since $\frac{dr}{ds} \sim p^4(0)$.

Finally, to show the retarded time $u_f$ can be chosen as desired, let $u(s)$ denote the $u$ coordinate of the geodesic $\gamma$ at time $s$.  Then
\[
	\vert \dot{u}(s) \vert \sim p^3(s) \sim \frac{p^4(0)}{r(s)^2},
\]
and hence $\vert u(s) - u(0)\vert \leq C$ for all $s \in [0,\infty)$, for some constant $C$.

\subsection{Global Energy Estimates for the Curvature Components} \label{subsec:introcurvature}
The global energy estimates for the Weyl curvature components can now be outlined.  They are carried out in detail in Section \ref{section:curvature}.  The Bianchi equations take the schematic form,
\begin{equation} \label{eq:introBianchi2}
	\nablaslash_3 \uppsi_p = \Dslash_{\uppsi_p} \uppsi_{p'}' + E_p, \qquad \nablaslash_4 \uppsi_{p'}' = \Dslash_{\uppsi_{p'}'} \uppsi_{p} + c\tr \chi \uppsi_{p'}' + E_{p' + \frac{3}{2}},
\end{equation}
where $c$ is a constant (which is different for the different $\uppsi_{p'}'$) and $E_p$ is an error which will decay, according to the $p$ notation, like $\frac{1}{r^p}$.  Similarly, $E_{p' + \frac{3}{2}}$ is an error which will decay like $\frac{1}{r^{p'+\frac{3}{2}}}$.  Recall from equation \eqref{eq:introBianchi} that the errors $E_p$ and $E_{p' + \frac{3}{2}}$ contain linear terms involving $\Gamma$, nonlinear terms of the form $\Gamma \cdot \psi$ and $\Gamma \cdot \mathcal{T}$, and projected covariant derivatives of components of the energy momentum tensor $\nabla \mathcal{T}$.  Using \eqref{eq:introBianchi2} to compute $Div \left( r^w \vert \uppsi_p \vert^2 e_3 \right)$, $Div \left( r^w \vert \uppsi_{p'}' \vert^2 e_4 \right)$, after summing a cancellation will occur in the terms involving angular derivatives, as discussed in Section \ref{subsec:schematic}, and they can be rewritten as a spherical divergence.  If the weight $w$ is chosen correctly, a cancellation\footnote{This cancellation is exploited for each Bianchi pair except $(\uppsi, \uppsi') = (\alpha,\beta)$, for which a slightly weaker weight is chosen.  See Remark \ref{rmk:peeling}.} also occurs in the $c\tr \chi \uppsi_{p'}'$ term (which, since $\tr \chi$ looks like $\frac{2}{r}$ to leading order, cannot be included in the error $E_{p'+ \frac{3}{2}}$) and one is then left with,
\begin{equation} \label{eq:introenid}
	\int_{\mathcal{B}} Div \left( r^w \vert \uppsi_p \vert^2 e_3 \right) + Div \left( r^w \vert \uppsi_{p'}' \vert^2 e_4 \right)
	=
	\int_{\mathcal{B}} r^w \left( \uppsi_p \cdot E_p + \uppsi_{p'}' \cdot E_{p'+ \frac{3}{2}} \right),
\end{equation}
where $Div$ denotes the spacetime divergence and $\mathcal{B}$ denotes a spacetime ``bulk'' region bounded to the past by the initial characteristic hypersurfaces, and to the future by constant $v$ and constant $u$ hypersurfaces.  See Figure \ref{fig:br}.  Note that this procedure will generate additional error terms but they can be treated similarly to those arising from the errors in \eqref{eq:introBianchi2} and hence are omitted here.  See Section \ref{section:curvature} for the details.

If the curvature fluxes are defined as,
\[
	 F^1_{v_0,v'}(u') = \sum_{\uppsi_p} \int_{\{u=u'\} \cap \{ v_0 \leq v \leq v'\}} r^{w(\uppsi_p, \uppsi_{p'}')} \vert \uppsi_p \vert^2,
\]
\[
	 F^2_{u_0,u'}(v') = \sum_{\uppsi_{p'}'} \int_{\{v=v'\} \cap \{ u_0 \leq u \leq u'\}} r^{w(\uppsi_p, \uppsi_{p'}')} \vert \uppsi_{p'}' \vert^2,
\]
then by the divergence theorem, when the above identity \eqref{eq:introenid} is summed over all Bianchi pairs $(\uppsi_p, \uppsi_{p'}')$, the left hand side becomes
\[
	 F^1_{v_0,v}(u) + F^2_{u_0,u}(v) - F^1_{v_0,v}(u_0) - F^2_{u_0,u}(v_0).
\]
Due to the relation between the weights $w(\uppsi_p, \uppsi_{p'}')$ and $p,p'$, and the bounds assumed for $\Gamma$ and $\mathcal{T}$ through the bootstrap argument, the right hand side of \eqref{eq:introenid} can be controlled by,
\[
	 \int_{u_0}^u F^1_{v_0,v}(u') du' + \frac{C}{v_0},
\]
for some constant $C$ (which, of course, arises from inserting the bootstrap assumptions).  It is this step where one sees the manifestation of the null condition in the Bianchi equations.  Dropping the $F^2_{u_0,u}(v)$ term on the left yields,
\[
	 F^1_{v_0,v}(u) \leq F^1_{v_0,v}(u_0) + F^2_{u_0,u}(v_0) + \int_{u_0}^u F^1_{v_0,v}(u') du' + \frac{C}{v_0},
\]
and hence, by the Gr\"onwall inequality, $F^1_{v_0,v}(u)$ can be controlled by initial data and the term $\frac{C}{v_0}$.  Returning to the inequality,
\[
	 F^1_{v_0,v}(u) + F^2_{u_0,u}(v) \leq F^1_{v_0,v}(u_0) + F^2_{u_0,u}(v_0) + \int_{u_0}^u F^1_{v_0,v}(u') du' + \frac{C}{v_0},
\]
and inserting the above bounds for $F^1_{v_0,v}(u)$, $F^2_{u_0,u}(v)$ can now also be similarly controlled.

\subsection{Global Transport Estimates for the Ricci Coefficients} \label{subsec:introRicci}
Turning now to the global estimates for the Ricci coefficients, which are treated in detail in Section \ref{section:Ricci}, in the $p$-index notation the null structure equations take the form,
\[
	 \nablaslash_3 \Riccit = E_p, \qquad \nablaslash_4 \Riccif = - \frac{p}{2} \tr \chi \Riccif + E_{p+2},
\]
where again $E_p$ is an error which decays, according to the $p$-index notation, like $\frac{1}{r^p}$ and $E_{p+2}$ decays like $\frac{1}{r^{p+2}}$.  Recall from equation \eqref{eq:intronullstructure} that $E_{p}$ and $E_{p+2}$ contain linear terms involving $\Gamma, \psi, \mathcal{T}$, and quadratic terms of the form $\Gamma \cdot \Gamma$.  The $\nablaslash_4 \Riccif$ equations can be rewritten as
\[
	 \nablaslash_4 \left( r^p \Riccif\right) = r^p E_{p+2}.
\]
To estimate the $\overset{(4)}{\Gamma}$ one then uses the identity, for a function $h$ on $\mathcal{M}$,
\[
	 \partial_v \int_{S_{u,v}} h d \mu_S = \int_{S_{u,v}} \nablaslash_4 h + \tr \chi h d \mu_S,
\]
where the $\tr \chi h$ term comes from the derivative of the volume form on $S_{u,v}$, with $h = r^{2p-2} \vert \Riccif \vert^2$.  The $r^{-2}$ factor serves to cancel the $\tr \chi$ term (which, recall, behaves like $\frac{2}{r}$ and so is not globally integrable in $v$).  Hence,
\begin{align*}
	 \partial_v \int_{S_{u,v}} r^{2p-2} \vert \Riccif \vert^2 d\mu_S 
	 &
	 =
	 \int_{S_{u,v}} r^{2p-2}\Riccif \cdot E_{p+2} d\mu_S
	 \\
	 &
	 =
	 \mathcal{O} \left( \frac{1}{r^2} \right),
\end{align*}
since the volume form is of order $r^2$.  Integrating in $v$ from the initial hypersurface $\{ v=v_0\}$ then gives,
\begin{equation} \label{eq:introRiccifbounds}
	 r^{2p-2} \int_{S_{u,v}} \vert \Riccif \vert^2 d\mu_S \leq C \left( r^{2p-2} \int_{S_{u,v}} \vert \Riccif \vert^2 d\mu_S \right) \bigg\vert_{v=v_0} + \frac{C}{v_0}.
\end{equation}
Note that the error $E_{p+2}$ is integrated over a $u =$ constant hypersurface.  These are exactly the regions on which the integrals of the Weyl curvature components were controlled in Section \ref{subsec:introcurvature}, and it is for this reason the curvature terms in the error $E_{p+2}$ can be controlled in \eqref{eq:introRiccifbounds}.

Since the volume form is of order $r^2$, the bound \eqref{eq:introRiccifbounds} is consistent with $\Riccif$ decaying like $\frac{1}{r^p}$ and, after repeating the above with appropriate derivatives of $\Riccif$, this pointwise decay can be obtained using Sobolev inequalities on the spheres.

It is not a coincidence that the $\frac{p}{2}$ coefficient of $\tr \chi \Riccif$ in the $\nablaslash_4 \Riccif$ equation is exactly that which is required to obtain $\frac{1}{r^p}$ decay for $\Riccif$.  In fact some of the $\Riccif$ will decay faster than this but the other null structure equations are required, along with elliptic estimates, to obtain this.  It is therefore the $\frac{p}{2}$ coefficient which determines the $p$ index given here to the $\overset{(4)}{\Gamma}$ as it restricts the decay which can be shown to hold using only the $\nablaslash_4 \overset{(4)}{\Gamma}$ equations.  Note the difference with \cite{DHR} where the authors are not constrained by this coefficient as they there integrate ``backwards'' from future null infinity.

Turning now to the equations in the 3 direction, the $\Riccit$ quantities are estimated using the identity,
\[
	 \partial_u \int_{S_{u,v}} h d\mu_S = \int_{S_{u,v}} \nablaslash_3 h + \tr \chibar h d\mu_S,
\]
with $h = r^{2p-2} \vert \Riccit \vert^2$.  It does not now matter that $\tr \chibar$ only decays like $\frac{1}{r}$ since the integration in $u$ will only be up to the finite value $u_f$.

Suppose first that $\Riccit$ satisfies
\begin{equation} \label{eq:introRiccit}
	 \nablaslash_3 \Riccit = E_{p+1} + E_p^0,
\end{equation}
where $E_{p+1}$ decays like $\frac{1}{r^{p+1}}$ and $E_p^0$ decays like $\frac{1}{r^p}$ but only contains Weyl curvature, energy momentum tensor and $\overset{(4)}{\Gamma}$ terms which have already been estimated (the energy momentum tensor estimates are outlined below as they present the greatest difficulty but in the logic of the proof are estimated first).
Then,
\[
	 \left\vert \partial_u \int_{S_{u,v}} r^{2p-2} \vert \Riccit \vert^2 d \mu_S \right\vert
	 \leq 
	 \int_{S_{u,v}} r^{2p-2} \left( \vert \Riccit \vert^2 + \vert E_{p+1}\vert^2 + \vert E_{p}^0\vert^2 \right) d \mu_S.
\]
Integrating from $u_0$ to $u$ and inserting the bootstrap assumptions and the previously obtained bounds for $E_p^0$, the Gr\"onwall inequality then gives,
\begin{align*}
	 \int_{S_{u,v}} r^{2p-2} \vert \Riccit \vert^2 d \mu_S 
	 &
	 \leq
	 \left( \int_{S_{u,v}} r^{2p-2} \vert \Riccit \vert^2 d \mu_S\right) \bigg\vert_{u=u_0} + \frac{C}{v_0} + C \int_{u_0}^{u_f} r^{2p} \vert E_{p}^0\vert^2 d u'
	 \\
	 &
	 \leq
	 C\left( \bc + \frac{1}{v_0} \right),
\end{align*}
where $\bc$ controls the size of the initial data.  Note that it was important that the only error terms which have not already been estimated are of the form $E_{p+1}$, and not $E_p$, in order to gain the $\frac{1}{v_0}$ smallness factor.  It turns out that there is a \emph{reductive} structure in the null structure equations so that, provided they are estimated in the correct order, each $\overset{(3)}{\Gamma}$ satisfies an equation of the form \eqref{eq:introRiccit} where $E_p^0$ now also contains $\overset{(3)}{\Gamma}$ terms which have been estimated previously.  Hence all of the $\overset{(3)}{\Gamma}$ can be estimated with smallness factors.

\subsection{Elliptic Estimates and Ricci Coefficients at the Top Order} \label{subsec:introelliptic}
The procedure in Section \ref{subsec:introRicci} is used to estimate the Ricci coefficients, along with their derivatives at all but the top order, in $L^2$ of the spheres of intersection of constant $u$ and constant $v$ hypersurfaces.  The derivatives of Ricci coefficients at the top order are estimated only in $L^2$ on null hypersurfaces.  These estimates are obtained using elliptic equations on the spheres for some of the Ricci coefficients, coupled to transport equations for certain auxilliary quantities.  This procedure is familiar from many other works (e.\@g.\@ \cite{ChKl}, \cite{Ch}) and forms the content of Section \ref{section:Riccitop}.  It should be noted that these estimates are only required here for estimating the components of the energy momentum tensor.  If one were to restrict the semi-global problem of Theorem \ref{thm:main2} to the case of the vacuum Einstein equations \eqref{eq:Einsteinvacuum} then the estimates for the Ricci coefficients and curvature components could be closed with a loss (i.\@e.\@ without knowing that angular derivatives of Ricci coefficients lie at the same degree of differentiability as the Weyl curvature components) as only the null structure equations of the form \eqref{eq:intronullstructure} would be used, and these elliptic estimates would not be required.  See Section \ref{subsec:schematic}.

\subsection{Global Estimates for the Energy Momentum Tensor Components}
At the zeroth order the estimates for the energy momentum tensor components follow directly from the bounds \eqref{eq:intromomentum}, which show that the size of the region $\supp(f\vert_{P_x}) \subset P_x$ on which the integral in \eqref{eq:EV} is taken is decaying as $r(x) \to \infty$, and the fact that $f$ is conserved along trajectories of the geodesic flow.  For example, using the volume form for $P_x$ defined in Section \ref{subsec:massshell}, if $\sup_{\{v=v_0\}} \vert f \vert \leq \bc$,
\[
	T_{33}(x) 
	= 
	4 T^{44}(x) 
	\leq 
	4\bc \int_0^C \int_{\vert p^1 \vert, \vert p^2 \vert \leq \frac{C}{r^2}} p^4 \sqrt{\det \gslash(x)} dp^1 dp^2 dp^4
	\leq
	\frac{C\bc}{r^2},
\]
since $\sqrt{\det \gslash} \leq C r^2$.  In fact, provided the derivatives of $f$ can be estimated, the estimates for the derivatives of $\mathcal{T}$ are obtained in exactly the same way.

\subsection{Global Estimates for Derivatives of $f$} \label{subsec:introf}
A fundamental new aspect of this work arises in obtaining estimates for the derivatives of $f$.  Recall from Section \ref{subsec:schematic} that, in order to close the bootstrap argument, it is crucial that the energy momentum tensor components $\mathcal{T}$, and hence $f$, are estimated at one degree of differentiability greater than the Weyl curvature components, i.\@e.\@ $k$ derivatives of $f$ must be estimated using only $k-1$ derivatives of $\psi$.  Written in components with respect to the frame\footnote{Recall $e_1,e_2,e_3,e_4$ is the double null frame for $\mathcal{M}$, defined in Section \ref{subsec:coords} using the $(u,v,\theta^1,\theta^2)$ coordinate system on $\mathcal{M}$.  Note the slight abuse of notation here as $e_1,e_2,e_3,e_4$ act only on functions on $\mathcal{M}$, whilst $f$ is a function on $P$.  As $e_1 = \partial_{\theta^1}$ in the $(u,v,\theta^1,\theta^2)$ coordinate system, $e_1 f$ is used to denote $\partial_{\theta^1}f$ in the $(u,v,\theta^1,\theta^2, p^1,p^2,p^4)$ coordinate system for $P$.  Similarly for $e_2f,e_3f,e_4f$.} $e_1,e_2,e_3,e_4, \partial_{p^1},\partial_{p^2},\partial_{p^4}$ for $P$, the Vlasov equation \eqref{eq:vlas} takes the form,
\[
	X(f) = p^{\mu} e_{\mu} (f) - p^{\nu} p^{\lambda} \Gamma_{\nu \lambda}^{\mu} \partial_{p^{\mu}} f = 0,
\]
where $\Gamma_{\nu \lambda}^{\mu}$ denote the Ricci coefficients of $\mathcal{M}$.  See \eqref{eq:Riccitab1}--\eqref{eq:Riccitab5} below.  One way to estimate derivatives of $f$ is to commute this equation with suitable vector fields and integrate along trajectories of the geodesic flow.  If $V$ denotes such a vector field, commuting will give,
\[
	X(Vf) = E,
\]
where $E$ is an error involving terms of the form $V(\Gamma_{\nu \lambda}^{\mu})$.  At first glance this seems promising as derivatives of the Ricci coefficients should live at the same level of differentiability as the Weyl curvature components $\psi$.  This is not the case for all of the $\Gamma_{\nu \lambda}^{\mu}$ however, for example if $V$ involves an angular derivative then $V(\Gamma_{4A}^B)$, for $A,B = 1,2$, will contain two angular derivatives of the vector field $b$.  See \eqref{eq:metric} below for the definition of $b$ and \eqref{eq:Riccitab3} for $\Gamma_{4A}^B$.  The vector field $b$ is estimated through an equation of the form,
\[
	\nablaslash_3 b = \Gamma + \Gamma \cdot b,
\]
and hence, commuting twice with angular derivatives and using the elliptic estimates described in Section \ref{subsec:introelliptic} will only give estimates for two angular derivatives of $b$ by first order derivatives of $\psi$ and $\mathcal{T}$.  The angular derivatives of the spherical Christoffel symbols $\Gammaslash$, see \eqref{eq:defGammaslash} below, which also appear when commuting the Vlasov equation give rise to similar issues.

Whilst it may still be the case that $E$ as a whole (rather than each of its individual terms) can be estimated just at the level of $\psi$, a different approach is taken here in order to see more directly that derivatives of $f$ can be estimated at the level of $\psi$.  This approach, which is treated in detail in Section \ref{section:emtensor}, is outlined now.

Consider again a vector $V \in T_{(x,p)} P$.  Recall the form of the Vlasov equation \eqref{eq:vlasexp}.  Using this expression for $f$ and the chain rule,
\[
	Vf (x,p) = df\vert_{(x,p)} V = df\vert_{\exp_s(x,p)} \cdot d \exp_s \vert_{(x,p)} V,
\]
for any $s$, and hence, if $J(s) := d \exp_s \vert_{(x,p)} V$,
\begin{equation} \label{eq:introVJ}
	Vf (x,p) = J(s) f (\exp_s(x,p)).
\end{equation}
If $s<0$ is taken so that $\pi (\exp_s(x,p)) \in \{ v = v_0\}$ then the expression \eqref{eq:introVJ} relates a derivative of $f$ at $(x,p)$ to a derivative of $f$ on the initial hypersurface.  It therefore remains to estimate the components of $J(s)$, with respect to a suitable frame for $P$, uniformly in $s$ and independently of the point $(x,p)$.

The metric $g$ on the spacetime $\mathcal{M}$ can be used to define a metric on the tangent bundle $\TM$, known as the Sasaki metric \cite{Sa}, which by restriction defines a metric $\hat{g}$ on the mass shell $P$.  See Section \ref{section:Sasaki} where this metric is introduced.  With respect to this metric trajectories of the geodesic flow $s\mapsto \exp_s(x,p)$ are geodesics in $P$ and, for any vector $V\in T_{(x,p)}P$, $J(s) := d \exp_s \vert_{(x,p)} V$ is a Jacobi field along this geodesic (see Section \ref{section:Sasaki}).  Therefore $J(s)$ satisfies the Jacobi equation,
\begin{equation} \label{eq:introjacobi}
	\hat{\nabla}_X \hat{\nabla}_X J = \hat{R}(X,J)X,
\end{equation}
where $\hat{\nabla}$ denotes the induced connection on $P$, and $\hat{R}$ denotes the curvature tensor of $(P,\hat{g})$.  Equation \eqref{eq:introjacobi} is used, as a transport equation along the trajectories of the geodesic flow, to estimate the components of $J$.  The curvature tensor $\hat{R}$ can be expressed in terms of (vertical and horizontal lifts of) the curvature tensor $R$ of $(\mathcal{M},g)$ along with its first order covariant derivatives $\nabla R$.  See equation \eqref{eq:curvaturenotuseful}.  At first glance the presence of $\nabla R$ again appears to be bad.  On closer inspection, however, the terms involving covariant derivatives of $R$ are always derivatives in the ``correct'' direction so that they can be recovered by the transport estimates, and the components of $J$, and hence $Vf$, can be estimated at the level of $\psi$.

The above observations of course only explain how one can hope to close the estimates for $\mathcal{T}$ from the point of view of regularity.  In order to obtain global estimates for the components of $J$ one has to use the crucial fact that, according to the $p$-index notation, the right hand side of the Jacobi equation $\hat{R}(X,J)X$, when written in terms of $\psi, \mathcal{T}, p^1, p^2, p^3,p^4$, decays sufficiently fast as to be twice globally integrable along $s \mapsto \exp_s(x,p)$.  This can be viewed as a null condition for the Jacobi equation and is brought to light through further schematic notation introduced in Section \ref{subsec:furtherschematic}.

The fact that the right hand side of \eqref{eq:introjacobi} has sufficient decay in $r$ is perhaps not surprising.  Consider for example the term
\begin{equation} \label{eq:horexample}
	\hor_{(\gamma,\dot{\gamma})} \left( R(\dot{\gamma}, J^h) \dot{\gamma} \right),
\end{equation}
in $\hat{R}(X,J)X$.  Here $\gamma$ is a geodesic in $\mathcal{M}$ such that $\exp_s(x,p) = (\gamma(s), \dot{\gamma}(s))$ and $J^h$ is a vector field along $\gamma$ on $\mathcal{M}$ such that, together with another vector field $J^v$ along $\gamma$,
\[
	J_{(\gamma,\dot{\gamma})} = \hor_{(\gamma,\dot{\gamma})} (J^h) + \ver_{(\gamma,\dot{\gamma})} (J^v),
\]
with $\hor_{(\gamma,\dot{\gamma})}$ and $\ver_{(\gamma,\dot{\gamma})}$ denoting horizontal and vertical lifts at $(\gamma,\dot{\gamma})$ (defined in Section \ref{section:Sasaki}).  The slowest decaying $\psi$ and $\mathcal{T}$ are those which contain the most $e_3$ vectors.  Whenever such $\psi$ and $\mathcal{T}$ arise in \eqref{eq:horexample} however, they will typically be accompanied by $p^3(s)$, the $e_3$ component of $\dot{\gamma}(s)$, which (recall from Section \ref{subsec:introsuppf}) has fast $\frac{1}{r(s)^2}$ decay.  Similarly the non-decaying $p^4(s)$, the $e_4$ component of $\dot{\gamma}(s)$, can only appear in \eqref{eq:horexample} accompanied by the $\psi$ and $\mathcal{T}$ which contain $e_4$ vectors and hence have fast decay in $r$.  In particular, potentially slowly decaying terms involving $p^4(s)$ multiplying the $\psi$ and $\mathcal{T}$ which contain no $e_4$ vectors do not arise in \eqref{eq:horexample}.

Finally, since $Jf$ now itself is also conserved along $s\mapsto \exp_s(x,p)$, second order derivatives of $f$ can be obtained by repeating the above.  If $J_1,J_2$ denote Jacobi fields corresponding to vectors $V_1,V_2$ at $(x,p)$ respectively, then,
\[
	V_2 V_1 f (x,p) = J_2(s) J_1(s) f (\exp_s(x,p)).
\]
In order to control $V_2 V_1 f (x,p)$ it is therefore necessary to estimate the $J_2$ derivatives of the components of $J_1$ along $s\mapsto \exp_s(x,p)$.  This is done by commuting the Jacobi equation \eqref{eq:introjacobi} and showing that the important structure described above is preserved.  The Jacobi fields which are used, and hence the vectors $V$ used to take derivatives of $f$, have to be carefully chosen.  They are defined in Section \ref{subsec:overview}.

Note that this procedure can be repeated to obtain higher order derivatives of $f$.  Whilst the pointwise bounds on $\psi$ at lower orders mean that lower order derivatives of $f$ can be estimated pointwise, at higher orders this procedure will generate terms involving higher order derivatives of $\psi$ and hence higher order derivatives of $\mathcal{T}$ must be estimated in $L^2$ on null hypersurfaces.  In fact, at the very top order, $\mathcal{T}$ is estimated in the spacetime $L^2$ norm.

\subsection{Outline of the Paper}
In the next section coordinates are defined on the spacetime to be constructed, and on the mass shell $P$.  The Ricci coefficients and curvature components are introduced along with their governing equations.  In Section \ref{section:equations} the schematic form of the quantities and equations are given.  Three derivative operators are then introduced which are shown to preserve the schematic form of the equations under commutation.  Some facts about the Sasaki metric are recalled in Section \ref{section:Sasaki} and are used to describe certain geometric properties of the mass shell.  A precise statement of Theorem \ref{thm:main2} is given in Section \ref{section:ba}, along with the statement of a bootstrap theorem.  The proof of the bootstrap theorem is given in the following sections.  The main estimates are obtained for the energy momentum tensor components, Weyl curvature components and lower order derivatives of Ricci coefficients in Sections \ref{section:emtensor}, \ref{section:curvature} and \ref{section:Ricci} respectively.  The estimates for the Ricci coefficients at the top order are obtained in Section \ref{section:Riccitop}.  The results of these sections rely on the Sobolev inequalities of Section \ref{section:Sobolev}, and the decay estimates for the size of $\supp(f\vert_{P_x}) \subset P_x$ as $x$ approaches null infinity from Section \ref{section:suppf}.  The fact that the retarded time $u_f$ can be chosen to have the desired property, stated in Theorem \ref{thm:main2}, is also established in Section \ref{section:suppf}.  Finally, the completion of the proof of Theorem \ref{thm:main2}, through a \emph{last slice argument}, is outlined in Section \ref{section:lastslice}.

\subsection{Acknowledgements}
I thank Mihalis Dafermos for introducing me to this problem and for his support and encouragement.  I am also grateful to Gustav Holzegel, Jonathan Luk and Igor Rodnianski for many helpful discussions, and to Cl\'{e}ment Mouhot, Jacques Smulevici and J\'{e}r\'{e}mie Szeftel for useful comments on an earlier version of the manuscript.  This work was supported by the UK Engineering and Physical Sciences Research Council (EPSRC) grant EP/H023348/1 for the University of Cambridge Centre for Doctoral Training, the Cambridge Centre for Analysis.

\section{Basic Setup} \label{section:setup}
Throughout this section consider a smooth spacetime $(\mathcal{M},g)$ where $\mathcal{M} = [u_0,u'] \times [v_0,v') \times S^2$, for some $u_0 < u' \leq u_f$, $v_0 < v' \leq \infty$, is a manifold with corners and $g$ is a smooth Lorentzian metric on $\mathcal{M}$ such that $(\mathcal{M},g)$, together with a continuous function $f:P\to [0,\infty)$, smooth on $P\smallsetminus Z$, where $Z$ denotes the zero section, satisfy the Einstein--Vlasov system \eqref{eq:EV}--\eqref{eq:vlas}.

\subsection{Coordinates and Frames} \label{subsec:coords}
A point in $\mathcal{M}$ will be denoted $(u,v,\theta^1,\theta^2)$.  It is implicitly understood that two coordinate charts are required on $S^2$.  The charts will be defined below using two coordinate charts on $S_{u_0,v_0} = \{u = u_0\} \cap \{ v = v_0\}$.  Assume $u$ and $v$ satisfy the Eikonal equation
\[
  g^{\mu \nu} \partial_{\mu} u \partial_{\nu} u = 0, 
  \qquad
  g^{\mu \nu} \partial_{\mu} v \partial_{\nu} v = 0.
\]
Define null vector fields
\[
  L^{\mu} := - 2 g^{\mu \nu}\partial_{\nu} u, \qquad \underline{L}^{\mu} := - 2 g^{\mu \nu}\partial_{\nu} v,
\]
and the function $\Omega$ by
\[
  2\Omega^{-2} = -g(L,\underline{L}).
\]
Let $(\theta^1,\theta^2)$ be a coordinate system in some open set $U_1$ on the initial sphere $S_{u_0,v_0}$.  These functions can be extended to define a coordinate system $(u,v,\theta^1,\theta^2)$ on an open subset of the spacetime as follows.  Define $\theta^1,\theta^2$ on $\{ u = u_0\}$ by solving
\[
  L(\theta^A) = 0, \qquad \text{ for } A = 1,2.
\]
Then extend to $u>u_0$ by solving
\[
  \underline{L}(\theta^A) = 0, \qquad \text{ for } A = 1,2.
\]
This defines coordinates $(u,v,\theta^1,\theta^2)$ on the region $D(U_1)$ defined to be the image of $U_1$ under the diffeomorphisms generated by $L$ on $\{u=u_0\}$, then by the diffeomorphisms generated by $\underline{L}$.  Coordinates can be defined on another open subset of the spacetime by considering coordinates in another region $U_2 \subset S_{u_0,v_0}$ and repeating the procedure.  These two coordinate charts will cover the entire region of the spacetime in question provided the charts $U_1,U_2$ cover $S_{u_0,v_0}$.  The choice of coordinates on $U_1,U_2$ is otherwise arbitrary.

The spheres of constant $u$ and $v$ will be denoted $S_{u,v}$ and the restriction of $g$ to these spheres will be denoted $\gslash$.  A vector field $V$ on $\mathcal{M}$ will be called an $S_{u,v}$ vector field if $V_x \in T_{x} S_{u(x),v(x)}$ for all $x \in \mathcal{M}$.  Similarly for $(r,0)$ tensors.  A one form $\xi$ is called an $S_{u,v}$ one form if $\xi(L) = \xi(\underline{L}) = 0$.  Similarly for $(0,s)$, and for general $(r,s)$ tensors.

In these coordinates the metric takes the form
\begin{equation} \label{eq:metric}
	g = -2 \Omega^2(du \otimes dv + dv \otimes du) + \gslash_{AB}( d \theta^A - b^A dv) \otimes (d\theta^B - b^B dv),
\end{equation}
where $b$ is some vector field tangent to the spheres $S_{u,v}$ which is zero on the initial hypersurfaces $\{u=u_0\},\{v=v_0\}$.  Note that, due to the remaining gauge freedom, $\Omega$ can be specified on $\{u=u_0\}$ and $\{v=v_0\}$.  Since, in Theorem \ref{thm:main3}, it is assumed that $\left\vert \frac{1}{\Omega^2} - 1 \right\vert$ and $r^3\vert \nablaslash_4 \log \Omega \vert$ are small on $\{u=u_0\}$, it is convenient to set $\Omega =1$ on $\{u=u_0\}$ so that they both vanish.

Integration of a function $\phi$ on $S_{u,v}$ is defined as
\[
	\int_{S_{u,v}} \phi d\mu_{S_{u,v}} = \sum_{i=1,2} \int_{\theta^1} \int_{\theta^2} \phi \tau_i \sqrt{\det \gslash} d\theta^2 d \theta^1,
\]
where $\tau_1,\tau_2$ is a partition of unity subordinate to $D_{U_1},D_{U_2}$ at $u,v$.

Define the double null frame
\begin{equation} \label{eq:nullframe}
  e_A = \partial_{\theta^A}, \text{ for } A = 1,2, \qquad e_3 = \frac{1}{\Omega^2} \partial_u, \qquad e_4 = \partial_v + b^A \partial_{\theta^A},
\end{equation}
and let $(p^{\mu}; \mu = 1,2,3,4)$, denote coordinates on each tangent space to $\mathcal{M}$ conjugate to this frame, so that the coordinates $(x^{\mu},p^{\mu})$ denote the point
\[
	p^{\mu}e_{\mu}\vert_x \in T_x\mathcal{M},
\]
where $x = (x^{\mu})$.  This then gives a frame, $\{e_{\mu}, \partial_{p^{\mu}} \mid \mu = 1,2,3,4 \}$, on $T\mathcal{M}$.  The Vlasov equation \eqref{eq:vlas} written with respect to this frame takes the form
\[
	p^{\mu} e_{\mu}(f) - \Gamma^{\mu}_{\nu \lambda} p^{\nu} p^{\lambda} \partial_{p^{\mu}} f = 0,
\]
where $\Gamma^{\mu}_{\nu \lambda}$ are the Ricci coefficients of $g$ with respect to the null frame \eqref{eq:nullframe}.  For $f$ as a function on the mass shell $P$, this reduces to,
\[
	p^{\mu} e_{\mu}(f) - \Gamma^{\hat{\mu}}_{\nu \lambda} p^{\nu} p^{\lambda} \partial_{\overline{p}^{\hat{\mu}}} f = 0,
\]
where $\hat{\mu}$ now runs over $1,2,4$, and $\overline{p}^1,\overline{p}^2,\overline{p}^4$ denote the restriction of the coordinates $p^1,p^2,p^4$ to $P$, and $\partial_{\overline{p}^{\hat{\mu}}}$ denote the partial derivatives with respect to this restricted coordinate system.  Using the mass shell relation \eqref{eq:massshell} below one can easily check,
\begin{equation} \label{eq:pbar}
	\partial_{\overline{p}^{1}} = \partial_{p^1} + \frac{\gslash_{1A}p^A}{2p^4} \partial_{p^3}, \qquad
	\partial_{\overline{p}^{2}} = \partial_{p^2} + \frac{\gslash_{2A}p^A}{2p^4} \partial_{p^3}, \qquad
	\partial_{\overline{p}^{4}} = \partial_{p^1} - \frac{p^3}{p^4} \partial_{p^3}.
\end{equation}

Note that Greek indices, $\mu,\nu,\lambda$, etc.\@ will always be used to sum over the values $1,2,3,4$, whilst capital Latin indices, $A,B,C$, etc.\@ will be used to denote sums over only the spherical directions $1,2$.  In Section \ref{section:emtensor} lower case latin indices $i,j,k$, etc.\@ will be used to denote summations over the values $1,\ldots,7$.

\begin{remark}
  A seemingly more natural null frame to use on $\mathcal{M}$ would be
  \begin{equation} \label{eq:natframe}
	 e_A = \partial_{\theta^A}, \text{ for } A = 1,2, \qquad e_3 = \frac{1}{\Omega} \partial_u, \qquad e_4 = \frac{1}{\Omega} \left( \partial_v + b^A \partial_{\theta^A} \right).
  \end{equation}
  Dafermos--Holzegel--Rodnianski \cite{DHR} use the same ``unnatural'' frame for regularity issues on the event horizon.  The reason for the choice here is slightly different and is related to the fact that $\omegabar$, defined below, is zero in this frame.
\end{remark}

\subsection{Null Geodesics and the Mass Shell} \label{subsec:massshell}
Recall that the mass shell $P \subset T\mathcal{M}$ is defined to be the set of future pointing null vectors.  Using the definition of the coordinates $p^{\mu}$ and the form of the metric given in the previous section one sees that, since all of the particles have zero mass, i.\@e.\@ since $f$ is supported on $P$, the relation
\begin{equation} \label{eq:massshell}
	-4p^3p^4 + \gslash_{AB}p^Ap^B = 0,
\end{equation}
is true in the support of $f$.  The identity \eqref{eq:massshell} is known as the \emph{mass shell relation}.

The mass shell $P$ is a 7 dimensional hypersurface in $T\mathcal{M}$ and can be parameterised by coordinates $(u,v,\theta^1,\theta^2,p^1,p^2,p^4)$, with $p^3$ defined by \eqref{eq:massshell}.

To make sense of the integral in the definition of the energy momentum tensor \eqref{eq:EV} one needs to define a suitable volume form on the mass shell, $P_x$, over each point $x\in \mathcal{M} \cap \{u\leq u_f\}$.  Since $P_x$ is a null hypersurface it is not immediately clear how to do this.  Given such an $x$, the metric on $\mathcal{M}$ defines a metric on $T_x \mathcal{M}$,
\[
	- 4dp^3 dp^4 + \gslash_{AB}(x) dp^Adp^B,
\]
which in turn defines a volume form on $T_x\mathcal{M}$,
\[
	2 \sqrt{\det \gslash} dp^3 \wedge dp^4 \wedge dp^1 \wedge dp^2.
\]
A canonical one-form normal to $P_x$ can be defined as the differential of the function $\Lambda_X : T_x \mathcal{M} \to \mathbb{R}$ which measures the \emph{length} of $X\in T_x \mathcal{M}$,
\[
	\Lambda_x(X) := g(X,X).
\]
Taking the normal $-\frac{1}{2} d \Lambda_x$ to $P_x$, the volume form (in the $(u,v,\theta^1,\theta^2,p^1,p^2,p^4)$ coordinate system) can be defined as
\[
	\frac{\sqrt{\det \gslash}}{p^4} dp^4 \wedge dp^1 \wedge dp^2.
\]
This is the unique volume form on $P_x$ \emph{compatible} with the normal $-\frac{1}{2} d \Lambda_x$ in the sense that
\[
	-\frac{1}{2} d \Lambda_x \wedge \left( \frac{\sqrt{\det \gslash}}{p^4} dp^4 \wedge dp^1 \wedge dp^2 \right) = 2 \sqrt{\det \gslash} dp^3 \wedge dp^4 \wedge dp^1 \wedge dp^2,
\]
and if $\xi$ is another 3-form on $P_x$ such that 
\[
	-\frac{1}{2} d \Lambda_x \wedge \xi = 2 \sqrt{\det \gslash} dp^3 \wedge dp^4 \wedge dp^1 \wedge dp^2,
\]
then
\[
	\xi = \frac{\sqrt{\det \gslash}}{p^4} dp^4 \wedge dp^1 \wedge dp^2 \qquad \text{on } P_x.
\]
See Section 5.6 of \cite{SaWu}.

The energy momentum tensor at $x\in \mathcal{M}$ therefore takes the form
\[
	T_{\mu \nu}(x) = \int_{-\infty}^{\infty} \int_{-\infty}^{\infty} \int_{0}^{\infty} f p_{\mu}p_{\nu} \frac{1}{p^4} \sqrt{\det \gslash} \mathrm{d} p^4 \mathrm{d} p^1 \mathrm{d} p^2.
\]

\subsection{Ricci Coefficients and Curvature Components} \label{subsec:alphabeta}
Define the Ricci coefficients
\begin{equation*}
  \begin{array}{ll}
	 \chi_{AB} = g(\nabla_{e_A} e_4, e_B), \qquad \qquad & \chibar_{AB} = g(\nabla_{e_A} e_3, e_B), \\
	 \eta_A = -\frac{1}{2} g(\nabla_{e_3} e_A, e_4), & \etabar_A = -\frac{1}{2} g(\nabla_{e_4} e_A, e_3), \\
	 \omega = \frac{1}{2} g(\nabla_{e_4} e_3, e_4).
  \end{array}
\end{equation*}
The null second fundamental forms $\chi, \chibar$ are decomposed into their trace and trace free parts
\begin{align*}
	\tr \chi = \gslash^{AB} \chi_{AB}, & \qquad \hat{\chi}_{AB} = \chi_{AB} - \frac{1}{2} \tr \chi \gslash_{AB}, \\
	\tr \chibar = \gslash^{AB} \chibar_{AB}, & \qquad \hat{\chibar}_{AB} = \chibar_{AB} - \frac{1}{2} \tr \chibar \gslash_{AB}.
\end{align*}
Note that due to the choice of frame, since $e_3$ is an affine geodesic vector field, $\omegabar := \frac{1}{2} g(\nabla_{e_3} e_4, e_3) = 0$.  Also note that in this frame $\zeta_A := \frac{1}{2} g(\nabla_{e_A} e_4, e_3) = - \etabar_A$.  The Christoffel symbols of $(S_{u,v},\gslash)$ with respect to the frame $e_1,e_2$ are denoted $\Gammaslash_{AB}^C$,
\begin{equation} \label{eq:defGammaslash}
	\nablaslash_{e_A} e_B = \Gammaslash_{AB}^C e_C.
\end{equation}

Define also the null Weyl curvature components
\begin{equation*}
  \begin{array}{ll}
	 \alpha_{AB} = W(e_A,e_4,e_B,e_4), \qquad \qquad & \alphabar_{AB} = W(e_A,e_4,e_B,e_4), \\
	 \beta_A = \frac{1}{2} W(e_A,e_4,e_3,e_4), & \betabar_A = \frac{1}{2} W(e_A,e_3,e_3,e_4), \\
	 \rho = \frac{1}{4} W(e_4,e_3,e_4,e_3), & \sigma = \frac{1}{4} {}^* W(e_4,e_3,e_4,e_3).
  \end{array}
\end{equation*}
Here\footnote{Recall that the scalar curvature $R$ vanishes for solutions of the massless Einstein--Vlasov system.}
\begin{equation} \label{eq:Weyl}
	W_{\alpha \beta \gamma \delta} = R_{\alpha \beta \gamma \delta} - \frac{1}{2}(g_{\alpha \gamma} Ric_{\beta \delta} + g_{\beta \delta} Ric_{\alpha \gamma} - g_{\beta \gamma} Ric_{\alpha \delta} - g_{\alpha \delta} Ric_{\beta \gamma}),
\end{equation}
is the Weyl, or conformal, curvature tensor of $(\mathcal{M},g)$ and ${}^* W$ denotes the hodge dual of $W$ (see \cite{ChKl}, \cite{Ch}).

Define the $S_{u,v}$ (0,2)-tensor $\slashed{T}$ to be the restriction of the energy momentum tensor defined in equation \eqref{eq:EV} to vector fields tangent to the spheres $S_{u,v}$:
\[
	\slashed{T}(Y,Z) := Y^{A}Z^B \int_{P_x} f p_{A}p_{B}, \qquad \text{for} \qquad Y = Y^{A}e_A, Z = Z^Ae_A.
\]
Similarly let $\slashed{T}_3, \slashed{T}_4$ denote the $S_{u,v}$ 1-forms defined by restricting the 1-forms $T(e_3, \cdot), T(e_4, \cdot)$ to vector fields tangent to the spheres $S_{u,v}$:
\[
	\slashed{T}_3 (Y) = Y^{A} \int_{P_x} f p_{3}p_{A}, \qquad \slashed{T}_4 (Y) = Y^{A} \int_{P_x} f p_{4}p_{A}.
\]
Finally, let $\slashed{T}_{33}, \slashed{T}_{44}, \slashed{T}_{34}$ denote the functions
\[
	\slashed{T}_{33} = \int_{P_x} f p_{3}p_{3}, \quad \slashed{T}_{44} = \int_{P_x} f p_{4}p_{4}, \quad \slashed{T}_{34} = \int_{P_x} f p_{3}p_{4}.
\]

\subsection{The Minkowski Values}
For the purpose of renormalising the null structure and Bianchi equations, define the following \emph{Minkowski values} of the metric quantities using the function $r:=v-u + r_0$, with $r_0 >0$ a constant chosen to make sure $r\geq \inf_{u} \area(S_{u,v_0})$,
\[
  \Omega_{\circ}^2 = 1, \qquad \gslash^{\circ} = r^2 \gamma, \qquad b_{\circ} = 0,
\]
where $\gamma$ is the round metric on the unit sphere.  Similarly, define
\[
 tr\chi_{\circ} = \frac{2}{r},\qquad tr\chibar_{\circ} = -\frac{2}{r},
\]
and let ${\Gammaslash^{\circ}}_{AB}^C$ denote the spherical Christoffel symbols of the metric $\gslash^{\circ} = r^2 \gamma$, so that,
\[
	\nablaslash^{\circ}_{e_A} e_B = {\Gammaslash^{\circ}}_{AB}^C e_C,
\]
where $\nablaslash^{\circ}$ is the Levi-Civita connection of $\gslash^{\circ}$.  These are the only non-identically vanishing Ricci coefficients in Minkowski space.  All curvature components vanish, as do all components of the energy momentum tensor.

Note that the function $r$ in general does not have the geometric interpretation as the area radius of the spheres $S_{u,v}$.  Note also that
\[
	\frac{1}{C} v \leq r \leq Cv,
\]
in the region $u_0 \leq u \leq u_f, v_0 \leq v < \infty$, for some constant $C > 0$.

\subsection{The Renormalised Null Structure and Bianchi Equations} \label{subsec:nsbianchi}
The Bianchi equations,
\[
	\nabla^{\mu}W_{\mu \nu \lambda \rho} = \frac{1}{2} \left( \nabla_{\lambda} T_{\nu \rho} - \nabla_{\delta} T_{\beta \gamma} \right),
\]
written out in full using the table of Ricci coefficients,
\begin{equation} \label{eq:Riccitab1}
	\nabla_{e_A} e_B = \Gammaslash^C_{AB} e_C + \frac{1}{2} \chi_{AB} e_3 + \frac{1}{2} \chibar_{AB} e_4,
\end{equation}
\begin{equation} \label{eq:Riccitab2}
	\nabla_{e_A} e_3 = {\chibar_A}^B e_B - \etabar_A e_3, \quad \nabla_{e_A} e_4 = {\chi_A}^B e_B + \etabar_A e_4,
\end{equation}
\begin{equation} \label{eq:Riccitab3}
	\nabla_{e_3} e_A = {\chibar_A}^B e_B + \eta_A e_3, \quad \nabla_{e_4}e_A = \left[ {\chi_A}^B - e_A(b^B) \right] e_B + \etabar_A e_4,
\end{equation}
\begin{equation} \label{eq:Riccitab4}
	\nabla_{e_3}e_4 = 2\eta^Ae_A, \quad \nabla_{e_4}e_3 = -\omega e_3 + 2\etabar^B e_B,
\end{equation}
\begin{equation} \label{eq:Riccitab5}
	\nabla_{e_3} e_3 = 0, \quad \nabla_{e_4} e_4 = \omega e_4,
\end{equation}
take the form\footnote{See \cite{Ch} for a detailed derivation in the vacuum case.  Recall that $\zeta = - \etabar$ and $\omegabar = 0$ in the frame used here, and that the scalar curvature $R$ vanishes for solutions of the massless Einstein--Vlasov system.}
\begin{align}
  \begin{split}
  \nablaslash_3 \alpha + \frac{1}{2} tr\chibar \alpha = 
  & \ 
  \nablaslash \hat{\otimes} \beta - 3(\hat{\chi} \rho + {}^*\hat{\chi} \sigma) + (4\eta - \etabar) \hat{\otimes} \beta
  +  \tilde{\nablaslash} \Tslash_4 - \nablaslash_4 \Tslash 
  \\
  & \
  - \frac{1}{2} \chi \Tslash_{34} - \frac{1}{2} \chibar \Tslash_{44} - \Tslash \times \chi + \etabar \otimes \Tslash_4
  \\
  & \
  -\frac{1}{4} \left( \nablaslash_3 \Tslash_{44} - 4 \eta \cdot \Tslash_4 - \nablaslash_4 \Tslash_{34} + 2 \eta \cdot \Tslash_4 + \omega \Tslash_{34} \right) \gslash
  \\
  & \
  + \frac{1}{2} \left( \epsslash \nablaslash \Tslash_4 - \epsslash \cdot (\chi \times \Tslash) - \epsslash \cdot (\etabar \otimes \Tslash_4) \right) \epsslash,
  \end{split} \label{eq:bianchifirst}
  \\
  \begin{split}
  \nablaslash_4 \beta + 2 tr\chi \beta = 
  & \
   \divslash \alpha + \omega \beta + (\eta^{\#} - 2\etabar^{\#})\cdot \alpha
   \\
   & \
   - \frac{1}{2} \left( \nablaslash \Tslash_{44} - \nablaslash_4 \Tslash_4 - 2\chi \cdot \Tslash_4 + \omega \Tslash_4 - \Tslash_{44} \eta \right),
   \end{split}
	\\
  \begin{split}
  \nablaslash_3 \beta + tr \chibar \beta =
  & \
  \nablaslash \rho + {}^*\nablaslash \sigma + 2 \hat{\chi}^{\#} \cdot \betabar + 3(\eta \rho + {}^*\eta \sigma)
  \\
  & \
  + \frac{1}{2}\left( \nablaslash \Tslash_{34} - \nablaslash_4 \Tslash_3 - \chibar \cdot \Tslash_4 - \chi \cdot \Tslash_3 - \Tslash_{34} \etabar - \omega \Tslash_3 + 2 \etabar \cdot \Tslash \right),
  \end{split}
  \\
  \begin{split}
  \nablaslash_4 \rho + \frac{3}{2} tr \chi \rho =
  & \
  \divslash \beta - \frac{1}{2} \hat{\chibar}^{\#} \cdot \alpha + (\etabar ,\beta)
  - \frac{1}{4}\left( \nablaslash_3 \Tslash_{44} - \nablaslash_4 \Tslash_{34} - 2 \eta \cdot \Tslash_4 + \omega \Tslash_{34} \right),
  \end{split}
  \\
  \begin{split}
  \nablaslash_4 \sigma + \frac{3}{2} tr\chi \sigma = 
  & - \curlslash \beta - \etabar \wedge \beta + \frac{1}{2} \hat{\chibar} \wedge \alpha
  \\
  &
  - \frac{1}{2} \left( \epsslash \cdot \nablaslash \Tslash_4 - \epsslash \cdot (\chi \times \Tslash) - \epsslash \cdot (\etabar \otimes \Tslash_4) \right),
  \end{split}
\end{align}
\begin{align}
  \begin{split}
  \nablaslash_3 \rho + \frac{3}{2} tr \chibar \rho =
  &
  - \divslash \betabar - \frac{1}{2} \hat{\chi}^{\#} \cdot \alphabar - ( 2\eta + \etabar, \betabar)
  \\
  &
  - \frac{1}{4} \left( \nablaslash_4 \Tslash_{33} - \nablaslash_3 \Tslash_{34} + 2\omega \Tslash_{33} - 4 \etabar \cdot \Tslash_3 + 2 \eta \cdot \Tslash_3 \right),
  \end{split}
  \\
  \begin{split}
  \nablaslash_3 \sigma + \frac{3}{2} tr\chibar \sigma =
  &
  - \curlslash \betabar - (2\eta + \etabar) \wedge \betabar - \frac{1}{2} \hat{\chi} \wedge \alphabar
  \\
  &
  + \frac{1}{2} \left( \epsslash \cdot \nablaslash \Tslash_3 - \epsslash \cdot ( \chibar \times \Tslash) + \epsslash \cdot(\etabar \otimes \Tslash_3) \right), 
  \end{split}
  \\
  \begin{split}
  \nablaslash_4 \betabar + tr \chi \betabar =
  &
  - \nablaslash \rho + {}^* \nablaslash \sigma - \omega \betabar + 2 \hat{\chibar}^{\#} \cdot \beta - 3(\etabar \rho - {}^*\etabar \sigma)
  \\
  &
  - \frac{1}{2}\left( \nablaslash \Tslash_{34} - \nablaslash_3 \Tslash_4 - \chibar \cdot \Tslash_4 - \chi \cdot \Tslash_3 + \Tslash_{34} \eta + 2 \eta \cdot \Tslash \right),
  \end{split}
  \\
  \begin{split}
  \nablaslash_3 \betabar + 2 tr \chibar \betabar = 
  &
  - \divslash \alphabar - (\eta^{\#} + 2\etabar^{\#}) \cdot \alphabar
  \\
  &
  + \frac{1}{2}\left( \nablaslash \Tslash_{33} - \nablaslash_3 \Tslash_3 - 2\chibar \cdot \Tslash_3 + 2 \Tslash_{33} \etabar + \Tslash_{33} \eta \right),
  \end{split}
  \\
  \begin{split}
  \nablaslash_4 \alphabar + \frac{1}{2} tr\chi \alphabar =
  &
  - \nablaslash \hat{\otimes} \betabar - 2 \omega \alphabar - 3(\hat{\chibar} \rho - {}^*\hat{\chibar} \sigma) - 5 \etabar \hat{\otimes} \betabar
  + \tilde{\nablaslash}\Tslash_3 - \nablaslash_3 \Tslash - \Tslash \times \chibar 
  \\
  &
  + \Tslash_3 \otimes \etabar - \frac{1}{2} \Tslash_{33} \chi - \frac{1}{2} \Tslash_{34} \chibar + \eta \otimes \Tslash_3 + \Tslash_3 \otimes \eta
  \\
  &
  - \frac{1}{4} \left( \nablaslash_4 \Tslash_{33} - \nablaslash_3 \Tslash_{34} + 2\omega \Tslash_{33} - 4\etabar \cdot \Tslash_3 + 2 \eta \cdot \Tslash_3 \right) \gslash
  \\
  &
  + \frac{1}{2} \left( \epsslash \cdot \nablaslash \Tslash_3 - \epsslash \cdot (\chibar \times \Tslash) + \epsslash \cdot ( \etabar \otimes \Tslash_3) \right) \epsslash .
  \end{split} \label{eq:bianchilast}
\end{align}
Here for an $S_{u,v}$ 1-form $\xi$, $\tilde{\nablaslash} \xi$ denotes the transpose of the derivative of $\xi$,
\[
	(\tilde{\nablaslash} \xi)_{AB} := (\nablaslash_B \xi)_A.
\]
The left Hodge-dual ${}^{*}$ is defined on $S_{u,v}$ one forms and $(0,2)$ $S_{u,v}$ tensors by
\[
	{}^{*} \xi_A = \epsslash_{AB} \gslash^{BC} \xi_C \quad \text{and} \quad {}^{*} \xi_{AB} = \epsslash_{AC} \gslash^{CD} \xi_{DB},
\]
respectively.  Here $\epsslash$ denotes the volume form associated with the metric $\gslash$ and, for a $(0,2)$ $S_{u,v}$ tensor $\xi$,
\[
	\epsslash \cdot \xi = \epsslash^{AB} \xi_{AB}.
\]
The symmetric traceless product of two $S_{u,v}$ one forms is defined by
\[
	(\xi \hat{\otimes} \xi ')_{AB} = \xi_A \xi_B' + \xi_B \xi_A' - \gslash_{AB} \left( \gslash^{CD} \xi_C \xi_D' \right),
\]
and the anti-symmetric products are defined by
\[
	\xi \wedge \xi' = \epsslash^{AB} \xi_A \xi_B' \quad \text{and} \quad \xi \wedge \xi' = \epsslash^{AB} \gslash^{CD} \xi_{AC} \xi_{BD}',
\]
for two $S_{u,v}$ one forms and $S_{u,v}$ $(0,2)$ tensors respectively.  Also,
\[
	(\xi \times \xi')_{AB} = \gslash^{CD} \xi_{AC} \xi'_{BD},
\]
for $S_{u,v}$ $(0,2)$ tensors $\xi,\xi'$.  The symmetric trace free derivative of an $S_{u,v}$ 1-form is defined as
\[
	(\nablaslash \hat{\otimes} \xi)_{AB} = \nablaslash_A \xi_B + \nablaslash_B \xi_A - (\divslash \xi) \gslash_{AB}.
\]
Finally define the $\gslash$ inner product of two $(0,n)$ $S_{u,v}$ tensors
\[
	(\xi,\xi') = \gslash^{A_1B_1} \cdots \gslash^{A_nB_n} \xi_{A_1,\ldots ,A_n} \xi_{B_1,\ldots ,B_n}',
\]
and the norm of a $(0,n)$ $S_{u,v}$ tensor
\[
	\vert \xi \vert^2 = \gslash^{A_1B_1} \cdots \gslash^{A_nB_n} \xi_{A_1,\ldots ,A_n} \xi_{B_1,\ldots ,B_n}.
\]
The notation $\vert \cdot \vert$ will also later be used when applied to components of $S_{u,v}$ tensors to denote the standard absolute value on $\mathbb{R}$.  See Section \ref{section:Sobolev}.  It will always be clear from the context which is meant, for example if $\xi$ is an $S_{u,v}$ 1-form then $\vert \xi \vert$ denotes the $\gslash$ norm as above, whilst $\vert \xi_A \vert$ denotes the absolute value of $\xi(e_A)$.

The null structure equations for the Ricci coefficients and the metric quantities in the 3 direction, suitably renormalised using the Minkowski values, take the form
\begin{align}
  \nablaslash_3 \omega = & \ 2 (\eta, \etabar) - \vert \etabar \vert^2 - \rho - \frac{1}{2} \Tslash_{34},
  \label{eq:nsfirst}
  \\
  \nablaslash_3 \etabar = & -\frac{1}{2} (tr\chibar - tr\chibar_{\circ})(\etabar - \eta) + \frac{1}{r}(\etabar - \eta) - \hat{\chibar} \cdot (\etabar - \eta) + \betabar - \frac{1}{2} \Tslash_3,
  \label{eq:nsetabar}
  \\
  \begin{split}
  \nablaslash_3(tr\chibar - tr\chibar_{\circ}) = & -\frac{1}{2}(tr\chibar - tr\chibar_{\circ})^2 + \frac{2}{r} (tr\chibar - tr\chibar_{\circ})
  - \frac{2}{r^2} \left( 1 - \frac{1}{\Omega^2} \right)
  \\
  &
  - \vert \hat{\chibar} \vert^2 - \Tslash_{33}, 
  \end{split}
  \\
  \nablaslash_3 \hat{\chibar} = & - (tr\chibar - tr\chibar_{\circ})\hat{\chibar} + \frac{2}{r} \hat{\chibar} - \alphabar, 
  \\
  \nablaslash_3 b = & \ 2(\eta - \etabar) + \hat{\chibar} \cdot b + \frac{1}{2} (\tr \chibar - \tr \chibar_{\circ}) b - \frac{1}{r}b, 
  \\
  \begin{split}
  \left( \nablaslash_3 (\gslash - \gslash^{\circ})\right)_{AB} = & \left( 1 - \frac{1}{\Omega^2}\right) \tr \chibar_{\circ} \gslash_{AB}^{\circ} + \left( \tr \chibar - \tr \chibar_{\circ} \right) \gslash_{AB}^{\circ} 
  \\
  & \ + 2\hat{\chibar}_{AB} - {\hat{\chibar}_A}^{\; \; C} \left( \gslash_{BC} - \gslash_{BC}^{\circ} \right) - {\hat{\chibar}_B}^{\; \; C} \left( \gslash_{AC} - \gslash_{AC}^{\circ} \right),
  \end{split}
\end{align}
and in the 4 direction
\begin{align}
  \nablaslash_4 \eta + \frac{1}{2} \tr \chi \eta = & \ \frac{1}{2} (\tr \chi - \tr \chi_{\circ}) \etabar + \frac{1}{r} \etabar - \hat{\chi} \cdot (\eta - \etabar) - \beta - \frac{1}{2} \Tslash_4,
  \\
  \begin{split}
  \nablaslash_4(tr\chi - tr\chi_{\circ}) + \tr \chi (tr\chi - tr\chi_{\circ}) = & \ \frac{1}{2}(tr\chi - tr\chi_{\circ})^2 + \omega (tr\chi - tr\chi_{\circ}) + \frac{2}{r} \omega 
  \\
  &
  - \vert \hat{\chi} \vert^2 - \Tslash_{44},
  \end{split}
  \\
  \nablaslash_4 \hat{\chi} + \tr \chi \hat{\chi} = & \ \omega\hat{\chi} - \alpha, 
  \\
  \nablaslash_4 \left( 1 - \frac{1}{\Omega^2} \right) = & \ \omega - \omega \left( 1 - \frac{1}{\Omega^2} \right).
  \label{eq:nslast}
\end{align}
Through most of the text, when referring to the \emph{null structure equations} it is the above equations which are meant.  The following null structure equations on the spheres will also be used in Section \ref{section:Riccitop},
\begin{align}
	K 
	&
	= \frac{1}{2} \hat{\chi} \cdot \hat{\chibar} - \frac{1}{4} \tr \chi \tr \chibar - \rho + \frac{1}{2} \Tslash_{34} \label{eq:Gauss}
	\\
	 \divslash \hat{\chi}
	 &
	 =
	 \frac{1}{2} \nablaslash \left( \tr \chi - \frac{2}{r} \right) - \frac{1}{2} \left( \tr \chi - \frac{2}{r} \right) \etabar + \frac{1}{r} \etabar + \hat{\chi} \cdot \etabar - \beta + \frac{1}{2} \Tslash_4, \label{eq:Codazziout}
	 \\
	 \divslash \hat{\chibar}
	 &
	 =
	 \frac{1}{2} \nablaslash \left( \tr \chibar + \frac{2}{r} \right) + \frac{1}{2} \left( \tr \chibar + \frac{2}{r} \right) \etabar - \frac{1}{r} \etabar - \hat{\chibar} \cdot \etabar - \betabar + \frac{1}{2} \Tslash_3, \label{eq:Codazziin}
\end{align}
\[
	 \curlslash \eta = \sigma - \frac{1}{2} \hat{\chi} \wedge \hat{\chibar}, \qquad \curlslash \etabar = \frac{1}{2} \hat{\chi} \wedge \hat{\chibar} - \sigma,
\]
where $K$ denotes the Gauss curvature of the spheres $(S_{u,v},\gslash)$.

The additional propagation equations for $\hat{\chi}, \hat{\chibar}$,
\begin{align}
	 \nablaslash_4 \hat{\chibar} + \frac{1}{2} \tr \chi \hat{\chibar}
	 &
	 =
	 \nablaslash \hat{\otimes} \etabar - \omega \hat{\chibar} - \frac{1}{2} \tr \chibar \hat{\chi} + \etabar \hat{\otimes} \etabar + \hat{\Tslash},
	 \label{eq:chibarhat4}
	 \\
	 \nablaslash_3 \hat{\chi} + \frac{1}{2} \tr \chibar \hat{\chi}
	 &
	 =
	 \nablaslash \hat{\otimes} \eta - \frac{1}{2} \tr \chi \hat{\chibar} + \eta \hat{\otimes} \eta + \hat{\Tslash},
	 \label{eq:chihat3}
\end{align}
will also be used in Section \ref{section:Riccitop} to derive propagation equations for the mass aspect function $\mu, \mubar$ defined later.  Here $\hat{\Tslash} = \Tslash - \Tslash_{34} \gslash$ is the trace free part of $\Tslash$.

The following first variational formulas for the induced metric on the spheres will also be used,
\begin{align}
	\mathcal{L}_{e_3} \gslash 
	&
	= 2 \chibar, \label{eq:firstvar3}
	\\
	\mathcal{L}_{e_4} \gslash 
	&
	= 2 \chi, \label{eq:firstvar4}
\end{align}
where $\mathcal{L}$ denotes the Lie derivative.

There are additional null structure equations but, since they will not be used here, are omitted.

\section{The Schematic Form of the Equations and Commutation} \label{section:equations}
In this section schematic notation is introduced for the Ricci coefficients, curvature components and components of the energy momentum tensor, which is used to isolate the structure in the equations that is important for the proof of Theorem \ref{thm:main2}.  A collection of differential operators is introduced and it is shown that this structure remains present after commuting the equations by any of the operators in the collection.  This section closely follows Section 3 of \cite{DHR} where this notation was introduced.

\subsection{Schematic Notation} \label{subsec:schemnot}
Consider the collection of Ricci coefficients\footnote{The quantities $\frac{1}{\Omega^2} - 1$, $b$ and $\gslash - \gslash^{\circ}$ are, of course, metric quantities.  They are however, in Section \ref{section:Ricci}, estimated systematically along with the Ricci coefficients.  Any discussion of the ``Ricci coefficients'' from now on will hence implicitly refer also to these metric quantities.} which are schematically denoted $\Gamma$,
\[
  \Gamma = \frac{1}{\Omega^2} - 1, b, \gslash - \gslash^{\circ}, \tr \chi - \tr \chi_{\circ}, \tr \chibar - \tr \chibar_{\circ}, \hat{\chi}, \hat{\chibar}, \eta, \etabar, \omega.
\]
Note that the $\Gamma$ are normalised so that each of the corresponding quantities in Minkowski space is equal to zero.  In the proof of the main result it will be shown that each $\Gamma$ converges to zero as $r\to \infty$ in the spacetimes considered.  Each $\Gamma$ will converge with a different rate in $r$ and so, to describe these rates, each $\Gamma$ is given an index, $p$, to encode the fact that, as will be shown in the proof of the main result, $r^p \vert \Gamma_p \vert$ will be uniformly bounded.  The $p$-indices are given as follows,
\begin{equation*}
  \begin{array}{ll}
	 \Gamma_0 = \frac{1}{\Omega^2} - 1, 
	 & 
	 \Gamma_1 = \hat{\chibar}, \gslash - \gslash^{\circ}, b , \eta, 
	 \\
	 \Gamma_2 = \etabar, \tr\chibar - \tr \chibar_{\circ}, \tr \chi - \tr \chi_{\circ}, \hat{\chi}, 
	 & 
	 \Gamma_3 = \omega,
  \end{array}
\end{equation*}
so that $\Gamma_1$ schematically denotes any of the quantities $\hat{\chibar}, \gslash - \gslash^{\circ}, b , \eta,$ etc.  It may be the case, for a particular $\Gamma_p$, that $\lim_{r\to \infty} r^p \vert \Gamma_p \vert$ is always zero in each of the spacetimes which are constructed here.  This means that some of the Ricci coefficients will decay with faster rates than those propagated in the proof of Theorem \ref{thm:main2}.  Some of these faster rates can be recovered a posteriori.

The notation $\overset{(3)}{\Gamma}$ will be used to schematically denote any $\Gamma$ for which the corresponding null structure equation of \eqref{eq:nsfirst}--\eqref{eq:nslast} it satisfies is in the $\nablaslash_3$ direction,
\[
	\overset{(3)}{\Gamma} = \hat{\chibar}, \gslash - \gslash^{\circ}, b, \etabar, \tr\chibar - \tr \chibar_{\circ}, \omega.
\]
Similarly, $\overset{(4)}{\Gamma}$ will schematically denote any $\Gamma$ for which the corresponding null structure equation of \eqref{eq:nsfirst}--\eqref{eq:nslast} is in the $\nablaslash_4$ direction,
\[
	\overset{(4)}{\Gamma} = \left( \frac{1}{\Omega^2} - 1 \right), \eta, \tr \chi - \tr\chi_{\circ}, \hat{\chi}.
\]
Finally, $\Riccit$ will schematically denote any $\Gamma_p$ which has also been denoted $\overset{(3)}{\Gamma}$.  So, for example, $\hat{\chibar}$ may be schematically denoted $\overset{(3)}{\Gamma}_1$.  Similarly, $\Riccif$ will schematically denote any $\Gamma_p$ which has also been denoted $\overset{(4)}{\Gamma}$.

Consider now the collection of Weyl curvature components, which are schematically denoted $\psi$,
\[
  \psi = \alpha, \beta, \rho, \sigma, \betabar, \alphabar.
\]
Each $\psi$ is similarly given a $p$-index,
\begin{equation} \label{eq:curvdef}
  \psi_1 = \alphabar, \qquad \psi_2 = \betabar, \qquad \psi_3 = \rho, \sigma, \qquad \psi_{\frac{7}{2}} = \beta, \qquad \psi_4 = \alpha.
\end{equation}
to encode the fact that, as again will be shown, $r^p\vert \psi_p\vert$ is uniformly bounded in each of the spacetimes which are constructed.

When deriving energy estimates for the Bianchi equations in Section \ref{section:curvature}, a special divergence structure present in the terms involving angular derivatives is exploited.  For example, the $\nablaslash_3 \alpha$ equation is contracted with $\alpha$ (multiplied by a suitable weight) and integrated by parts over spacetime.  The $\nablaslash_4 \beta$ equation is similarly contracted with $\beta$ and integrated by parts.  When the two resulting identities are summed, a cancellation occurs in the terms involving angular derivatives leaving only a spherical divergence which vanishes due to the integration on the spheres.  The $\nablaslash_3 \alpha$ equation is thus \emph{paired} with the $\nablaslash_4 \beta$ equation.  To highlight this structure, consider the ordered pairs,
\[
	(\alpha, \beta), (\beta, (\rho, \sigma)), ((\rho, \sigma), \betabar), (\betabar, \alphabar).
\]
Each of these ordered pairs will be schematically denoted $(\uppsi_p, \uppsi_{p'}')$, with the subscripts $p$ and $p'$ as in \eqref{eq:curvdef}, and referred to as a \emph{Bianchi pair}.

The components of the energy momentum tensor are schematically denoted $\mathcal{T}$,
\[
	\mathcal{T} = \Tslash, \Tslash_3, \Tslash_4, \Tslash_{33}, \Tslash_{44}, \Tslash_{34},
\]
and each $\mathcal{T}$ is similarly given a $p$-index,
\[
	\mathcal{T}_2 =  \Tslash_{33}, 
	\quad 
	\mathcal{T}_3 = \Tslash_{3},
	\quad 
	\mathcal{T}_4 = \Tslash, \Tslash_{34},
	\quad
	\mathcal{T}_5 = \Tslash_{4},
	\quad
	\mathcal{T}_6 = \Tslash_{44},
\]
to encode the fact that $r^p \vert \mathcal{T}_p \vert$ will be shown to be uniformly bounded.

Finally, for a given $p\in \mathbb{R}$, let $\rp_p$ denote any smooth function $\rp_p : \mathcal{M} \to \mathbb{R}$, depending only on $r$, which behaves like $\frac{1}{r^p}$ to infinite order, i.\@e.\@ any function such that, for any $k\in \mathbb{N}_0$, there is a constant $C_k$ such that $r^{k+p} \vert (\partial_v)^k \rp_p \vert \leq C_k$, where the derivative is taken in the $(u,v,\theta^1,\theta^2)$ coordinate system.  In addition, the tensor field $\rp_p \gslash^{\circ}$ may also be denoted $\rp_p$.  Note that $r^{k+p} \vert (\partial_v)^k (\rp_p\gslash^{\circ}_{AB}) \vert_{\gslash^{\circ}} \leq C_k$.  For example,
\[
	\tr \chi_{\circ} = \rp_1, \qquad \tr \chibar_{\circ} \gslash^{\circ} = \rp_1.
\]

\subsection{The Schematic Form of the Equations}
Using the notation of the previous section, the null structure and Bianchi equations can be rewritten in schematic form.  For example the null structure equation \eqref{eq:nsetabar} can be rewritten,
\[
  \nablaslash_3 \overset{(3)}{\Gamma}_2 = \Gamma_2\cdot(\Gamma_2 + \Gamma_1) + \rp_1(\Gamma_2 + \Gamma_1) + \Gamma_1 \cdot (\Gamma_2 + \Gamma_1) + \psi_2 + \mathcal{T}_2.
\]
Here and in the following, $\Gamma_{p_1} \cdot \Gamma_{p_2}$ denotes (a constant multiple of) an arbitrary contraction between a $\Gamma_{p_1}$ and a $\Gamma_{p_2}$.  In the estimates later, the Cauchy--Schwarz inequality $\vert \Gamma_{p_1} \cdot \Gamma_{p_2} \vert \leq C \vert \Gamma_{p_1}\vert \vert \Gamma_{p_2} \vert$ will always be used and so the precise form of the contraction will be irrelevant.  Similarly for $\rp_{p_1} \Gamma_{p_2}$.

Rewriting the equations in this way allows one to immediately read off the rate of decay in $r$ of the right hand side.  In the above example one sees that $\nablaslash_3 \overset{(3)}{\Gamma}_2$ is equal to a combination of terms whose overall decay is, according to the $p$-index notation, like $\frac{1}{r^2}$, consistent with the fact that applying $\nablaslash_3$ to a Ricci coefficient does not alter its $r$ decay (see Section \ref{subsec:commutation}).  Each of the null structure equations can be expressed in this way.
  
\begin{proposition}[cf.\@ Proposition 3.1 of \cite{DHR}]
  The null structure equations \eqref{eq:nsfirst}--\eqref{eq:nslast} can be written in the following schematic form
  \begin{IEEEeqnarray}{rccl}
	 \nablaslash_3 \Riccit & & & = E_3[\Riccit] \\
	 \nablaslash_4 \Riccif &\ + \ & \frac{p}{2} \tr \chi \ \Riccif \ & = E_4[\Riccif] \label{eq:firstRicci}
  \end{IEEEeqnarray}
  where\footnote{Stricly speaking, the terms listed in the error can be ``worse'' than the terms which actually appear in the equation in question.  For example, in $E_3[\Riccit]$, $\psi_p$ may actually refer to some $\psi_q$ where $q>p$.  In fact, in some of the null structure equations no curvature term actually appears.}
  \begin{align*}
	 E_3[\Riccit] = & \ \sum_{p_1+p_2 \geq p} \rp_{p_1} \cdot \Gamma_{p_2} + \sum_{p_1+p_2 \geq p} \Gamma_{p_1} \cdot \Gamma_{p_2} + \psi_p + \mathcal{T}_p \\
	 E_4[\Riccif] = & \ \sum_{p_1+p_2 \geq p + 2} \rp_{p_1} \cdot \Gamma_{p_2} + \sum_{p_1+p_2 \geq p + 2} \Gamma_{p_1} \cdot \Gamma_{p_2} + \psi_{p+2} + \mathcal{T}_{p + 2}.
  \end{align*}
\end{proposition}

This proposition allows us to see that the right hand sides of the $\nablaslash_3 \Riccit$ equations behave like $\frac{1}{r^p}$, whilst the right hand sides of the $\nablaslash_4 \Riccif$ equations behave like $\frac{1}{r^{p+2}}$.  This structure will be heavily exploited and should be seen as a manifestation of the null condition present in the Einstein equations.

\begin{remark} \label{rem:renorm}
  The term $\frac{p}{2} \tr \chi \ \Riccif$ on the left hand side of equation \eqref{eq:firstRicci} is not contained in the error since $\tr \chi$ behaves like $\frac{1}{r}$ and so this term only behaves like $\frac{1}{r^{p+1}}$.  This would thus destroy the structure of the error.  It is not a problem that this term appears however since, when doing the estimates, the following renormalised form of the equation will always be used:
  \[
	 \nablaslash_4 \left( r^{p} \Riccif \right) = r^{p} E_4[\Riccif ].
  \]
  This can be derived by differentiating the left hand side using the product rule, substituting equation \eqref{eq:firstRicci} and using the fact that $(\tr \chi_{\circ} - \tr \chi) \Riccif$ can be absorbed into the error.
  
  It is not a coincidence that the coefficient of this term is always $\frac{p}{2}$, it is the value of this coefficient which decides the rate of decay to be propagated for each $\Riccif$.  This will be elaborated on further in Section \ref{section:Ricci}.  This was not the case in \cite{DHR}; they have more freedom since they are integrating backwards from null infinity, rather than towards null infinity, and so can propagate stronger decay rates for some of the $\Riccif$.  These stronger rates could be recovered here using the ideas in Section \ref{section:Riccitop}, however it is perhaps interesting to note that the estimates can be closed with these weaker rates.
\end{remark}

The Bianchi equations can also be rewritten in this way.
\begin{proposition}[cf.\@ Proposition 3.3 of \cite{DHR}] \label{prop:Bianchi}
  For each Bianchi pair $(\uppsi_p, \uppsi_{p'}')$, the Bianchi equations \eqref{eq:bianchifirst}--\eqref{eq:bianchilast} can be written in the following schematic form
  \begin{equation*}
  \begin{array}{lccl}
	 \nablaslash_3 \uppsi_p & & & = \Dslash \uppsi_{p'}' + E_3[\uppsi_p],\\
	 \nablaslash_4 \uppsi_{p'}' & + & \gamma[\uppsi_{p'}'] \ tr \chi \ \uppsi_{p'}' & = \Dslash \uppsi_p +  E_4[\uppsi_{p'}'],
  \end{array}
  \end{equation*}
  where $\Dslash$ denotes the angular operator appearing in equation for the particular curvature component under consideration\footnote{So, for example, $\Dslash \uppsi_{p'}' = \nablaslash \hat{\otimes} \beta$ in the $\nablaslash_3{\alpha}$ equation, and $\Dslash \uppsi_{p} = - \nablaslash \rho + {}^* \nablaslash \sigma $ in the $\nablaslash_4 \betabar$ equation, etc.} and $\gamma[\uppsi_{p'}'] = \frac{p'}{2}$ for $\uppsi_{p'}' \neq \beta$, $\gamma[\beta] = 2$.  The error terms take the form
  \begin{align*}
	 E_3[\uppsi_p] = & \ \rp_1 \uppsi_p + \sum_{p_1+p_2 \geq p} \Gamma_{p_1} \cdot \psi_{p_2} + \sum_{p_1 + p_2 \geq p} \rp_{p_1} \mathfrak{D} \mathcal{T}_{p_2} + \sum_{p_1 + p_2 \geq p} \Gamma_{p_1} \cdot \mathcal{T}_{p_2}, \\
	 E_4[\uppsi_{p'}'] = & \ \sum_{p_1+p_2 \geq p' + \frac{3}{2}} \Gamma_{p_1} \cdot \psi_{p_2} + \sum_{p_1 + p_2 \geq p + 2} \rp_{p_1} \mathfrak{D} \mathcal{T}_{p_2} + \sum_{p_1 + p_2 \geq p + 2} \Gamma_{p_1} \cdot \mathcal{T}_{p_2},
  \end{align*}
  where $\mathfrak{D}$ is used to denote certain derivative operators which are introduced in Section \ref{subsec:commutation}.
\end{proposition}
When applied to $\mathcal{T}_p$, the operators $\mathfrak{D}$ should not alter the rate of decay so again this schematic form allows one to easily read off the $r$ decay rates of the errors.  This structure of the errors will again be heavily exploited.  The first summation in $E_4[\uppsi_{p'}']$ can in fact actually always begin at $p'+2$ except for in $E_4[\beta]$ where the term $\eta^{\#} \cdot \alpha$ appears.  Also the terms,
\[
	 \sum_{p_1 + p_2 \geq p} \rp_{p_1} \mathfrak{D} \mathcal{T}_{p_2} + \sum_{p_1 + p_2 \geq p} \Gamma_{p_1} \cdot \mathcal{T}_{p_2},
\]
in $E_3[\uppsi_{p}]$ can be upgraded to,
\[
	 \sum_{p_1 + p_2 \geq p + \frac{1}{2}} \rp_{p_1} \mathfrak{D} \mathcal{T}_{p_2} + \sum_{p_1 + p_2 \geq p + \frac{1}{2}} \Gamma_{p_1} \cdot \mathcal{T}_{p_2},
\]
in $E_3[\alpha]$ and $E_3[\beta]$.  These points are important and will be returned to in Section \ref{section:curvature}.

\subsection{The Commuted Equations} \label{subsec:commutation}
As discussed in the introduction, the Ricci coefficients and curvature components will be estimated in $L^2$ using the null structure and Bianchi equations respectively\footnote{The former in $L^2$ on the spheres, the latter in $L^2$ on null hypersurfaces.}.  In order to deal with the nonlinearities some of the error terms are estimated in $L^{\infty}$ on the spheres.  These $L^{\infty}$ bounds are obtained from $L^2$ estimates for higher order derivatives via Sobolev inequalities.  These higher order $L^2$ estimates are obtained through commuting the null structure and Bianchi equations with suitable differential operators, showing that the structure of the equations are preserved, and then proceeding as for the zero-th order case.  It is shown in this section that the structure of the equations are preserved under commutation.

It is also necessary to obtain higher order estimates for components of the energy momentum tensor in order to close the estimates for the Bianchi and Null structure equations.  Rather than commuting the Vlasov equation, which leads to certain difficulties, these estimates are obtained by estimating components of certain Jacobi fields on the mass shell.  See Section \ref{section:emtensor}.

Define the set of differential operators $\{ \nablaslash_3, r\nablaslash_4, r\nablaslash\}$ acting on covariant $S_{u,v}$ tensors of any order\footnote{Note that $\nablaslash_3$ and $r\nablaslash_4$ preserve the rank of a tensor, whilst $r\nablaslash$ takes tensors of rank $(0,n)$ to tensors of rank $(0,n+1)$.}, and let $\mathfrak{D}$ denote an arbitrary element of this set.  These operators are introduced because of the \emph{Commutation Principle} of \cite{DHR}:
\begin{quotation}
  \emph{Commutation Principle: Applying any of the operators $\mathfrak{D}$ to any of the $\Gamma, \psi, \mathcal{T}$ should not alter its rate of decay.}
\end{quotation}
This will be shown to hold in $L^2$, though until then it serves as a useful guide to interpret the structure of the commuted equations.

If $\xi$ is an $S_{u,v}$ tensor field, $\mathfrak{D}^k \xi$ will be schematically used to denote any fixed $k$-tuple $\mathfrak{D}_k\mathfrak{D}_{k-1} \ldots \mathfrak{D}_1 \xi$ of operators applied to $\xi$, where each $\mathfrak{D}_i \in \{ \nablaslash_3, r\nablaslash_4, r\nablaslash\}$.

In order to derive expressions for the commuted Bianchi equations in this schematic notation, the following commutation lemma will be used.  Recall first the following lemma which relates projected covariant derivatives of a covariant $S_{u,v}$ tensor to derivatives of its components.
\begin{lemma} \label{lemma:componentderivatives}
	Let $\xi$ be a $(0,k)$ $S_{u,v}$ tensor.  Then,
	\begin{align*}
		(\nablaslash_3 \xi)_{A_1\ldots A_k} 
		=
		&
		e_3(\xi_{A_1\ldots A_k}) - \sum_{i=1}^k {\chibar_{A_i}}^B \xi_{A_1 \ldots A_{i-1} B A_{i+1} \ldots A_k},
		\\
		(\nablaslash_4 \xi)_{A_1\ldots A_k} 
		=
		&
		e_4(\xi_{A_1\ldots A_k}) - \sum_{i=1}^k \left( {\chi_{A_i}}^B - \nablaslash_A b^B + b^C \Gammaslash_{AC}^B \right) \xi_{A_1 \ldots A_{i-1} B A_{i+1} \ldots A_k},
	\end{align*}
	and
	\begin{align*}
		(\nablaslash_B \xi)_{A_1 \ldots A_k}
		=
		e_B(\xi_{A_1\ldots A_k}) - \sum_{i=1}^k \Gammaslash_{B A_i}^C \xi_{A_1 \ldots A_{i-1} C A_{i+1} \ldots A_k}.
	\end{align*}
\end{lemma}

The commutation lemma then takes the following form.

\begin{lemma}[cf.\@ Lemma 7.3.3 of \cite{ChKl} or Lemma 3.1 of \cite{DHR}] \label{lemma:commutation}
	If $\xi$ is a $(0,k)$ $S_{u,v}$ tensor then,
	\begin{multline*}
		[\nablaslash_4, \nablaslash_B] \xi_{A_1\ldots A_k}
		=
		- {\chi_B}^C \nablaslash_C \xi_{A_1\ldots A_k}
		\\
		+ \sum_{i=1}^k  \left( \chi_{A_iB} \etabar^C - \etabar_{A_i} {\chi_B}^C + {}^* \beta_B {\epsslash_{A_i}}^C + \frac{1}{2} {\gslash_B}^C \Tslash_{4A_i} - \frac{1}{2} \gslash_{A_i B} {\Tslash_4}^C \right) \xi_{A_1\ldots A_{i-1} C A_{i+1} \ldots A_k},
	\end{multline*}
	\begin{multline*}
		[\nablaslash_3, \nablaslash_B] \xi_{A_1\ldots A_k}
		=
		\left( \eta_B + \etabar_B \right) \nablaslash_3 \xi_{A_1\ldots A_k}  - {\chibar_B}^C \nablaslash_C \xi_{A_1\ldots A_k}
		\\
		+ \sum_{i=1}^k  \left( \chibar_{A_iB} \eta^C - \eta_{A_i} {\chibar_B}^C - {}^* \betabar_B {\epsslash_{A_i}}^C + \frac{1}{2} {\gslash_B}^C \Tslash_{3A_i} - \frac{1}{2} \gslash_{A_i B} {\Tslash_3}^C \right) \xi_{A_1\ldots A_{i-1} C A_{i+1} \ldots A_k},
	\end{multline*}
	\begin{multline*}
		[\nablaslash_3, \nablaslash_4] \xi_{A_1\ldots A_k}
		=
		2\left( \eta^C - \etabar^C \right) \nablaslash_C \xi_{A_1\ldots A_k} + \omega \nablaslash_3 \xi_{A_1 \ldots A_k}
		\\
		+ 2 \sum_{i=1}^k  \left( \etabar_{A_i} \eta^C - \eta_{A_i}\etabar^C - \sigma {\epsslash_{A_i}}^C \right) \xi_{A_1\ldots A_{i-1} C A_{i+1} \ldots A_k},
	\end{multline*}
	and
	\begin{equation*}
		[\nablaslash_B, \nablaslash_C] \xi_{A_1\ldots A_k}
		=
		K \sum_{i=1}^k  \gslash_{BA_i} \xi_{A_1\ldots A_{i-1} C A_{i+1} \ldots A_k} - \gslash_{CA_i} \xi_{A_1\ldots A_{i-1} B A_{i+1} \ldots A_k},
	\end{equation*}
	where $K$ is the Gauss curvature of $(S_{u,v},\gslash)$.
\end{lemma}

\begin{proof}
	The proof of the first identity follows by writing
	\begin{align*}
		\nablaslash_4 \nablaslash_B \xi_{A_1\ldots A_k}
		=
		&
		\nabla_4\nabla_B \xi_{A_1\ldots A_k} + \etabar_B \nablaslash_4 \xi_{A_1\ldots A_k} - \sum_{i=1}^k \etabar_{A_i} {\chi_B}^C \xi_{A_1\ldots A_{i-1} C A_{i+1} \ldots A_k},
		\\
		\nablaslash_B \nablaslash_4 \xi_{A_1\ldots A_k}
		=
		&
		\nabla_B\nabla_4 \xi_{A_1\ldots A_k} + \etabar_B \nablaslash_4 \xi_{A_1\ldots A_k} 
		\\
		&
		+ {\chi_B}^C \nablaslash_C \xi_{A_1\ldots A_k}
		- \sum_{i=1}^k \etabar_{A_i} {\chi_B}^C \xi_{A_1\ldots A_{i-1} C A_{i+1} \ldots A_k},
	\end{align*}
	and using
	\begin{align*}
		[\nabla_4,\nabla_B]
		=
		&
		\sum_{i=1}^k {R_{4BA_i}}^C \xi_{A_1\ldots A_{i-1} C A_{i+1} \ldots A_k}
		\\
		=
		&
		\sum_{i=1}^k  \left( {}^* \beta_B {\epsslash_{A_i}}^C + \frac{1}{2} {\gslash_B}^C \Tslash_{4A_i} - \frac{1}{2} \gslash_{A_i B} {\Tslash_4}^C \right) \xi_{A_1\ldots A_{i-1} C A_{i+1} \ldots A_k},
	\end{align*}
	where the last line follows by using equation \eqref{eq:Weyl} to write,
	\[
		{R_{4BA_i}}^C = {W_{4BA_i}}^C + \frac{1}{2} \left( {\gslash_B}^C \Tslash_{4A_i} - \gslash_{A_i B}{\Tslash_4}^C \right).
	\]
	Similarly for the second one uses
	\begin{align*}
		\nablaslash_3 \nablaslash_B \xi_{A_1\ldots A_k}
		=
		&
		\nabla_3\nabla_B \xi_{A_1\ldots A_k} + \eta_B \nablaslash_3 \xi_{A_1\ldots A_k} - \sum_{i=1}^k \eta_{A_i} {\chibar_B}^C \xi_{A_1\ldots A_{i-1} C A_{i+1} \ldots A_k},
		\\
		\nablaslash_B \nablaslash_3 \xi_{A_1\ldots A_k}
		=
		&
		\nabla_B\nabla_3 \xi_{A_1\ldots A_k} - \etabar_B \nablaslash_3 \xi_{A_1\ldots A_k} 
		\\
		&
		+ {\chibar_B}^C \nablaslash_C \xi_{A_1\ldots A_k}
		- \sum_{i=1}^k \chibar_{A_iB} \eta^{C} \xi_{A_1\ldots A_{i-1} C A_{i+1} \ldots A_k},
	\end{align*}
	and for the third,
	\begin{align*}
		\nablaslash_3 \nablaslash_4 \xi_{A_1\ldots A_k}
		=
		&
		\nabla_3 \nabla_4 \xi_{A_1\ldots A_k} + 2\eta^C \nabla_C \xi_{A_1\ldots A_k} - 2 \sum_{i=1}^k \eta_{A_i} \etabar^C \xi_{A_1\ldots A_{i-1} C A_{i+1} \ldots A_k},
		\\
		\nablaslash_4 \nablaslash_3 \xi_{A_1\ldots A_k}
		=
		&
		\nabla_4 \nabla_3 \xi_{A_1\ldots A_k} - \omega \nablaslash_3 \xi_{A_1\ldots A_k} + 2\etabar^C \nablaslash_C \xi_{A_1\ldots A_k}
		\\
		&
		- 2 \sum_{i=1}^k \etabar_{A_i} \eta^C \xi_{A_1\ldots A_{i-1} C A_{i+1} \ldots A_k}.
	\end{align*}
	If $\Rslash$ denotes the curvature tensor of $(S_{u,v},\gslash)$, the last follows from writing,
	\[
		[\nablaslash_B,\nablaslash_C] \xi_{A_1\ldots A_k}
		=
		\sum_{i=1}^k {\Rslash_{BCA_i}}^D \xi_{A_1\ldots A_{i-1} D A_{i+1} \ldots A_k},
	\]
	and the fact that,
	\[
		\Rslash_{BCA_iD} = K \left( \gslash_{BA_i} \gslash_{CD} - \gslash_{BD} \gslash_{CA_i} \right).
	\]
\end{proof}
The above Lemma implies that the terms arising from commutation take the following schematic form,
\begin{align} \label{eq:schematiccom}
\begin{split}
	[\nablaslash_4,\nablaslash] \xi =
	&
	-\frac{1}{r} \nablaslash \xi + \sum_{p_1 \geq 2} \Gamma_{p_1} \cdot \nablaslash \xi + \sum_{p_1+p_2 \geq 3} (\rp_{p_1} + \Gamma_{p_1}) \cdot (\Gamma_{p_2} + \psi_{p_2} + \mathcal{T}_{p_2} ) \cdot \xi,
	\\
	[\nablaslash_3,\nablaslash] \xi =
	&
	\sum_{p_1 \geq 1} \Gamma_{p_1} \cdot \nablaslash_3 \xi + \sum_{p_1 + p_2 \geq 1} \rp_{p_1} \cdot (\rp_{p_2} + \Gamma_{p_2}) \cdot r \nablaslash \xi
	\\
	&
	+ \sum_{p_1+p_2 \geq 2} (\rp_{p_1} + \Gamma_{p_1}) \cdot (\Gamma_{p_2} + \psi_{p_2} + \mathcal{T}_{p_2} ) \cdot \xi,
	\\
	[\nablaslash_3,\nablaslash_4] \xi =
	&
	\sum_{p_1 + p_2 \geq 2} \rp_{p_1} \cdot \Gamma_{p_2} \cdot r \nablaslash \xi + \sum_{p_1 \geq 3} \Gamma_{p_1} \cdot \nablaslash_3 \xi
	\\
	&
	+ \sum_{p_1+p_2 \geq 3} (\rp_{p_1} + \Gamma_{p_1}) \cdot (\Gamma_{p_2} + \psi_{p_2} ) \cdot \xi.
\end{split}
\end{align}

The commuted Bianchi equations can then be written as follows.

\begin{proposition}[cf.\@ Proposition 3.4 of \cite{DHR}] \label{prop:commutedBianchi}
  For any integer $k\geq 1$ the commuted Bianchi equations, for each Bianchi pair $(\uppsi_p,\uppsi'_{p'})$, take the form\footnote{Note that $\divslash, \curlslash$ when applied to a $(0,n)$ tensor are always defined with respect to the last index etc.}
  \begin{IEEEeqnarray}{rccl}
	 \nablaslash_3 (\mathfrak{D}^k\uppsi_p) & & & = \Dslash (\mathfrak{D}^k\uppsi_{p'}') + E_3[\mathfrak{D}^k\uppsi_p], \nonumber \\\label{eq:commBianchi3}\\
	 \nablaslash_4 (\mathfrak{D}^k\uppsi_{p'}') &\ + \ & \gamma[\uppsi_{p'}'] \ \tr \chi \ \mathfrak{D}^k\uppsi_{p'}' \ & = \Dslash (\mathfrak{D}^k\uppsi_p) +  E_4[\mathfrak{D}^k\uppsi_{p'}'], \nonumber \\ \label{eq:commBianchi4}
  \end{IEEEeqnarray}
  where
  \begin{align}
	 E_3[\mathfrak{D}^k \uppsi_p] = & \ \mathfrak{D}(E_3[\mathfrak{D}^{k-1}\uppsi_{p}]) + \Lambda_1 \cdot ( \mathfrak{D}^k \uppsi_p + \mathfrak{D}^k\uppsi_{p'}') + \Lambda_1 \cdot ( \mathfrak{D}^{k-1}\uppsi_p + \mathfrak{D}^{k-1}\uppsi_{p'}') ,\label{eq:commutedBianchierror3} \\
	 \begin{split}
	 E_4[\mathfrak{D}^k \uppsi_{p'}'] = & \ \mathfrak{D}(E_4[\mathfrak{D}^{k-1}\uppsi_{p'}']) + E_4[ \mathfrak{D}^{k-1} \uppsi_{p'}'] + \Lambda_1 \cdot \mathfrak{D}^k \uppsi_p + \Lambda_2 \cdot \mathfrak{D}^k\uppsi_{p'}'  \\
	 & + \Lambda_1 \cdot \mathfrak{D}^{k-1} \uppsi_p + \Lambda_2' \cdot \mathfrak{D}^{k-1}\uppsi_{p'}', \label{eq:commutedBianchierror4}
	 \end{split}
  \end{align}
  and, for $p = 1,2$, $\Lambda_p$ denotes some fixed sum of contractions of $\rp$, $\Gamma$, $\psi$ and $\mathcal{T}$ such that $\Lambda_p$ decays, according to the $p$-index notation, like $\frac{1}{r^p}$.  Explicitly
  	\begin{align*}
  		\Lambda_1 
		= 
		& 
		\sum_{p_1+p_2 + p_3 \geq 1} \rp_{p_1} ( \rp_{p_2} + \Gamma_{p_2}) \cdot ( \rp_{p_3} + \Gamma_{p_3} + \psi_{p_3} + \mathcal{T}_{p_3}),
		\\
		\Lambda_2 =
		&
		\sum_{p_1+p_2 \geq 2} \rp_{p_1} ( \rp_{p_2} + \Gamma_{p_2}),
		\\
		\Lambda_2' = 
		&
		\sum_{p_1+p_2 + p_3 \geq 2} \rp_{p_1} ( \rp_{p_2} + \Gamma_{p_2}) \cdot ( \rp_{p_3} + \Gamma_{p_3} + \mathfrak{D} \Gamma_{p_3} + \psi_{p_3} + \mathcal{T}_{p_3}).
	\end{align*}
\end{proposition}

Note the presence of the first order derivative of $\Gamma$ in $\Lambda_2'$, whilst $\Lambda_1$ and $\Lambda_2$ contain only zeroth order terms.

\begin{remark}
  By the commutation principle and induction, it is clear that the first two terms in the error \eqref{eq:commutedBianchierror4} preserve the structure highlighted in Proposition \ref{prop:Bianchi}.  In the remaining terms, it is essential that $\mathfrak{D}^{k}\uppsi_{p'}'$ and $\mathfrak{D}^{k-1}\uppsi_{p'}'$ appear contracted with $\Lambda_2$ and $\Lambda_2'$, rather than $\Lambda_1$.  It will be clear in the proof below that it is the special form of the operators that cause this to occur.  Since, for each Bianchi pair $(\uppsi_p,\uppsi_{p'}')$, it is the case that $p \geq p' + \frac{1}{2}$, the $\Lambda_1 \cdot \mathfrak{D}^k \uppsi_p$ and $\Lambda_1 \cdot \mathfrak{D}^{k-1} \uppsi_p$ terms in $E_4[\mathfrak{D}^k \uppsi_{p'}']$ still preserve the form of the error.
  
  Similarly looking at the error \eqref{eq:commutedBianchierror3}, it is clear that the expected $r$ decay will be preserved from Proposition \ref{prop:Bianchi}.
  
  It will also be important later that $\Lambda_1$ and $\Lambda_2$ do not contain any derivatives of $\psi$ or $\Gamma$, whilst $\Lambda_2'$ only contains first order derivatives.
\end{remark}

\begin{proof}[Proof of Proposition \ref{prop:commutedBianchi}]
  The proof proceeds exactly as in Proposition 3.4 of \cite{DHR}, though one does need to be careful since some of the quantities decay slightly weaker here.  We consider only the $k=1$ case.  A simple induction argument completes the proof for $k>1$.
  
  Consider first the $\nablaslash_4 \uppsi_{p'}'$ equations.  Using the schematic form of the commutation formulae \eqref{eq:schematiccom},
  \begin{align*}
  	\nablaslash_4 r \nablaslash_4 \uppsi_{p'}'
	=
	&
	\nablaslash_4 \uppsi_{p'}' + r \nablaslash_4 \nablaslash_4 \uppsi_{p'}',
	\\
	\gamma \tr \chi r \nablaslash_4 \uppsi_{p'}'
	=
	&
	r \nablaslash_4 \left( \gamma \tr \chi \uppsi_{p'}' \right) - \gamma \uppsi_{p'}' r \nablaslash_4 \tr \chi,
	\\
	\Dslash r \nablaslash_4 \uppsi_p
	=
	&
	r \nablaslash_4 \left( \Dslash \uppsi_p \right) + \sum_{p_1 \geq 1} (\rp_{p_1} + \Gamma_{p_1}) \cdot r\nablaslash \uppsi_p
	\\
	&
	+ \sum_{p_1 + p_2 + p_3 \geq 2} \rp_{p_1} ( \rp_{p_2} + \Gamma_{p_2}) \cdot (\Gamma_{p_2} + \psi_{p_2} + \mathcal{T}_p) \cdot \uppsi_p.
  \end{align*}
  Now the Raychaudhuri equation\footnote{This is the unrenormalised null structure equation for $\tr \chi$ and can be derived from the $\nablaslash_4 \left( \tr \chi - \tr \chi_{\circ} \right)$ equation in Section \ref{subsec:nsbianchi}.},
  \[
  	\nablaslash_4 \tr \chi = - \frac{1}{2} (\tr \chi)^2 - \vert \hat{\chi} \vert^2 + \omega \tr \chi - \Tslash_{44},
  \]
  and the Bianchi equation for $\nablaslash_4 \uppsi_{p'}'$ imply that,
  \begin{align*}
  	\nablaslash_4 \uppsi_{p'}' - \gamma \uppsi_{p'}' r \nablaslash_4 \tr \chi
	=
	&
	\gamma \tr \chi \uppsi_{p'}'\left( \frac{ r}{2} \tr \chi - 1 \right) + \Dslash \uppsi_{p} + E_4 [ \uppsi_{p'}'] 
	\\
	&
	+ \gamma \left( \vert \hat{\chi} \vert^2 - \omega \tr \chi + \Tslash_{44} \right) \uppsi_{p'}'
	\\
	= 
	&
	\sum_{p_1 + p_2 + p_3 \geq 2} \rp_{p_1} (\rp_{p_2} + \Gamma_{p_2}) \cdot( \Gamma_{p_3} + \mathcal{T}_{p_3}) \cdot \uppsi_{p'}'
	\\
	&
	+ \frac{1}{r} \mathfrak{D} \uppsi_p + E_4 [\uppsi_{p'}'].
  \end{align*}
  Note the cancellation.  Hence
  \[
  	\nablaslash_4 (r \nablaslash_4 \uppsi_{p'}') + \gamma[\uppsi_{p'}'] \tr \chi \ r\nablaslash_4 \uppsi_{p'}' = \Dslash (r \nablaslash_4 \uppsi_p) +  E_4[r \nablaslash_4 \uppsi_{p'}'],
  \]
  where
  \begin{align*}
  	E_4[r \nablaslash_4 \uppsi_{p'}'] 
	=
	&
	r \nablaslash_4 E_4 [ \uppsi_{p'}' ] + E_4 [\uppsi_{p'}'] + \sum_{p_1 \geq 1} (\rp_{p_1} + \Gamma_{p_1}) \cdot \mathfrak{D} \uppsi_p
	\\
	&
	+ \sum_{p_1 + p_2 + p_3 \geq 2} \rp_{p_1} (\rp_{p_2} + \Gamma_{p_2}) \cdot( \Gamma_{p_3} + \mathcal{T}_{p_3}) \cdot \uppsi_{p'}'
	\\
	&
	+ \sum_{p_1 + p_2 + p_3 \geq 2} \rp_{p_1} ( \rp_{p_2} + \Gamma_{p_2}) \cdot (\Gamma_{p_3} + \psi_{p_3} + \mathcal{T}_{p_3}) \cdot \uppsi_p.
  \end{align*}
  Similarly, using again the schematic expressions \eqref{eq:schematiccom},
  \begin{align*}
  	\nablaslash_4 \nablaslash_3 \uppsi_{p'}'
	=
	&
	\nablaslash_3 \nablaslash_4  \uppsi_{p'}' + \sum_{p_1 \geq 3} \Gamma_{p_1} \cdot \nablaslash_3  \uppsi_{p'}' + \sum_{p_1+p_2 \geq 2} \rp_{p_1} \cdot \Gamma_{p_2} \cdot r \nablaslash \uppsi_{p'}'
	\\
	&
	+ \sum_{p_1+p_2+p_3 \geq 3} (\rp_{p_1} + \Gamma_{p_1}) \cdot (\Gamma_{p_2} + \psi_{p_2}) \cdot \uppsi_{p'}',
	\\
	\gamma \tr \chi \nablaslash_3  \uppsi_{p'}'
	=
	&
	\nablaslash_3 \left( \gamma \tr \chi \uppsi_{p'}' \right) + \sum_{p_1 \geq 2} (\rp_{p_1} + \mathfrak{D} \Gamma_{p_2} ) \cdot \uppsi_{p'}'
	\\
	\Dslash \nablaslash_3  \uppsi_{p}
	=
	&
	\nablaslash_3 \Dslash  \uppsi_{p} + \sum_{p_1 \geq 1} \Gamma_{p_1} \cdot \nablaslash_3  \uppsi_{p} + \sum_{p_1 + p_2 \geq 1} \rp_{p_1} (\rp_{p_2} + \Gamma_{p_2}) \cdot r \nablaslash  \uppsi_{p}
	\\
	&
	+ \sum_{p_1+p_2 \geq 2} (\rp_{p_1} + \Gamma_{p_1}) \cdot (\Gamma_{p_2} + \psi_{p_2} + \mathcal{T}_{p_2}) \cdot  \uppsi_{p},
  \end{align*}
  and hence,
  \[
  	\nablaslash_4 (\nablaslash_3 \uppsi_{p'}') + \gamma[\uppsi_{p'}'] \tr \chi \ \nablaslash_3 \uppsi_{p'}' = \Dslash (\nablaslash_3 \uppsi_p) +  E_4[ \nablaslash_3 \uppsi_{p'}'],
  \]
  where
  \begin{align*}
  	E_4[ \nablaslash_3 \uppsi_{p'}']
	=
	&
	\nablaslash_3 E_4[ \uppsi_{p'}']
	+
	\sum_{p_1+p_2 \geq 2} \rp_{p_1} \cdot \Gamma_{p_2} \cdot \mathfrak{D} \uppsi_{p'}'
	+
	\sum_{p_1 + p_2 \geq 1} \rp_{p_1} (\rp_{p_2} + \Gamma_{p_2}) \cdot \mathfrak{D} \uppsi_{p}
	\\
	&
	+
	\sum_{p_1+p_2+p_3 \geq 3} (\rp_{p_1} + \Gamma_{p_1}) \cdot (\rp_{p_2} + \Gamma_{p_2} + \mathfrak{D} \Gamma_{p_2} + \psi_{p_2}) \cdot \uppsi_{p'}'
	\\
	&
	+
	\sum_{p_1+p_2 \geq 2} (\rp_{p_1} + \Gamma_{p_1}) \cdot (\Gamma_{p_2} + \psi_{p_2} + \mathcal{T}_{p_2}) \cdot  \uppsi_{p}.
  \end{align*}
  Finally,
	\begin{align*}
		\nablaslash_4 r \nablaslash \uppsi_{p'}'
		=
		&
		\nablaslash \uppsi_{p'}' + r \nablaslash_4 \nablaslash \uppsi_{p'}'
		\\
		=
		&
		\sum_{p_1 \geq 2} \Gamma_{p_1} \cdot r \nablaslash \uppsi_{p'}' + r \sum_{p_1+p_2 \geq 3} (\rp_{p_1} + \Gamma_{p_1}) \cdot (\Gamma_{p_2} + \psi_{p_2} + \mathcal{T}_{p_2}) \cdot \uppsi_{p'}',
		\\
		\gamma \tr \chi r \nablaslash \uppsi_{p'}'
		=
		&
		r \nablaslash \left( \gamma \tr \chi \uppsi_{p'}' \right) + \mathfrak{D} \Gamma_2 \cdot \uppsi_{p'}',
		\\
		\Dslash r \nablaslash \uppsi_{p}
		=
		&
		r \nablaslash \Dslash \uppsi_{p} + r \sum_{p_1+p_2 \geq 3}(\rp_{p_1} + \Gamma_{p_1}) \cdot (\rp_{p_2} + \Gamma_{p_2} + \psi_{p_2} + \mathcal{T}_{p_2}) \cdot \uppsi_p,
	\end{align*}
	where the Gauss equation \eqref{eq:Gauss} has been used for the third equality.  Note also the cancellation which occurs in the first equality.  Hence,
	\[
	  	\nablaslash_4 (r \nablaslash \uppsi_{p'}') + \gamma[\uppsi_{p'}'] \tr \chi \ r\nablaslash \uppsi_{p'}' = \Dslash (r \nablaslash \uppsi_p) +  	E_4[r \nablaslash \uppsi_{p'}'],
  	\]
  	where
	\begin{align*}
		E_4[r \nablaslash \uppsi_{p'}']
		=
		&
		r\nablaslash E_4[ \uppsi_{p'}']
		+
		\sum_{p_1 \geq 2} \Gamma_{p_1} \cdot \mathfrak{D} \uppsi_{p'}' 
		\\
		&
		+ \sum_{p_1+p_2 + p_3 \geq 2} \rp_{p_1} (\rp_{p_2} + \Gamma_{p_2}) \cdot (\Gamma_{p_3} + \mathfrak{D} \Gamma_{p_3} + \psi_{p_3} + \mathcal{T}_{p_3}) \cdot \uppsi_{p'}'
		\\
		&
		+
		\sum_{p_1+p_2 + p_3 \geq 2} \rp_{p_1} (\rp_{p_2} + \Gamma_{p_2}) \cdot (\rp_{p_3} + \Gamma_{p_3} + \psi_{p_3} + \mathcal{T}_{p_3}) \cdot \uppsi_p.
	\end{align*}
	
	The schematic expressions for the $\nablaslash_3 \uppsi_p$ equations follow similarly.
\end{proof}

Similarly, the commuted null structure equations can be schematically written as follows.

\begin{proposition}[cf.\@ Proposition 3.5 of \cite{DHR}] \label{prop:commutednullstructure}
  For any integer $k\geq 1$ the commuted null structure equations take the form,
  \begin{equation*}
	 \begin{array}{rccl}
		\nablaslash_3 (\mathfrak{D}^k \Riccit) & & & = E_3[\mathfrak{D}^k \Riccit],\\
		\nablaslash_4 (\mathfrak{D}^k \Riccif) & + & \frac{p}{2} \ tr \chi \ \mathfrak{D}^k \Riccif & = E_4[\mathfrak{D}^k \Riccif],
	 \end{array}
  \end{equation*}
  where,
  \begin{align}
	 E_3[\mathfrak{D}^k \Riccit] = & \ \mathfrak{D} (E_3[\mathfrak{D}^{k-1} \Riccit]) + \Lambda_1 \cdot ( \mathfrak{D}^{k} \Riccit + \mathfrak{D}^{k-1} \Riccit) \label{eq:nullstructureerror3} \\
	 E_4[\mathfrak{D}^k \Riccif] = & \ \mathfrak{D} (E_4[\mathfrak{D}^{k-1} \Riccif]) + E_4[\mathfrak{D}^{k-1} \Riccif] + \Lambda_2 \cdot \mathfrak{D}^{k} \Riccif + \Lambda_2' \cdot\mathfrak{D}^{k-1} \Riccif ,\label{eq:nullstructureerror4}
  \end{align}
  and again,
  \begin{align*}
  		\Lambda_1 
		= 
		& 
		\sum_{p_1+p_2 + p_3 \geq 1} \rp_{p_1} ( \rp_{p_2} + \Gamma_{p_2}) \cdot ( \rp_{p_3} + \Gamma_{p_3} + \psi_{p_3} + \mathcal{T}_{p_3}),
		\\
		\Lambda_2 =
		&
		\sum_{p_1+p_2 \geq 2} \rp_{p_1} ( \rp_{p_2} + \Gamma_{p_2}),
		\\
		\Lambda_2' = 
		&
		\sum_{p_1+p_2 + p_3 \geq 2} \rp_{p_1} ( \rp_{p_2} + \Gamma_{p_2}) \cdot ( \rp_{p_3} + \Gamma_{p_3} + \mathfrak{D} \Gamma_{p_3} + \psi_{p_3} + \mathcal{T}_{p_3}).
	\end{align*}
\end{proposition}

\begin{proof}
	The proof is similar to that of Proposition \ref{prop:commutedBianchi}, though slightly simpler as there are no terms involving $\Dslash$.
\end{proof}

\begin{remark}
  Note that again in \eqref{eq:nullstructureerror4}, $\mathfrak{D}^{k} \overset{(4)}{\Gamma_p}$ and $\mathfrak{D}^{k-1} \overset{(4)}{\Gamma_p}$ only appear multiplying terms which decay like $\frac{1}{r^2}$.  Note also that again $\Lambda_1,\Lambda_2$ contain no derivative terms, whilst $\Lambda_2'$ contains only first order derivatives.
\end{remark}

\section{The Sasaki Metric} \label{section:Sasaki}
The Lorentizian metric $g$ on $\mathcal{M}$ induces a metric, $\gbar$, on $T\mathcal{M}$, known as the Sasaki metric, which in turn induces a metric on $P$ by restriction.  This was first introduced in the context of Riemannian geometry by Sasaki \cite{Sa}.  Certain properties of this metric will be used when estimating derivatives of $f$ later.  The goal of this section is to define the metric and compute certain components of its curvature tensor in terms of the curvature of $(\mathcal{M},g)$.  It is then shown that trajectories of the null geodesic flow of $(\mathcal{M},g)$ are geodesics in $P$ (or more generally that trajectories of the full geodesic flow are geodesics in $T\mathcal{M}$) with respect to this metric and derivatives of the exponential map are Jacobi fields along these geodesics.  This fact will be used in Section \ref{section:emtensor} to estimate derivatives of $f$.  Most of this section is standard and is recalled here for convenience.

\subsection{Vertical and Horizontal Lifts}
Given $(x,p)\in \TM$, $\gbar_{(x,p)}$ is defined by splitting $T_{(x,p)} \TM$ into it's so-called \emph{vertical} and \emph{horizontal} parts.  This is done using the connection of $g$ on $\mathcal{M}$.

Given $v\in T_x \mathcal{M}$, its vertical lift at $p\in T_x\mathcal{M}$, denoted $\ver_{(x,p)} (v)$ is defined to be the vector tangent to the curve $c_{(x,p),V} : (-\varepsilon,\varepsilon) \to T\mathcal{M}$ defined by,
\[
	c_{(x,p),V}(s) = (x,p + sv)
\]
at $s=0$,
\[
	\ver_{(x,p)} (v) = {c_{(x,p),V}}'(0).
\]

To define the horizontal lift of $v$ at $(x,p)$, first let $c:(-\varepsilon,\varepsilon) \to \mathcal{M}$ denote a curve in $\mathcal{M}$ such that $c(0) = x$, $c'(0) = v$.  Extend $p$ to a vector field along $c$ by parallel transport using the Levi-Civita connection of $g$ on $\mathcal{M}$,
\[
	\nabla_{c'} p = 0.
\]
The horizontal lift of $v$ at $(x,p)$, denoted $\hor_{(x,p)}(v)$, is then defined to be the tangent vector to the curve $c_{(x,p),H} : (-\varepsilon,\varepsilon) \to T\mathcal{M}$ defined by
\[
	c_{(x,p),H}(s) = (c(s),p),
\]
at $s=0$,
\[
	\hor_{(x,p)} (v) = {c_{(x,p),H}}'(0).
\]
It is straightforward to check this is independent of the particular curve $c$, as long as $c(0) = x, c'(0) = v$.

Given the coordinates $ p^1,\ldots, p^4$ on $T_x \mathcal{M}$ conjugate to $e_1, \ldots , e_4$, the double null frame on $\mathcal{M}$, one has a frame for $\TM$ given by $e_1, \ldots, e_4, \partial_{p^1} , \ldots , \partial_{p^4}$.  If $v\in T_x \mathcal{M}$ is written with respect to the double null frame as $v = v^{\mu} e_{\mu}$, then
\[
	\ver_{(x,p)}(v) = v^{\mu} \partial_{p^{\mu}},
\]
and
\[
	\hor_{(x,p)}(v) = v^{\mu} e_{\mu} - v^{\mu} p^{\nu} \Gamma^{\lambda}_{\mu \nu} \partial_{p^{\lambda}},
\]
where $\Gamma^{\lambda}_{\mu \nu}$ are the Ricci coefficients of the frame $e_1, \ldots, e_4$.

\begin{example} \label{ex:generator}
	The generator of the geodesic flow, $X$, at $(x,p) \in \TM$ is given by
	\[
		X = p^{\mu} e_{\mu} - p^{\mu} p^{\nu} \Gamma^{\lambda}_{\mu \nu} \partial_{p^{\lambda}} = \hor_{(x,p)} (p).
	\]
\end{example}

The vertical and horizontal subspaces of $T_{(x,p)}\TM$ are defined as,
\begin{align*}
	\mathcal{V}_{(x,p)} := & \ \ver_{(x,p)} (T_x\mathcal{M}) = \{ \ver_{(x,p)} (v) \mid v\in T_x\mathcal{M}\}, \\
	\mathcal{H}_{(x,p)} := & \ \hor_{(x,p)} (T_x\mathcal{M}) = \{ \hor_{(x,p)} (v) \mid v\in T_x\mathcal{M}\},
\end{align*}
respectively.  Note that $\mathcal{V}_{(x,p)}$ is just $T_{(x,p)}T_x \mathcal{M}$, the tangent space to the fibre of $T \mathcal{M}$.  One clearly has the following.

\begin{proposition}
	The tangent space to $\TM$ at $(x,p)$ can be written as the direct sum
	\[
		T_{(x,p)} \TM = \mathcal{V}_{(x,p)} \oplus \mathcal{H}_{(x,p)}.
	\]
\end{proposition}

Since each vector in $T_{(x,p)} T\mathcal{M}$ can be uniquely decomposed into its horizontal and vertical components, the following defines $\gbar$ on all pairs of vectors in $T_{(x,p)} T \mathcal{M}$.

\begin{definition}
	The \emph{Sasaki metric}, $\gbar$ on $\TM$ is defined as follows.  For $(x,p) \in T \mathcal{M}$ and $X,Y\in T_x \mathcal{M}$,
	\begin{align*}
		\gbar_{(x,p)} ( \hor_{(x,p)} (X), \hor_{(x,p)}(Y) ) = & \ g_x (X,Y) \\
		\gbar_{(x,p)} ( \hor_{(x,p)} (X), \ver_{(x,p)}(Y) ) = & \ 0 \\
		\gbar_{(x,p)} ( \ver_{(x,p)} (X), \ver_{(x,p)}(Y) ) = & \ g_x (X,Y).
	\end{align*}
\end{definition}

\subsection{The Connection and Curvature of the Sasaki Metric}
Since the Sasaki metric $\gbar$ is defined in terms of the metric $\gbar$ on $\mathcal{M}$, the connection and curvature of $\gbar$ can be computed in terms of the connection and curvature of $g$.  The computations are exactly the same as in Riemannian geometry.  See \cite{Ko}.

\begin{proposition} \label{prop:sasakiconnection}
	Let $\nablabar$ denote the Levi-Civita connection of the Sasaki metric $\gbar$.  Given $(x,p) \in \TM$ and vector fields $X,Y \in \Gamma (T \mathcal{M})$ on $\mathcal{M}$,
	\begin{enumerate}
		\item $\nablabar_{\hor_{(x,p)}(X)}\hor_{(x,p)}(Y) = \hor_{(x,p)}(\nabla_X Y) - \frac{1}{2} \ver_{(x,p)}(R_x(X,Y)p)$,
		\item $\nablabar_{\hor_{(x,p)}(X)}\ver_{(x,p)}(Y) = \ver_{(x,p)}(\nabla_X Y) + \frac{1}{2} \hor_{(x,p)}(R_x(p,Y)X)$,
		\item $\nablabar_{\ver_{(x,p)}(X)}\hor_{(x,p)}(Y) = \frac{1}{2} \hor_{(x,p)}(R_x(p,X)Y)$,
		\item $\nablabar_{\ver_{(x,p)}(X)}\ver_{(x,p)}(Y) = 0$,
	\end{enumerate}
	where $\nabla$ is the connection and $R$ is the curvature tensor of $(\mathcal{M},g)$.
\end{proposition}

\begin{proposition} \label{prop:sasakicurvature}
	Given $(x,p) \in \TM$ and vectors $X,Y,Z \in T_x \mathcal{M}$, then 
	\begin{multline*}
		\overline{R}_{(x,p)} (\hor_{(x,p)}(X), \hor_{(x,p)}(Y)) \hor_{(x,p)}(Z) = 
		\\
		\frac{1}{2} \ver_{(x,p)}\Big( (\nabla_Z R) (X,Y)p \Big)
		+ \hor_{(x,p)} \bigg( R(X,Y)Z + \frac{1}{4} R(p,R(Z,Y)p)X
		\\
		+ \frac{1}{4} R(p,R(X,Z)p)Y 
		+ \frac{1}{2} R(p,R(X,Y)p)Z \bigg),
	\end{multline*}
	and
	\begin{multline*}
		\overline{R}_{(x,p)} (\hor_{(x,p)}(X), \ver_{(x,p)}(Y)) \hor_{(x,p)}(Z) = 
		\\
		\ver_{(x,p)} \bigg( \frac{1}{2} R(X,Z)Y + \frac{1}{4} R(R(p,Y)Z,X)p \bigg)
		+ \frac{1}{2} \hor_{(x,p)}\Big( (\nabla_X R) (p,Y) Z \Big),
	\end{multline*}
	where $\overline{R}$ denotes the curvature tensor of $\gbar$, and $R$ the curvature tensor of $g$.
\end{proposition}

The proofs of Proposition \ref{prop:sasakiconnection} and Proposition \ref{prop:sasakicurvature} follow by direct computation.  See \cite{Ko} and also \cite{GuKa}.  The remaining components of $\overline{R}$ can be computed similarly but are not used here.

One important property of the Sasaki metric is the following.
\begin{proposition} \label{prop:geodesicflow}
	When equipped with the Sasaki metric, trajectories of the geodesic flow, $s\mapsto \exp_s(x,p)$, are geodesics in $\TM$.
\end{proposition}
\begin{proof}
	The tangent vector to a trajectory of the geodesic flow is given by the generator $X$.  As noted above, this is given at $(x,p)\in \TM$ by 
	\[
		X_{(x,p)} = p^{\mu} \hor_{(x,p)} (e_{\mu}).
	\]
	A trajectory of the geodesic flow takes the form $(\gamma(s),\dot{\gamma}(s))$ where $\gamma$ is a geodesic in $\mathcal{M}$.  Hence, by Proposition \ref{prop:sasakiconnection},
	\begin{align*}
		\nablabar_X X =
		& \
		X(\dot{\gamma}^{\mu}) \hor_{(\gamma,\dot{\gamma})} (e_{\mu}) + \dot{\gamma}^{\mu}\dot{\gamma}^{\nu} \nablabar_{\hor_{(\gamma,\dot{\gamma})} (e_{\nu})} \hor_{(\gamma,\dot{\gamma})} (e_{\mu})
		\\
		=
		& \ 
		- \dot{\gamma}^{\nu} \dot{\gamma}^{\lambda} \Gamma_{\nu \lambda}^{\mu} \hor_{(\gamma,\dot{\gamma})} (e_{\mu}) + \dot{\gamma}^{\mu}\dot{\gamma}^{\nu} \hor_{(\gamma,\dot{\gamma})} (\nabla_{e_{\nu}} e_{\mu})
		- \frac{1}{2} \ver_{(\gamma,\dot{\gamma})} (R(\dot{\gamma} ,\dot{\gamma} )\dot{\gamma} )
		\\
		=
		& \
		0.
	\end{align*}
\end{proof}

\subsection{Curvature of the Mass Shell}

\begin{proposition} \label{prop:curvaturemassshell}
	If $\hat{R}$ denotes the curvature of the mass shell $P$ then, if $(x,p) \in \TM$ and $X,Y,Z\in T_x\mathcal{M}$, the following formula for certain components of $\hat{R}$ are true.
	\begin{multline*}
		\hat{R}_{(x,p)} (\hor_{(x,p)}(X), \hor_{(x,p)}(Y)) \hor_{(x,p)}(Z) =
		\\
		\frac{1}{2} \ver_{(x,p)}\Big( (\nabla_Z R) (X,Y)p \Big)
		+ \hor_{(x,p)} \bigg( R(X,Y)Z + \frac{1}{4} R(p,R(Z,Y)p)X
		\\
		+ \frac{1}{4} R(p,R(X,Z)p)Y 
		+ \frac{1}{2} R(p,R(X,Y)p)Z \bigg),
	\end{multline*}
	and
	\begin{multline*}
		\hat{R}_{(x,p)} (\hor_{(x,p)}(X), \ver_{(x,p)}(Y)) \hor_{(x,p)}(Z) =
		\\
		\ver_{(x,p)} \bigg( \frac{1}{2} R(X,Z)Y + \frac{1}{4} R(R(p,Y)Z,X)p \bigg)
		\\
		+ \frac{1}{2} \hor_{(x,p)}\Big( (\nabla_X R) (p,Y) Z \Big)
		+ \frac{1}{4p^4} g( R(X,Z)Y,p) V,
	\end{multline*}
	where $V = \partial_{p^3}$ is transverse to the mass shell $P$.
\end{proposition}

\begin{proof}
	Throughout $N = \partial_{p^4} + \frac{p^3}{p^4} \partial_{p^3} + \frac{p^A}{p^4} \partial_{p^A}$ will denote the normal to the mass shell, $P$, such that $\overline{g} (N,V) = -2$.
	
	Each identity can be shown by first writing the curvature of $P$ in terms of the curvature of $\TM$, $N$ and $V$.  If $A,B,C \in \Gamma( T\TM)$ denote vector fields on $\TM$ then, since,
	\[
		\hat{\nabla}_A B = \overline{\nabla}_A B + \frac{1}{2} \overline{g} (\overline{\nabla}_A B, N) V,
	\]
	where $\hat{\nabla}$ is the induced connection on $P$, one easily deduces,
	\begin{align*}
		\hat{R}(A,B)C =
		& \
		\overline{R}(A,B)C + \frac{1}{2} \overline{g} (\overline{R}(A,B)C, N) V
		\\
		& \
		+ \frac{1}{2} \overline{g} (\overline{\nabla}_B C, N) \left(\overline{\nabla}_A V + \frac{1}{2} \overline{g} (\overline{\nabla}_A V, N) V \right)
		\\
		& \
		- \frac{1}{2} \overline{g} (\overline{\nabla}_A C, N) \left( \overline{\nabla}_B V + \frac{1}{2} \overline{g} (\overline{\nabla}_B V, N) V \right).
	\end{align*}
	
	To obtain the first identity note that, by Proposition \ref{prop:sasakiconnection},
	\[
		\overline{\nabla}_{\hor_{(x,p)}(Y)} \hor_{(x,p)}(Z) = \hor_{(x,p)} \left( \nabla_Y Z \right) - \frac{1}{2} \ver_{(x,p)} \left( R(Y,Z) p \right),
	\]
	and so
	\begin{align*}
		\overline{g}(\overline{\nabla}_{\hor_{(x,p)}(Y)} \hor_{(x,p)}(Z) , N) =
		& \
		- \frac{1}{2} \overline{g} \left( \ver_{(x,p)} \left( R(Y,Z) p \right) , N \right)
		\\
		=
		& \
		- \frac{1}{2p^4} g \left(  R(Y,Z) p , p \right)
		\\
		=
		& \
		0.
	\end{align*}
	Similarly
	\[
		\overline{g}(\overline{\nabla}_{\hor_{(x,p)}(X)} \hor_{(x,p)}(Z) , N) = 0.
	\]
	Finally, by Proposition \ref{prop:sasakicurvature},
	\begin{align*}
		\overline{g} \left( \overline{R} (\hor_{(x,p)}(X), \hor_{(x,p)}(Y)) \hor_{(x,p)}(Z) , N \right) = 
		\frac{1}{2 p^4} g( (\nabla_Z R) (X,Y)p, p) = 0.
	\end{align*}
	Hence
	\[
		\hat{R} (\hor_{(x,p)}(X), \hor_{(x,p)}(Y)) \hor_{(x,p)}(Z) = \overline{R} (\hor_{(x,p)}(X), \hor_{(x,p)}(Y)) \hor_{(x,p)}(Z),
	\]
	and the formula follows from Proposition \ref{prop:sasakicurvature}.
	
	For the second identity note that, as above,
	\[
		\overline{g} ( \overline{\nabla}_{\hor_{(x,p)}(X)} \hor_{(x,p)}(Z) , N) = 0,
	\]
	and that
	\[
		\overline{g} ( \overline{\nabla}_{\ver_{(x,p)}(Y)} \hor_{(x,p)}(Z) , N) = 0,
	\]
	since, by Proposition \ref{prop:sasakiconnection}, $\overline{\nabla}_{\ver_{(x,p)}(Y)} \hor_{(x,p)}(Z)$ is horizontal.  The result follows from Proposition \ref{prop:sasakicurvature}.
\end{proof}

\subsection{Derivatives of the Exponential Map}
Recall the definition of the exponential map (or geodesic flow) for $(x,p) \in \TM$,
\[
	\exp_s(x,p) = (\gamma_{x,p}(s), \dot{\gamma}_{x,p}(s)),
\]
where $\gamma_{x,p}$ is the unique geodesic in $\mathcal{M}$ such that $\gamma_{x,p}(0) = x$, $\dot{\gamma}_{x,p}(0) = p$.

Derivatives of the particle density function $f$ are estimated using the fact that derivatives of the exponential map are Jacobi fields as follows.  Consider $(x,p) \in \TM$ and $V\in T_{(x,p)}\TM$.  Using the Vlasov equation,
\[
	f(x,p) = f(\exp_s(x,p)),
\]
and the chain rule one obtains,
\[
	V(f)(x,p) = d f\mid_{(x,p)} (V) = d f\mid_{\exp_s(x,p)} \cdot d \exp_s \mid_{(x,p)} (V) = J(f) (\exp_s(x,p)).
\]
By Proposition \ref{prop:geodesicflow}, $s \mapsto \exp_s(x,p)$ is a geodesic in $P$ (or in $\TM$).  Below it will be shown that $J := d \exp_s \mid_{(x,p)} (V)$ is a Jacobi field along this geodesic, and moreover $J(0)$ and $(\hat{\nabla}_XJ)(0)$ are computed.  By taking $s<0$ so that $\exp_s(x,p)$ lies on the initial hypersurface $\{ u = u_0\}$, this then gives an expression for $V(f)$ in terms of initial data which can be estimated using the Jacobi equation.  In practice it is convenient to split $V$ into its horizontal and vertical parts.

\begin{proposition} \label{prop:jacobi}
	If $v\in T_x \mathcal{M}$ and $H = \hor_{(x,p)}(v) \in \mathcal{H}_{(x,p)} \subset T_{(x,p)} \TM$ is the horizontal lift of $v$, then $d \exp_s \mid_{(x,p)} (H)$ is a Jacobi field, $J_H$, along $\exp_s(x,p)$ such that 
	\[
		J_H\mid_{s=0} = H, \qquad \text{and} \qquad \hat{\nabla}_X J_H\mid_{s=0} = \frac{1}{2} \ver_{(x,p)} (R(p,v)p).
	\]
	If $V_1 = \ver_{(x,p)}(v) \in \mathcal{V}_{(x,p)} \subset T_{(x,p)} \TM$ is the vertical lift of $v$, then $d \exp_s \mid_{(x,p)} (V_1)$ is a Jacobi field, $J_{V_1}$, along $\exp_s(x,p)$ such that 
	\[
		J_{V_1}\mid_{s=0} = V_1, \qquad \text{and} \qquad \hat{\nabla}_X J_{V_1}\mid_{s=0} = \hor_{(x,p)}(v) + \frac{1}{2} \hor_{(x,p)} (R(p,v)p).
	\]
	Here $X$, the generator of the geodesic flow, is tangent to the curve $s\mapsto \exp_s(x,p)$.
\end{proposition}

\begin{proof}
	Let $c_H:(-\varepsilon,\varepsilon) \to \TM$ be a curve in $\TM$ such that $c_H(0) = (x,p)$, $c_H ' (0) = H$.  Then, by Proposition \ref{prop:geodesicflow}, $(s,s_1) \mapsto \exp_s(c_H(s_1))$ defines a variation of geodesics.  Hence
	\[
		d\exp_s\mid_{(x,p)} (H) = \frac{d}{ds_1} \left( \exp_s(c_H(s_1)) \right) \bigg\vert_{s_1 = 0},
	\]
	is a Jacobi field along $\exp_s(x,p)$.
	
	Since $\exp_0(x,p) = (x,p)$ is the identity map,
	\[
		J_H\mid_{s=0} = d\exp_0\mid_{(x,p)} (H) = H.
	\]
	Now,
	\begin{align*}
		\hat{\nabla}_X J_H \mid_{s=0} 
		=
		& \
		\frac{\hat{\nabla}}{ds} \frac{\partial}{\partial s_1} \left( \exp_s(c_H(s_1)) \right) \Big\vert_{s=0,s_1=0}
		\\
		=
		& \
		\frac{\hat{\nabla}}{ds_1} \frac{\partial}{\partial s} \left( \exp_s(c_H(s_1)) \right) \Big\vert_{s=0,s_1=0}
		\\
		=
		& \
		\hat{\nabla}_H X
		\\
		=
		& \
		H(p^{\mu})\hor_{(x,p)}(e_{\mu}) + p^{\mu}\hat{\nabla}_H \hor_{(x,p)}(e_{\mu})
		\\
		=
		& \
		- v^{\mu}p^{\nu} \Gamma^{\lambda}_{\mu \nu} \hor_{(x,p)}(e_{\mu}) 
		\\
		& \
		+ p^{\mu}\left( \overline{\nabla}_H \hor_{(x,p)}(e_{\mu}) + \frac{1}{2} \overline{g}(\overline{\nabla}_H \hor_{(x,p)}(e_{\mu}),N)V \right)
		\\
		=
		& \
		- v^{\mu}p^{\nu} \Gamma^{\lambda}_{\mu \nu} \hor_{(x,p)}(e_{\mu}) + p^{\mu} \hor_{(x,p)} (\nabla_v e_{\mu})
		\\
		& \
		- \frac{1}{2} p^{\mu} \ver_{(x,p)} (R(v, e_{\mu}) p) - \frac{p^{\mu}}{4p^4} g (R(v, e_{\mu})p,p)
		\\
		=
		& \
		\frac{1}{2} \ver_{(x,p)} (R(p,v) p),
	\end{align*}
	by Proposition \ref{prop:sasakiconnection}.
	
	Similarly, if $c_{V_1} : (-\varepsilon,\varepsilon) \to \TM$ is a curve such that $c_{V_1}(0) = (x,p), c_{V_1}'(0) = V_1$ (for example $c_{V_1}(s_1) = (x, p + s_1 v)$), then $(s,s_1) \mapsto \exp_s(c_{V_1}(s_1))$ is again a variation of geodesics.  Therefore
	\[
		d\exp_s\mid_{(x,p)} (V_1) = \frac{d}{ds_1} \left( \exp_s(c_{V_1}(s_1)) \right) \bigg\vert_{s_1 = 0},
	\]
	is again a Jacobi field.  Clearly, as above,
	\[
		J_{V_1}\mid_{s=0} = d\exp_0\mid_{(x,p)} (V_1) = V_1.
	\]
	The first derivative can again be computed, using Proposition \ref{prop:sasakiconnection}, as follows.
	\begin{align*}
		\hat{\nabla}_X J_{V_1} \mid_{s=0} 
		=
		& \
		\frac{\hat{\nabla}}{ds} \frac{\partial}{\partial s_1} \left( \exp_s(c_{V_1}(s_1)) \right) \Big\vert_{s=0,s_1=0}
		\\
		=
		& \
		\frac{\hat{\nabla}}{ds_1} \frac{\partial}{\partial s} \left( \exp_s(c_{V_1}(s_1)) \right) \Big\vert_{s=0,s_1=0}
		\\
		=
		& \
		\hat{\nabla}_{V_1} X
		\\
		=
		& \
		V_1(p^{\mu}) \hor_{(x,p)} (e_{\mu}) + p^{\mu} \hat{\nabla}_{V_1} \hor_{(x,p)} (e_{\mu})
		\\
		=
		& \
		\hor_{(x,p)} (v) + p^{\mu} \left( \overline{\nabla}_{V_1} \hor_{(x,p)} (e_{\mu}) + \frac{1}{2} \overline{g}(\overline{\nabla}_{V_1} \hor_{(x,p)} (e_{\mu}), N)V \right)
		\\
		=
		& \
		\hor_{(x,p)} (v) + \frac{1}{2} \hor_{(x,p)}(R(p,v)p).
	\end{align*}
\end{proof}

\section{The Main Theorem and Bootstrap Assumptions} \label{section:ba}

\subsection{Characteristic Initial Data}
In Theorem \ref{thm:main3} below, characteristic initial data, prescribed on the hypersurfaces $\{v=v_0\}$, $\{ u = u_0\}$, satisfying a certain smallness condition is considered.  Of course, in the setting of Theorem \ref{thm:main}, such data arises as induced data on two transversely intersecting null hypersurfaces, whose existence is guaranteed by a Cauchy stability argument and an application of a result of Klainerman--Nicol\`{o} \cite{KlNi} on the vacuum equations \eqref{eq:Einsteinvacuum}.  See Section \ref{subsec:vacuumappeal} and Section \ref{subsec:Cauchystability} where this argument is discussed.  Characteristic initial data for Theorem \ref{thm:main3} can, however, be prescribed independently of the setting of Theorem \ref{thm:main}.  Suppose ``free data'', consisting of a ``seed'' $S_{u,v}$--tensor density of weight $-1$, $\hat{\gslash}^{v_0}$, on $[u_0,u_f] \times S^2$, a ``seed'' $S_{u,v}$--tensor density of weight $-1$, $\hat{\gslash}^{u_0}$, on $[v_0,\infty ) \times S^2$, and a compactly supported function $f_0 : P\vert_{\{ v= v_0\}} \to [0,\infty)$, along with certain quantities on the sphere of intersection $S_{u_0,v_0}$, are given.  Here $P\vert_{\{ v= v_0\}}$ denotes the mass shell over the initial hypersurfaces $\{ v = v_0\}$.  The characteristic constraint equations for the system \eqref{eq:EV}--\eqref{eq:vlas} take the form of ordinary differential equations and can be integrated to give all of the geometric quantities $\Gamma, \psi, \mathcal{T}$, along with their derivatives, on $\{v=v_0\}$ and $\{u=u_0\}$ once the above ``free data'' is prescribed.  These geometric quantities, on $\{v=v_0\}$ and $\{u=u_0\}$,  are what is referred to as the ``characteristic initial data'' in the statement of Theorem \ref{thm:main3}.  Appropriate smallness conditions can be made for the ``free data'' (and their derivatives), along with appropriate decay conditions for the seed $\hat{\gslash}^{u_0}$, in order to ensure the conditions of Theorem \ref{thm:main3} are met.

The prescription of such characteristic ``free data'', and the determining of the geometric quantities from them, will not be discussed further here.  The interested reader is directed to \cite{Ch}, where this is discussed in great detail in a related setting.  See also \cite{DHR}.

\subsection{The Main Existence Theorem}
Define the norms
\begin{align*}
	F^1_{v_0,v}(u) 
	&
	:= 
	\sum_{k=0}^3 \sum_{\mathfrak{D}^k} \int_{v_0}^v \int_{S_{u,v'}}  r^5 \vert \mathfrak{D}^k\alpha \vert^2 + r^4 \vert \mathfrak{D}^k\beta\vert^2 + r^2 (\vert \mathfrak{D}^k\rho \vert^2 + \vert \mathfrak{D}^k\sigma \vert^2)
	+ \vert \mathfrak{D}^k\betabar \vert^2 d \mu_{S_{u,v'}} d v' , 
	\\
	F^2_{u_0,u}(v)
	&
	:= 
	\sum_{k=0}^3 \sum_{\mathfrak{D}^k} \int_{u_0}^u \int_{S_{u',v}}  r^5 \vert \mathfrak{D}^k\beta\vert^2 + r^4 (\vert \mathfrak{D}^k\rho \vert^2 + \vert \mathfrak{D}^k\sigma \vert^2) + r^2 \vert \mathfrak{D}^k\betabar \vert^2
	+ \vert \alphabar \vert^2 d \mu_{S_{u',v}} \Omega^2 d u',
\end{align*}
where the second summation is taken over $\mathfrak{D}^k \in \{ \nablaslash_3, r\nablaslash_4, r\nablaslash\}^k$ and $d \mu_{S_{u,v}}$ denotes the volume measure on $S_{u,v}$.

The main theorem can now be stated more precisely as follows.
\begin{theorem} \label{thm:main3}
	There exists a $v_0$ large and an $\bc$ small such that the following holds.
	
	Given smooth characteristic initial data for the massless Einstein--Vlasov system \eqref{eq:EV}--\eqref{eq:vlas} on the characteristic initial hypersurfaces $\{v = v_0\}$, $\{ u = u_0\}$, suppose the data on $\{v=v_0\}$ satisfy
	\begin{align*}
		F^2_{u_0,u_f} (v_0) + \sum_{k\leq3} \int_{u_0}^{u_f} \int_{S_{u',v_0}} r^6 \vert \mathfrak{D}^k \alpha \vert^2 d \mu_{S_{u',v_0}} du' 
		+
		\sum_{k\leq 1} \sup_{u_0\leq u\leq u_f} \Vert r^{\frac{7}{2}} \mathfrak{D}^k \alpha \Vert_{L^4(S_{u,v_0})}
		&
		< \bc,
		\\
		\sum_{\Riccif} \sum_{k\leq 3} \sup_{u_0 \leq u \leq u_f} r(u,v_0)^{2p-2} \int_{S_{u,v_0}} \vert \mathfrak{D}^k \Riccif \vert^2 d\mu_{S_{u,v_0}} 
		&
		< \bc,
		\\
		\sup_{\{v=v_0\}} \sum_{k = 0}^4 \sum_{i_1,\ldots,i_k = 1}^7 \vert \tilde{E}_{i_1} \ldots \tilde{E}_{i_k} f \vert 
		&
		< \bc,
		\\
		\sum_{k\leq 2} \sum_{\Thetaf_p} \sup_{u_0 \leq u \leq u_f} r(u,v_0)^{2p-2} \int_{S_{u,v_0}} \vert \mathfrak{D}^k \Thetaf_p\vert^2 d \mu_{S_{u,v_0}} 
		+ \sum_{\Thetaf_p} \int_{u_0}^{u_f} \int_{S_{u',v_0}} r^{2p-2} \vert \mathfrak{D}^3 \Thetaf_p \vert^2 d\mu_{S_{u',v_0}} du' 
		&
		< \bc,
		\\
		\int_{u_0}^{u_f} \int_{S_{u',v_0}} \left\vert (r\nablaslash)^3 \left( \Gammaslash - \Gammaslash^{\circ} \right) \right\vert^2 d\mu_{S_{u',v_0}} du' 
		&
		< \bc,
	\end{align*}
	and the data on $\{u = u_0\}$ satisfy,
	\begin{align*}
		F^1_{v_0,\infty} (u_0) + \sum_{k\leq3} \int_{v_0}^{\infty} \int_{S_{u_0,v'}} r^{-2} \vert \mathfrak{D}^k \alphabar \vert^2 d \mu_{S_{u_0,v'}} dv' 
		&
		< \bc,
		\\
		\sum_{k\leq 1} \sum_{\psi_p \neq \alpha} \sup_{v_0\leq v < \infty} \Vert r^{p - \frac{1}{2}} \mathfrak{D}^k \psi_p \Vert_{L^4(S_{u,v_0})}
		&
		< \bc,
		\\
		\sum_{\Riccit} \sum_{k\leq 3} \sup_{v_0 \leq v < \infty} r(u_0,v)^{2p-2} \int_{S_{u_0,v}} \vert \mathfrak{D}^k \Riccit \vert^2 d\mu_{S_{u_0,v}} 
		&
		< \bc,
		\\
		f \vert_{\pi^{-1}(\{u=u_0\})} 
		&
		= 0,
		\\
		\sum_{k\leq 2} \sum_{\Thetat_p} \sup_{v_0 \leq v < \infty} r(u_0,v)^{2p-2} \int_{S_{u_0,v}} \vert \mathfrak{D}^k \Thetat_p\vert^2 d \mu_{S_{u_0,v}}
		&
		\\
		+
		\sum_{k\leq 1} \sup_{v_0 \leq v < \infty} r(u_0,v)^{6} \int_{S_{u_0,v}} \vert \mathfrak{D}^k \kappa \vert^2 + \vert \mathfrak{D}^k r\nablaslash \kappa \vert^2 d \mu_{S_{u_0,v}}
		&
		\\
		+ \int_{v_0}^{\infty} \int_{S_{u_0,v'}} \sum_{\Thetat_p \neq \kappa}  r^{2p-2} \vert \mathfrak{D}^3 \Thetat_p \vert^2 + r^4 \vert \mathfrak{D}^2 r\nablaslash \kappa \vert^2 d\mu_{S_{u',v_0}} du'
		&
		< \bc,
		\\
		\sum_{k\leq 2} \sup_{v_0 \leq v < \infty} \int_{S_{u_0,v}} \left\vert \mathfrak{D}^k \left( \Gammaslash - \Gammaslash^{\circ} \right) \right\vert^2 d \mu_{S_{u_0,v}} 
		&
		\\
		+ \int_{v_0}^v  \int_{S_{u_0,v'}} r^{-2} \left\vert \mathfrak{D}^3 r\nablaslash b \right\vert^2 d \mu_{S_{u_0,v'}} 
		&
		< \bc,
	\end{align*}
	Here $\kappa$ and the $\Theta$ variables are defined as certain combinations of Ricci coefficients and Weyl curvature components in Section \ref{section:Riccitop}, and $\tilde{E}_1,\ldots,\tilde{E}_7$ is a frame for $P$ defined by,
\[
	\tilde{E}_i = E_i \text{ for } i = 1,\ldots, 4, \qquad \tilde{E}_i = p^4 E_i \text{ for } i = 5,6,7,
\]
where $E_1,\ldots,E_7$ is a frame for $P$ defined in Section \ref{subsec:fjacobi}.  Suppose also that,
	\[
		0\leq p^4 \leq C_{p^4}, 
		\quad 
		0 \leq r^2 p^3 \leq C_{p^3} p^4, 
		\quad 
		\vert r^2 p^A \vert \leq C_{p^A} p^4 \quad \text{ for } A=1,2,
	\]
	in $\supp f \vert_{P\vert_{\{ v = v_0 \}}}$, for some fixed constants $C_{p^1}, \ldots, C_{p^4}$ independent of $v_0$, and that, in each of the two spherical coordinate charts, the components of the metric satisfy,
	\[
		\left\vert \gslash_{AB} - \gslash_{AB}^{\circ} \right\vert \leq C r, \qquad \left\vert \gslash^{AB} - \gslash^{AB}_{\circ} \right\vert \leq \frac{C}{r^3},
	\]
	for some constant $C$ uniformly on the initial hypersurfaces $\{ u = u_0 \}, \{ v = v_0\}$.
	
	Then there exists a unique solution of the Einstein--Vlasov system \eqref{eq:EV}--\eqref{eq:vlas} on $\mathcal{M} = [u_0,u_f] \times [v_0,\infty) \times S^2$, attaining the data on $\{ u = u_0\}$, $\{v = v_0\}$, such that
	\begin{align*}
		&
		\sup_{u,v} \Bigg(
		F^1_{v_0,v}(u) + F^2_{u_0,u}(v)
		+ \sum_{k\leq 3} \sum_{\Gamma_p} r(u,v)^{2p-2} \int_{S_{u,v}} \vert \mathfrak{D}^k \Gamma_p \vert^2 d\mu_{S_{u,v}}
		\\
		&
		+ \sum_{k\leq2} \sum_{\mathcal{T}_p} r(u,v)^p \vert \mathfrak{D}^k \mathcal{T}_p \vert
		+ \sum_{k\leq 3} \sum_{\mathcal{T}_p} \Bigg( \int_{u_0}^u \int_{S_{u',v}} r(u',v)^{2p-2} \vert \mathfrak{D}^k \mathcal{T}_p \vert^2 d\mu_{S_{u',v}} du' 
		\\
		&
		+ \int_{v_0}^v \int_{S_{u,v'}} r(u,v')^{2p-4} \vert \mathfrak{D}^k \mathcal{T}_p \vert^2 d\mu_{S_{u,v'}} dv'  \Bigg)
		\\
		&
		+ \sum_{k\leq 4} \sum_{\mathcal{T}_p} \int_{u_0}^u \int_{v_0}^v \int_{S_{u',v'}} r(u',v')^{2p-4} \vert \mathfrak{D}^k \mathcal{T}_p \vert^2 d\mu_{S_{u',v'}} dv' du' 
		\Bigg)
		\leq \overline{C},
	\end{align*}
	where $\overline{C}$ is a constant which can be made arbitrarily small provided $\bc$ and $\frac{1}{v_0}$ are taken sufficiently small.  Moreover one also has explicit decay rates for the size of $\supp f\vert_{P_x} \subset P_x$ as $v(x) \to \infty$ and explicit bounds on weighted $L^2$ norms of the $\Theta$ variables.  See Section \ref{section:suppf} and Section \ref{section:Riccitop} respectively.  Finally, if $u_f$ was chosen sufficiently large, $f=0$ on the mass shell over any point $x\in \mathcal{M}$ such that $u(x) \geq u_f - 1$.
\end{theorem}

The $L^4$ norms of the Weyl curvature components are required for the Sobolev inequalities on the null hypersurfaces.  See Section \ref{section:Sobolev}.

\subsection{Bootstrap Assumptions}

The proof of Theorem \ref{thm:main3} is obtained through a bootstrap argument, so consider the following bootstrap assumptions for Ricci coefficients
\begin{equation} \label{eq:Ricciba}
	r^{2p-2} \int_{S_{u,v}} \vert \mathfrak{D}^k \Gamma_p \vert^2 d \mu_{S_{u,v}} \leq \overline{C},
\end{equation}
for $k = 0,1,2$, the spherical Christoffel symbols,
\begin{align} \label{eq:Gammaslashba}
\begin{split}
	\int_{S_{u,v}} \vert \mathfrak{D}^{k} \left( \Gammaslash - \Gammaslash^{\circ} \right) \vert^2 d \mu_{S_{u,v}} +
	\int_{S_{u,v}} \vert \mathfrak{D}^{k} \nablaslash_3 \left( \Gammaslash - \Gammaslash^{\circ} \right) \vert^2 d \mu_{S_{u,v}}
	&
	\\
	+ \int_{S_{u,v}} \vert \mathfrak{D}^{k} r\nablaslash_4 \left( \Gammaslash - \Gammaslash^{\circ} \right) \vert^2 d \mu_{S_{u,v}}
	&
	\leq \overline{C},
\end{split}
\end{align}
for $k = 0,1, 2$, for Weyl curvature components
\begin{equation} \label{eq:curvatureba}
	F^1_{v_0,v}(u) \leq \overline{C}, \qquad F^2_{u_0,u}(v) \leq \overline{C},
\end{equation}
and for the energy momentum tensor components,
\begin{align} \label{eq:emba}
\begin{split}
	&
	\sum_{k\leq 2} r^p \vert \mathfrak{D} \mathcal{T}_p \vert + \sum_{k\leq 3} \Bigg( \int_{u_0}^u \int_{S_{u',v}} r(u',v)^{2p-2} \vert \mathfrak{D}^k \mathcal{T}_p \vert^2 d\mu_{S_{u',v}} du' 
	\\
	&
	+ \int_{v_0}^v \int_{S_{u,v'}} r(u,v')^{2p-4} \vert \mathfrak{D}^k \mathcal{T}_p \vert^2 d\mu_{S_{u,v'}} dv'  \Bigg)
	\\
	&
	+ \sum_{k\leq 4} \sum_{\mathcal{T}_p} \int_{u_0}^u \int_{v_0}^v \int_{S_{u',v'}} r(u',v')^{2p-4} \vert \mathfrak{D}^k \mathcal{T}_p \vert^2 d\mu_{S_{u',v'}} dv' du' 
	\leq \overline{C},
\end{split}
\end{align}
where $\overline{C}$ is some small constant\footnote{In Sections \ref{section:Sobolev} and \ref{section:emtensor} smallness conditions on $\overline{C}$ will be made, but it is otherwise arbitrary.}.  Moreover, since a derivative of $b$ appears in the expression for $\nablaslash_4 e_A$, consider also the bootstrap assumption for an additional derivative of $b$,
\begin{equation} \label{eq:bba}
	\int_{v_0}^v \int_{S_{u,v'}} r^{-2} \vert \mathfrak{D}^3 (r \nablaslash) b \vert^2 d \mu_{S_{u,v'}} dv' \leq \overline{C},
\end{equation}
and also for $\Gammaslash - \Gammaslash^{\circ}$ at the top order,
\begin{equation} \label{eq:Gammaslashbanull}
	 \int_{u_0}^u \int_{S_{u',v}} \left\vert (r \nablaslash)^3 \left( \Gammaslash - \Gammaslash^{\circ} \right) \right\vert^2 d \mu_{S_{u',v}} du' \leq \overline{C}.
\end{equation}

Recall that $\Gammaslash - \Gammaslash^{\circ}$ is a geometric object, an $S_{u,v}$ $(1,2)$ tensor, and so its covariant derivatives are well defined.

Note that, since the volume form of $S_{u,v}$ grows like $r^2$, \eqref{eq:Ricciba} is consistent with the expectation that $\Gamma_p$ behaves like $\frac{1}{r^p}$.  Moreover, \eqref{eq:Gammaslashba}, \eqref{eq:Gammaslashbanull} is consistent with the expectation that $\vert \Gammaslash - \Gammaslash^{\circ} \vert$ decays like $\frac{1}{r}$, or equivalently (by Proposition \ref{prop:gslash} below) that the components $\Gammaslash^C_{AB} - {\Gammaslash^{\circ}}_{AB}^C$ behave like 1 with respect to $r$.  Since the ${\Gammaslash^{\circ}}_{AB}^C$ behave like 1, this implies the components ${\Gammaslash}_{AB}^C$ also behave like 1 and hence that $r\vert \Gammaslash \vert \leq C$ where,
\[
	\vert \Gammaslash \vert^2 = \gslash^{A A'} \gslash^{B B'} \gslash_{C C'} \Gammaslash^C_{AB} \Gammaslash^{C'}_{A'B'}.
\]
These pointwise bounds for lower order derivatives are derived from the bootstrap assumptions \eqref{eq:Gammaslashba}, \eqref{eq:Gammaslashbanull} via Sobolev inequalities in Section \ref{section:Sobolev}.  The covariant derivatives of $\Gammaslash$ are defined in each coordinate system as,
\begin{align*}
	\nablaslash_3 \Gammaslash^C_{AB} =
	& \ 
	e_3 \left( \Gammaslash^C_{AB} \right) - {\chibar_A}^D \Gammaslash_{DB}^C - {\chibar_B}^D \Gammaslash_{AD}^C + {\chibar_D}^C \Gammaslash_{AB}^D,
	\\
	\nablaslash_4 \Gammaslash^C_{AB} =
	& \ 
	e_4 \left( \Gammaslash^C_{AB} \right) - \left( {\chi_A}^D - \nablaslash_A b^D + b^E \Gammaslash_{AE}^D \right) \Gammaslash_{DB}^C
	\\
	& 
	- \left( {\chi_B}^D - \nablaslash_B b^D + b^E \Gammaslash_{BE}^D \right) \Gammaslash_{AD}^C 
	+ \left( {\chi_D}^C - \nablaslash_D b^C + b^E \Gammaslash_{DE}^C \right) \Gammaslash_{AB}^D,
	\\
	\nablaslash_D \Gammaslash^C_{AB} =
	& \ 
	e_D \left( \Gammaslash^C_{AB} \right) -\Gammaslash_{DA}^E \Gammaslash_{EB}^C - \Gammaslash_{DB}^E \Gammaslash_{AE}^C + \Gammaslash_{DE}^C \Gammaslash_{AB}^E,
\end{align*}

Finally, note also that \eqref{eq:bba} is consistent with $b = \Gamma_1$.  Since $b$ is only estimated on an outgoing null hypersurface at the top order though, the Sobolev Inequalities of the next section only allow us to conclude
\[
	 r^{\frac{1}{2}} \vert \mathfrak{D} r \nablaslash b \vert \leq C,
\]
unlike at lower orders where the Sobolev inequalities will give,
\[
	 r\vert b \vert, r \vert \mathfrak{D} b \vert \leq C.
\]
Here and throughout the remainder of the paper $C$ will denote a numerical constant which can change from line to line.

\subsection{The Bootstrap Theorem}

Theorem \ref{thm:main3} will follow from the following bootstrap theorem, Theorem \ref{thm:main4}, via a \emph{last slice} argument.

\begin{theorem} \label{thm:main4}
	There exist $\bc$, $\overline{C}$ small and $v_0$ large such that the following is true.  Given initial data satisfying the restrictions of Theorem \ref{thm:main3}, let $\mathcal{A}$  denote a characteristic rectangle of the form $\mathcal{A} = [u_0,u']\times [v_0, v'] \times S^2 \subset \mathcal{M}$, with $u_0 < u' \leq u_f$, $v_0 < v' <\infty$, such that a solution to the Einstein--Vlasov system \eqref{eq:EV}--\eqref{eq:vlas}, attaining the given data, exists in $\mathcal{A}$ and, for any $x \in \mathcal{A}$, the bootstrap assumptions \eqref{eq:Ricciba}--\eqref{eq:Gammaslashbanull} hold for $(u,v) = (u(x),v(x))$.
	
	If $x\in \mathcal{A}$, then the bounds \eqref{eq:Ricciba}--\eqref{eq:Gammaslashbanull} in fact hold at $x$ with the constant $\overline{C}$ replaced by $\frac{\overline{C}}{2}$.
\end{theorem}

Sections \ref{section:Sobolev}--\ref{section:Riccitop} are devoted to the proof of Theorem \ref{thm:main4}, which follows from Propositions \ref{prop:tmain}, \ref{prop:curvmain}, \ref{prop:Ricci4main}, \ref{prop:Ricci3main}, \ref{prop:Gammaslashmain}, \ref{prop:bmain}.  The proof of Theorem \ref{thm:main3}, using a last slice argument, is outlined in Section \ref{section:lastslice}.

\section{Sobolev Inequalities} \label{section:Sobolev}
The Sobolev inequalities shown in this section will allow one to obtain $L^{\infty}$ estimates on the spheres for quantities through the $L^2$ bootstrap estimates \eqref{eq:Ricciba}--\eqref{eq:Gammaslashbanull}.  They are shown to hold in the setting of Theorem \ref{thm:main4}, i.\@ e.\@ for $x \in \mathcal{A}$, and are derived from the \emph{isoperimetric inequality} for each sphere $S_{u,v}$: if $f$ is a function which is integrable on $S_{u,v}$ with integrable derivative, then $f$ is square integrable and
\begin{equation} \label{eq:iso}
	\int_{S_{u,v}} \left( f - \overline{f} \right)^2 d \mu_{S_{u,v}} \leq I(S_{u,v}) \left( \int_{S_{u,v}} \vert \nablaslash f \vert d \mu_{S_{u,v}} \right)^2,
\end{equation}
where
\[
	\overline{f} := \frac{1}{\area (S_{u,v})} \int_{S_{u,v}} f d \mu_{S_{u,v}},
\]
denotes the average of $f$ on $S_{u,v}$, and $I(S_{u,v})$ denotes the isoperimetric constant of $S_{u,v}$:
\[
	I(S_{u,v}) := \sup_{U} \frac{\min \{\area(U),\area(U^c)\}}{\left( \mathrm{Perimeter} (\partial U) \right)^2},
\]
where the supremum is taken over all domains $U \subset S_{u,v}$ with $C^1$ boundary $\partial U$ in $S_{u,v}$.

The following Sobolev inequalities are standard, see e.\@g.\@ Chapter 5.2 of \cite{Ch}.
\begin{lemma} \label{lem:sob1}
	Given a compact Riemannian manifold $(S,\gslash)$, let $\xi$ be a tensor field on $S$ which is square integrable with square integrable first covariant derivative.  Then $\xi \in L^4(S)$ and
	\[
		\frac{1}{(\area(S))^{1/4}} \left( \int_S \vert \xi \vert^4 d\mu_S \right)^{1/4} \leq C \sqrt{I'(S)} \left( \int_S \vert \nablaslash \xi \vert^2 + \frac{1}{\area (S)} \vert \xi\vert^2 d\mu_S \right)^{1/2},
	\]
	where $I'(S) := \max \{ I(S), 1\}$, and $C$ is a numerical constant independent of $(S,\gslash)$.
\end{lemma}

\begin{lemma} \label{lem:sob2}
	If $\xi$ is a tensor field on $S$ such that $\xi,\nablaslash \xi \in L^4(S)$, then
	\[
		\sup_S \vert \xi \vert \leq C \sqrt{I'(S)} (\area(S))^{1/4} \left( \int_S \vert \nablaslash \xi \vert^4 + \frac{1}{(\area (S))^{2}} \vert \xi\vert^4 d\mu_S \right)^{1/4},
	\]
	where again $C$ is independent of $(S,\gslash)$.
\end{lemma}

Under the assumption that the components of $\gslash$ satisfy,
\[
	\left\vert \gslash_{AB} - \gslash^{\circ}_{AB} \right\vert \leq C r,
\]
for some constant $C>0$.  It follows that 
\[
	\vert \area(S_{u,v}, \gslash) - 4 \pi r^2 \vert = \vert \area(S_{u,v}, \gslash) - \area(S_{u,v}, \gslash^{\circ}) \vert < C r
\]
and hence there exist constants $c,C >0$ such that
\begin{equation} \label{eq:arearadius}
	cr(u,v) \leq \sqrt{\area(S_{u,v})} \leq Cr(u,v).
\end{equation}
Using this fact, Lemma \ref{lem:sob1} and Lemma \ref{lem:sob2} can be combined to give
\begin{equation} \label{eq:sobmain}
	r \Vert \xi \Vert_{L^{\infty}(S_{u,v})} \leq C I'(S_{u,v}) (\Vert r\nablaslash r\nablaslash \xi \Vert_{L^2(S_{u,v})} + \Vert r\nablaslash \xi \Vert_{L^2(S_{u,v})} + \Vert \xi \Vert_{L^2(S_{u,v})} ).
\end{equation}

Thus, in order to gain global pointwise control over the Ricci coefficients, curvature components and energy momentum tensor components, it remains to gain control over the isoperimetric constants $I(S_{u,v})$.  We first show the above bounds on the components of $\gslash$ hold under the following bootstrap assumptions.  Let $\mathcal{A}' \subset \mathcal{A}$ be the set of points $x \in \mathcal{A}$ such that,
\begin{align}
	\Omega \leq C_0, \label{eq:auxba1}
	\\
	\left\vert \tr \chibar + \frac{2}{r} \right\vert \leq \frac{C_0}{r^2}, \label{eq:auxba2}
	\\
	\left\vert \hat{\chibar} \right\vert \leq \frac{C_0}{r}, \label{eq:auxba3}
\end{align}
for some constant $C_0$, for all points $y \in \mathcal{A}$ such that $u(y) \leq u(x)$, $v(y) \leq v(x)$.

\begin{proposition} \label{prop:gslash}
	If $x\in \mathcal{A}'$ then, in each of the two spherical charts defined in Section \ref{subsec:coords},
	\[
		\left\vert \gslash_{AB} - \gslash^{\circ}_{AB} \right\vert \leq Cr,
		\quad \text{and} \quad
		\left\vert \gslash^{AB} - \gslash_{\circ}^{AB} \right\vert \leq \frac{C}{r^3},
	\]
	at $x$.  In particular,
	\[
		\left\vert \gslash_{AB} \right\vert \leq Cr^2,
		\quad \text{and} \quad
		\left\vert \gslash^{AB} \right\vert \leq \frac{C}{r^2}.
	\]
\end{proposition}

\begin{proof}
	Recall the first variation formula \eqref{eq:firstvar3} which implies that,
	\begin{align*}
		\partial_u \left( r^{-2} \gslash_{AB} \right)
		=
		\frac{2}{r^{3}} \gslash_{AB} + \frac{2\Omega}{r^2} \chibar_{AB},
	\end{align*}
	and hence,
	\[
		\partial_u \left( \log \det \frac{\gslash}{r^2} \right) 
		= 
		r^2 \gslash^{AB} \partial_u \left( r^{-2} \gslash_{AB} \right)
		=
		2 \Omega \left( \tr \chibar + \frac{2}{r} \right) + \frac{4}{r} \left( 1 - \Omega \right).
	\]
	This gives,
	\[
		e^{- \int_{u_0}^u \left\vert 2 \Omega \left( \tr \chibar + \frac{2}{r} \right) + \frac{4}{r} \left( 1 - \Omega \right) \right\vert du'}
		\leq
		\frac{\det \frac{\gslash}{r^2}(u)}{\det \frac{\gslash}{r^2}(u_0)}
		\leq
		e^{\int_{u_0}^u \left\vert 2 \Omega \left( \tr \chibar + \frac{2}{r} \right) + \frac{4}{r} \left( 1 - \Omega \right) \right\vert du'},
	\]
	and hence, using the assumption that
	\[
		c \det \gamma \leq \det \frac{\gslash}{r^2}(u_0) \leq C \det \gamma,
	\]
	for some constants $C,c >0$, where $\gamma$ is the round metric, and the bootstrap assumptions \eqref{eq:auxba1}--\eqref{eq:auxba3},
	\[
		c \det \gamma \leq \det \frac{\gslash}{r^2}(u) \leq C \det \gamma.
	\]
	
	Let $\lambda$ and $\Lambda$ denote the eigenvalues of $\frac{\gslash_{AB}}{r^2}$ such that $0<\lambda \leq \Lambda$.  There exists $v = (v_1,v_2)$ such that $\max\{ \vert v_1 \vert, \vert v_2 \vert \}$ and,
	\[
		\frac{\gslash}{r^2} v = \Lambda v,
	\] 
	i.\@e.\@,
	\[
		\Lambda v_1 = \frac{\gslash_{11}}{r^2} v_1 + \frac{\gslash_{12}}{r^2} v_2 
		\quad \text{ and } \quad
		\Lambda v_2 = \frac{\gslash_{21}}{r^2} v_1 + \frac{\gslash_{22}}{r^2} v_2 .
	\]
	Hence,
	\[
		\Lambda \leq 2 \sum_{A,B} \frac{\vert \gslash_{AB} \vert}{r^2},
	\]
	and
	\[
		\left\vert r^2 \gslash^{AB} \right\vert 
		= 
		\left\vert \frac{r^{-2} \gslash_{A'B'}}{\det r^{-2} \gslash} \right\vert
		=
		\frac{ \left\vert r^{-2} \gslash_{A'B'} \right\vert }{\lambda \Lambda}
		\geq
		\frac{1}{\lambda}
		\geq
		\frac{1}{\Lambda},
	\]
	i.\@e.\@
	\[
		1 \leq \Lambda \left\vert r^2 \gslash^{AB} \right\vert.
	\]
	This implies that,
	\[
		\sum_{A,B} \vert \chibar_{AB}\vert^2
		\leq
		\sum_{A,B,C,D} \vert \chibar_{AB} \vert \vert \chibar_{CD} \vert
		\leq
		\Lambda^2 r^4 \gslash^{AC} \gslash^{BD} \chibar_{AB} \chibar_{CD}
		\leq
		2 \vert \chibar \vert^2 \sum_{A,B} \vert \gslash_{AB} \vert^2.
	\]
	Returning now to the first variational formula \eqref{eq:firstvar3},
	\[
		\left\vert \gslash_{AB}(u) - \gslash_{AB}(u_0) \right\vert
		\leq
		\int_{u_0}^u \left\vert \partial_u \gslash_{AB} \right\vert du'
		\leq
		\int_{u_0}^u 2 \left\vert \Omega \chibar_{AB} \right\vert du',
	\]
	summing over $A,B$ and using the above bounds for the components of $\chibar$ this gives,
	\begin{multline*}
		\sum_{A,B} \left\vert \gslash_{AB}(u) - \gslash_{AB}(u_0) \right\vert
		\leq
		\int_{u_0}^u 4 \Omega \vert \chibar \vert \sum_{A,B} \vert \gslash_{AB} \vert d u'
		\\
		\leq
		\int_{u_0}^u 4 \Omega \vert \chibar \vert \sum_{A,B} \vert \gslash_{AB}(u') - \gslash_{AB} (u_0) \vert d u'
		+
		\int_{u_0}^u 4 \Omega \vert \chibar \vert  d u' \sum_{A,B} \vert \gslash_{AB}(u_0) \vert.
	\end{multline*}
	Using again the bootstrap assumptions \eqref{eq:auxba1}--\eqref{eq:auxba3} and the fact that,
	\[
		\sum_{A,B} \left\vert \gslash_{AB}(u_0) \right\vert \leq C r(u_0)^2,
	\]
	the Gr\"{o}nwall inequality implies,
	\[
		\sum_{A,B} \left\vert \gslash_{AB}(u) - \gslash_{AB}(u_0) \right\vert \leq Cr.
	\]
	The first result then follows from the fact that,
	\[
		\sum_{A,B} \left\vert \gslash_{AB}(u_0) - \gslash_{AB}^{\circ} \right\vert \leq Cr.
	\]
	The result for $\gslash^{AB}$ follows from this and the bounds on $\det \gslash$ above.
\end{proof}

If $\xi$ is a $(0,k)$ $S_{u,v}$ tensor such that $\vert \xi \vert \leq \frac{C}{r^p}$, then this proposition implies that the components of $\xi$ satisfy
\[
	\vert \xi_{A_1\ldots A_k} \vert \leq \frac{C}{r^{p-k}}.
\]
This fact will be used in Section \ref{section:emtensor}, together with the bootstrap assumptions and Sobolev inequalities, to give bounds on the components of the Ricci coefficients, Weyl curvature components and energy momentum tensor components.

\begin{lemma}
	If $x \in \mathcal{A}'$ then, for $u = u(x)$, $v = v(x)$,
	\[
		I(S_{u,v}) \leq \frac{1}{\pi},
	\]
	so that the constant $I'(S_{u,v})$ in Lemma \ref{lem:sob1} and Lemma \ref{lem:sob2} is equal to $1$.
\end{lemma}

\begin{proof}
The proof proceeds as in Chapter 5 of \cite{Ch}.

\end{proof}

Combining the equation \eqref{eq:sobmain} with the bootstrap assumptions \eqref{eq:Ricciba} then gives
\begin{equation} \label{eq:Riccisob}
	\sup_{u,v} \left(r^p \Vert \mathfrak{D}^k \Gamma_p \Vert_{L^{\infty}(S_{u,v})} \right) \leq C \overline{C},
\end{equation}
for $k = 0,1$.  In particular, by taking $\overline{C}$ to be sufficiently small, the bootstrap assumptions \eqref{eq:auxba1}--\eqref{eq:auxba3} can be recovered with better constants.  Hence $\mathcal{A}' \subset \mathcal{A}$ is open, closed, connected and non-empty, and therefore $\mathcal{A}' = \mathcal{A}$.

Note that, provided $\overline{C}$ is taken to be sufficiently small, this implies that
\[
	\frac{1}{2} \leq \Omega^2 \leq 2.
\]
This fact will be used throughout.

Finally, to obtain pointwise bounds for curvature on the spheres from bounds on $F^1_{v_0,v}(u), F^2_{u_0,u}(v)$, an additional Sobolev inequality is required.

\begin{lemma} \label{lem:sobflux}
	If $\xi$ is an $S_{u,v}$ tensor then, for any weight $q$,
	\[
		\sup_{v_0 \leq v' \leq v} \Vert r^q \xi \Vert_{L^4(S_{u,v'})} \leq C \Bigg( \Vert r^q \xi \Vert_{L^4 (S_{u,v_0})}
		+ \bigg( \int_{v_0}^v r^{2q-2} \int_{S_{u,v'}} \vert r \nablaslash_4 \xi \vert^2 + \vert r \nablaslash \xi \vert^2 + \vert \xi \vert^2 d \mu_{S_{u,v'}} dv' \bigg)^{1/2} \Bigg),
	\]
	and
	\[
		\sup_{u_0 \leq u' \leq u} \Vert r^q \xi \Vert_{L^4(S_{u',v})} \leq C \Bigg( \Vert r^q \xi \Vert_{L^4(S_{u_0,v})}
		+ \bigg( \int_{u_0}^u r^{2q-1} \int_{S_{u',v}} \vert \nablaslash_3 \xi \vert^2 + \vert r \nablaslash \xi \vert^2 + \vert \xi \vert^2 d \mu_{S_{u',v}} du' \bigg)^{1/2} \Bigg).
	\]
\end{lemma}

Lemma \ref{lem:sobflux} together with Lemma \ref{lem:sob2}, equation \eqref{eq:arearadius} and the bound on the isoperimetric constant combine to give the inequalities, 
\begin{multline*}
	\sup_{v_0 \leq v' \leq v} \Vert r^{\frac{w+3}{2}} \xi \Vert_{L^{\infty}(S_{u,v'})}
	\leq
	C \Bigg(
	\Vert r^{\frac{w+2}{2}} \xi \Vert_{L^{4}(S_{u,v_0})} + \Vert r^{\frac{w+2}{2}} r\nablaslash \xi \Vert_{L^{4}(S_{u,v_0})}
	\\
	+ \left( \int_{v_0}^v r^{w} \int_{S_{u,v'}} \vert r \nablaslash_4 r \nablaslash \xi \vert^2 + \vert (r \nablaslash)^2 \xi \vert^2 + \vert r \nablaslash_4 \xi \vert^2 + \vert r \nablaslash \xi \vert^2 + \vert \xi \vert^2 d \mu_{S_{u,v'}} dv' \right)^{\frac{1}{2}} \Bigg),
\end{multline*}
and
\begin{multline*}
	\sup_{u_0 \leq u' \leq u} \Vert r^{\frac{w+2}{2}} \xi \Vert_{L^{\infty}(S_{u',v})}
	\leq
	C \Bigg(
	\Vert r^{\frac{w+1}{2}} \xi \Vert_{L^{4}(S_{u_0,v})} + \Vert r^{\frac{w+1}{2}} r\nablaslash \xi \Vert_{L^{4}(S_{u_0,v})}
	\\
	+ \left( \int_{u_0}^u r^{w} \int_{S_{u',v}} \vert \nablaslash_3 r \nablaslash \xi \vert^2 + \vert (r \nablaslash)^2 \xi \vert^2 + \vert \nablaslash_3 \xi \vert^2 + \vert r \nablaslash \xi \vert^2 + \vert \xi \vert^2 d \mu_{S_{u',v}} du' \right)^{\frac{1}{2}} \Bigg),
\end{multline*}
for any weight $w$.  The bootstrap assumptions \eqref{eq:curvatureba} then give the following pointwise bounds on curvature,
\begin{equation} \label{eq:curvsob}
	\sup_{u,v} \left(r^p \Vert \mathfrak{D}^k \psi_p \Vert_{L^{\infty}(S_{u,v})} \right) \leq C,
\end{equation}
for $k = 0,1$.

Finally, note also that, whilst \eqref{eq:Riccisob} give the pointwise bounds
\[
	r \vert \mathfrak{D}^k b \vert \leq C,
\]
for $k = 0,1$, the bootstrap assumption \eqref{eq:bba} together with Lemma \ref{lem:sobflux} give the additional pointwise bounds,
\[
	r^{\frac{1}{2}} \vert \mathfrak{D} r \nablaslash b \vert \leq C,
\]
and \eqref{eq:Gammaslashba}, \eqref{eq:Gammaslashbanull} give,
\[
	r \vert \mathfrak{D}^k \Gammaslash \vert \leq C,
\]
for $k=0,1$, as discussed at the end of Section \ref{section:ba}.

\section{Geometry of Null Geodesics and the Support of $f$} \label{section:suppf}
The decay of the components of the energy momentum tensor come from the decay of the size of the support of $f$ in $P_x$ as $r(x) \to \infty$.  The estimates on the size of the support of $f$ form the content of this section.  It will also be shown that, provided $u_f$ is chosen suitably large, the matter is supported to the past of the hypersurface $\{ u = u_f -1\}$.  Recall that the results of this section are shown in the setting of Theorem \ref{thm:main4}, so that they hold for points $x\in \mathcal{A}$.

Throughout this section $\gamma$ will denote a null geodesic emanating from $\{v = v_0\}$ (so that $v(\gamma(0)) = v_0$) such that $(\gamma(0),\dot{\gamma}(0)) \in \supp(f)$.  The tangent vector to $\gamma$ at time $s$ can be written with respect to the double null frame as
\[
	\dot{\gamma}(s) = p^A(s) e_A + p^3(s) e_3 + p^4(s) e_4.
\]
Note that in the next section notation will change slightly ($\gamma(0)$ there will be a point in $\{v>v_0\} \cap \pi(\supp(f))$ and the parameter $s$ will always be negative).

Recall that, by assumption, $(\gamma(0),\dot{\gamma}(0)) \in \supp(f)$ implies that,
\[
	0\leq p^4(0) \leq C_{p^4}, 
	\quad 
	0 \leq r(0)^2 p^3(0) \leq C_{p^3} p^4(0), 
\]
\[
	\vert r(0)^2 p^A(0) \vert \leq C_{p^A} p^4(0) \quad \text{ for } A=1,2,
\]
for some fixed constants $C_{p^1}, \ldots, C_{p^4}$ independent of $v_0$.

The main result of this section is the following.

\begin{proposition} \label{prop:suppfmain}
	For such a geodesic $\gamma$,
	\[
		\frac{1}{2} p^4(0) \leq p^4(s) \leq \frac{3}{2} p^4(0),
		\qquad
		0 \leq r(s)^2 p^3(s) \leq 2 C_{p^3} p^4(0),
	\]
	\[
		\vert r(s)^2 p^A(s) \vert \leq 2 C_{p^A} p^4(0) \quad \text{ for } A=1,2,
	\]
	for all $s\geq 0$ such that $\gamma(s) \in \mathcal{A}$, provided $v_0$ is taken suitably large.
\end{proposition}

The proof of Proposition \ref{prop:suppfmain} is obtained through a bootstrap argument, so suppose $s_1 \in (0,\infty)$ is such that,
\begin{align}
	\frac{1}{4} p^4(0)
	&
	\leq
	p^4(s)
	\leq
	2 p^4(0)
	\label{eq:suppfba1}
	\\
	r(s)^2 p^3(s)
	&
	\leq
	2 C_{p^3} p^4(0)
	\label{eq:suppfba2}
	\\
	\vert r(s)^2 p^A(s) \vert
	&
	\leq
	2 C_{p^3} p^4(0)
	\quad \text{for } A=1,2,
	\label{eq:suppfba3}
\end{align}
for all $0\leq s\leq s_1$.  Clearly the set of all such $s_1$ is a non-empty, closed, connected subset of $(0,\infty)$.  The goal is to show it is also open, and hence equal to $(0,\infty)$, by improving the constants.

The following fact, proved assuming the above bootstrap assumptions hold, is used for integrating the geodesic equations in the proof of Proposition \ref{prop:suppfmain}.

\begin{lemma} \label{lem:suppf1}
	Along such a geodesic $\gamma$,
	\[
		\frac{p^4(s)}{\dot{r}(s)} \leq 2,
	\]
	where $r(s) = r(\gamma(s))$ and $\dot{r}(s) = \frac{dr}{ds}(s)$, provided $v_0$ is taken sufficiently large and $\gamma(s) \in \mathcal{A}$.
\end{lemma}

\begin{proof}
	Recall that 
	\[
		\left\vert \frac{1}{\Omega^2} - 1\right\vert \leq \frac{1}{2},
	\]
	provided $\overline{C}$ is sufficiently small, and so $\Omega^2 > \frac{1}{2}$.  Since $\dot{u}(s) = \frac{p^3(s)}{\Omega^2}$ and $\dot{v}(s) = p^4(s)$, this and the bootstrap assumptions \eqref{eq:suppfba1}, \eqref{eq:suppfba2} then imply that,
	\[
		\left\vert \frac{p^4(s)}{\dot{r}(s)} \right\vert 
		=
		\left\vert \frac{p^4(s)}{p^4(s) - \frac{p^3(s)}{\Omega^2}} \right\vert
		\leq
		\frac{1}{1 - \frac{2C_{p^3} p^4(0) \frac{1}{r(s)^2}}{\frac{1}{2}{\frac{1}{4} p^4(0)}}}
		=
		\frac{1}{1 - \frac{16 C_{p^3}}{r(s)^2}}
		\leq
		2,
	\]
	provided $v_0$, and hence $r(0)$ is taken sufficiently large.
\end{proof}

\begin{proof}[Proof of Proposition \ref{prop:suppfmain}]
	The geodesic equation for $p^4$, 
	\[
		\dot{p}^4(s) + \Gamma^4_{\mu \nu}(s) p^{\mu}(s) p^{\nu}(s) = 0,
	\] 
	written using the notation for the Ricci coefficients introduced in Section \ref{subsec:alphabeta} takes the form
	\begin{align*}
		\dot{p}^4(s) = 
		&
		\ \frac{1}{2r} \gslash_{AB} p^A(s)p^B(s) - \frac{1}{4} \left( \tr \chibar + \frac{2}{r} \right) \gslash_{AB} p^A(s)p^B(s)
		\\
		&
		- \frac{1}{2} \hat{\chibar}_{AB} p^A(s)p^B(s) - 2 \etabar_A p^A(s) p^4(s) - \omega (p^4(s))^2.
	\end{align*}
	Using the fact that the pointwise bounds for $\Gamma$ give,
	\[
		\left\vert \gslash_{AB} \right\vert \leq Cr^2,
		\quad
		\left\vert \tr \chibar + \frac{2}{r} \right\vert \leq \frac{C}{r^2},
		\quad
		\left\vert \hat{\chibar}_{AB} \right\vert \leq Cr,
		\quad
		\left\vert \etabar_A \right\vert \leq \frac{C}{r},
		\quad
		\vert \omega \vert \leq \frac{C}{r^3},
	\]
	the bootstrap assumptions \eqref{eq:suppfba1}--\eqref{eq:suppfba3} then imply that,
	\[
		\left\vert \dot{p}^4(s) \right\vert
		\leq
		\frac{C \left( p^4(0) \right)^2}{r(s)^3}.
	\]
	Hence, for any $s\in[0,s_1]$,
	\[
		\left\vert p^4(s) - p^4(0) \right\vert
		\leq
		\frac{C p^4(0)}{r(0)^2}
	\]
	by Lemma \ref{lem:suppf1}, and,
	\[
		\left( 1 - \frac{C}{r(0)^2} \right) p^4(0)
		\leq
		p^4(s)
		\leq
		\left( 1 + \frac{C}{r(0)^2} \right) p^4(0).
	\]
	Taking $v_0$, and hence $r(0)$, sufficiently large then gives,
	\[
		\frac{1}{2} p^4(0)
		\leq
		p^4(s)
		\leq
		\frac{3}{2} p^4(0),
	\]
	improving the bootstrap assumption \eqref{eq:suppfba1}.
	
	Consider now the geodesic equation for $p^3(s)$,
	\begin{align*}
		\dot{p}^3(s) =
		&
		- \frac{1}{2r} \gslash_{AB} p^A(s) p^B(s) - \frac{1}{2} \hat{\chi}_{AB} p^A(s) p^B(s) - \frac{1}{4} \left( \tr \chi - \frac{2}{r} \right) \gslash_{AB} p^A(s) p^B(s)
		\\
		&
		- \left( \eta_A - \etabar_A \right) p^A(s) p^3(s) + \omega p^3(s) p^4(s).
	\end{align*}
	Recalling that,
	\[
		\dot{r}(s) = p^4(s) - \frac{p^3(s)}{\Omega^2},
	\]
	this then gives,
	\begin{align*}
		\frac{d r(s)^2 p^3(s)}{ds} =
		&
		2r \left( p^4(s) - \frac{p^3(s)}{\Omega^2} \right) p^3(s) - \frac{r}{2} \gslash_{AB} p^A(s) p^B(s) - \frac{r^2}{2} \hat{\chi}_{AB} p^A(s) p^B(s)
		\\
		&
		- \frac{r^2}{4} \left( \tr \chi - \frac{2}{r} \right) \gslash_{AB} p^A(s) p^B(s) - r^2 \left( \eta_A - \etabar_A \right) p^A(s) p^3(s) 
		\\
		&
		+ r^2 \omega p^3(s) p^4(s)
		\\
		=
		&
		-2r \frac{p^3(s)}{\Omega^2} p^3(s)  - \frac{r^2}{2} \hat{\chi}_{AB} p^A(s) p^B(s) - \frac{r^2}{4} \left( \tr \chi - \frac{2}{r} \right) \gslash_{AB} p^A(s) p^B(s)
		\\
		&
		 - r^2 \left( \eta_A - \etabar_A \right) p^A(s) p^3(s) + r^2 \omega p^3(s) p^4(s),
	\end{align*}
	where the mass shell relation \eqref{eq:massshell} has been used to obtain the cancellation.  Inserting the pointwise bounds for the components of $\Gamma$ and the bootstrap assumptions \eqref{eq:suppfba1}--\eqref{eq:suppfba3}, this gives,
	\[
		\left\vert \frac{d r(s)^2 p^3(s)}{ds} \right\vert \leq \frac{C \left( p^4(0) \right)^2}{r(s)^2}.
	\]
	Again, integrating from $s=0$ gives,
	\[
		\left\vert r(s)^2 p^3(s) - r(0)^2 p^3(0) \right\vert \leq \frac{C p^4(0)}{r(0)},
	\]
	and using the assumption on $p^3(0)$,
	\[
		r(0)^2 p^3(0) - \frac{C p^4(0)}{r(0)} 
		\leq 
		r(s)^2 p^3(s) 
		\leq 
		C_{p_3} \left( 1 + \frac{C}{r(0)} \right) p^4(0).
	\]
	If $v_0$, and hence $r(0)$ is sufficiently large this gives,
	\[
		0 
		\leq 
		r(s)^2 p^3(s) 
		\leq 
		\frac{3}{2}C_{p_3} p^4(0),
	\]
	hence improving the bootstrap assumption \eqref{eq:suppfba2}.
	
	Finally, consider the geodesic equation for $p^A(s)$, for $A=1,2$,
	\begin{align*}
		\dot{p}^A = 
		& 
		- \frac{2\dot{r}}{r} p^A - \Gammaslash^A_{BC} p^B p^C + \left( 1 - \frac{1}{\Omega^2} \right) \frac{2}{r} p^A p^3 
		\\
		& 
		- \left( \tr \chibar + \frac{2}{r} \right) p^A p^3 - \left( \tr \chi - \frac{2}{r} \right) p^A p^4 - 2 {\hat{\chibar}_B}^Ap^Bp^3
		\\
		& 
		- 2 {\hat{\chi}_B}^Ap^Bp^4 + \left( \nablaslash_{e_B}b\right)^A p^B p^4 - b^C \Gammaslash^A_{BC} p^Bp^4 - 2 p^3p^4(\eta^A + \etabar^A),
	\end{align*}
	which similarly gives,
	\begin{align*}
		\frac{d r(s)^2 p^A(s)}{ds} = 
		& 
		- r^2 \Gammaslash^A_{BC} p^B p^C + \left( 1 - \frac{1}{\Omega^2} \right) 2r p^A p^3 - r^2 \left( \tr \chibar + \frac{2}{r} \right) p^A p^3 
		\\
		& 
		- r^2 \left( \tr \chi - \frac{2}{r} \right) p^A p^4 - 2 r^2 {\hat{\chibar}_B}^Ap^Bp^3 - 2r^2 {\hat{\chi}_B}^Ap^Bp^4
		\\
		& 
		 + r^2 \left( \nablaslash_{e_B}b\right)^A p^B p^4 - r^2 b^C \Gammaslash^A_{BC} p^Bp^4 - 2 r^2 p^3p^4(\eta^A + \etabar^A),
	\end{align*}
	and using the bootstrap assumptions \eqref{eq:suppfba1}--\eqref{eq:suppfba3},
	\[
		\left\vert \frac{d r(s)^2 p^A(s)}{ds} \right\vert \leq \frac{C \left( p^4(0) \right)^2}{r(s)^2}.
	\]
	Integrating and again taking $v_0$ large similarly gives,
	\[
		\left\vert r(s)^2 p^A(s) \right\vert \leq \frac{3}{2} C_{p_A} p^4(0),
	\]
	improving the bootstrap assumption \eqref{eq:suppfba3}.
	
	The set of all $s_1$ such that \eqref{eq:suppfba1}--\eqref{eq:suppfba3} hold for all $0 \leq s \leq s_1$ is therefore a non-empty, open, closed, connected subset of $(0,\infty)$, and hence equal to $(0,\infty)$.
\end{proof}

Finally we can show that $\pi(\text{supp}(f))$ is contained in $\{ u \leq u_f - 1\}$ for some $u_f$ large.

\begin{proposition} \label{prop:uf}
	For a geodesic $\gamma$ as above,
	\[
		u(s) \leq u_f - 1,
	\]
	for all $s \geq 0$ provided $u_f$ is chosen sufficiently large and $\gamma(s) \in \mathcal{A}$.
\end{proposition}
\begin{proof}
	Recall that $\dot{u}(s) = \frac{p^3(s)}{\Omega^2}$.  Since $\Omega^2 \geq \frac{1}{2}$, Proposition \ref{prop:suppfmain} implies that
	\[
		\vert \dot{u}(s) \vert 
		\leq 
		4 C_{p^3} \frac{p^4(0)}{r(s)^2},
	\]
	and so
	\[
		u(s) 
		\leq 
		u(0) + 4 C_{p^3} \int_0^s \frac{p^4(0)}{r(s')^2} ds' 
		\leq 
		u(0) + 64 C_{p^3} \int_{r(0)}^{r(s)} \frac{1}{r^2} dr \leq u(0) + \frac{32 C_{p^3}}{r(0)}.
	\]
	The result then holds if
	\[
		u_f > \sup_{\{v=v_0\}} \left( u + \frac{32 C_{p^3}}{r} \right) + 1.
	\]
\end{proof}

\section{Estimates for the Energy Momentum Tensor} \label{section:emtensor}
Recall the notation from Section \ref{section:equations}, and the set $\mathcal{A}$ from Theorem \ref{thm:main4}.  The main result of this section is the following.
\begin{proposition} \label{prop:tmain}
	If $x \in \mathcal{A}$ then, for $u = u(x)$, $v = v(x)$, $0\leq k\leq 2$,
	\[
		\left( r^p \vert \mathfrak{D}^k \mathcal{T}_p \vert \right) (u,v) \leq C\bc,
	\]
  for $k \leq 3$,
	\begin{equation} \label{eq:em3in}
		\int_{u_0}^u \int_{S_{u',v}} r(u',v)^{2p-2} \vert \mathfrak{D}^k \mathcal{T}_p \vert^2 d\mu_{S_{u',v}} d u' \leq C \bc,
	\end{equation}
	and,
	\begin{equation} \label{eq:em3out}
		\int_{v_0}^v \int_{S_{u,v'}} r(u,v')^{2p-4} \vert \mathfrak{D}^k \mathcal{T}_p \vert^2 d\mu_{S_{u,v'}} d v' \leq C \bc,
	\end{equation}
	and for $k\leq 4$,
	\begin{equation} \label{eq:em4}
		\int_{u_0}^u \int_{v_0}^v \int_{S_{u',v'}} r(u',v')^{2p-4} \vert \mathfrak{D}^k \mathcal{T}_p \vert^2 d\mu_{S_{u',v'}} d v' d u' \leq C \bc,
	\end{equation}
	for some constant $C$.
\end{proposition}
Recall from Section \ref{section:ba} that $\bc$ describes the size of the data.

The main difficulty in the proof of Proposition \ref{prop:tmain}, and in fact the main new difficulty in this work, is estimating derivatives of $f$.  In Section \ref{subsec:testimates}, Proposition \ref{prop:tmain} is reduced to Proposition \ref{prop:emmain}, a statement about derivatives of $f$.  In particular, in Section \ref{subsec:testimates} it is seen how the zeroth order bounds, $r^p \vert \mathcal{T}_p \vert \leq C$, are obtained using the results of Section \ref{section:Sasaki}.  A collection of operators $\tilde{\mathfrak{D}}$, which act on functions $h : P \to \mathbb{R}$, akin to the collection $\mathfrak{D}$ introduced in Section \ref{subsec:commutation}, is defined and used in the formulation of Proposition \ref{prop:emmain}.  In Section \ref{subsec:furtherschematic} additional schematic notation is introduced.  This notation is used throughout the remainder of Section \ref{section:emtensor}.  In Section \ref{subsec:overview} seven more operators, $V_{(0)}, \ldots, V_{(6)}$, are introduced and Proposition \ref{prop:emmain} is further reduced to Proposition \ref{prop:emmain2}, which involves bounds on combinations of the six operators $V_{(1)},\ldots,V_{(6)}$ applied to $f$.  The main observation is that the Vlasov equation can be used to replace $V_{(0)} f := r \hor_{(x,p)} (e_4) f$ with a combination of operators from $V_{(1)}, \ldots, V_{(6)}$ (such that the coefficients have desirable weights) applied to $f$.  The operators $V_{(1)}, \ldots, V_{(6)}$, in Section \ref{subsec:jacobidef}, are then used to define corresponding Jacobi fields (with respect to the Sasaki metric, defined in Section \ref{section:Sasaki}) $J_{(1)}, \ldots, J_{(6)}$.  In Section \ref{subsec:fjacobi} two frames for the mass shell, $\{E_i\}$ and $\{F_i\}$, are defined.  In Sections \ref{subsec:lowerjacobi} and \ref{subsec:higherjacobi} bounds for the components, with respect to the frame $\{E_i\}$, of the Jacobi fields, along with their derivatives, are obtained.  It is transport estimates for the Jacobi equation which are used to obtain these bounds.  For such estimates it is convenient to use the parallel frame $\{F_i\}$, and it is therefore also important to control the change of frame matrix $\Xi$, also defined in Section \ref{subsec:fjacobi}.  In Section \ref{subsec:proofemmain2} it is shown how Proposition \ref{prop:emmain2} follows from the bounds on derivatives of the components of the Jacobi fields obtained in Sections \ref{subsec:lowerjacobi} and \ref{subsec:higherjacobi}, thus completing the proof of Proposition \ref{prop:tmain}.

For a function $h : P\to \mathbb{R}$, define $\mathcal{T}_p [h]$ by replacing $f$ with $h$ in the definition of $\mathcal{T}_p$.  So for example,
\[
	\Tslash_{44}[h] = \int_{P_x} h \ p_4 p_4 d \mu_{P_x},
	\qquad \text{and} \qquad
	\Tslash [h]_{AB} = \int_{P_x} h \ p_A p_B d \mu_{P_x}.
\]
In particular
\[
	\mathcal{T}_p[f] = \mathcal{T}_p.
\]
This notation will be used throughout this section.  Finally, it is assumed throughout this section that $x \in \mathcal{A}$.

\subsection{Estimates for $\mathcal{T}$ Assuming Estimates for $f$} \label{subsec:testimates}

Consider the set of operators $\{ e_3, re_4, re_1, re_2\}$.  The notation $\tilde{\mathfrak{D}}$ will be used to schematically denote an arbitrary operator in this set.  These operators act on functions $h : P \to \mathbb{R}$ on the mass shell where, for example,
\[
	e_3(h) = \frac{1}{\Omega^2} \partial_u h, \qquad e_4(h) = \partial_v h + b^A \partial_{\theta^A} h,
\]
in the coordinate system $(u,v,\theta^1,\theta^2, p^1,p^2, p^4)$ on $P$ (as usual it is assumed we are working in one of the two fixed spherical coordinate charts).

Given a collection of derivative operators from the set $\{ r\nablaslash, \nablaslash_3, r \nablaslash_4 \}$, say $\mathfrak{D}^k$, this will act on $(0,m)$ $S_{u,v}$ tensors and give a $(0,l+m)$ $S_{u,v}$ tensor, where $l \leq k$ is the number of times $r\nablaslash$ appears in $\mathfrak{D}^k$.  Let $\tilde{\mathfrak{D}}^k_{C_1,\ldots,C_l}$ denote the corresponding collection of derivative operators in $\tilde{\mathfrak{D}}$ where, in $\mathfrak{D}^k (e_{C_1} ,\ldots, e_{C_l})$, each $r\nablaslash_4$ is replaced by $re_4$, each $\nablaslash_3$ is replaced by $e_3$, and each $r\nablaslash_{C_i}$ is replaced by $re_{C_i}$.  So for example if $k=4$ and
\[
	\mathfrak{D}^4 = (r\nablaslash_4) (r\nablaslash) \nablaslash_3 (r\nablaslash),
\]
then
\[
	\tilde{\mathfrak{D}}^4_{C_1,C_2} = re_4(re_{C_1}(e_3(re_{C_2}( \cdot)))).
\]

Using the results of Section \ref{section:suppf}, the proof of the $k=0$ case of Proposition \ref{prop:tmain} can immediately be given.  First, however, Proposition \ref{prop:emmain}, a result about $\tilde{\mathfrak{D}}$ derivatives of $f$, is stated.  The full proof of Proposition \ref{prop:tmain}, assuming Proposition \ref{prop:emmain}, is then given after Proposition \ref{prop:emcommutation}, Proposition \ref{prop:emerrors} and Lemma \ref{lem:logp3}, which relate $\mathfrak{D}$ derivatives of $\mathcal{T}$ to $\tilde{\mathfrak{D}}$ derivatives of $f$.

\begin{proposition} \label{prop:emmain}
	If $x\in \mathcal{A}$ then, for $u = u(x)$, $v = v(x)$, $k = 1, 2$, 
	\[
		\left( r^p \left\vert \mathcal{T}_p \left[ \tilde{D}^k f \right] \right\vert \right) (u,v) \leq C \bc,
	\]
	for $1 \leq k \leq 3$,
	\[
		\int_{u_0}^u \int_{S_{u',v}} r(u',v)^{2p-2} \left\vert \mathcal{T}_p \left[ \tilde{D}^k f \right] \right\vert^2 d\mu_{S_{u',v}} d u' \leq C\bc,
	\]
	and for $1 \leq k\leq 4$,
	\[
		\int_{u_0}^u \int_{v_0}^v \int_{S_{u',v'}} r(u',v')^{2p-4} \left\vert \mathcal{T}_p \left[ \tilde{D}^k f \right] \right\vert^2 d\mu_{S_{u',v'}} d v' d u' \leq C \bc.
	\]
\end{proposition}
In the above,
\[
	\left\vert \mathcal{T}_p [\tilde{\mathfrak{D}}^k f] \right\vert^2
	=
	\gslash^{C_1D_1} \ldots \gslash^{C_lD_l} \mathcal{T}_p [\tilde{\mathfrak{D}}^k_{C_1,\ldots,C_l} f] \cdot \mathcal{T}_p [\tilde{\mathfrak{D}}^k_{D_1,\ldots,D_l} f ].
\]

The proof of Proposition \ref{prop:emmain} is given in Sections \ref{subsec:overview}--\ref{subsec:proofemmain2}.

\begin{proposition} \label{prop:emcommutation}
	Given $h:P\to \mathbb{R}$,
	\begin{align*}
	\nablaslash_3 \mathcal{T}_p [h] =
		& \
		\mathcal{T}_p [e_3(h)] + (\Gamma_2 + \Gamma_1 + \rp_1) \cdot \mathcal{T}_p [h] + \mathcal{T}_p [e_3(\log p^3) h]
		\\
		r \nablaslash_4 \mathcal{T}_p [h] =
		& \
		\mathcal{T}_p [re_4(h)] + (r\Gamma_2 +  r \nablaslash b + r \Gammaslash \cdot \Gamma_1 + 1) \cdot \mathcal{T}_p [h] + \mathcal{T}_p [r e_4(\log p^3) h]
		\\
		r \nablaslash_C \mathcal{T}_p [h] =
		& \
		\mathcal{T}_p [re_C(h)] + \left( r\Gammaslash \cdot \mathcal{T}_p [h] \right) (e_C) + \mathcal{T}_p [r e_C(\log p^3) h],
	\end{align*}
	where the last line is true for $C = 1,2$.
\end{proposition}

\begin{proof}
	This follows by directly computing the derivatives of each $\mathcal{T}_p$.  For example,
	\begin{align*}
		\nablaslash_4 \Tslash_{34}[h](x)
		=
		& \
		e_4 \left( \int_{P_x} h p_3 p_4 \frac{\sqrt{\det \gslash}}{p^4} dp^1 dp^2 dp^4 \right)
		\\
		=
		& \
		4 e_4 \left( \int_{P_x} h p^3 \sqrt{\det \gslash} dp^1 dp^2 dp^4 \right)
		\\
		=
		& \
		4 \int_{P_x} \left( e_4(h) + h \tr \chi \right) p^3 \sqrt{\det \gslash} dp^1 dp^2 dp^4
		\\
		& \
		+ 4  \int_{P_x} h e_4(p^3) \sqrt{\det \gslash} dp^1 dp^2 dp^4
		\\
		=
		& \
		\Tslash_{34}[e_4(h)] + \tr \chi \Tslash_{34}[h] + \Tslash_{34}[ e_4(\log p^3) h],
	\end{align*}
	and
	\begin{align*}
		\left( \nablaslash_4 \Tslash_4 [h] \right)^A
		=
		& \
		2 e_4 \left( \int_{P_x} h p^3 p^A \frac{\sqrt{\det \gslash}}{p^4} dp^1 dp^2 dp^4  \right) + \Gamma_{4B}^A {\Tslash_3}^B
		\\
		=
		& \
		\left( \Tslash_4 [e_4(h)] \right)^A + \tr \chi {\Tslash_4[h]}^A + \left( [{\chi_B}^A - \nablaslash_B b^A - \Gammaslash_{AC}^B \cdot b^C] \cdot \Tslash_4[h]^B \right)
		\\
		& \
		+ 2 \int_{P_x} h e_4(p^3) p^A \frac{\sqrt{\det \gslash}}{p^4} dp^1 dp^2 dp^4
		\\
		=
		& \
		\left( \Tslash_4 [e_4(h)] \right)^A + \tr \chi {\Tslash_4[h]}^A + \left( [\chi - \nablaslash b - \Gammaslash \cdot b] \cdot \Tslash_4[h] \right)^A
		\\
		& \
		+ \Tslash_4 [e_4(\log p^3)]^A.
	\end{align*}
	The other derivatives are similar.
\end{proof}

\begin{proposition} \label{prop:emerrors}
	For any $k \geq 1$,
	\[
		\mathfrak{D}^k \mathcal{T}_p [f] (e_{C_1}, \ldots, e_{C_l})
		=
		\mathcal{T}_p [\tilde{\mathfrak{D}}^k_{C_1,\ldots,C_l} f] + E\left[ \mathfrak{D}^k \mathcal{T}_p \right] (e_{C_1}, \ldots, e_{C_l})
	\]
	where
	\begin{align*}
		E\left[ \mathfrak{D}^k \mathcal{T}_p \right] (e_{C_1}, \ldots, e_{C_l})
		=
		&
		\left( \mathfrak{D} E\left[ \mathfrak{D}^{k-1} \mathcal{T}_p \right] \right) (e_{C_1}, \ldots, e_{C_l})
		+ \mathcal{T}_p \left[ \left( \tilde{\mathfrak{D}} \log p^3 \right) \tilde{\mathfrak{D}}^{k-1} f \right] (e_{C_1},\ldots,e_{C_l})
		\\
		&
		+ \left( \sum_{p_1+p_2\geq 0} \rp_{p_1} \cdot \Gamma_{p_2} + r\nablaslash b + r \Gammaslash \cdot (\Gamma_1 + 1) \right) \cdot \mathcal{T}_p [\tilde{\mathfrak{D}}^{k-1} f] (e_{C_1},\ldots,e_{C_l})
		,
	\end{align*}
	and
	\[
		E\left[ \mathcal{T}_p \right] = 0.
	\]
	Here $l\leq k$ is the number of times $r\nablaslash$ appears in $\mathfrak{D}^k$.
\end{proposition}

\begin{proof}
	The proof follows by induction by writing
	\[
		\mathfrak{D}^k \mathcal{T}_p = \mathfrak{D} (\mathfrak{D}^{k-1} \mathcal{T}_p) = \mathfrak{D} \left( \mathcal{T}_p\left[ \tilde{\mathfrak{D}}^{k-1} f \right] + E\left[ \mathfrak{D}^{k-1} \mathcal{T}_p \right] \right)
	\]
	and using the previous proposition.
\end{proof}

\begin{lemma} \label{lem:logp3}
	 For $1\leq k\leq 4$,
	 \[
		  \vert \tilde{D}^k \log p^3 \vert
		  \leq 
		  C \sum_{k'\leq k-1} \left( r \vert \mathfrak{D}^{k'} \Gammaslash \vert + \vert \mathfrak{D}^{k'} r \nablaslash b \vert + \sum_{p_1 + p_2 \geq 0} \rp_{p_1} \left( \vert \mathfrak{D}^{k'} \Gamma_{p_2} \vert + \rp_{p_2} \right) \right),
	 \]
	 in $\supp(f)$.
\end{lemma}

\begin{proof}
	Recall that
	\[
		p^3 = \frac{\gslash_{AB} p^A p^B}{4p^4}.
	\]
	Using the first variation formula \eqref{eq:firstvar3},
	\[
		e_3(p^3)
		=
		\frac{\chibar_{AB} p^A p^B}{2 p^4},
	\]
	and hence,
	\[
		\vert e_3(\log p^3) \vert
		=
		\left\vert \frac{e_3(p^3)}{p^3} \right\vert
		=
		\left\vert \frac{2\hat{\chibar}_{AB} p^A p^B + \tr \chibar \gslash_{AB}p^Ap^B}{\gslash_{CD}p^Cp^D} \right\vert
		\leq
		\left( 2 \vert \hat{\chibar} \vert + \vert \tr \chi\vert \right),
	\]
	by the Cauchy--Schwarz inequality.  Similarly, using \eqref{eq:firstvar4},
	\[
		e_4(p^3) = \frac{ \left( \chi_{AB} + \gslash_{BC} e_A(b^C) + \gslash_{AC} e_B(b^C) \right)p^A p^B}{2 p^4},
	\]
	and hence,
	\[
		\vert r e_4(\log p^3) \vert \leq r \left( \vert \hat{\chi} \vert + \vert \tr \chi \vert + \vert \nablaslash b\vert + \vert b \vert \vert \Gammaslash \vert \right).
	\]
	Finally,
	\[
		e_C (p^3) = \frac{e_C(\gslash_{AB}) p^A p^B}{4p^4} = \frac{\left( \gslash_{DB} \Gammaslash_{CA}^D + \gslash_{AD} \Gammaslash_{CB}^D \right) p^A p^B}{4p^4},
	\]
	hence similarly,
	\[
		\vert r e_C(\log p^3) \vert
		\leq
		2 r \vert \Gammaslash \vert.
	\]
	The higher order derivatives follow similarly.
\end{proof}

\begin{proof}[Proof of Proposition \ref{prop:tmain}]
	For $k=0$ the result follows using the bounds on $p^1,p^2,p^3,p^4$ in $\supp(f)$,
	\[
		\vert p^1\vert, \vert p^2\vert, \vert p^3\vert \leq \frac{Cp^4}{r^2}, \qquad \vert p^4\vert \leq C,
	\]
	from Proposition \ref{prop:suppfmain}.
	
	Note that
	\[
		\sup_{P\vert_{\{v = v_0\}}} \vert f \vert \leq \bc,
	\]
  and hence, since $f$ is preserved by the geodesic flow,
  \[
		\sup_P \vert f \vert \leq \bc.
  \]
  This fact and the above bounds imply
  \[
		\left\vert \int_{P_x} f p^4 d\mu_{P_x} \right\vert 
		=
		\left\vert \int_{0}^C \int_{\vert p^1 \vert, \vert p^2 \vert 
		\leq 
		\frac{C}{r^2}} f p^4 \frac{\sqrt{\det \gslash(x)}}{p^4} dp^1 dp^2 dp^4 \right\vert \leq \frac{C\bc}{r(x)^2},
  \]
  as $\left\vert \sqrt{\det \gslash (x)} \right\vert \leq C r(x)^2$.  One then easily sees,
  \[
		\vert \Tslash_{33} \vert
		\leq
		4 \left\vert \int_{P_x} f p^4 p^4 d\mu_{P_x} \right\vert
		\leq
		\frac{C\bc}{r^2},
  \]
  \[
		\vert \Tslash_{34} \vert
		\leq
		4 \left\vert \int_{P_x} f p^3 p^4 d\mu_{P_x} \right\vert
		\leq
		\frac{C\bc}{r^4},
  \]
  and
  \[
		\vert \Tslash_{44} \vert
		\leq
		4 \left\vert \int_{P_x} f p^3 p^3 d\mu_{P_x} \right\vert
		\leq
		\frac{C}{r^4} \left\vert \int_{P_x} f p^4 p^4 d\mu_{P_x} \right\vert
		\leq
		\frac{C\bc}{r^6}.
  \]
  Moreover,
  \[
		\vert \Tslash_{3A} \vert
		\leq
		2 \left\vert \gslash_{AA'} \int_{P_x} f p^{A'} p^4 d\mu_{P_x} \right\vert
		\leq
		C \left\vert \int_{P_x} f p^4 p^4 d\mu_{P_x} \right\vert
		\leq
		\frac{C\bc}{r^2},
  \]
  \[
		\vert \Tslash_{4A} \vert
		\leq
		2 \left\vert \gslash_{AA'} \int_{P_x} f p^{A'} p^3 d\mu_{P_x} \right\vert
		\leq
		\frac{C}{r^2} \left\vert \int_{P_x} f p^4 p^4 d\mu_{P_x} \right\vert
		\leq
		\frac{C\bc}{r^4},
  \]
  and
  \[
		\vert \Tslash_{AB} \vert
		\leq
		\left\vert \gslash_{AA'} \gslash_{BB'} \int_{P_x} f p^{A'} p^{B'} d\mu_{P_x} \right\vert
		\leq
		C \left\vert \int_{P_x} f p^4 p^4 d\mu_{P_x} \right\vert
		\leq
		\frac{C\bc}{r^2},
  \]
  so that,
  \[
		\vert \Tslash_3 \vert 
		=
		\sqrt{\left\vert \gslash^{AB} \Tslash_{3A} \Tslash_{3B} \right\vert} 
		\leq
		C \sqrt{\frac{1}{r^2} \frac{1}{r^2} \frac{1}{r^2}}
		=
		\frac{C\bc}{r^3},
  \]
  \[
		\vert \Tslash_4 \vert 
		=
		\sqrt{\left\vert \gslash^{AB} \Tslash_{4A} \Tslash_{4B} \right\vert} 
		\leq
		\frac{C\bc}{r^5},
  \]
  and
  \[
		\vert \Tslash \vert 
		=
		\sqrt{\left\vert \gslash^{AC} \gslash^{BD} \Tslash_{AB} \Tslash_{CD} \right\vert} 
		\leq
		\frac{C\bc}{r^4}.
  \]
  
  For first order derivatives of $\mathcal{T}_p$, Proposition \ref{prop:emerrors} and the pointwise bounds on $\Gamma, \Gammaslash, r \nablaslash b$ imply that,
  \[
		\vert \mathfrak{D} \mathcal{T}_p \vert
		\leq
		C \left( \left\vert \mathcal{T}_p \left[ \tilde{\mathfrak{D}} f \right] \right\vert + \vert \mathcal{T}_p \vert + \left\vert \mathcal{T}_p \left[ \tilde{\mathfrak{D}} \log p^3 \cdot f \right] \right\vert \right),
  \]
  and hence Proposition \ref{prop:emmain} and Lemma \ref{lem:logp3} imply,
  \[
		\vert \mathfrak{D} \mathcal{T}_p \vert
		\leq 
		\frac{C\bc}{r^p}.
  \]
  Similarly for the second order derivatives, the pointwise estimates on $\Gamma, \Gammaslash, r \nablaslash b, \mathfrak{D}\Gamma, \mathfrak{D} \Gammaslash, \mathfrak{D} r \nablaslash b$ imply that,
  \[
		\vert \mathfrak{D}^2 \mathcal{T}_p \vert
		\leq
		C \left( \sum_{k'\leq 2} \left\vert \mathcal{T}_p \left[ \tilde{\mathfrak{D}}^{k'} f \right] \right\vert 
		+ \vert \mathfrak{D} \mathcal{T}_p \vert
		+ \left\vert \mathfrak{D} \mathcal{T}_p \left[ \tilde{\mathfrak{D}} \log p^3 \cdot f \right] \right\vert
		+ \left\vert \mathcal{T}_p \left[ \tilde{\mathfrak{D}} \log p^3 \cdot \tilde{\mathfrak{D}} f \right] \right\vert \right),
  \]
  and hence, by Proposition \ref{prop:emcommutation},
  \begin{align*}
  		\vert \mathfrak{D}^2 \mathcal{T}_p \vert
		\leq
		&
		C \Bigg( \sum_{k'\leq 2} \left\vert \mathcal{T}_p \left[ \tilde{\mathfrak{D}}^{k'} f \right] \right\vert 
		+ \vert \mathfrak{D} \mathcal{T}_p \vert
		+ \left\vert \mathcal{T}_p \left[ \tilde{\mathfrak{D}}^2 \log p^3 \cdot f \right] \right\vert
		 + \left\vert \mathcal{T}_p \left[ \tilde{\mathfrak{D}} \log p^3 \cdot f \right] \right\vert
		\\
		&
		+ \left\vert \mathcal{T}_p \left[ \tilde{\mathfrak{D}} \log p^3 \cdot \tilde{\mathfrak{D}} f \right] \right\vert
		+ \left\vert \mathcal{T}_p \left[ \tilde{\mathfrak{D}} \log p^3 \cdot \tilde{\mathfrak{D}} \log p^3 \cdot \tilde{\mathfrak{D}} f \right] \right\vert
		\Bigg).
  \end{align*}
  Proposition \ref{prop:emmain}, Lemma \ref{lem:logp3} and the above bounds for $\vert \mathfrak{D} \mathcal{T}_p \vert$ therefore give,
  \[
  		\vert \mathfrak{D}^2 \mathcal{T}_p \vert \leq \frac{C\bc}{r^p}.
  \]
  
  For the third and fourth order derivatives, Proposition \ref{prop:emmain}, Lemma \ref{lem:logp3}, Proposition \ref{prop:emcommutation}, the pointwise estimates for $\Gamma, \Gammaslash, r \nablaslash b, \mathfrak{D}\Gamma, \mathfrak{D} \Gammaslash, \mathfrak{D} r \nablaslash b$ and the pointwise bounds on $\mathcal{T}, \mathfrak{D} \mathcal{T}, \mathfrak{D}^2 \mathcal{T}$ obtained above similarly give
  \begin{align*}
  	\vert \mathfrak{D}^3 \mathcal{T} \vert
	\leq
	C 
	\sum_{k'\leq 3} \left\vert \mathcal{T}_p \left[ \tilde{\mathfrak{D}}^{k'} f \right] \right\vert
	+ \frac{C \bc}{r^p}
	\Bigg(
	\sum_{p_1 + p_2 \geq 0} \vert \rp_{p_1} \mathfrak{D}^2 \Gamma_{p_2} \vert
	+ \vert \mathfrak{D}^2 r \nablaslash b\vert
	+ \vert r \mathfrak{D}^2 \Gammaslash \vert
	+ 1
	\Bigg),
  \end{align*}
  and
  \begin{align*}
  	\vert \mathfrak{D}^4 \mathcal{T} \vert
	\leq
	&
	C \left(
	\sum_{k'\leq 4} \left\vert \mathcal{T}_p \left[ \tilde{\mathfrak{D}}^{k'} f \right] \right\vert
	+
	\vert \mathfrak{D}^3 \mathcal{T} \vert \right)
	\\
	&
	+ \frac{C \bc}{r^p} \sum_{k' = 2,3}
	\Bigg(
	\sum_{p_1 + p_2 \geq 0} \vert \rp_{p_1} \mathfrak{D}^{k'} \Gamma_{p_2} \vert
	+ \vert \mathfrak{D}^{k'} r \nablaslash b\vert
	+ \vert r \mathfrak{D}^{k'} \Gammaslash \vert
	+ 1
	\Bigg).
  \end{align*}
  The estimates \eqref{eq:em3in} and \eqref{eq:em4} now follow using Proposition \ref{prop:emmain} and the bootstrap assumptions for derivatives of $\Gamma, r\nablaslash b$ and $\Gammaslash$.
  
  To obtain \eqref{eq:em3out} first compute,
  \begin{align*}
  	&
	\partial_u \left( \int_{v_0}^v \int_{S_{u,v'}} r(u,v')^{2p-4} \vert \mathfrak{D}^3 \mathcal{T}_p \vert^2 d\mu_{S_{u,v'}} d v' \right)
	\\
	&
	\qquad
	\leq
	C  \int_{v_0}^v \int_{S_{u,v'}} \nablaslash_3 \left( r(u,v')^{2p-4} \vert \mathfrak{D}^3 \mathcal{T}_p \vert^2 \right) + \tr \chibar r(u,v')^{2p-4} \vert \mathfrak{D}^3 \mathcal{T}_p \vert^2 d\mu_{S_{u,v'}} d v'
	\\
	&
	\qquad
	\leq
	C \int_{v_0}^v \int_{S_{u,v'}} r(u,v')^{2p-4} \left( \vert \mathfrak{D}^3 \mathcal{T}_p \vert^2 + \vert \mathfrak{D}^4 \mathcal{T}_p \vert^2 \right) d\mu_{S_{u,v'}} d v'.
  \end{align*}
  The result then follows by integrating from $u_0$ to $u$ and using \eqref{eq:em4}.
\end{proof}

\subsection{Schematic Notation} \label{subsec:furtherschematic}
To deal with some of the expressions which arise in the remainder of this section it is convenient to introduce further schematic notation.  Like the previous schematic notation introduced in Section \ref{subsec:schemnot}, this notation will make it easy to read off the overall $r$ decay of complicated expressions.

Throughout most of this section we will consider a point $(x,p) \in P \cap \supp(f)$ and the trajectory of the geodesic flow $s \mapsto \exp_s(x,p)$ through this point.  The trajectory will be followed backwards to the initial hypersurface, so $s$ will be negative.  Note that
\[
	\exp_s(x,p) = (\gamma(s),\dot{\gamma}(s)),
\]
where $\gamma$ is the unique geodesic in $\mathcal{M}$ such that $\gamma(0) = x$, $\dot{\gamma}(0) = p$.  The expressions $\exp_s(x,p)$ and $(\gamma(s), \dot{\gamma}(s))$ will be used interchangeably.  Also $\dot{\gamma}^{\mu}(s)$ and $p^{\mu}(s)$ will both be used to denote the $\mu$ component of $\dot{\gamma}(s)$ with respect to the frame $e_1,e_2,e_3,e_4$.  Note the slight change in notation from Section \ref{section:suppf} where $\gamma(0)$ lay on the initial hypersurface $\{v=v_0\}$ and $s$ was positive.

Recall from Section \ref{section:suppf} that, for such a geodesic, $\dot{\gamma}^4(s)$ will remain bounded in $s$ (and in fact will be comparable to $\dot{\gamma}^4(0)$), whilst $\dot{\gamma}^1(s), \dot{\gamma}^2(s), \dot{\gamma}^3(s)$ will all decay like $\frac{1}{r(s)^2}$.  The notation $\dot{\upgamma}_0$ will be used to schematically denote $\dot{\gamma}^4(s)$ and $\dot{\upgamma}_2$ to schematically denote any of $\dot{\gamma}^1(s), \dot{\gamma}^2(s), \dot{\gamma}^3(s)$,
\[
	\dot{\upgamma}_0 = \dot{\gamma}^4 , \qquad \dot{\upgamma}_2 = \dot{\gamma}^1, \dot{\gamma}^2, \dot{\gamma}^3,
\]
so that,
\[
	r^p \vert \dot{\upgamma}_p \vert \leq C p^4(0).
\]

Certain vector fields, $K$, along $\gamma$ will be considered later.  If $K^1,\ldots,K^4$ denote the components of $K$ with respect to the frame\footnote{Not the frame $e_1,\ldots,e_4$.  The reason for this will be explained later.} $\frac{1}{r} e_1,\frac{1}{r} e_2, e_3, e_4$, so that,
\[
	K = K^1 \frac{1}{r} e_1 + K^2 \frac{1}{r} e_2 + K^3 e_3 + K^4 e_4,
\]
then it will be shown that, for all such $K$ considered, $K^3$ will always be bounded along $\gamma$, and $K^1,K^2,K^4$ can grow at most like $r(s)$.  Therefore $\mathcal{K}_0$ will be used to schematically denote $K^3$ and $\mathcal{K}_{-1}$ will schematically denote $K^1,K^2,K^4$,
\[
	\mathcal{K}_0 = K^3, \qquad \mathcal{K}_{-1} = K^1,K^2,K^4,
\]
so it is always true that,
\[
	r^p \vert \mathcal{K}_{p} \vert \leq C.
\]

Finally, let $\eslash$, $\eslash^{-1}$ schematically denote the following quantities,
\[
	\eslash = e_1,e_2, \qquad \eslash^{-1} = e_1^{\flat}, e_2^{\flat},
\]
where $e_A^{\flat}$ denotes the $S_{u,v}$ one-form $\gslash(e_A, \cdot)$, for $A=1,2$.  This notation will be used for schematic expressions involving the components of Weyl curvature components, Ricci coefficients and energy momentum tensor components.  If $\xi$ is a $(0,k)$ $S_{u,v}$ tensor, write,
\[
	\xi \eslash^k,
\]
to schematically denote any of the components,
\[
	\xi_{A_1\ldots A_k} = \xi(e_{A_1},\ldots,e_{A_k}).
\]
Note that, if $r^p \vert \xi \vert_{\gslash} \leq C$, then the bounds, $\vert \gslash^{AB} \vert \leq \frac{C}{r^2}$ imply that,
\[
	\vert \xi(e_{A_1},\ldots, e_{A_k}) \vert \leq C r^{k-p},
\]
where $\vert \cdot \vert$ here denotes the usual absolute value on $\mathbb{R}$.  For example, the sum of the components $\alpha_{AB} + \Tslash_{AB}$ decays like $\frac{1}{r^2}$ and is schematically written,
\[
	\alpha_{AB} + \Tslash_{AB} = \psi_4 \eslash^2 + \mathcal{T}_4 \eslash^2 = \sum_{p_1+p_2 \geq 2} (\psi_{p_1} + \mathcal{T}_{p_1}) \eslash^{-p_2},
\]
where in the summation $p_1 = 4, p_2 = - 2$ for each of the two terms.  Looking at the summation on the right hand side, it is straightforward to read off that each term should decay like $\frac{1}{r^2}$.

Similarly, if $\xi$ is a $(k,0)$ $S_{u,v}$ tensor, write,
\[
	\xi e^{-k},
\]
to denote any of the components
\[
	\xi^{A_1 \ldots A_k} = \xi(e_{A_1}^{\flat},\ldots,e_{A_k}^{\flat}).
\]
Since $\vert \gslash_{AB} \vert \leq C r^2$, if $r^p \vert \xi \vert_{\gslash} \leq C$ then,
\[
	\vert \xi^{A_1 \ldots A_k} \vert \leq C r^{-p-k}.
\]
For example, allowing $\psi$ to also denote $\psi^{\sharp}$, $\mathcal{T}$ to also denote $\mathcal{T}^{\sharp}$ etc.\@,
\[
	\alpha^{AB} + \Tslash^{AB} = \psi_4 \eslash^{-2} + \mathcal{T}_4 \eslash^{-2} = \sum_{p_1 + p_2 \geq 6} ( \psi_{p_1} + \mathcal{T}_{p_1}) \eslash^{-p_2},
\]
where in the summation now $p_1 = 4, p_2 = 2$ for each of the two terms.  Again, the subscript of the summation on the right hand side allows us to immediately read off that the components $\alpha^{AB}, \Tslash^{AB}$ decay like $\frac{1}{r^6}$.

Finally, if $\xi$ is a $(k_1,k_2)$ $S_{u,v}$ tensor, write,
\[
	\xi \eslash^{k_2 - k_1},
\]
to schematically denote any of the components,
\[
	{\xi^{A_1\ldots A_{k_1}}}_{B_1\ldots B_{k_2}}.
\]
For example,
\[
	{\alpha_A}^B + {\Tslash_A}^B = \sum_{p_1+p_2 \geq 4} (\psi_{p_1} + \mathcal{T}_{p_2}) \eslash^{-p_2},
\]
and it can immediately be read off from the subscript of the summation that the components ${\alpha_A}^B, {\Tslash_A}^B$ decay like $\frac{1}{r^4}$.  Note that, in this notation, it is clearly not necessarily the case that no $e_1,e_2, e_1^{\flat}, e_2^{\flat}$ appear in the expression
\[
	\xi \eslash^0.
\]

\subsection{Vector Fields on the Mass Shell} \label{subsec:overview}
Consider the vectors $V_{(1)},\ldots,V_{(6)} \in T_{(x,p)}P$ defined by,
\[
	 V_{(A)} = \hor_{(x,p)} (e_A) + \frac{p^4}{r} \partial_{\overline{p}^A}, \quad V_{(3)} = \hor_{(x,p)} (e_3),
\]
\[
	 V_{(4)} = r\hor_{(x,p)} (e_4) + p^4\partial_{\overline{p}^4}, \quad V_{(4+A)} = \frac{p^4}{r^2}\partial_{\overline{p}^A},
\]
for $A = 1,2$.  The proof of Proposition \ref{prop:emmain} reduces to the following.
\begin{proposition} \label{prop:emmain2}
	At any point $x \in \mathcal{A}$, if $u = u(x)$, $v = v(x)$ then, for $i_1,i_2,i_3,i_4 = 1,\ldots,6$,
	\[
		\left( r^p \left\vert \mathcal{T}_p \left[ V_{(i_1)} f \right] \right\vert \right) (u,v) \leq C \bc,
		\qquad
		\left( r^p \left\vert \mathcal{T}_p \left[ V_{(i_2)} V_{(i_1)} f \right] \right\vert \right) (u,v) \leq C \bc,
	\]
	\[
		\int_{u_0}^u \int_{S_{u',v}} r(u',v)^{2p-2} \left\vert \mathcal{T}_p \left[ V_{(i_3)} V_{(i_2)} V_{(i_1)} f \right] \right\vert^2 d\mu_{S_{u',v}} d u' \leq C\bc,
	\]
	and
	\[
		\int_{u_0}^u \int_{v_0}^v \int_{S_{u',v'}} r(u',v')^{2p-4} \left\vert \mathcal{T}_p \left[ V_{(i_4)} V_{(i_3)} V_{(i_2)} V_{(i_1)} f \right] \right\vert^2 d\mu_{S_{u',v'}} d v' d u' \leq C \bc.
	\]
\end{proposition}

The vectors $V_{(1)},\ldots,V_{(6)}$, together with $V_{(0)}$ given by,
\[
	V_{(0)} = r\hor_{(x,p)} (e_4),
\]
form a basis for $T_{(x,p)}P$.  They are preferred to the operators $\tilde{\mathfrak{D}}$ introduced in Section \ref{subsec:testimates} as, in view of Proposition \ref{prop:jacobi}, it is much more natural to work with vectors divided into their horizontal and vertical parts.  It will be shown below that $\vert V_{(i)} f \vert$ is uniformly bounded for $i = 0, \ldots, 6$.  

\begin{remark}
	It is not the case that $\vert \hor_{(x,p)}(e_A) f \vert$ is uniformly bounded, for $A = 1,2$.  In fact, $\vert \hor_{(x,p)}(e_A) f \vert$ grows at the rate $r(x)$ as $r(x) \to \infty$.  It is for this reason the term $\frac{p^4}{r} \partial_{\overline{p}^A}$ also appears in $V_{(A)}$.  A cancellation will later be seen to occur in these two terms, so that $\vert V_{(A)} f \vert$ is uniformly bounded.
	
	In Section \ref{subsec:jacobidef}, the vector fields $V_{(1)}, \ldots, V_{(6)}$ are used to define corresponding Jacobi fields, $J_{(1)},\ldots, J_{(6)}$, along $\exp_s(x,p)$.  The boundedness of $\vert V_{(i)} f \vert$ will follow from bounds on the components of $J_{(i)}$ for $i = 1,\ldots, 6$.  Whilst it is true that $\vert V_{(0)} f \vert$ is uniformly bounded, the appropriate bounds for the components of the Jacobi field corresponding to $V_{(0)}$ do not hold.  This is the reason $V_{(0)}$ derivatives of $f$ are treated separately in the proof of Proposition \ref{prop:emmain} below, and do not appear in Proposition \ref{prop:emmain2}.  See also the discussion in Remark \ref{rem:horef}.
	
	Similarly, the components of the Jacobi field corresponding to $ p^4 \partial_{\overline{p}^4}$ do not satisfy the appropriate bounds.  The $r \hor_{(x,p)}(e_4)$ term in $V_{(4)}$ appears for this reason.  The bound on $\vert p^4 \partial_{\overline{p}^4} f\vert$ can easily be recovered from the bound on $\vert V_{(4)} f\vert$ using the below observation that the Vlasov equation can be used to re-express $r \hor_{(x,p)} (e_4) f$ in terms of other derivatives of $f$.
\end{remark}

Below is a sketch of how Proposition \ref{prop:emmain} follows from Proposition \ref{prop:emmain2}.  

\begin{proof}[Proof of Proposition \ref{prop:emmain}]
Recall the point $(x,p)\in \supp (f)$ is fixed.  The goal is to estimate $\tilde{\mathfrak{D}}^k_{C_1,\ldots,C_l}f$ where $\tilde{\mathfrak{D}} \in \{ re_1, re_2, e_3,re_4\}$, $k\leq4$, and $l\leq k$ is the number of times $re_1$ or $re_2$ appears in $\tilde{\mathfrak{D}}^k$.  Since $V_{(0)}, \ldots, V_{(6)}$ span $T_{(x,p)}P$, clearly $\tilde{\mathfrak{D}}^k_{C_1,\ldots,C_l}f$ can be written as a combination of terms of the form $V_{(i_1)} \ldots V_{(i_{k'})} f$, where $k'\leq k$ and $i_1,\ldots,i_{k'} = 0,\ldots,6$.  It remains to check that the $V_{(0)}$ can be eliminated and then that the coefficients in the resulting expressions behave well.  This is done in the following steps.
\begin{enumerate}
  \item 	
  First rewrite $\tilde{\mathfrak{D}}^k_{C_1,\ldots,C_l}f$ as
  \[
	 \tilde{\mathfrak{D}}^k_{C_1,\ldots,C_l}f = r^l \hat{\mathfrak{D}}^k_{C_1,\ldots,C_l}f + \sum_{k'=1}^l C_{k'} r^{p_{k'}} \hat{\mathfrak{D}}^{k-k'}_{C_1,\ldots,C_l}f,
  \]
where $\hat{\mathfrak{D}} \in \{ e_1,e_2,e_3,re_4\}$, the $C_{k'}$ are constants and $p_{k'}$ are powers such that $p_{k'} \leq l$.  The terms in the sum are lower order and so, by induction, can be viewed as having already been estimated ``at the previous step'', so they are ignored from now on.  The $r^l$ factor in the first term will vanish when the norm is taken with the metric $\gslash$ and so is also ignored.

  \item 	Rewrite each $\hat{\mathfrak{D}}$ in terms of the vectors $V_{(0)},\ldots, V_{(6)}$ defined above,
  \begin{align*}
  	e_A
	=
	&
	V_{(A)} + r^2 \left( \frac{p^B}{p^4} \Gammaslash_{AB}^C + \frac{p^3}{p^4} {\chibar_A}^C + {\chi_A}^C \right) V_{(4+C)}
	+ \left( \frac{p^B}{2p^4} \chibar_{AB} + \etabar_A \right) V_{(4)}
	\\
	&
	- \left( \frac{p^B}{2p^4} \chibar_{AB} + \etabar_A \right) V_{(0)},
	\\
	e_3
	=
	&
	V_{(3)} + r^2 \left( \frac{p^A}{p^4} {\chibar_A}^B + 2 \eta^B \right) V_{(4+A)},
	\\
	re_4
	=
	&
	\left( 1 - \frac{rp^A}{p^4} \etabar_A - r \omega \right) V_{(0)} + r^3 \left( \frac{p^A}{p^4} {\chi_A}^B - \frac{p^A}{p^4} e_A(b^B) + \frac{2p^3}{p^4} \etabar^B \right) V_{(4+B)}
	\\
	&
	+ r \left( \frac{p^A}{p^4} \etabar_A + \omega \right) V_{(4)}.
  \end{align*}
  \item 	In the resulting expression bring all of the coefficients of the vectors $V_{(0)},\ldots,V_{(6)}$ out to the front to get
  \[
	 \hat{\mathfrak{D}}^k_{C_1,\ldots,C_l}f = \sum_{1\leq k' \leq k} \sum_{i_1,\ldots, i_{k'}} d_{i_1 \ldots i_{k'}} V_{(i_1)} \ldots V_{(i_{k'})} f,
  \]
  where the $d_{i_1 \ldots i_{k'}}$ are combinations of $\rp_p$ terms and derivatives of components of $\Gamma, \Gammaslash$ and $b$.  Clearly $d_{i_1 \ldots i_{k'}}$ involves at most $k-k'$ derivatives of the components of $\Gamma$ and $\Gammaslash$, and $k-k' + 1$ derivatives of components of $b$.  Moreover, using the bootstrap assumptions from Section \ref{section:ba} they are bounded with respect to $r$.  Hence $\vert d_{i_1 \ldots i_{k'}} \vert \leq C$ for $k' = 3,4$, $\int_{S_{u,v}} r^{-2} \vert d_{i_1 \ldots i_{k'}} \vert^2 d \mu_{S_{u,v}} \leq C$ for $k'=2$ and $\int_{v_0}^v \int_{S_{u,v'}} r^{-4} \vert d_{i_1 \ldots i_{k'}} \vert^2 d \mu_{S_{u,v'}} dv' \leq C$ for $k' = 1$.
  \item
	 For each $V_{(i_1)} \ldots V_{(i_{k'})} f$ in the above expression containing at least one $V_{(0)}$, commute to bring the innermost to the inside.  Relabelling if necessary, this gives
	 \begin{align*}
		V_{(i_1)} \ldots V_{(i_{k'})} f = 
		&
		\
		d_{i_1 \ldots i_{k'-1} 0} V_{(i_1)} \ldots V_{(i_{k'-1})} V_{(0)} f 
		\\
		&
		+
		\sum_{1\leq k'' \leq k' -1} \sum_{j_1,\ldots, j_{k''}} d_{j_1 \ldots j_{k''}} V_{(j_1)} \ldots V_{(j_{k''})} f,
	 \end{align*}
	 for some (new) $d_{j_1 \ldots j_{k''}}$ as above.
  \item
	 Use the Vlasov equation $p^{\mu} \hor_{(x,p)} (e_{\mu}) f = 0$ (see Example \ref{ex:generator}) to rewrite
	 \[
		V_{(0)} f = - \frac{rp^3}{p^4} \hor_{(x,p)} (e_3) f - \frac{rp^A}{p^4} \hor_{(x,p)} (e_A) f.
	 \]
	 Rewrite this right hand side in terms of $V_{(1)},\ldots,V_{(6)}$ and repeat Step (3) to bring the coefficients to the outside of the expression.
  \item
	 Repeat steps (4) and (5) to eliminate all of the $V_{(0)}$ terms to leave
	 \[
		\hat{\mathfrak{D}}^k_{C_1,\ldots,C_l}f = \sum_{1\leq k' \leq k} \sum_{i_1,\ldots, i_{k'}=1}^6 d_{i_1 \ldots i_{k'}} V_{(i_1)} \ldots V_{(i_{k'})} f,
	 \]
	 where the $d_{i_1 \ldots i_{k'}}$ have the correct form as above.
\end{enumerate}
The result now clearly follows from Proposition \ref{prop:emmain2}.

\end{proof}

\begin{remark} \label{rem:horef}
	The vector $V_{(0)} = r \hor_{(x,p)}(e_4)$ should be compared to the vector
	\[
		S := v \hor_{(x,p)} (e_4) + u \hor_{(x,p)} (e_3).
	\] 
	In the context of Theorem \ref{thm:main3}, where $u_0 \leq u \leq u_f$ in $\supp(f)$, the $u \hor_{(x,p)} (e_3)$ term is not the dominant one in $S$.  Recall also that $v$ here is comparable to $r$.
	
	Given any vector $V \in T_{(x,p)} P$, in the sections to follow it is shown that there exists a vector field $J$ on $P$ such that $Jf$ satisfies the Vlasov equation and $J$ coincides with $V$ at the point $(x,p) \in P$.  For $V_{(1)}, \ldots, V_{(6)}$ it will be shown that the corresponding $J$ are all of size 1 (independent of the point $x$) at the initial hypersurface $\{ v = v_0 \}$.  Whilst this is not the case for the vector field $J$ corresponding to $S$, a form of the following observation was used in the proof of Proposition \ref{prop:emmain}.  The vector field $J$ corresponding to $S$ has a large component, but this component points in the $X$ direction and hence vanishes when applied to $f$.  A manifestation of this fact is brought to light through the fact that $[X_M, S] = X_M$, where $X_M$ denotes the generator of the null Minkowski geodesic flow and, with a slight abuse of notation, $S$ now denotes the vector field $S = v \hor_{(x,p)} (e_4) + u \hor_{(x,p)} (e_3)$ on the mass shell over Minkowski space.
	
	Note that this observation is not specific to the massless case, i.\@e.\@ the identity $[X_M, S] = X_M$ is still true if $X_M$ now denotes the Minkowski geodesic flow restricted to the hypersurface $P_m = \{ (x,p) \in \TM \mid p \text{ future directed, } g_{\text{Mink}} (p,p) = -m^2\}$ for $m>0$.  A form of this observation is used in the work \cite{FaJoSm}.
\end{remark}

\subsection{The Jacobi Fields} \label{subsec:jacobidef}
Define vector fields $J_{(1)}, \ldots ,J_{(6)}$ along the trajectory of the geodesic flow $s \mapsto \exp_s(x,p)$ by 
\[
	J_{(i)}(s) = d \exp_{s} |_{(x,p)} V_{(i)},
\]
for $i = 1,\ldots ,6$.

Since
\[
	f(x,p) = f(\exp_s (x,p)),
\]
by the chain rule,
\[
	V_{(i)}f(x,p) = df|_{(x,p)} V_{(i)} = df|_{\exp_s(x,p)} \cdot d \exp_{s} |_{(x,p)} V_{(i)} = J_{(i)}(s) f.
\]
Taking $s<0$ so that $\exp_s(x,p)$ lies on the mass shell over the initial hypersurface $\{ v = v_0\}$, this relates $V_{(i)}f(x,p)$ to intial data.  By Proposition \ref{prop:jacobi}, $J_{(i)}$ is in fact a Jacobi field and hence $J_{(i)}(s)$ can be controlled using the Jacobi equation.

Note that so far the Jacobi fields are only defined along the trajectory $s \mapsto \exp_s(x,p)$.  Since higher order derivatives of $f$ will be taken it is necessary to define them in a neighbourhood of the geodesic $s \mapsto \exp_s(x,p)$ in $P$.  They are in general defined differently depending on what the higher order derivatives to be taken are.  When considering the quantity
\[
	V_{(i_k)} \ldots V_{(i_1)} f,
\]
for $2 \leq k \leq 4$ the Jacobi fields are extended so that
\[
	J_{(i_k)} \ldots J_{(i_1)} \vert_{s=0} = V_{(i_k)} \ldots V_{(i_1)},
\]
as follows.

If $k=2$ then define a curve $c_1 : (-\epsilon, \epsilon) \to P$, for some small $\epsilon > 0$, such that 
\[
	c_1(0) = (x,p), \qquad c_1'(s_1) = V_{(i_2)},
\]
i.\@e.\@ $c_1$ is the integral curve of $V_{(i_2)}$ through $(x,p)$.  Set $J_{(i_1)} = V_{(i_1)}$ along $c_1$ and let $J_{(i_1)}(s,s_1) = d \exp_s |_{c_1(s_1)} V_{(i_1)}$.  Now the expression $J_{(i_2)} J_{(i_1)} f$ is defined along $\exp_s(x,p)$ and has the desired property that $J_{(i_2)} J_{(i_1)}|_{s=0} = V_{(i_2)} V_{(i_1)}$.

Similarly, if $k=3$ define a variation of curves $c_2: (-\epsilon, \epsilon)^2 \to P$ so that
\[
	c_2(0,0) = (x,p), \qquad \frac{\partial c_2}{\partial s_1}(s_1,0) = V_{(i_3)}, \qquad \frac{\partial c_2}{\partial s_2}(s_1,s_2) = V_{(i_2)}.
\]
So the curve $s_1 \mapsto c_2(s_1,0)$ is the integral curve of $V_{(i_3)}$ through $(x,p)$, and, for fixed $s_1$, the curve $s_2 \mapsto c_2(s_1,s_2)$ is the integral curve of $V_{(i_2)}$ through $c_2(s_1,0)$.  Set
\[
	J_{(i_1)}(s,s_1,s_2) = d \exp_s |_{c_2(s_1,s_2)} V_{(i_1)},
\]
for $s_1,s_2 \in (-\epsilon, \epsilon), s<0$, and,
\[
	J_{(i_2)}(s,s_1,0) = d \exp_s |_{c_2(s_1,0)} V_{(i_2)},
\]
for $s_1 \in (-\epsilon, \epsilon), s<0$.  Now the expression $J_{(i_3)} J_{(i_2)} J_{(i_1)} f$ is defined along $\exp_s(x,p)$ and moreover,
\[
	J_{(i_3)} J_{(i_2)} J_{(i_1)} |_{s=0} = V_{(i_3)} V_{(i_2)} V_{(i_1)}.
\]

Finally, if $k=4$, similarly define $c_3: (-\epsilon, \epsilon)^3 \to P$ so that
\[
	c_3(0,0,0) = (x,p), \qquad \frac{\partial c_3}{\partial s_1}(s_1,0,0) = V_{(i_4)},
\]
\[
	\frac{\partial c_3}{\partial s_2}(s_1,s_2,0) = V_{(i_3)} \qquad \frac{\partial c_3}{\partial s_3}(s_1,s_2,s_3) = V_{(i_2)},
\]
and similarly set
\begin{align*}
	J_{(i_1)}(s,s_1,s_2,s_3) 
	&
	= d \exp_s |_{c_3(s_1,s_2,s_3)} V_{(i_1)},
	\\
	J_{(i_2)}(s,s_1,s_2,0) 
	&
	= d \exp_s |_{c_2(s_1,s_2,0)} V_{(i_2)},
	\\
	J_{(i_3)}(s,s_1,0,0) 
	&
	= d \exp_s |_{c_2(s_1,0,0)} V_{(i_3)},
\end{align*}
for $s_1,s_2,s_3 \in (-\varepsilon,\varepsilon), s<0$.

\subsection{Two Frames for $P$ and Components of Jacobi Fields} \label{subsec:fjacobi}
Let $s_* \geq 0$ be the time such that $\pi(\exp_{-s_*}(x,p)) \in \{v = v_0\}$, where $\pi:P\to \mathcal{M}$ is the natural projection.  The definition of the Jacobi fields of Section \ref{subsec:jacobidef} imply that, for $1 \leq k\leq 4$,
\[
	V_{(i_k)} \ldots V_{(i_1)} f (x,p) = J_{(i_k)} \ldots J_{(i_1)} f|_{s=-s_*},
\]
and so Proposition \ref{prop:emmain2} will follow from appropriate estimates for $J_{(i_k)} \ldots J_{(i_1)} f |_{s=-s_*}$.

Recall from Section \ref{section:setup} that $\overline{p}^1,\overline{p}^2, \overline{p}^4$ denote the restrictions of $p^1,p^2,p^3$ to $P$ and $\partial_{\overline{p}^1}, \partial_{\overline{p}^2}, \partial_{\overline{p}^4}$ denote the corresponding partial derivatives with respect to the $(u,v,\theta^1,\theta^2,\overline{p}^1,\overline{p}^2, \overline{p}^4)$ coordinate system for $P$.  For every $(x,p)\in P$ define the frame $E_1,\ldots,E_7$ of horizontal and vertical vectors for $P$ by
\[
	E_1
	= \hor_{(x,p)} \left( \frac{1}{r} e_1 \right),
	\quad
	E_2
	= \hor_{(x,p)} \left( \frac{1}{r} e_2 \right),
	\quad
	E_3
	= \hor_{(x,p)} \left( e_3 \right),
\]
\[
	E_4
	= \hor_{(x,p)} \left( e_4 \right),
	\quad
	E_5
	= \frac{1}{r(x,p)} \partial_{\overline{p}^1},
	\quad
	E_6
	= \frac{1}{r(x,p)} \partial_{\overline{p}^2},
	\quad
	E_7
	= \partial_{\overline{p}^4}.
\]
Recall the expressions \eqref{eq:pbar}, which imply,
\[
	E_5
	= \frac{1}{r(x,p)} \ver_{(x,p)} \left( e_1 + \frac{\gslash_{1A}p^A}{2p^4} e_3 \right),
	\quad
	E_6
	= \frac{1}{r(x,p)} \ver_{(x,p)} \left( e_2 + \frac{\gslash_{2A}p^A}{2p^4} e_3 \right),
\]
\[
	E_7
	= \ver_{(x,p)} \left( e_4 - \frac{p^3}{p^4} e_3 \right).
\]
The vectors $\frac{1}{r} e_A$, for $A = 1,2$, are used rather than the vectors $e_A$ which grow like $r$.  For $J\in \{ J_{(1)},\ldots,J_{(6)}\}$ let $J^j$ denote the components of $J$ with respect to this frame.  So
\[
	J = J^j E_j.
\]
Define also the frame $F_1,\ldots,F_7$ for $P$ along $s\mapsto \exp_s(x,p)$ by
\[
	F_i
	= \partrans_{(\gamma,\dot{\gamma})} \left( E_i \right)
\]
for $i=1,\ldots,7$.  Here $\gamma (s) = \pi (\exp_s(x,p))$ denotes the geodesic in $\mathcal{M}$, so that $\exp_s(x,p) = (\gamma(s),\dot{\gamma}(s))$, and $\partrans_{(\gamma,\dot{\gamma})}$ denotes parallel transport along $(\gamma(s),\dot{\gamma}(s))$.

Let $\Xi$ denote the change of basis matrix from $\{ F_i \}$ to $\{ E_i\}$, so that
\begin{equation} \label{eq:Xidef}
	E_i = {\Xi_i}^j F_j.
\end{equation}
Note that, at $s=0$,
\[
	 {\Xi_i}^j \vert_{s=0} = {\delta_i}^j.
\]

\begin{remark}
	In the following, when tensor fields on $P$ are written in components, these will always be components with respect to the frame $\{E_i\}$.  So if $J$ is a Jacobi field then $J^i$ denote the components such that
  \[
	 J = J^i E_i.
  \]
  When writing components with respect to the parallelly transported frame $\{F_i\}$, the matrix $\Xi$ will always be used.  So,
  \[
	 J = J^i {\Xi_i}^j F_j.
  \]
  Latin indices $i,j$ will always run from $1,\ldots, 7$.
\end{remark}

It will be necessary in the following sections to estimate the components of $\Xi$ and $\Xi^{-1}$, and certain derivatives, along $(\gamma(s),\dot{\gamma}(s))$.  It is therefore necessary to derive equations satisfied by the components of $\Xi$ and $\Xi^{-1}$.

\begin{proposition}
	The components of the matrices $\Xi$ and $\Xi^{-1}$ satisfy the equations,
	\begin{equation} \label{eq:Xi}
	 	\frac{d {\Xi_i}^j}{ds} (s) = \left( \hat{\nabla}_X E_i \right)^k {\Xi_k}^j(s),
	\end{equation}
	and
	\begin{equation} \label{eq:Xiinverse}
		 \frac{d {{\Xi^{-1}}_i}^j}{ds}(s) = - {{\Xi^{-1}}_i}^k \left( \hat{\nabla}_X E_k \right)^j (s),
	\end{equation}
	respectively, for $i,j = 1,\ldots,7$.  Here $\left( \hat{\nabla}_X E_i \right)^k$ denote the components of $\hat{\nabla}_X E_i$ with respect to $E_1,\ldots,E_7$.
\end{proposition}

\begin{proof}
Using the fact that,
\[
	 \hat{\nabla}_X F_j = 0,
\]
for all $j$, equation \eqref{eq:Xidef} gives,
\[
	 \hat{\nabla}_X E_i = X\left( {\Xi_i}^j \right) F_j.
\]
This can be written as the system of equations \eqref{eq:Xi}.

Similarly, writing $F_j = {{\Xi^{-1}}_i}^j E_j$ gives,
\[
	 0 = \hat{\nabla}_X F_i = \left[ X({{\Xi^{-1}}_i}^j) + {{\Xi^{-1}}_i}^k \left( \hat{\nabla}_X E_k \right)^j \right] E_j.
\]
This yields the system \eqref{eq:Xiinverse}.
\end{proof}

Now,
\[
	J_{(i_2)} J_{(i_1)} f = J_{(i_2)} \left( {J_{(i_1)}}^{j_1} \right) E_{j_1} f + {J_{(i_2)}}^{j_2} {J_{(i_1)}}^{j_1} E_{j_2} E_{j_1} f,
\]
so estimates for $J_{(i_2)} J_{(i_1)} f \vert_{s=-s_*}$ follow from estimates for the components ${J_{(i_2)}}^{j_2} {J_{(i_1)}}^{j_1}$ and derivatives $J_{(i_2)} \left( {J_{(i_1)}}^{j_1} \right)$ at $s= -s_*$ since $E_{j_1} f \vert_{s=-s_*}$ and $E_{j_2}E_{j_1} f \vert_{s=-s_*}$ are assumed to be bounded pointwise by assumption.  See Theorem \ref{thm:main3}.  Higher order derivatives can similarly be expressed in terms of derivatives of components of Jacobi fields.  This is discussed further in Section \ref{subsec:proofemmain2}.  It is hence sufficient to just estimate the derivatives of components,
\[
	J_{(i_k)} \ldots J_{(i_{2})} \left( {J_{(i_1)}}^j  \right),
\]
for $k = 1,\ldots, 4, j = 1,\ldots,7$.  These estimates are obtained in the next two subsections in Propositions \ref{prop:Jbounds}, \ref{prop:JJbounds}, \ref{prop:JJJbounds}, \ref{prop:JJJJbounds}.  In Section \ref{subsec:proofemmain2} they are then used to prove Proposition \ref{prop:emmain2}.

\subsection{Pointwise Estimates for Components of Jacobi Fields at Lower Orders} \label{subsec:lowerjacobi}
For the fixed point $(x,p) \in P \cap \supp(f)$, recall that $s_* = s_*(x,p)$ denotes the parameter time $s$ such that $\pi (\exp_{-s_*} (x,p)) \in \{ v = v_0 \}$.  The goal of this section is to show that the components of the Jacobi fields $J_{(1)}, \ldots, J_{(6)}$, with respect to the frame $E_1 ,\ldots , E_7$, are bounded, independently of $(x,p)$, at the parameter time $s= -s_*$, and then similarly for the first order derivatives $J_{(i_2)} ( {J_{(i_1)}}^j)$, for $i_1,i_2 = 1,\ldots, 6$, $j = 1,\ldots, 7$.  Second and third order derivatives of the Jacobi fields are estimated in Section \ref{subsec:higherjacobi}.

The estimates for the components of $J = J_{(1)},\ldots, J_{(6)}$ are obtained using the Jacobi equation
\[
	 \hat{\nabla}_X \hat{\nabla}_X J = \hat{R}(X,J) X,
\]
which in components takes the form,
\begin{equation} \label{eq:jacobicomponents}
	\frac{d^2 J^k {\Xi_k}^j}{ds^2} = \left( \hat{R}(X,J)X \right)^k {\Xi_k}^j.
\end{equation}
Recall that $\hat{R}$ denotes the curvature tensor of the induced Sasaki metic on $P$ and, by Proposition \ref{prop:curvaturemassshell}, $\hat{R}(X,J)X$ can be expressed in terms of the curvature tensor $R$ of $( \mathcal{M},g)$ as follows,
\begin{align} \label{eq:curvaturenotuseful}
\begin{split}
	\hat{R}(X,J)X =
	& \
	\hor_{(\gamma,\dot{\gamma})} \left( R(\dot{\gamma},J^h) \dot{\gamma} + \frac{3}{4} R(\dot{\gamma},R(\dot{\gamma},J^h) \dot{\gamma}) \dot{\gamma} + \frac{1}{2} (\nabla_{\dot{\gamma}} R)(\dot{\gamma},J^v) \dot{\gamma} \right)
	\\
	& \
	+ \frac{1}{2} {}^T\ver_{(\gamma,\dot{\gamma})} \left( (\nabla_{\dot{\gamma}} R)(\dot{\gamma},J^h) \dot{\gamma} + \frac{1}{2} R(\dot{\gamma},R(\dot{\gamma},J^v) \dot{\gamma}) \dot{\gamma}  \right).
\end{split}
\end{align}
Here $J^h$ and $J^v$ are defined by
\[
	J^h = J^A \frac{1}{r} e_A + J^3 e_3 + J^4e_4, \qquad J^v = J^{4+A}\frac{1}{r} e_A + \left( \frac{J^{4+A} \gslash_{AB} p^B}{2rp^4} - \frac{J^7 p^3}{p^4} \right) e_3 + J^7 e_4,
\]
so that
\[
	J\vert_{(\gamma,\dot{\gamma})} = \hor_{(\gamma,\dot{\gamma})}(J^h\vert_{(\gamma,\dot{\gamma})}) + {}^T\ver_{(\gamma,\dot{\gamma})}(J^v\vert_{(\gamma,\dot{\gamma})}).
\]
For a vector $Y \in T_{\gamma} \mathcal{M}$, ${}^T \ver_{(\gamma,\dot{\gamma})}(Y)$ denotes the projection of $\ver_{(\gamma,\dot{\gamma})}(Y)$ to $P$.  So, for each $(y,q) \in P$ on the trajectory of the geodesic flow through $(x,p)$, 
\[
	 J^h \vert_{(y,q)}, J^v \vert_{(y,q)} \in T_y \mathcal{M}.
\]
It is tempting to view $J^h,J^v$ as vector fields on $\mathcal{M}$, though this is not strictly the case as the value of $J^h\vert_{(y,q)},J^v \vert_{(y,q)}$ depends not only on $y$ but also on $q$.  Some care therefore needs to be taken here.

Given that derivatives of components of the energy momentum tensor appear in the Bianchi equations as error terms, it is important to estimate derivatives of $\mathcal{T}$, and hence the components of $J_{(1)}, \ldots, J_{(6)}$, at one degree of differentiability greater than the Weyl curvature components $\psi$.  The presence of the derivatives of $R$ in the expression \eqref{eq:curvaturenotuseful} therefore seem, at first glance, to be bad.  On closer inspection however, one sees that such terms are always horizontal or vertical lifts of derivatives of $R$ in the $\dot{\gamma}$ direction.  Since the components of $J_{(1)}, \ldots, J_{(6)}$ are estimated by integrating the Jacobi equation \eqref{eq:jacobicomponents} twice in $s$, the fact that the potentially problematic terms, when integrated in $s$, in fact lie at the same level of differentiability as $R$ can be taken advantage of.  In other words, the derivatives of $R$ appearing on the right hand side of the Jacobi equation \eqref{eq:jacobicomponents} always point in exactly the correct direction so that transport estimates for the Jacobi equation \eqref{eq:jacobicomponents} recover this loss.  In order to exploit this fact, it is convenient to rewrite the expression \eqref{eq:curvaturenotuseful}, using Proposition \ref{prop:sasakiconnection}, as,
\begin{align} \label{eq:curvatureuseful}
\begin{split}
	 \hat{R}(X,J)X = 
	 & \
	 \hor_{(\gamma,\dot{\gamma})} \bigg[ R(\dot{\gamma},J^h) \dot{\gamma} + \frac{1}{2} R(\dot{\gamma},R(\dot{\gamma},J^h) \dot{\gamma}) \dot{\gamma} - \frac{1}{2} X((J^v)^{\mu}) R(\dot{\gamma},e_{\mu}) \dot{\gamma}
	 \\
	 & \
	 - \frac{1}{2}(J^v)^{\mu} R(\dot{\gamma},\nabla_{\dot{\gamma}} e_{\mu}) \dot{\gamma} - \frac{1}{2} X(\dot{\gamma}^{\alpha}) \left( R(e_{\alpha},J^v) \dot{\gamma} + R(\dot{\gamma},J^v)e_{\alpha} \right) 
	 \\
	 & \
	 - \frac{1}{2} \dot{\gamma}^{\alpha} \left( R(\nabla_{\dot{\gamma}} e_{\alpha},J^v) \dot{\gamma} + R(\dot{\gamma},J^v) \nabla_{\dot{\gamma}} e_{\alpha} \right) \bigg]
	 \\
	 & \
	 + \frac{1}{2} {}^T\ver_{(\gamma,\dot{\gamma})} \bigg[ R(\dot{\gamma}, R(\dot{\gamma},J^v)\dot{\gamma})\dot{\gamma} - X((J^h)^{\mu}) R(\dot{\gamma},e_{\mu})\dot{\gamma}
	 \\
	 & \
	 -X(\dot{\gamma}^{\alpha}) \left( R(e_{\alpha},J^h )\dot{\gamma} + R(\dot{\gamma},J^h)e_{\alpha} \right)
	 \\
	 & \
	 - \dot{\gamma}^{\alpha} \left( R( \nabla_{\dot{\gamma}} e_{\alpha},J^h )\dot{\gamma} + R(\dot{\gamma},J^h) \nabla_{\dot{\gamma}} e_{\alpha} \right) - (J^h)^{\mu} R(\dot{\gamma}, \nabla_{\dot{\gamma}} e_{\mu}) \dot{\gamma} \bigg]
	 \\
	 &
	 + \frac{1}{2} \hat{\nabla}_X \hor_{(\gamma,\dot{\gamma})} \left( R(\dot{\gamma},J^v) \dot{\gamma} \right) + \frac{1}{2} \hat{\nabla}_X {}^T\ver_{(\gamma,\dot{\gamma})} \left( R(\dot{\gamma},J^h) \dot{\gamma} \right).
\end{split}
\end{align}

The above observations explain how \eqref{eq:jacobicomponents} can be used to give good estimates for $J_{(1)}, \ldots, J_{(6)}$ from the point of view of regularity.  In order to obtain global estimates however, it is also important to see that \eqref{eq:curvatureuseful} has the correct behaviour in $r$ so as to be twice globally integrable.  This can be seen by rewriting \eqref{eq:curvatureuseful} in terms of $\psi$, $\mathcal{T}$ and using the bootstrap assumptions \eqref{eq:curvatureba}, \eqref{eq:emba}, along with the the asymptotics for $p^1,p^2, p^3, p^4$ obtained in Section \ref{section:suppf}, being sure to allow certain components of $J$ to grow like $r$.  Consider, for example, just the first term $\hor_{(\gamma,\dot{\gamma})} \left( R(\dot{\gamma},J^h) \dot{\gamma} \right)$ in \eqref{eq:curvatureuseful}.  Recall the identity,
\[
	R_{\alpha \beta \gamma \delta} = W_{\alpha \beta \gamma \delta} + \frac{1}{2}(g_{\alpha \gamma}T_{\beta \delta} + g_{\beta \delta}T_{\alpha \gamma} - g_{\beta \gamma} T_{\alpha \delta} - g_{\alpha \delta} T_{\beta \gamma}).
\]
For a vector field $K$ along $\gamma$ in $\mathcal{M}$, let $K^{\mu}$ denote the components of $K$ with respect to $\frac{1}{r}e_1,\frac{1}{r}e_2,e_3,e_4$,
\[
	K = \frac{1}{r} K^A e_A + K^3e_3 + K^4e_4.
\]
Using the form of the metric in the double null frame,
\begin{align*}
	R(\dot{\gamma},K)\dot{\gamma}
	=
	&
	g^{\mu\nu} R(\dot{\gamma},K,\dot{\gamma}, e_{\mu})e_{\nu}
	\\
	=
	&
	-\frac{1}{2} R(\dot{\gamma},K,\dot{\gamma}, e_{3})e_{4} - \frac{1}{2} R(\dot{\gamma},K,\dot{\gamma}, e_{4})e_{3} + \gslash^{AB} R(\dot{\gamma}, K,\dot{\gamma}, e_A)e_B.
\end{align*}

Hence,
\begin{align} \label{eq:curvaturelong}
\begin{split}
	&
	R(\dot{\gamma}, K)\dot{\gamma} =
	\\
	&
	\qquad
	-\frac{1}{2} \bigg\{ 
	K^4\Big[ \dot{\gamma}^3 \dot{\gamma}^3( 4\rho + 2 \Tslash_{34}) + 2 \dot{\gamma}^3 \dot{\gamma}^A ( \beta_A + \Tslash_{4A})
	+ \dot{\gamma}^A \dot{\gamma}^B (\alpha_{AB} + \frac{1}{2} \gslash_{AB} \Tslash_{44}) \Big]
	\\
	&
	\qquad
	+ K^3 \Big[ -2 \dot{\gamma}^3 \dot{\gamma}^4 (2 \rho + \Tslash_{34}) + \dot{\gamma}^3 \dot{\gamma}^A (2 \betabar_A - \Tslash_{3A}) - \dot{\gamma}^4 \dot{\gamma}^A (2 \beta_A + \Tslash_{4A})
	\\
	&
	\qquad
	- \dot{\gamma}^A \dot{\gamma}^B (\rho \gslash_{AB} + \sigma \epsslash_{AB} - \frac{1}{2} \gslash_{AB} \Tslash_{34} + \Tslash_{AB}) \Big]
	\\
	&
	\qquad
	+ \frac{1}{r} K^C \Big[ - \dot{\gamma}^3 \dot{\gamma}^4 (2\beta_C + \Tslash_{4C}) + \dot{\gamma}^3 \dot{\gamma}^3 (-2 \betabar_C + \Tslash_{3C}) 
	\\
	&
	\qquad
	+ \dot{\gamma}^3 \dot{\gamma}^A( \rho\gslash_{AC} + 3 \sigma \epsslash_{AC} + \Tslash_{AC} - \frac{1}{2} \gslash_{AC} \Tslash_{34} ) - \dot{\gamma}^4 \dot{\gamma}^A(\alpha_{AC} + \frac{1}{2} \gslash_{AC} \Tslash_{44})
	\\
	&
	\qquad
	+ \dot{\gamma}^A \dot{\gamma}^B ( - \gslash_{AB} \beta_C + \gslash_{AC} \beta_B + \frac{1}{2} \gslash_{AB} \Tslash_{4C} - \frac{1}{2} \gslash_{AC} \Tslash_{4B}) \Big]
	\bigg\} e_3
	\\
	&
	\qquad
	-\frac{1}{2} \bigg\{ 
	K^4\Big[ -2 \dot{\gamma}^3 \dot{\gamma}^4 ( 2\rho + \Tslash_{34}) + \dot{\gamma}^3 \dot{\gamma}^A ( 2\betabar_A - \Tslash_{3A}) - \dot{\gamma}^4 \dot{\gamma}^A(2 \beta_A + \Tslash_{4A})
	\\
	&
	\qquad
	- \dot{\gamma}^A \dot{\gamma}^B(\rho \gslash_{AB} + \sigma \epsslash_{AB} - \frac{1}{2} \gslash_{AB} \Tslash_{34} + \Tslash_{AB}) \Big]
	\\
	&
	\qquad
	+ K^3 \Big[ \dot{\gamma}^4 \dot{\gamma}^4 (4 \rho + 2 \Tslash_{34}) - 2 \dot{\gamma}^4 \dot{\gamma}^A(2\betabar_A - \Tslash_{3A}) 
	+ \dot{\gamma}^A \dot{\gamma}^B( \alphabar_{AB} + \frac{1}{2}\gslash_{AB} \Tslash_{33})
	\Big]
	\\
	&
	\qquad
	+ \frac{1}{r} K^C \Big[
	\dot{\gamma}^3 \dot{\gamma}^4 (2 \betabar_C - \Tslash_{3C}) + \dot{\gamma}^4 \dot{\gamma}^4 (2 \beta_C + \Tslash_{4C})
	\\
	&
	\qquad
	- \dot{\gamma}^3 \dot{\gamma}^A( \alphabar_{AC} + \frac{1}{2} \gslash_{AC} \Tslash_{33}) + \dot{\gamma}^4 \dot{\gamma}^A (\rho \gslash_{AC} + 3\sigma \epsslash_{AC} + \Tslash_{AC} - \frac{1}{2} \gslash_{AC} \Tslash_{34})
	\\
	&
	\qquad
	+ \dot{\gamma}^A \dot{\gamma}^B( \gslash_{AB} \betabar_C - \gslash_{AC} \betabar_B + \frac{1}{2} \gslash_{AB} \Tslash_{3C} - \frac{1}{2} \gslash_{AC} \Tslash_{3B})
	\Big]
	\bigg\} e_4
	\\
	&
	\qquad
	+ \bigg\{
	K^4 \Big[
	- \dot{\gamma}^3 \dot{\gamma}^4(2 \beta^D + {\Tslash_4}^D) + \dot{\gamma}^3 \dot{\gamma}^3(-2\betabar^D + {\Tslash_3}^D) 
	\\
	&
	\qquad
	+ \dot{\gamma}^3 \dot{\gamma}^A( \rho \delta_A^D + 3 \sigma {\epsslash_A}^D + {\Tslash_A}^D - \frac{1}{2} \delta_A^D \Tslash_{34})
	- \dot{\gamma}^4 \dot{\gamma}^A ({\alpha_A}^D + \frac{1}{2} \delta_A^D \Tslash_{44})
	\\
	&
	\qquad
	+ \dot{\gamma}^A \dot{\gamma}^B ( - \gslash_{AB}\beta^D + \delta_A^D \beta_B + \frac{1}{2} \gslash_{AB} {\Tslash_4}^D - \frac{1}{2} \delta_A^D \Tslash_{4B})
	\Big]
	\\
	&
	\qquad
	+ K^3 \Big[
	\dot{\gamma}^3 \dot{\gamma}^4 (2 \betabar^D - {\Tslash_3}^D) + \dot{\gamma}^4 \dot{\gamma}^4(2\beta^D + {\Tslash_4}^D) - \dot{\gamma}^3 \dot{\gamma}^A({\alphabar_A}^D + \frac{1}{2} \delta_A^D \Tslash_{33})
	\\
	&
	\qquad
	+ \dot{\gamma}^4 \dot{\gamma}^A ( \rho \delta_A^D + 3 \sigma {\epsslash_A}^D + {\Tslash_A}^D - \frac{1}{2} \delta_A^D \Tslash_{34})
	\\
	&
	\qquad
	+ \dot{\gamma}^A \dot{\gamma}^B ( \gslash_{AB} \betabar^D - \delta_A^D \betabar_B + \frac{1}{2} {\Tslash_3}^D - \frac{1}{2} \delta_A^D \Tslash_{3B})
	\Big]
	\\
	&
	\qquad
	+ \frac{1}{r} K^C \Big[
	\dot{\gamma}^4 \dot{\gamma}^4( {\alpha_C}^D + \frac{1}{2} \delta_C^D \Tslash_{44}) + \dot{\gamma}^3 \dot{\gamma}^3( {\alphabar_C}^D+ \frac{1}{2} \delta_C^D \Tslash_{33})
	\\
	&
	\qquad
	+ \dot{\gamma}^3 \dot{\gamma}^4 ( -2\rho \delta_C^D - 2{T_C}^D + \delta_C^D \Tslash_{34})
	+ \dot{\gamma}^A \dot{\gamma}^B \big( -\rho(\gslash_{AB} \delta_C^D - \delta_A^D \gslash_{BC})
	\\
	&
	\qquad
	+ \dot{\gamma}^4 \dot{\gamma}^A( \gslash_{AC} \beta^D + \delta_A^D \beta_C - 2\delta_C^D \beta_A + \delta_C^D \Tslash_{4A} - \frac{1}{2} \gslash_{AC} {\Tslash_4}^D - \frac{1}{2} \delta_A^D \Tslash_{4C})
	\\
	&
	\qquad
	+ \dot{\gamma}^3 \dot{\gamma}^A( - \gslash_{AC} \betabar^D - \delta_A^D \betabar_C + 2\delta_C^D \betabar_A + \delta_C^D \Tslash_{3A} - \frac{1}{2} \gslash_{AC} {\Tslash_3}^D - \frac{1}{2} \delta_A^D \Tslash_{3C})
	\\
	&
	\qquad
	+ \frac{1}{2}(\gslash_{AB} {\Tslash_C}^D + \delta_A^D \Tslash_{AB} - \gslash_{BC} {\Tslash_A}^D - \delta_A^D \Tslash_{BC})  \big)
	\Big] \bigg\} e_D.
	\end{split}
\end{align}

This can schematically be written as
\begin{equation} \label{eq:schematiccurvature}
	R(\dot{\gamma}, K)\dot{\gamma} = \sum_{p_1 + \ldots + p_6 \geq \frac{5}{2}} \rp_{p_1} \mathcal{K}_{p_2} \dot{\upgamma}_{p_3} \dot{\upgamma}_{p_4}(\gslash + 1) (\psi_{p_5} + \mathcal{T}_{p_5}) \eslash^{-p_6}\left( e_3 + e_4 + \frac{1}{r} \eslash \right).
\end{equation}
The summation can actually always begin at 3, except for terms involving $\beta$.  This fact is important when estimating higher order derivatives of Jacobi fields in Section \ref{subsec:higherjacobi} and will be returned to then.

Recall, from Section \ref{subsec:furtherschematic}, that $\mathcal{K}_{-1}$ is used to schematically denote $K^1$, $K^2$, or $K^4$.  It is important to denote $K^1,K^2, K^4$ as such since the $E_1$, $E_2$ and $E_4$ components of some of $J_{(1)}, \ldots, J_{(6)}$ will be allowed to grow at rate $r$.

Clearly, in order to use the Jacobi equation \eqref{eq:jacobicomponents} to estimate the components of $J_{(1)}, \ldots, J_{(6)}$, several additional points need to first be addressed.  First, it is obviously important to understand how the matrix $\Xi$, along with its inverse $\Xi^{-1}$, behaves along $(\gamma, \dot{\gamma})$.  Moreover, since the terms in $\left( \hat{R}(X,J)X \right)^i$ involving derivatives of $R$ have to be integrated by parts, it is also important to understand how the derivative of $\Xi$, $\frac{d \Xi}{ds}$, behaves along $(\gamma, \dot{\gamma})$.  An understanding of the behaviour of these matrices is obtained in Proposition \ref{prop:matrixA} below using the equations \eqref{eq:Xi}, \eqref{eq:Xiinverse}.  Since the components of $\hat{\nabla}_X E_i$ appear in equations \eqref{eq:Xi}, \eqref{eq:Xiinverse}, they are written in schematic notation in Proposition \ref{prop:framederivatives}.

Secondly, it is necessary to understand the initial conditions,
\[
	\left( J^i {\Xi_i}^j \right) (0) = J^j(0),
	\qquad
	\left( \frac{ d J^i {\Xi_i}^j}{ds} \right) (0) = \left( \hat{\nabla}_X J \right) (0)
\]
for the Jacobi equation \eqref{eq:jacobicomponents}.  These initial conditions are computed in Proposition \ref{prop:zeroic}.  The reader is encouraged on first reading to first set $\hat{R}(X,J)X$ equal to zero (i.\@e.\@ to consider the Jacobi fields on a fixed Minkowski background) in order to first understand the argument in this simpler setting.  The Jacobi equation \eqref{eq:jacobicomponents} can, in this case, be explicitly integrated and explicit expressions for the $F_1,\ldots,F_7$ components of the Jacobi fields, $(J^i {\Xi_i}^j) (s)$, can be obtained.  It is clear that, even in this simplified setting, in order to obtain the appropriate boundedness statements, certain cancellations must occur in certain terms arising from $J^j(0)$ and certain terms arising from $(\hat{\nabla}_X J)^j (0)$ for some of the Jacobi fields.  Lemma \ref{lem:sstar} below is used to exploit these cancellations in the general setting.

Finally, it is convenient to write some remaining quantities appearing in the expression \eqref{eq:curvatureuseful} in schematic notation.  This is done in Proposition \ref{prop:geodesicframederivative}.

The zeroth order estimates for the components of $J_{(1)},\ldots,J_{(6)}$ are then obtained in Proposition \ref{prop:Jbounds}, with Lemma \ref{lem:generaljacobi} being used to make the presentation more systematic.

To obtain estimates for first order derivatives of the components of $J_{(i_1)}$, for $i_1 = 1,\ldots,6$, the Jacobi equation \eqref{eq:jacobicomponents} is commuted with $J_{(i_2)}$, for $i_2 = 1,\ldots,6$.  The fact that $J_{(i_2)}$ is a Jacobi field along $s \mapsto \exp_s (x,p) = ( \gamma(s) , \dot{\gamma}(s))$, a curve in $P$ whose tangent vector is $X$, guarantees that $[X,J_{(i_2)}] = 0$, i.\@e.\@ $J_{(i_2)}$ commutes with $\frac{d}{ds}$.  It is now crucial to ensure that the schematic form of the error terms is preserved on applying $J_{(i_2)}$, e.\@g.\@ $J_{(i_2)} \left( \left( R(X,J)X \right)^j \right)$ must have the same, globally twice integrable, behaviour in $r$ as $\left( R(X,J)X \right)^j$, for $j = 1,\ldots,7$.  Moreover, in obtaining the zeroth order estimates, it was important that the bound
\[
	\left\vert \left( R(X,J_{(i_1)})X \right)^j \right\vert
	\leq
	C (p^4(s))^2,
\]
was true in order that the right hand side of the Jacobi equation \eqref{eq:jacobicomponents} could be twice integrated in $s$.  It is therefore also important to also ensure that,
\[
	\left\vert J_{(i_2)} \left( \left( R(X,J_{(i_1)})X \right)^j \right) \right\vert
	\leq
	C (p^4(s))^2.
\]
Note that this property is completely independent of the behaviour in $r$.  Proposition \ref{prop:schematicpreserved} is motivated by showing these properties of the error terms are preserved.  In order for this to be so, it quickly becomes apparent that, at the zeroth order, it is necessary to show that $\vert {J_{(i_1)}}^j \vert \leq C p^4$ for $j = 5,6,7$.  Also, on inspection of Proposition \ref{prop:schematicpreserved}, one sees that $J(\dot{\gamma}^A)$ contains terms of the form,
\[
	\frac{1}{r} \left( J^{4+A} - \frac{\dot{\gamma}^4}{r} J^A \right),
\]
for $A = 1,2$.  The presence of such terms means that, in order to see that $J_{(i_2)}(\dot{\gamma}^A)$ has the correct $\frac{1}{r^2}$ behaviour, it is necessary to ensure that, for each $J = J_{(1)} ,\ldots, J_{(6)}$, either $J^{4+A}$ is not merely bounded at $s = -s_*$, but behaves like $\frac{1}{r}$ along $(\gamma, \dot{\gamma})$, or that an appropriate cancellation occurs between the $J^{4+A}$ and $J^A$ terms.  It is hence necessary, at the zeroth order, to not just show boundedness of the components at $s= -s_*$, but to understand their behaviour along $(\gamma, \dot{\gamma})$ in more detail.  To gain this understanding it also becomes necessary to understand properties of the change of frame matrices, $\Xi$ and $\Xi^{-1}$, in more detail.  See Proposition \ref{prop:matrixA} and Proposition \ref{prop:Jbounds}.  In order to further commute the Jacobi equation \eqref{eq:jacobicomponents}, to estimate second and third order derivatives of the components of the Jacobi fields, it is also necessary to understand more detailed properties of first order derivatives of the components of the Jacobi fields.

As preliminaries to the estimates for the first order derivatives of the components of the Jacobi fields, which are treated in Proposition \ref{prop:Jbounds}, relevant properties of first order derivatives of $\Xi$ and $\Xi^{-1}$ are understood in Proposition \ref{prop:JmatrixA}, along with relevant properties of first order derivatives of the initial conditions for the Jacobi equation in Proposition \ref{prop:Jic}.

The following Lemma, recall, will be used to exploit certain cancellations in terms arising from the initial conditions for the Jacobi equation \eqref{eq:jacobicomponents}.

\begin{lemma} \label{lem:sstar}
	For $s\in [-s_*,0]$,
	\[
		\left\vert r(0) + p^4(0)s - r(s) \right\vert \leq \frac{C}{r(s)}.
	\]
	where $C$ is a constant independent of $(x,p)$.
\end{lemma}

\begin{proof}
	Note that
	\[
		\dot{r} (s) = X(r)(s) = p^4(s) - \frac{1}{\Omega^2} p^3(s),
	\]
	and so
	\begin{align*}
		r(0)
		=
		& \
		r(s) + \int^0_{-s_*} \dot{r} (s) ds
		\\
		=
		& \
		r(s) + \int^0_{-s_*} p^4 (s) ds - \int^0_{-s_*} \frac{1}{\Omega^2} p^3(s) ds.
	\end{align*}
	Recall that $\vert \frac{1}{\Omega^2} p^3(s)\vert \leq \frac{C}{r^2}p^4(s)$, and the geodesic equation for $p^4$,
	\begin{align*}
		\dot{p}^4(s) = 
		&
		\ \frac{1}{2r} \gslash_{AB} p^A(s)p^B(s) - \frac{1}{4} \left( \tr \chibar + \frac{2}{r} \right) \gslash_{AB} p^A(s)p^B(s)
		\\
		&
		- \frac{1}{2} \hat{\chibar}_{AB} p^A(s)p^B(s) - 2 \etabar_A p^A(s) p^4(s) - \omega (p^4(s))^2,
	\end{align*}
	which implies that,
	\[
		\vert p^4(s) - p^4(0) \vert \leq \int_s^0 \vert \dot{p}^4 (s') \vert ds' \leq \int_s^0 \frac{C (p^4(s'))^2}{r(s')^3} ds' \leq \frac{Cp^4(s)}{r(s)^2}.
	\]
	Hence,
	\[
		\vert r(0) + sp^4(0) - r(s) \vert = \left\vert \int_s^0 p^4(s') - p^4(0) - \frac{1}{\Omega^2} p^3(s') ds' \right\vert \leq \frac{C}{r(s)}.
	\]
\end{proof}

Note that Lemma \ref{lem:sstar} in particular implies that
\[
	\left\vert r(0) - p^4(0)s_* \right\vert \leq C,
\]
and also,
\[
	\left\vert \frac{sp^4(0)}{r(0)} \right\vert \leq C,
\]
for $s\in [-s_*,0]$.

In the following two propositions, terms arising in the equations \eqref{eq:Xi}, \eqref{eq:Xiinverse} for $\Xi$ and $\Xi^{-1}$, and in the expression \eqref{eq:curvatureuseful} for $\hat{R}(X,J)X$ are respectively written in schematic form.

\begin{proposition} \label{prop:framederivatives}
  In schematic notation, if $i,j = 1,\ldots,7$, then
  \begin{multline*}
		  \left(\hat{\nabla}_X E_i\right)^j =
		  \sum_{p_1+ \ldots +p_4\geq 2} \dot{\upgamma}_{p_1} \rp_{p_2} (\gslash + 1) \bigg[ \rp_{p_3} + \Gamma_{p_3} + \rp_{p_3} \left( r \Gammaslash + r\nablaslash b + r \Gammaslash \cdot b \right)
		  \\
		  + \sum_{q_1+q_2\geq p_3 + \frac{1}{2}} \dot{\upgamma}_{q_1} ( \psi_{q_2} + \mathcal{T}_{q_2}) \bigg] \eslash^{-p_4},
  \end{multline*}
  and moreover, if $i=1,\ldots, 4, j = 5,6,7$ or vice versa,
  \begin{equation*}
		  \left(\hat{\nabla}_X E_i\right)^j =
		  \sum_{p_1+ \ldots +p_4\geq 2} \dot{\upgamma}_{p_1} \rp_{p_2} (\gslash + 1) \bigg[ \sum_{q_1+q_2\geq p_3 + \frac{1}{2}} \dot{\upgamma}_{q_1} ( \psi_{q_2} + \mathcal{T}_{q_2}) \bigg] \eslash^{-p_4},
  \end{equation*}
  Also,
  \begin{equation*}
		  \left(\hat{\nabla}_X E_A\right)^{4+B} =
		  \sum_{p_1+ \ldots +p_4\geq 3} \dot{\upgamma}_{p_1} \rp_{p_2} (\gslash + 1) \bigg[ \sum_{q_1+q_2\geq p_3 + \frac{1}{2}} \dot{\upgamma}_{q_1} ( \psi_{q_2} + \mathcal{T}_{q_2}) \bigg] \eslash^{-p_4},
  \end{equation*}
  for $A,B = 1,2$.
\end{proposition}

Note that the second summation in each line guarantees that terms involving Weyl curvature components and energy momentum tensor components decay slightly better than the others.  This fact is important and will be returned to in Proposition \ref{prop:JJmatrixA}.  Note also that, if $r\nablaslash b$ is just regarded as $\mathfrak{D} \Gamma_1$, the terms involving $r\nablaslash b$ above also decay slightly better.  This extra decay is important at higher orders because of the weaker bounds we have for $\mathfrak{D}^3 r \nablaslash b$.

\begin{proof}
	 Using Proposition \ref{prop:sasakiconnection} and the table of Ricci coefficients \eqref{eq:Riccitab1}--\eqref{eq:Riccitab5}, one derives,
	 \begin{align*}
		\hat{\nabla}_X E_A
		=
		&
		\left[\dot{\gamma}^C \Gammaslash^B_{AC} + \dot{\gamma}^3 {\hat{\chibar}_A}^B + \dot{\gamma}^4 {\hat{\chi}_A}^B - \dot{\gamma}^4 e_A(b^B) \right] E_B 
		\\
		&
		+ \left[ \frac{\dot{\gamma}^3}{r} \left( \frac{1}{\Omega^2} - 1\right) + \frac{\dot{\gamma}^3}{2} \left( \tr \chibar + \frac{2}{r} \right) + \frac{\dot{\gamma}^4}{2} \left( \tr \chi - \frac{2}{r} \right) \right] E_A
		\\
		&
		+ \left[ \frac{\dot{\gamma}^B}{2r} \hat{\chi}_{AB} + \frac{\gslash_{AB}\dot{\gamma}^B}{4r} \left(\tr \chi - \frac{2}{r} \right) + \frac{\gslash_{AB}\dot{\gamma}^B}{2r^2} + \frac{\dot{\gamma}^3}{r} \eta_A \right] E_3
		\\
		&
		+ \left[ \frac{\dot{\gamma}^B}{2r} \hat{\chibar}_{AB} + \frac{\gslash_{AB}\dot{\gamma}^B}{4r} \left(\tr \chibar + \frac{2}{r} \right) - \frac{\gslash_{AB}\dot{\gamma}^B}{2r^2} + \frac{\dot{\gamma}^3}{r} \etabar_A \right] E_4
		\\
		&
		- \frac{1}{2r} {}^T \ver_{(\gamma,\dot{\gamma})} \left( R(\dot{\gamma}, e_A)\dot{\gamma} \right),
	\end{align*}
	
	\begin{align*}
		\hat{\nabla}_X E_3
		=
		&
		\left[r \dot{\gamma}^A {\hat{\chibar}_A}^B + \dot{\gamma}^B \left( \frac{r}{2}\left( \tr \chibar + \frac{2}{r} \right) - 1  \right) + 2 \dot{\gamma}^4 r\etabar^B \right] E_B 
		- \left[ \dot{\gamma}^A \etabar_A + \dot{\gamma}^4\omega \right] E_3
		\\
		&
		- \frac{1}{2} {}^T \ver_{(\gamma,\dot{\gamma})} \left( R(\dot{\gamma}, e_3)\dot{\gamma} \right),
	\end{align*}
	
	\begin{align*}
		\hat{\nabla}_X E_4
		=
		&
		\left[r \dot{\gamma}^A {\hat{\chi}_A}^B + \dot{\gamma}^B \left( \frac{r}{2}\left( \tr \chi - \frac{2}{r} \right) + 1  \right) + 2r \dot{\gamma}^3 \eta^B \right] E_B 
		+ \left[ \dot{\gamma}^A \etabar_A + \dot{\gamma}^4\omega \right] E_4
		\\
		&
		- \frac{1}{2} {}^T \ver_{(\gamma,\dot{\gamma})} \left( R(\dot{\gamma}, e_4)\dot{\gamma} \right),
	\end{align*}
	
	\begin{align*}
		\hat{\nabla}_X E_{4 + A}
		=
		&
		\frac{1}{2r} \hor_{(\gamma,\dot{\gamma})} \left( R(\dot{\gamma}, e_A) \dot{\gamma} \right) + \frac{\gslash_{AB} \dot{\gamma}^B}{4r\dot{\gamma}^4} \hor_{(\gamma,\dot{\gamma})} \left( R(\dot{\gamma}, e_3) \dot{\gamma} \right)
		\\
		&
		+ \left[\dot{\gamma}^C \Gammaslash^B_{AC} + \dot{\gamma}^3 {\hat{\chibar}_A}^B + \dot{\gamma}^4 {\hat{\chi}_A}^B - \dot{\gamma}^4 e_A(b^B) \right] E_{4+B} 
		\\
		&
		+ \left[ \frac{\dot{\gamma}^3}{r} \left( \frac{1}{\Omega^2} - 1\right) + \frac{\dot{\gamma}^3}{2} \left( \tr \chibar + \frac{2}{r} \right) + \frac{\dot{\gamma}^4}{2} \left( \tr \chi - \frac{2}{r} \right) \right] E_{4+A}
		\\
		&
		+ \left[ \frac{\dot{\gamma}^B}{2r} \hat{\chibar}_{AB} + \frac{\gslash_{AB}\dot{\gamma}^B}{4r} \left(\tr \chibar + \frac{2}{r} \right) - \frac{\gslash_{AB}\dot{\gamma}^B}{2r^2} + \frac{\dot{\gamma}^3}{r} \etabar_A \right] E_7,
		\\
	\end{align*}
	
	\begin{align*}
		\hat{\nabla}_X E_7
		=
		&
		\frac{1}{2} \hor_{(\gamma,\dot{\gamma})} \left( R(\dot{\gamma}, e_4) \dot{\gamma} \right) - \frac{\dot{\gamma}^3}{2\dot{\gamma}^4} \hor_{(\gamma,\dot{\gamma})} \left( R(\dot{\gamma}, e_3) \dot{\gamma} \right)
		\\
		&
		+ \left[r \dot{\gamma}^A {\hat{\chi}_A}^B + \dot{\gamma}^B \left( \frac{r}{2}\left( \tr \chi - \frac{2}{r} \right) + 1  \right) + 2r \dot{\gamma}^3 \eta^B \right] E_{4+B} 
		\\
		&
		+ \left[ \dot{\gamma}^A \etabar_A + \dot{\gamma}^4\omega \right] E_7,
	\end{align*}
	where, for a vector $Y$ on $\mathcal{M}$, ${}^T \ver_{(\gamma,\dot{\gamma})}(Y)$ denotes the projection of $\ver_{(\gamma,\dot{\gamma})}(Y)$ to $P$.  The terms involving the curvature $R$ of $\mathcal{M}$ can be found explicitly in terms of $\psi,\mathcal{T}$ by setting $K = e_A, e_3$ or $e_4$ in the expression \eqref{eq:curvaturelong}.  For example, in the above expression for $\hat{\nabla}_X E_{4+A}$,
	\begin{align*}
		& \hor_{(\gamma,\dot{\gamma})} \left( R(\dot{\gamma}, e_A) \dot{\gamma} \right)
		=
		\\
		&
		\qquad \qquad 
		r \Big[
		\dot{\gamma}^4 \dot{\gamma}^4( {\alpha_A}^D + \frac{1}{2} \delta_A^D \Tslash_{44}) + \dot{\gamma}^3 \dot{\gamma}^3( {\alphabar_A}^D+ \frac{1}{2} \delta_A^D \Tslash_{33})
		\\
		&
		\qquad \qquad
		+ \dot{\gamma}^3 \dot{\gamma}^4 ( -2\rho \delta_A^D - 2{T_A}^D + \delta_A^D \Tslash_{34})
		+ \dot{\gamma}^C \dot{\gamma}^B \big( -\rho(\gslash_{BC} \delta_A^D - \delta_C^D \gslash_{AB})
		\\
		&
		\qquad \qquad
		+ \dot{\gamma}^4 \dot{\gamma}^C( \gslash_{AC} \beta^D + \delta_C^D \beta_A - 2\delta_A^D \beta_C + \delta_A^D \Tslash_{4C} - \frac{1}{2} \gslash_{AC} {\Tslash_4}^D - \frac{1}{2} \delta_C^D \Tslash_{4A})
		\\
		&
		\qquad \qquad
		+ \dot{\gamma}^3 \dot{\gamma}^C( - \gslash_{AC} \betabar^D - \delta_C^D \betabar_A + 2\delta_A^D \betabar_C + \delta_A^D \Tslash_{3C} - \frac{1}{2} \gslash_{AC} {\Tslash_3}^D - \frac{1}{2} \delta_C^D \Tslash_{3A})
		\\
		&
		\qquad \qquad
		+ \frac{1}{2}(\gslash_{BC} {\Tslash_A}^D + \delta_C^D \Tslash_{BC} - \gslash_{AB} {\Tslash_C}^D - \delta_C^D \Tslash_{AB})  \big)
		\Big] E_{D}
		\\
		&
		\qquad \qquad
		- \frac{1}{2} \Big[ - \dot{\gamma}^3 \dot{\gamma}^4 (2\beta_C + \Tslash_{4C}) + \dot{\gamma}^3 \dot{\gamma}^3 (-2 \betabar_C + \Tslash_{3C}) 
		\\
		&
		\qquad \qquad
		+ \dot{\gamma}^3 \dot{\gamma}^C( \rho\gslash_{AC} - 3 \sigma \epsslash_{AC} + \Tslash_{AC} - \frac{1}{2} \gslash_{AC} \Tslash_{34} ) - \dot{\gamma}^4 \dot{\gamma}^C(\alpha_{AC} + \frac{1}{2} \gslash_{AC} \Tslash_{44})
		\\
		&
		\qquad \qquad
		+ \dot{\gamma}^B \dot{\gamma}^C ( - \gslash_{BC} \beta_A + \gslash_{AC} \beta_B + \frac{1}{2} \gslash_{BC} \Tslash_{4A} - \frac{1}{2} \gslash_{AC} \Tslash_{4B}) 
		\Big] E_3
		\\
		&
		\qquad \qquad
		- \frac{1}{2} \Big[
		\dot{\gamma}^3 \dot{\gamma}^4 (2 \betabar_A - \Tslash_{3A}) + \dot{\gamma}^4 \dot{\gamma}^4 (2 \beta_A + \Tslash_{4A})
		\\
		&
		\qquad \qquad
		- \dot{\gamma}^3 \dot{\gamma}^C( \alphabar_{AC} + \frac{1}{2} \gslash_{AC} \Tslash_{33}) + \dot{\gamma}^4 \dot{\gamma}^C (\rho \gslash_{AC} - 3\sigma \epsslash_{AC} + \Tslash_{AC} - \frac{1}{2} \gslash_{AC} \Tslash_{34})
		\\
		&
		\qquad \qquad
		+ \dot{\gamma}^B \dot{\gamma}^C( \gslash_{BC} \betabar_A - \gslash_{AC} \betabar_B + \frac{1}{2} \gslash_{BC} \Tslash_{3A} - \frac{1}{2} \gslash_{AC} \Tslash_{3B})
		\Big] E_4.
	\end{align*}
	The result follows by inspecting each of the terms and writing in schematic notation.
\end{proof}

Using the bootstrap assumptions for the pointwise bounds on $\Gammaslash, \nablaslash b, \Gamma, \psi, \mathcal{T}$ and the fact that $r^p\vert \dot{\upgamma}_p \vert \leq p^4$, Proposition \ref{prop:framederivatives} in particular gives,
\[
	 \left\vert \left( \hat{\nabla}_X E_i \right)^j (s) \right\vert \leq \frac{C p^4}{r(s)^2},
\]
for $i,j=1,\ldots,7$,
\[
	 \left\vert \left( \hat{\nabla}_X E_i \right)^j (s) \right\vert \leq \frac{C (p^4)^2}{r(s)^2},
\]
for all $i=1,\ldots,4, j=5,6,7$ or vice versa, and
\[
	 \left\vert \left( \hat{\nabla}_X E_A \right)^{4+B} (s) \right\vert \leq \frac{C (p^4)^2}{r(s)^3},
\]
for $A,B = 1,2$.  These facts are crucial for showing the schematic form of the error term in the Jacobi equation (see Proposition \ref{prop:schematicpreserved} below) is preserved after taking derivatives.  Recall, by Proposition \ref{prop:suppfmain},
\[
	 cp^4(s) \leq p^4(0) \leq Cp^4(s),
\]
for all $s\in [-s_*(x,p),0]$, for some constants $c,C$ which are independent of the point $(x,p) \in P\cap \supp (f)$.  So $p^4$ in the above expressions can either be taken to be evaluated at time $s$ or time 0.

\begin{proposition} \label{prop:geodesicframederivative}
	In schematic notation,
	\begin{multline*}
		\nabla_{\dot{\gamma}} \left( \frac{1}{r} \eslash \right), \nabla_{\dot{\gamma}} e_3, \nabla_{\dot{\gamma}} e_4
		=
		\\
		\sum_{p_1 + \ldots + p_4 \geq 1} \rp_{p_1} \dot{\upgamma}_{p_2} ( \gslash + 1 ) \left( \rp_{p_3} + \Gamma_{p_3} + \rp_{p_3} \left( r \Gammaslash + r \nablaslash b + r \Gammaslash \cdot b\right) \right) \eslash^{-p_4} \left( e_3 + e_4 + \frac{1}{r} \eslash \right).
	\end{multline*}
\end{proposition}

\begin{proof}
	Using the table of Ricci coefficients \eqref{eq:Riccitab1}--\eqref{eq:Riccitab5} one can compute,
	\begin{align*}
		\nabla_{\dot{\gamma}} \left( \frac{1}{r} e_A \right)
		=
		&
		\left[\dot{\gamma}^C \Gammaslash^B_{AC} + \dot{\gamma}^3 {\hat{\chibar}_A}^B + \dot{\gamma}^4 {\hat{\chi}_A}^B - \dot{\gamma}^4 e_A(b^B) \right] \frac{1}{r} e_B 
		\\
		&
		+ \left[ \frac{\dot{\gamma}^3}{r} \left( \frac{1}{\Omega^2} - 1\right) + \frac{\dot{\gamma}^3}{2} \left( \tr \chibar + \frac{2}{r} \right) + \frac{\dot{\gamma}^4}{2} \left( \tr \chi - \frac{2}{r} \right) \right] \frac{1}{r} e_A
		\\
		&
		+ \left[ \frac{\dot{\gamma}^B}{2r} \hat{\chi}_{AB} + \frac{\gslash_{AB}\dot{\gamma}^B}{4r} \left(\tr \chi - \frac{2}{r} \right) + \frac{\gslash_{AB}\dot{\gamma}^B}{2r^2} + \frac{\dot{\gamma}^3}{r} \eta_A \right] e_3
		\\
		&
		+ \left[ \frac{\dot{\gamma}^B}{2r} \hat{\chibar}_{AB} + \frac{\gslash_{AB}\dot{\gamma}^B}{4r} \left(\tr \chibar + \frac{2}{r} \right) - \frac{\gslash_{AB}\dot{\gamma}^B}{2r^2} + \frac{\dot{\gamma}^3}{r} \etabar_A \right] e_4,
	\end{align*}
	
	\begin{align*}
		\nabla_{\dot{\gamma}} e_3
		=
		&
		\left[r \dot{\gamma}^A {\hat{\chibar}_A}^B + \dot{\gamma}^B \left( \frac{r}{2}\left( \tr \chibar + \frac{2}{r} \right) - 1  \right) + 2 \dot{\gamma}^4 r\etabar^B \right] \frac{1}{r} e_B 
		- \left[ \dot{\gamma}^A \etabar_A + \dot{\gamma}^4\omega \right] e_3,
	\end{align*}
	
	\begin{align*}
		\nabla_{\dot{\gamma}} e_4
		=
		&
		\left[r \dot{\gamma}^A {\hat{\chi}_A}^B + \dot{\gamma}^B \left( \frac{r}{2}\left( \tr \chi - \frac{2}{r} \right) + 1  \right) + 2r \dot{\gamma}^3 \eta^B \right] \frac{1}{r} e_B 
		+ \left[ \dot{\gamma}^A \etabar_A + \dot{\gamma}^4\omega \right] e_4.
	\end{align*}
\end{proof}

In the next proposition, estimates for the components of the matrices $\Xi$, $\Xi^{-1}$ are obtained.

\begin{proposition} \label{prop:matrixA}
	If $v_0$ is sufficiently large then the matrix $\Xi$ satisfies,
	\[
		\left\vert {\Xi_i}^j(s) - {\delta_i}^j \right\vert \leq \frac{C}{r(s)}, \qquad \left\vert \frac{d {\Xi_i}^j}{ds}(s) \right\vert \leq \frac{Cp^4}{r(s)^2},
	\]
	for all $i,j = 1,\ldots,7$.  Moreover, if $i=1,\ldots,4,j=5,6,7$ or vice versa,
	\[
		\left\vert {\Xi_i}^j(s) \right\vert \leq \frac{Cp^4}{r(s)}, \qquad \left\vert \frac{d {\Xi_i}^j}{ds}(s) \right\vert \leq \frac{C(p^4)^2}{r(s)^2}
	\]
	and
	\[
		\left\vert {\Xi_A}^{4+B}(s) \right\vert \leq \frac{Cp^4}{r(s)^2}, \qquad \left\vert \frac{d {\Xi_A}^{4+B}}{ds} (s) \right\vert \leq \frac{C(p^4)^2}{r(s)^3},
	\]
	for $A,B = 1,2$.  Similarly, for $\Xi^{-1}$,
	\[
		\left\vert {{\Xi^{-1}}_i}^j(s) - {\delta_i}^j \right\vert \leq \frac{C}{r(s)}, \qquad \left\vert \frac{d {{\Xi^{-1}}_i}^j}{ds}(s) \right\vert \leq \frac{Cp^4}{r(s)^2},
	\]
	for all $i,j = 1,\ldots,7$.  Moreover, if $i=1,\ldots,4,j=5,6,7$ or vice versa,
	\[
		\left\vert {{\Xi^{-1}}_i}^j(s) \right\vert \leq \frac{Cp^4}{r(s)},  \qquad \left\vert \frac{d {{\Xi^{-1}}_i}^j}{ds}(s) \right\vert \leq \frac{C(p^4)^2}{r(s)^2}
	\]
	and
	\[
		\left\vert {{\Xi^{-1}}_A}^{4+B}(s) \right\vert \leq \frac{Cp^4}{r(s)^2}, \qquad \left\vert \frac{d {{\Xi^{-1}}_A}^{4+B}}{ds} (s) \right\vert \leq \frac{C(p^4)^2}{r(s)^3},
	\]
	for $A,B = 1,2$.  Here $C$ is a constant independent of $(x,p)$.
\end{proposition}

\begin{proof}
	The proof proceeds by a bootstrap argument.  Assume, for some $s\in [-s_*,0]$, that
	\begin{equation} \label{eq:matrixaba1}
		\left\vert {\Xi_i}^j(s') - {\delta_i}^j \right\vert \leq \frac{C_1}{r(s')},
	\end{equation}
	for $i,j = 1,\ldots,7$ and that
	\begin{equation} \label{eq:matrixaba2}
		\left\vert {\Xi_i}^j(s') \right\vert \leq \frac{C_1p^4}{r(s')},
	\end{equation}
	for $i=1,\ldots,4,j=5,6,7$ or vice versa, for all $s' \in [s,0]$, where $C_1 > 4$ is a constant which will be chosen later.  These inequalities are clearly true for $s=0$.  For any $i,j=1,\ldots,7$, equation \eqref{eq:Xi} and Proposition \ref{prop:framederivatives} imply that,
	\[
		  \left\vert \frac{d {\Xi_i}^j (s')}{ds} \right\vert \leq \frac{Cp^4(s')}{r(s')^2} + \frac{CC_1p^4(s')}{r(s')^3},
	\]
	for all $s'\in [s,0]$.  Hence,
	\begin{align*}
		\left\vert {\Xi_i}^j(s) - {\delta_i}^j \right\vert
		&
		\leq
		\int_s^0 \left\vert \frac{d {\Xi_i}^j (s')}{ds} \right\vert ds'
		\\
		&
		\leq 
		\int_s^0 \frac{Cp^4(s')}{r(s')^2} + \frac{CC_1p^4(s')}{r(s')^3} ds'
		\\
		&
		\leq 
		\int_{r(s)}^{r(0)} \frac{C}{r^2} + \frac{CC_1}{r^3} dr
		\\
		&
		\leq 
		\frac{C}{r(s)} + \frac{CC_1}{r(s)^2}
		\\
		&
		\leq 
		\frac{C}{r(s)} \left( 1 + \frac{C_1}{v_0} \right).
	\end{align*}
	Choose $C_1>4$ large so that $C_1 > 4C$, and $v_0$ large so that $1 + \frac{C_1}{v_0} < 2$, i.\@e.\@ $v_0 > C_1$.  Then,
	\[
		  \left\vert {\Xi_i}^j(s) - {\delta_i}^j \right\vert \leq \frac{C_1}{4} \frac{2}{r(s)} = \frac{C_1}{2}\frac{1}{r(s)}.
	\]
	The set of $s\in [-s_*,0]$ such that the bootstrap assumptions \eqref{eq:matrixaba1} hold is therefore non-empty, open, closed and connected, and hence equal to $[-s_*,0]$.
	
	Suppose now $i=1,\ldots,4$, $j = 5,6,7$ or vice versa.  Equation \eqref{eq:Xi} and Proposition \ref{prop:framederivatives} then imply,
	\[
		  \left\vert \frac{d {\Xi_i}^j (s')}{ds} \right\vert \leq \frac{C \left(p^4(s')\right)^2}{r(s')^2} + \frac{CC_1\left(p^4(s')\right)^2}{r(s')^3},
	\]
	using now the second bootstrap assumptions \eqref{eq:matrixaba2}.  Proceeding as before, this implies that
	\[
		  \left\vert {\Xi_i}^j(s) \right\vert \leq \frac{C_1}{2}\frac{p^4(s)}{r(s)},
	\]
	if $C_1,v_0$ are sufficiently large, where we use the fact that $cp^4(0) \leq p^4(s) \leq Cp^4(0)$.  Hence \eqref{eq:matrixaba2} also holds for all $s\in [-s_*,0]$.
	
	Returning now to equation \eqref{eq:Xi} and setting $i=A, j = 4+B$, for $A,B = 1,2$, the final part of Proposition \ref{prop:framederivatives} gives
	\[
		  \left\vert \frac{d {\Xi_A}^{4+B}}{ds} (s) \right\vert \leq \frac{C(p^4(s))^2}{r(s)^3},
	\]
	for all $s \in [-s_*,0]$.  Integrating then gives the final part of the proposition for $\Xi$.  The result for $\Xi^{-1}$ follows identically, using equation \eqref{eq:Xiinverse}.
\end{proof}

In the next proposition the initial conditions for the Jacobi equation \eqref{eq:jacobicomponents} are computed.

\begin{proposition} \label{prop:zeroic}
	The Jacobi fields $J_{(1)},\ldots,J_{(6)}$, along with their first order derivatives in the $X$ direction, take the following initial values.
	\begin{align*}
		J_{(A)}(0)
		=
		& \
		rE_A + p^4 E_{4+A},
		\\
		\hat{\nabla}_X J_{(A)} (0) 
		=
		& \
		p^4 E_A + \frac{1}{2r} \gslash_{AB} p^B E_3
		\\
		&
		+ \frac{1}{2r} \hor_{(x,p)}\left( R\left(p, p^4 e_A + \frac{1}{2} \gslash_{AB}p^B e_3 \right) p\right)
		+ \frac{1}{2} {}^T \ver_{(x,p)} ( R(p,e_A)p),
	\end{align*}
	for $A = 1,2$,
	\begin{align*}
		J_{(3)}(0)
		= 
		E_3,
		\qquad
		\hat{\nabla}_X J_{(3)} (0) 
		=
		\frac{1}{2} {}^T \ver_{(x,p)} ( R(p,e_3)p),
	\end{align*}
	
	\begin{align*}
		J_{(4)}(0) 
		= 
		& \
		rE_4 + p^4 E_7,
		\\
		\hat{\nabla}_X J_{(4)} (0) 
		=
		& \
		-p^3 E_3 + p^4 E_4
		+ \frac{1}{2} \hor_{(x,p)} \left( R(p,p^4e_4 - p^3 e_3)p \right)
		+ \frac{r}{2} {}^T \ver_{(x,p)} ( R(p,e_4)p),
	\end{align*}
	
	\begin{align*}
		J_{(4 + A)}(0)
		=
		& \
		\frac{p^4}{r} E_{4+A},
		\\
		\hat{\nabla}_X J_{(4 + A)} (0)
		=
		& \
		\frac{p^4}{r} E_A + \frac{1}{2r^2} \gslash_{AB} p^B E_3
		+ \frac{1}{2r^2} \hor_{(x,p)}\left( R\left(p, p^4 e_A + \frac{1}{2} \gslash_{AB}p^B e_3 \right) p\right),
	\end{align*}
	for $A = 1,2$, and
	\begin{align*}
		J_{(7)}(0)
		=
		& \
		\frac{p^4}{r} E_7,
		\\
		\hat{\nabla}_X J_{(7)} (0)
		=
		& \
		-\frac{p^3}{r} E_3 + \frac{p^4}{r} E_4
		+ \frac{1}{2r} \hor_{(x,p)} \left( R(p,p^4e_4 - p^3 e_3)p \right).
	\end{align*}
	The expressions involving the curvature tensor $R$ of $(\mathcal{M},g)$ can be written explicitly in terms of $\psi$ and $\mathcal{T}$ using the expression \eqref{eq:curvaturelong}.
\end{proposition}

\begin{proof}
	The proof follows directly from Proposition \ref{prop:jacobi}.
\end{proof}

The components of the Jacobi fields $J_{(1)}, \ldots, J_{(6)}$ can now be estimated along $(\gamma(s), \dot{\gamma}(s))$.  Recall that it is important, in order to show the schematic form of the Jacobi equation is preserved after commuting with Jacobi fields, to identify the leading order terms of some of the components.  The leading order term of ${J_{(3)}}^3$ is also identified in order to carry out a change of variables in the proof of Proposition \ref{prop:emmain2}.  See Section \ref{subsec:proofemmain2}.

The following lemma will be used.

\begin{lemma} \label{lem:generaljacobi}
	If $J^1(s),\ldots,J^7(s)$ are functions along $\exp_s(x,p)$ for $s\in[-s_*,0]$ then
	\begin{align*} \label{eq:generaljacobi}
	\begin{split}
		\left\vert J^j(s) - \left( J^j(0) + s \frac{dJ^k {\Xi_k}^j}{ds}(0) \right) \right\vert
		\leq
		&
		\frac{C}{r(s)} \sum_{i=1}^7 
		\left\vert J^i(0) + s \frac{d J^k {\Xi_k}^i}{ds} (0)  \right\vert
		\\
		&
		+
		C \sum_{i=1}^7 
		\int_s^0 \left\vert \frac{d J^k {\Xi_k}^i}{ds} (s') - \frac{d J^k {\Xi_k}^i}{ds} (0)  \right\vert ds',
	\end{split}
	\end{align*}
	for $j = 1,\ldots,4$,
	\begin{align*}
	\begin{split}
		&
		\left\vert J^{4+A}(s) - \left( J^{4+A}(0) + s \frac{dJ^k {\Xi_k}^{4+A}}{ds}(0) \right) \right\vert
		\\
		&
		\qquad \qquad
		\leq
		\frac{C p^4}{r(s)^2} \sum_{B=1}^2 
		\left\vert J^B(0) + s \frac{d J^k {\Xi_k}^B}{ds} (0)  \right\vert
		+
		\frac{C p^4}{r(s)} \sum_{i=3}^4 
		\left\vert J^i(0) + s \frac{d J^k {\Xi_k}^i}{ds} (0)  \right\vert
		\\
		&
		\qquad \qquad \quad
		+
		\frac{C}{r(s)} \sum_{i=5}^7 
		\left\vert J^i(0) + s \frac{d J^k {\Xi_k}^i}{ds} (0)  \right\vert
		\\
		&
		\qquad \qquad \quad
		+
		C \int_s^0 \left\vert \frac{d J^k {\Xi_k}^{4+A}}{ds} (s') - \frac{d J^k {\Xi_k}^{4+A}}{ds} (0)  \right\vert ds'
		\\
		&
		\qquad \qquad \quad
		+ 
		\frac{C}{r(s)} \sum_{i=1}^7 
		\int_s^0 \left\vert \frac{d J^k {\Xi_k}^i}{ds} (s') - \frac{d J^k {\Xi_k}^i}{ds} (0)  \right\vert ds',
	\end{split}
	\end{align*}
	for $A=1,2$, and
	\begin{align*}
	\begin{split}
		\left\vert J^7(s) \right\vert
		\leq
		&
		\frac{C p^4}{r(s)} \sum_{i=1}^4 
		\left\vert J^i(0) + s \frac{d J^k {\Xi_k}^i}{ds} (0)  \right\vert
		\\
		&
		+
		C \sum_{i=5}^7 
		\left\vert J^i(0) + s \frac{d J^k {\Xi_k}^i}{ds} (0)  \right\vert
		\\
		&
		+ 
		C \sum_{i=1}^7 
		\int_s^0 \left\vert \frac{d J^k {\Xi_k}^i}{ds} (s') - \frac{d J^k {\Xi_k}^i}{ds} (0)  \right\vert ds'.
	\end{split}
	\end{align*}
\end{lemma}

\begin{proof}
The proof follows by using the fundamental theorem of calculus to write,
\[
	J^k(s) {\Xi_k}^j (s)
	=
	J^j (0) + s \frac{d J^k {\Xi_k}^i}{ds} (0)
	- \int^0_s \frac{d J^k {\Xi_k}^i}{ds} (s') - \frac{d J^k {\Xi_k}^i}{ds} (0) ds',
\]
and using the estimates for the components of $\Xi^{-1}$ from Proposition \ref{prop:matrixA}.
\end{proof}

\begin{proposition} \label{prop:Jbounds}
	If $\overline{C}$ is sufficiently small, $s\in [-s_*,0]$, then
	\[
		\vert {J_{(A)}}^B(s) - {\delta_A}^B r(s) \vert \leq C, \qquad \vert {J_{(4)}}^4(s) - r(s) \vert \leq C,
	\]
	for $A,B = 1,2$,
	\[
		\vert {J_{(3)}}^3(s) - 1 \vert \leq \frac{C}{r(s)},
	\]
	and
	\[
		\vert {J_{(i)}}^j(s) \vert \leq C,
	\]
	for $i=1,\ldots,6$, $j=1,\ldots,4$ otherwise,
	\[
		\vert {J_{(A)}}^{4+B}(s) - {\delta_A}^B p^4(s) \vert \leq \frac{Cp^4(s)}{r(s)}, \qquad \vert {J_{(i)}}^{4+A}(s) \vert \leq \frac{Cp^4(s)}{r(s)},
	\]
	for all $A,B = 1,2$, $i \neq A$, and
	\[
		\vert {J_{(i)}}^7 (s) \vert \leq Cp^4(s),
	\]
	for all $i=1,\ldots,7$.  Here $C = C(\overline{C})$ is independent of the point $(x,p)$ and of $s$.
\end{proposition}

\begin{proof}
	The result is shown using a bootstrap argument.  For each $J\in \{ J_{(1)},\ldots,J_{(7)}\}$, assume that, for some constant $C_1 > 1$ which will be chosen later, $s\in[-s_*,0]$ is such that
	\begin{equation} \label{eq:Jboundsba1}
		\left\vert \frac{d J^k{\Xi_k}^j}{ds}(s') - \frac{d J^k{\Xi_k}^j}{ds}(0) \right\vert \leq \frac{C_1 p^4(s')}{r(s')^{\frac{3}{2}}},
	\end{equation}
	for $j=1,\ldots,4$,
	\begin{equation} \label{eq:Jboundsba2}
		\left\vert \frac{d J^k{\Xi_k}^{4+A}}{ds}(s') - \frac{d J^k{\Xi_k}^{4+A}}{ds}(0) \right\vert \leq \frac{C_1 \left(p^4(s')\right)^2}{r(s')^{2}},
	\end{equation}
	for $A=1,2$, and
	\begin{equation} \label{eq:Jboundsba3}
		\left\vert \frac{d J^k{\Xi_k}^7}{ds}(s') - \frac{d J^k{\Xi_k}^7}{ds}(0) \right\vert \leq \frac{C_1 \left(p^4(s')\right)^2}{r(s')^{\frac{3}{2}}},
	\end{equation}
	for all $s'\in [s,0]$.
	
	Suppose first that $i\neq 1,2,4$.  By Proposition \ref{prop:zeroic} and the fact that,
	\[
		\left\vert \frac{sp^4(0)}{r(0)} \right\vert \leq C,
	\]
	it follows that,
	\[
		\left\vert {J_{(i)}}^j(0) + s \left( \hat{\nabla}_X J_{(i)} \right)^j(0) \right\vert \leq C,
	\]
	for $j=1,\ldots,4$, and
	\[
		\left\vert {J_{(i)}}^j(0) + s \left( \hat{\nabla}_X J_{(i)} \right)^j(0) \right\vert \leq C p^4(0),
	\]
	for $j=5,6,7$.  Using the fact that,
	\[
		\left( \hat{\nabla}_X J_{(i)} \right)^j(0) = \frac{d{J_{(i)}}^k {\Xi_k}^j}{ds}(0),
	\]
	\[
		\int_s^0 \frac{p^4(s')}{r(s')^{\frac{3}{2}}} ds' \leq C\int_{r(s)}^{r(0)} \frac{1}{r^{\frac{3}{2}}} dr \leq \frac{C}{r(s)^{\frac{1}{2}}},
	\]
	\[
		\int_s^0 \frac{\left(p^4(s')\right)^2}{r(s')^{\frac{3}{2}}} ds' \leq \frac{C p^4(s)}{r(s)^{\frac{1}{2}}},
	\]
	\[
		\int_s^0 \frac{\left(p^4(s')\right)^2}{r(s')^2} ds' \leq \frac{C p^4(s)}{r(s)},
	\]
	and
	\[
		cp^4(0) \leq p^4(s) \leq Cp^4(0),
	\]
	for all $s \in [-s_*,0]$,
	Lemma \ref{lem:generaljacobi} immediately gives
	\[
		\vert {J_{(i)}}^j(s) \vert \leq C (1 + C_1),
	\]
	for $j=1,\ldots, 4$, and
	\[
		\vert {J_{(i)}}^{4+A}(s) \vert \leq C (1 + C_1)\frac{p^4(s)}{r(s)},
	\]
	for $A=1,2$, and
	\[
		\vert {J_{(i)}}^7(s) \vert \leq C (1 + C_1) p^4(s).
	\]
	Note also that,
	\[
		{J_{(3)}}^3(0) + s \left( \hat{\nabla}_X J_{(3)} \right)^3(0) = 1,
	\]
	and hence Lemma \ref{lem:generaljacobi} moreover gives,
	\[
		\vert {J_{(3)}}^3(s) - 1 \vert \leq \frac{C (1 + C_1)}{r(s)}.
	\]
	If $i=1,2$ then, using the fact that
	\[
		\vert p^4(s) - p^4(0) \vert \leq \frac{Cp^4(0)}{r(s)^2},
	\]
	(see the proof of Lemma \ref{lem:sstar}), Lemma \ref{lem:sstar} and Proposition \ref{prop:zeroic} imply that
	\[
		\left\vert {J_{(A)}}^A(0) + s \left( \hat{\nabla}_X J_{(A)} \right)^A(0) - r(s) \right\vert 
		\leq
		\left\vert \left( r(0) + sp^4(0) \right) - r(s) \right\vert
		+ \left\vert s \left( p^4(0) - \left( \hat{\nabla}_X J_{(A)} \right)^A(0) \right) \right\vert
		\leq
		C,
	\]
	for $A = 1,2$, and that,
	\[
		\left\vert {J_{(A)}}^j(0) + s \left( \hat{\nabla}_X J_{(A)} \right)^j(0) \right\vert \leq C,
	\]
	for $j=1,\ldots,4, j\neq A$, and
	\[
		\left\vert {J_{(A)}}^{4+A}(0) + s \left( \hat{\nabla}_X J_{(A)} \right)^{4+A}(0) - p^4(0) \right\vert
		=
		\left\vert s {J_{(A)}}^{4+A}(0) \right\vert \leq \frac{Cp^4(0)}{r(0)},
	\]
	\[
		\left\vert {J_{(A)}}^j(0) + s \left( \hat{\nabla}_X J_{(A)} \right)^j(0) \right\vert \leq \frac{Cp^4(0)}{r(0)},
	\]
	for $j=5,6,7, j\neq 4+A$.  Lemma \ref{lem:generaljacobi} then gives,
	\[
		\left\vert {J_{(A)}}^B (s) - {\delta_A}^B r(s) \right\vert \leq C(1+C_1),
	\]
	for $B=1,2$,
	\[
		\left\vert {J_{(A)}}^i (s) \right\vert \leq C(1+C_1),
	\]
	for $i=3,4$,
	\[
		\left\vert {J_{(A)}}^{4+B} (s) - {\delta_A}^B p^4(s) \right\vert \leq \frac{C(1+C_1)p^4(s)}{r(s)},
	\]
	for $B=1,2$, and
	\[
		\left\vert {J_{(A)}}^7 (s) \right\vert \leq C(1+C_1)p^4(s).
	\]
	Similarly, using the fact that,
	\[
		\left\vert {J_{(4)}}^4(0) + s \left( \hat{\nabla}_X J_{(4)} \right)^4(0) - r(s) \right\vert 
		=
		\left\vert \left( r(0) + sp^4(0) \right) - r(s) \right\vert
		+ \left\vert s \left( p^4(0) - \left( \hat{\nabla}_X J_{(4)} \right)^4(0) \right) \right\vert
		\leq
		C,
	\]
	etc.\@, the bounds for $J_{(4)}$,
	\[
		\left\vert {J_{(4)}}^4 (s) - r(s) \right\vert \leq C(1+C_1),
	\]
	and,
	\[
		\left\vert {J_{(4)}}^i (s) \right\vert \leq C(1+C_1),
	\]
	for $i=1,2,3$,
	\[
		\left\vert {J_{(4)}}^{4+A} (s) \right\vert \leq \frac{C(1+C_1)p^4(s)}{r(s)},
	\]
	for $A=1,2$, and
	\[
		\left\vert {J_{(4)}}^7 (s) \right\vert \leq C(1+C_1)p^4(s),
	\]
	can be obtained.
	
	It remains to recover the bootstrap assumptions \eqref{eq:Jboundsba1}--\eqref{eq:Jboundsba3} with better constants.  It will be shown that, for each $J = J_{(1)},\ldots,J_{(6)}$,
	\begin{equation} \label{eq:Jconc1}
		\int_s^0 \vert (\hat{R}(X,J)X)^i {\Xi_i}^j (s') \vert ds' \leq \frac{\overline{C} C(C_1) p^4(s)}{r(s)^{\frac{3}{2}}},
	\end{equation}
	for $j=1,\ldots,4$,
	\begin{equation} \label{eq:Jconc2}
		\int_s^0 \vert (\hat{R}(X,J)X)^i {\Xi_i}^{4+A} (s') \vert ds' \leq \frac{\overline{C} C(C_1) (p^4(s))^2}{r(s)^2},
	\end{equation}
	for $A=1,2$, and
	\begin{equation} \label{eq:Jconc3}
		\int_s^0 \vert (\hat{R}(X,J)X)^i {\Xi_i}^7 (s') \vert ds' \leq \frac{\overline{C} C(C_1) (p^4(s))^2}{r(s)^{\frac{3}{2}}},
	\end{equation}
	where $C(C_1)$ is a constant depending on $C_1$.  By integrating the Jacobi equation \eqref{eq:jacobicomponents} and taking $\overline{C}$ small depending on $C(C_1)$, the bootstrap assumptions \eqref{eq:Jboundsba1}--\eqref{eq:Jboundsba3} can then be recovered with better constants.  This implies that the set of $s\in[-s_*,0]$ such that \eqref{eq:Jboundsba1}--\eqref{eq:Jboundsba3} hold for all $s' \in [s,0]$ is non-empty, open and closed, and hence that \eqref{eq:Jboundsba1}--\eqref{eq:Jboundsba3} hold for all $s\in[-s_*,0]$.
	
	Consider first the first term in the expression \eqref{eq:curvatureuseful} for $\hat{R}(X,J)X$,
	\[
		\hor_{(\gamma,\dot{\gamma})} \left( R(\dot{\gamma},J^h) \dot{\gamma} \right).
	\]
	The components of this term with respect to $E_1,\ldots,E_4$ are exactly the components of $R(\dot{\gamma},J^h) \dot{\gamma}$ with respect to the frame $\frac{1}{r} e_1,\frac{1}{r}e_2,e_3,e_4$ for $\mathcal{M}$.  From the schematic expression \eqref{eq:schematiccurvature}, the pointwise bounds on the components\footnote{The pointwise bounds on the components of $\psi$ and $\mathcal{T}$ follow from the pointwise bounds on $\vert \psi \vert, \vert \mathcal{T}\vert$ and the fact that $\frac{1}{r^2} \vert \gslash_{AB} \vert , r^2 \vert \gslash^{AB} \vert \leq C$ in each of the spherical coordinate charts.} of $\psi,\mathcal{T}$ and the fact that $r^p \vert \dot{\upgamma}_p(s) \vert \leq Cp^4(s)$, one immediately sees that,
	\[
		\left\vert \left( \hor_{(\gamma,\dot{\gamma})} \left( R(\dot{\gamma},J^h) \dot{\gamma} \right) \right)^i \right\vert \leq \frac{\overline{C} C (p^4(s))^2}{r(s)^{\frac{5}{2}}},
	\]
	and hence, by Proposition \ref{prop:matrixA},
	\[
		\int_s^0 \left\vert \left( \hor_{(\gamma,\dot{\gamma})} \left( R(\dot{\gamma},J^h) \dot{\gamma} \right) \right)^i {\Xi_i}^j \right\vert ds' \leq 
		\frac{\overline{C} C p^4(s)}{r(s)^{\frac{3}{2}}},
	\]
	for $j = 1,2,3,4$.  Other than those in the bottom line, the remaining horizontal components in the expression \eqref{eq:curvatureuseful} can be treated similarly using also Proposition \ref{prop:geodesicframederivative} and the pointwise bounds on the components of $\Gamma, \Gammaslash, \nablaslash b$.  For the term
	\[
		\hat{\nabla}_X \hor_{(\gamma,\dot{\gamma})} \left( R(\dot{\gamma},J^v)\dot{\gamma}\right),
	\]
	in the bottom line of \eqref{eq:curvatureuseful}, write
	\begin{equation} \label{eq:JXhor}
		\hat{\nabla}_X \hor_{(\gamma,\dot{\gamma})} \left( R(\dot{\gamma},J^v)\dot{\gamma}\right) = X \left( ( R(\dot{\gamma},J^v)\dot{\gamma})^{\mu} \right) E_{\mu} + ( R(\dot{\gamma},J^v)\dot{\gamma})^{\mu} \hat{\nabla}_X E_{\mu},
	\end{equation}
	where $\mu$ runs from 1 to 4 in the summations.  The horizontal components of the second term of \eqref{eq:JXhor} can be estimated exactly as the others using Proposition \ref{prop:framederivatives}.  For the components of the first term, write,
	\[
		\int_s^0 X \left( ( R(\dot{\gamma},J^v)\dot{\gamma})^{\mu} {\Xi_{\mu}}^j \right) (s') ds' = ( R(\dot{\gamma},J^v)\dot{\gamma})^{\mu} {\Xi_{\mu}}^j (0) - ( R(\dot{\gamma},J^v)\dot{\gamma})^{\mu} {\Xi_{\mu}}^j (s),
	\]
	for $j=1,\ldots,7$.  Then using again the schematic expression \eqref{eq:schematiccurvature}, the pointwise bounds on $\psi$, $\mathcal{T}$ and Proposition \ref{prop:matrixA}, the terms in the Jacobi equation \eqref{eq:jacobicomponents} which the first term of \eqref{eq:JXhor} give rise to can be estimated,
	\begin{align*}
		\left\vert \int_s^0 X \left( ( R(\dot{\gamma},J^v)\dot{\gamma})^{\mu} \right) {\Xi_{\mu}}^j (s') ds' \right\vert
		\leq
		&
		\left\vert \int_s^0 X \left( ( R(\dot{\gamma},J^v)\dot{\gamma})^{\mu} {\Xi_{\mu}}^j \right) (s') ds' \right\vert
		\\
		&
		+ \left\vert \int_s^0  ( R(\dot{\gamma},J^v)\dot{\gamma})^{\mu} X \left( {\Xi_{\mu}}^j \right) (s') ds' \right\vert
		\\
		\leq
		&
		\frac{\overline{C} C (p^4(s))^2}{r(s)^{\frac{5}{2}}}.
	\end{align*}
	
	The vertical terms in \eqref{eq:curvatureuseful} are similarly estimated as follows.  Notice now that, ignoring the term
	\[
		\hat{\nabla}_X {}^T \ver_{(\gamma,\dot{\gamma})} \left( R(\dot{\gamma},J^h)\dot{\gamma} \right),
	\]
	in the bottom line of \eqref{eq:curvatureuseful} for now, each term contains at least three $\dot{\gamma}$ factors and moreover that, since Proposition \ref{prop:geodesicframederivative} guarantees that the terms involving $\nabla_{\dot{\gamma}} e_{\alpha}$ gain an extra power of decay.  Similarly, since $\nabla_{\dot{\gamma}} \dot{\gamma} = 0$, one can check,
	\begin{align*}
		X(\dot{\gamma}^A)
		=
		&
		\
		- \dot{\gamma}^A \dot{\gamma}^B \Gammaslash_{BC}^A - 2 \dot{\gamma}^3 \dot{\gamma}^B {\chibar_B}^A - \dot{\gamma}^4 \dot{\gamma}^B \left( 2{\chi_B}^A - e_B(b^A) - 2 \dot{\gamma}^3 \dot{\gamma}^4 ( \eta^A + \etabar^A) \right),
		\\
		X(\dot{\gamma}^3)
		=
		&
		\
		- \frac{1}{2} \dot{\gamma}^A \dot{\gamma}^B \chi_{AB} + \dot{\gamma}^3\dot{\gamma}^A (\etabar_A - \eta_A) + \dot{\gamma}^3 \dot{\gamma}^4 \omega,
		\\
		X(\dot{\gamma}^4)
		=
		&
		\
		- \frac{1}{2} \dot{\gamma}^A \dot{\gamma}^B \chibar_{AB} - 2 \dot{\gamma}^4 \dot{\gamma}^A \etabar_A - \dot{\gamma}^4 \dot{\gamma}^4 \omega,
	\end{align*}
	and hence the terms involving $X(\dot{\gamma}^{\alpha})$ also gain an extra power of $r$ decay.  Similarly for the vertical terms arising from the second term in \eqref{eq:JXhor}, by Proposition \ref{prop:framederivatives}.  Since
	\[
		{}^T \ver_{(\gamma,\dot{\gamma})} \left( R(\dot{\gamma},R(\dot{\gamma},J^v) \dot{\gamma}) \dot{\gamma} \right),
	\]
	is quadratic in $R$ this term also decays better.  Hence, using also Proposition \ref{prop:matrixA}, one sees all the vertical terms in \eqref{eq:curvatureuseful}, still ignoring the final term in the bottom line, can be controlled by\footnote{They will actually decay like $\frac{1}{r^{\frac{7}{2}}}$, but $\frac{1}{r^3}$ is sufficient.}
	\[
		\frac{\overline{C} C (p^4(s))^3}{r(s)^3}.
	\]
	For the final term, write
	\[
		\hat{\nabla}_X {}^T \ver_{(\gamma,\dot{\gamma})} \left( R(\dot{\gamma},J^h)\dot{\gamma} \right)
		=
		X \left( (R(\dot{\gamma},J^h)\dot{\gamma})^{\lambda} \right) E_{\tilde{\lambda}(\lambda)} + (R(\dot{\gamma},J^h)\dot{\gamma})^{\lambda}\hat{\nabla}_X E_{\tilde{\lambda}(\lambda)},
	\]
	where $\lambda$ runs over 1,2,4 and $\tilde{\lambda}(1) = 5, \tilde{\lambda}(2) = 6, \tilde{\lambda}(4) = 7$.  The components of the second term can be estimated as before (with the additional $r$ decay) by Proposition \ref{prop:framederivatives}.  The components of the first term can again be estimated after integrating,
	\begin{align*}
		&
		\left\vert \int_s^0 X \left( ( R(\dot{\gamma},J^h)\dot{\gamma})^{\lambda} {\Xi_{\tilde{\lambda}(\lambda)}}^j \right) (s') ds' \right\vert
		\\
		&
		\qquad \qquad
		=
		\left\vert ( R(\dot{\gamma},J^h)\dot{\gamma})^{\lambda} {\Xi_{\tilde{\lambda}(\lambda)}}^j (0) - ( R(\dot{\gamma},J^h)\dot{\gamma})^{\lambda}  {\Xi_{\tilde{\lambda}(\lambda)}}^j (s) \right\vert
		\\
		&
		\qquad \qquad
		\leq 
		\frac{\overline{C} C (p^4(s))^2}{r(s)^{\frac{5}{2}}},
	\end{align*}
	and hence
	\[
		\left\vert \int_s^0 X \left( ( R(\dot{\gamma},J^h)\dot{\gamma})^{\lambda} \right) {\Xi_{\tilde{\lambda}(\lambda)}}^j (s') ds' \right\vert
		\leq
		\frac{\overline{C} C (p^4(s))^2}{r(s)^{\frac{5}{2}}}.
	\]
	The bounds \eqref{eq:Jconc1}--\eqref{eq:Jconc3} are thus obtained.
\end{proof}

Suppose now $i_1,i_2 = 1,\ldots,6$.  Since $J_{(i_2)}$ is a Jacobi field along $(\gamma,\dot{\gamma})$, a curve with tangent vector field $X$, it is true that $[X,J_{(i_2)}] = 0$ and the Jacobi equation for the components of $J_{(i_1)}$ can be commuted with $J_{(i_2)}$ to give,
\[
	 \frac{d^2 J_{(i_2)} ({J_{(i_1)}}^k {\Xi_k}^j)}{ds^2} = J_{(i_2)} \left( (\hat{R}(X,J_{(i_1)}) X )^k {\Xi_k}^j \right).
\]
The goal now is to repeat the proof of Proposition \ref{prop:Jbounds} to get pointwise estimates for $J_{(i_2)}({J_{(i_1)}}^j)$ along $(\gamma,\dot{\gamma})$.  It is first necessary to show that the schematic form of $\hat{R}(X,J_{(i_1)})X$ is preserved after differentiating the components with respect to $J_{(i_2)}$.

As with the $\mathcal{K}$ notation introduced in Section \ref{subsec:furtherschematic}, for $J = J_{(1)},\ldots,J_{(6)}$, let the components be schematically denoted as follows,
\[
	\mathcal{J}_0 = J^3 , \qquad \mathcal{J}_{-1} = J^1, J^2, J^4, J^5, J^6, J^7.
\]
By Proposition \ref{prop:Jbounds}, it is always true that
\[
	r^p \vert \mathcal{J}_p \vert \leq C,
\]
for some constant $C$.\footnote{In fact, all of the $J^5,J^6,J^7$ components are uniformly bounded along $(\gamma,\dot{\gamma})$, though it is easier to treat them systematically if they are included in $\mathcal{J}_{-1}$.}

\begin{proposition} \label{prop:schematicpreserved}
	For $J = J_{(1)},\ldots,J_{(6)}$,
	\[
		J(\rp_p) = \sum_{p_1+p_2\geq p} \rp_{p_1} \mathcal{J}_{p_2},
	\]
	for any $\rp_p$ appearing in the schematic expressions of this section,
	\begin{align} \label{eq:Jgamma}
	\begin{split}
		J(\dot{\upgamma}_p)
		=
		&
		\sum_{p_1+\ldots +p_5 \geq p} \rp_{p_1} \mathcal{J}_{p_2} \dot{\upgamma}_{p_3} (\gslash + 1) (\rp_{p_4} + \Gamma_{p_4}) ( 1 + r\Gammaslash + r \nablaslash b + r \Gammaslash \cdot b ) \eslash^{-p_5}
		\\
		&
		+ J^7 \rp_p + \sum_{A=1,2} \frac{1}{r} \left( J^{4+A} - \frac{p^4}{r} J^A \right) + \frac{1}{2r} \frac{\gslash_{AB}p^B}{p^4} \left( J^{4+A} - \frac{p^4}{r} J^A \right),
	\end{split}
	\end{align}
	and
	\begin{align*}
		&
		J \left(  \sum_{p_1 +p_2 \geq p} (\gslash + 1) ( \Gamma_{p_1} + \psi_{p_1} + \mathcal{T}_{p_1}) ( 1 + r\Gammaslash + r \nablaslash b + r \Gammaslash \cdot b ) \eslash^{-p_2} \right)
		\\
		&
		\qquad \qquad
		=
		\sum_{p_1 +p_2 \geq p} (1 + \gslash + \mathfrak{D} \gslash ) ( \Gamma_{p_1} + \mathfrak{D} \Gamma_{p_1} + \psi_{p_1} + \mathfrak{D} \psi_{p_1} + \mathcal{T}_{p_1} + \mathfrak{D} \mathcal{T}_{p_1}) 
		\\
		&
		\qquad \qquad \qquad
		\times ( 1 + r\Gammaslash + r\mathfrak{D} \Gammaslash + r \nablaslash b + \mathfrak{D} r \nablaslash b ) \eslash^{-p_2}.
	\end{align*}
\end{proposition}

\begin{proof}
	In the schematic expressions of this section, $\rp_p$ always denotes (a constant multiple of) $\frac{1}{r^p}$.  One easily checks,
	\[
		J\left( \frac{1}{r^p} \right) = J^4 e_4 \left( \frac{1}{r^p} \right) + J^3 e_3 \left( \frac{1}{r^p} \right) = \frac{1}{pr^{p-1}} (J^3 - J^4) = \sum_{p_1+p_2\geq p} \rp_{p_1} \mathcal{J}_{p_2}.
	\]
	
	For the second part, writing $\hor_{(\gamma,\dot{\gamma})} (e_{\mu}) = e_{\mu} - p^{\nu} \Gamma_{\mu \nu}^{\lambda} \partial_{p^{\lambda}}$, by direct computation,
	\begin{align*}
		J(\dot{\gamma}^4)
		=
		&
		- \frac{1}{2r} J^A \dot{\gamma}^B \left( \hat{\chibar}_{AB} + \frac{1}{2} \gslash_{AB} \left( \tr \chibar + \frac{2}{r} \right) - \frac{\gslash_{AB}}{r} \right)
		\\
		&
		- J^4 \dot{\gamma}^4 \omega - \frac{1}{r} \dot{\gamma}^4 J^A \etabar_A + J^7,
		\\
		J(\dot{\gamma}^A)
		=
		&
		- \frac{1}{r} J^B \dot{\gamma}^C \Gammaslash_{BC}^A - \frac{1}{r} J^B\dot{\gamma}^3 \left( {\hat{\chibar}_{A}}^B + \frac{1}{2} {\delta_A}^B \left( \tr \chibar + \frac{2}{r} \right) - {\delta_A}^B \frac{1}{r} \right)
		\\
		&
		- \frac{1}{r} J^B\dot{\gamma}^4 \left( {\hat{\chi}_{A}}^B + \frac{1}{2} {\delta_A}^B \left( \tr \chi - \frac{2}{r} \right) \right)
		\\
		&
		- J^3\dot{\gamma}^B \left( {\hat{\chibar}_{A}}^B + \frac{1}{2} {\delta_A}^B \left( \tr \chibar + \frac{2}{r} \right) - {\delta_A}^B \frac{1}{r} \right)
		\\
		&
		- J^4\dot{\gamma}^B \left( {\hat{\chi}_{A}}^B + \frac{1}{2} {\delta_A}^B \left( \tr \chi - \frac{2}{r} \right) + {\delta_A}^B \frac{1}{r} - (\nablaslash_B b)^A + \Gammaslash_{BC}^A b^C \right)
		\\
		&
		- 2J^3 \dot{\gamma}^4 \eta^A - 2 J^4 \dot{\gamma}^3 \etabar^A + \frac{1}{r} \left( J^{4+A} - \frac{\dot{\gamma}^4}{r} J^A \right).
	\end{align*}
	One easily sees these two expressions have the desired schematic form.  For $J(\dot{\gamma}^3)$, recall,
	\[
		\dot{\gamma}^3 = \frac{\gslash_{AB} \dot{\gamma}^A \dot{\gamma}^B}{4\dot{\gamma}^4},
	\]
	so that
	\begin{align*}
		J(\dot{\gamma}^3)
		=
		&
		\frac{J(\gslash_{AB}) \dot{\gamma}^A \dot{\gamma}^B}{4\dot{\gamma}^4} + \frac{\gslash_{AB}J(\dot{\gamma}^A) \dot{\gamma}^B}{2\dot{\gamma}^4} - \frac{\gslash_{AB} \dot{\gamma}^A \dot{\gamma}^B}{4(\dot{\gamma}^4)^2} J(\dot{\gamma}^4)
		\\
		=
		&
		\left( \frac{J^C}{r} e_C (\gslash_{AB}) + J^3 e_3(\gslash_{AB}) + J^4e_4(\gslash_{AB}) \right)
		+ \frac{\gslash_{AB} J(\dot{\gamma}^A) \dot{\gamma}^B}{2 \dot{\gamma}^4} - \frac{\dot{\gamma}^3}{\dot{\gamma}^4} J(\dot{\gamma}^4).
	\end{align*}
	The result follows by expanding $e_C (\gslash_{AB}), e_3(\gslash_{AB}),e_4(\gslash_{AB})$ and using the previous two expressions.
	
	The last point is immediate from Lemma \ref{lemma:componentderivatives}.
\end{proof}

Note that it is the terms in the last line of \eqref{eq:Jgamma} which make it necessary to keep track of the leading order terms in some of the Jacobi fields.

\begin{remark} \label{rmk:hoschematic}
	One easily sees that the last point from Proposition \ref{prop:schematicpreserved} is true at higher orders, i.\@e.\@,
	\begin{align*}
		&
		J \left(  \sum_{p_1 +p_2 \geq p} (1+ \mathfrak{D}^k \gslash) ( \mathfrak{D}^k \Gamma_{p_1} + \mathfrak{D}^k \psi_{p_1} + \mathfrak{D}^k \mathcal{T}_{p_1}) ( 1 + r\mathfrak{D}^k \Gammaslash + \mathfrak{D}^k r \nablaslash b ) \eslash^{-p_2} \right)
		\\
		&
		\qquad \qquad \qquad
		=
		\sum_{p_1 +p_2 \geq p} (1 + \mathfrak{D}^k \gslash + \mathfrak{D}^{k+1} \gslash ) ( \mathfrak{D}^k \Gamma_{p_1} + \mathfrak{D}^{k+1} \Gamma_{p_1} + \mathfrak{D}^k \psi_{p_1} + \mathfrak{D}^{k+1} \psi_{p_1} 
		\\
		&
		\qquad \qquad \qquad \qquad
		+ \mathfrak{D}^k \mathcal{T}_{p_1} + \mathfrak{D}^{k+1} \mathcal{T}_{p_1}) 
		( 1 + r \mathfrak{D}^k \Gammaslash + r\mathfrak{D}^{k+1} \Gammaslash + \mathfrak{D}^k r \nablaslash b + \mathfrak{D}^{k+1} r \nablaslash b ) \eslash^{-p_2},
	\end{align*}
	for $k \geq 1$.  This fact will be used later when estimating higher order derivatives of the Jacobi fields.
\end{remark}

Using the bounds on the components of $J_{(1)},\ldots,J_{(6)}$ from Proposition \ref{prop:Jbounds}, Proposition \ref{prop:schematicpreserved} in particular guarantees that
\[
	\left\vert J(\dot{\upgamma}_p) (s) \right\vert \leq \frac{Cp^4}{r(s)^p},
\]
for $J = J_{(1)},\ldots,J_{(6)}$.

In order to mimic the strategy used to obtain the zeroth order estimates of the components of the Jacobi fields, estimates for $J_{(i_2)}({\Xi_k}^j)$ along $(\gamma,\dot{\gamma})$ are first obtained, then the initial conditions for $J_{(i_2)}({J_{(i_1)}}^j)$ are computed in Proposition \ref{prop:JmatrixA} and Proposition \ref{prop:Jic} respectively.

\begin{proposition} \label{prop:JmatrixA}
	If $v_0$ is sufficiently large, for $J = J_{(1)},\ldots,J_{(6)}$ the matrix $\Xi$ satisfies,
	\[
		\left\vert J({\Xi_k}^j) (s) \right\vert \leq \frac{C}{r(s)}, \qquad \left\vert \frac{d J({\Xi_k}^j)}{ds} (s) \right\vert \leq \frac{Cp^4(s)}{r(s)^2},
	\]
	for $k,j = 1,\ldots,7$.  Moreover, if $k=1,\ldots,4$, $j = 5,6,7$ or vice versa,
	\[
		\left\vert J({\Xi_k}^j) (s) \right\vert \leq \frac{C p^4(s)}{r(s)}, \qquad \left\vert \frac{d J({\Xi_k}^j)}{ds} (s) \right\vert \leq \frac{C(p^4(s))^2}{r(s)^2}.
	\]
	Similarly for $\Xi^{-1}$,
	\[
		\left\vert J({{\Xi^{-1}}_k}^j) (s) \right\vert \leq \frac{C p^4(s)}{r(s)}, \qquad \left\vert \frac{d J({{\Xi^{-1}}_k}^j)}{ds} (s) \right\vert \leq \frac{C(p^4(s))^2}{r(s)^2}.
	\]
	for $k,j = 1,\ldots,7$.  Moreover, if $k=1,\ldots,4$, $j = 5,6,7$ or vice versa,
	\[
		\left\vert J({{\Xi^{-1}}_k}^j) (s) \right\vert \leq \frac{C p^4(s)}{r(s)}, \qquad \left\vert \frac{d J({{\Xi^{-1}}_k}^j)}{ds} (s) \right\vert \leq \frac{C(p^4(s))^2}{r(s)^2}.
	\]
\end{proposition}

\begin{proof}
	The proof follows that of Proposition \ref{prop:matrixA} by first, for $s\in [-s_*,0]$, making the bootstrap assumptions,
	\[
		\left\vert J({\Xi_k}^j) (s') \right\vert \leq \frac{C_1}{r(s')},
	\]
	for $k,j = 1,\ldots,7$,
	\[
		\left\vert J({\Xi_k}^j) (s') \right\vert \leq \frac{C_1 p^4(s')}{r(s')},
	\]
	for $k = 1,\ldots,4, j=5,6,7$ or vice versa, for all $s'\in [s,0]$.  Note that at time $s=0$,
	\[
		J({\Xi_k}^j) \vert_{s=0} = 0,
	\]
	for all $J,k,j$.  Using the schematic expressions for the components of $\hat{\nabla}_X E_k$ from Proposition \ref{prop:framederivatives} and the fact that this schematic structure is pressured by Proposition \ref{prop:schematicpreserved}, the commuted equation for $\Xi$,
	\[
		\frac{d J({\Xi_k}^j)}{ds} (s) = J \left( (\hat{\nabla}_X E_k)^l {\Xi_l}^j (s)  \right)
	\]
	can be estimated exactly as in Proposition \ref{prop:matrixA}.  Similarly for $\Xi^{-1}$.
\end{proof}

The next proposition gives pointwise estimates for the initial conditions for the commuted Jacobi equation.  As was the case for the uncommuted equation, the leading order terms of some of the components have to be subtracted first.

\begin{proposition} \label{prop:Jic}
	At time $s=0$,
	\[
		\left\vert J_{(4)} ({J_{(A)}}^B) \big\vert_{s=0} - {\delta_A}^B r  \right\vert \leq C, \qquad 
		\left\vert J_{(4)} ( (\hat{\nabla}_X J_{(A)})^B) \big\vert_{s=0} - {\delta_A}^B p^4  \right\vert \leq \frac{Cp^4}{r},
	\]
	for $A,B = 1,2$,
	\[
		\left\vert J_{(4)} ({J_{(4)}}^4) \big\vert_{s=0} - r  \right\vert \leq C, \qquad 
		\left\vert J_{(4)} ( (\hat{\nabla}_X J_{(4)})^4) \big\vert_{s=0} - p^4  \right\vert \leq \frac{Cp^4}{r},
	\]
	and
	\[
		\left\vert J_{(i_2)} ({J_{(i_1)}}^j) \big\vert_{s=0} \right\vert \leq C, \qquad 
		\left\vert J_{(i_2)} ( (\hat{\nabla}_X J_{(i_1)})^j) \big\vert_{s=0} \right\vert \leq \frac{Cp^4}{r},
	\]
	for $i_1,i_2 = 1,\ldots,6$, $j=1,\ldots,4$ otherwise,
	\[
		\left\vert J_{(4)} ({J_{(A)}}^{4+B}) \big\vert_{s=0} - {\delta_A}^B p^4 \right\vert \leq \frac{Cp^4}{r}, 
	\]
	for $A,B = 1,2$,
	\[
		\left\vert J_{(i_2)} ({J_{(i_1)}}^{4+A}) \big\vert_{s=0} \right\vert \leq \frac{Cp^4}{r}, 
	\]
	for $i_1,i_2 = 1,\ldots,6$ otherwise,
	\[
		\left\vert J_{(i_2)} ( (\hat{\nabla}_X {J_{(i_1)})}^{4+A}) \big\vert_{s=0}  \right\vert \leq \frac{Cp^4}{r^2},
	\]
	for all $i_1,i_2 = 1,\ldots,6$, and
	\[
		\left\vert J_{(i_2)} ({J_{(i_1)}}^7) \big\vert_{s=0} \right\vert \leq Cp^4, \qquad 
		\left\vert J_{(i_2)} ( (\hat{\nabla}_X {J_{(i_1)})}^7) \big\vert_{s=0}  \right\vert \leq \frac{Cp^4}{r},
	\]
	for all $i_1,i_2 = 1,\ldots,6$.
\end{proposition}

\begin{proof}
	Consider the expressions for $J_{(i_1)} (0), \hat{\nabla}_X J_{(i_1)} (0)$ before Proposition \ref{prop:Jbounds}.  The proof follows by applying $V_{(i_2)}$ to the components, noting that,
	\[
		V_{(4)} (r) = r, \qquad V_{(4)} (p^4) = p^4,
	\]
	\[
		\vert V_{(i)}(r) \vert \leq C, \qquad \vert V_{(i)}(p^4) \vert \leq \frac{Cp^4}{r},
	\]
	for $i\neq 4$, and that, by Proposition \ref{prop:schematicpreserved}, one has the same pointwise bounds for $V_{(i_2)}$ applied to the terms involving curvature as one does for the terms involving curvature alone.
\end{proof}

\begin{proposition} \label{prop:JJbounds}
	For $s \in [-s_*,0]$, if $\overline{C}$ is sufficiently small then,
	\[
		\vert J_{(4)} \left( {J_{(A)}}^B \right) (s) - {\delta_A}^B r(s) \vert \leq C, \qquad \vert J_{(4)} \left( {J_{(4)}}^4 \right) (s) - r(s) \vert \leq C,
	\]
	for $A,B = 1,2$,
	\[
		\vert J_{(i_2)} \left( {J_{(i_1)}}^j \right) (s) \vert \leq C,
	\]
	for $i_1,i_2=1,\ldots,6$, $j=1,\ldots,4$ otherwise,
	\[
		\vert J_{(4)} \left( {J_{(A)}}^{4+B}\right) (s) - {\delta_A}^B p^4(s) \vert \leq \frac{Cp^4(s)}{r(s)}, \qquad \vert J_{(i_2)} \left( {J_{(i_1)}}^{4+A} \right) (s) \vert \leq \frac{Cp^4(s)}{r(s)},
	\]
	for all $A,B = 1,2$, $i_1,i_2 = 1,\ldots,6$ such that $(i_2,i_1) \neq (4,A)$, and
	\[
		\vert J_{(i_2)} \left( {J_{(i_1)}}^7 \right) (s) \vert \leq Cp^4(s),
	\]
	for all $i_1,i_2=1,\ldots,7$.  Here $C = C(\overline{C})$ is independent of the point $(x,p)$ and of $s$.
\end{proposition}

\begin{proof}
	The proof follows that of Proposition \ref{prop:Jbounds}.  The commuted Jacobi equation takes the form
	\[
		\frac{d^2 J_{(i_2)} \left( {J_{(i_1)}}^k {\Xi_k}^j \right)}{ds^2} = J_{(i_2)} \left( (\hat{R}(X,J_{(i_1)})X)^k {\Xi_k}^j \right),
	\]
	since $[X,J_{(i_2)}] = 0$.  Assume that $s\in [-s_*,0]$ is such that
	\[
		\left\vert \frac{d^2 J_{(i_2)} \left( {J_{(i_1)}}^k {\Xi_k}^j \right)}{ds^2}(s') - \frac{d^2 J_{(i_2)} \left( {J_{(i_1)}}^k {\Xi_k}^j \right)}{ds^2}(0) \right\vert
		\leq
		\frac{C_1p^4(s')}{r(s')^{\frac{3}{2}}},
	\]
	for $j = 1,\ldots,4$,
	\[
		\left\vert \frac{d^2 J_{(i_2)} \left( {J_{(i_1)}}^k {\Xi_k}^{4+A} \right)}{ds^2}(s') - \frac{d^2 J_{(i_2)} \left( {J_{(i_1)}}^k {\Xi_k}^{4+A} \right)}{ds^2}(0) \right\vert
		\leq
		\frac{C_1(p^4(s'))^2}{r(s')^{2}},
	\]
	for $A=1,2$, and
	\[
		\left\vert \frac{d^2 J_{(i_2)} \left( {J_{(i_1)}}^k {\Xi_k}^7 \right)}{ds^2}(s') - \frac{d^2 J_{(i_2)} \left( {J_{(i_1)}}^k {\Xi_k}^7 \right)}{ds^2}(0) \right\vert
		\leq
		\frac{C_1(p^4(s'))^2}{r(s')^{\frac{3}{2}}},
	\]
	for all $s'\in[s,0]$, for all $i_1,i_2 = 1,\ldots,6$, where $C_1$ is a large constant which will be chosen later.  For $j = 1,\ldots,4$ this immediately gives
	\[
		\left\vert J_{(i_2)} \left( {J_{(i_1)}}^k {\Xi_k}^j \right)(s) - \left( J_{(i_2)} \left( {J_{(i_1)}}^k {\Xi_k}^j \right)(0) + s \frac{d^2 J_{(i_2)} \left( {J_{(i_1)}}^k {\Xi_k}^j \right)}{ds^2}(0) \right) \right\vert
		\leq
		\frac{CC_1}{\sqrt{r(s)}}.
	\]
	By Proposition \ref{prop:JmatrixA}, and Proposition \ref{prop:Jbounds},
	\[
		\left\vert {J_{(i_1)}}^k J_{(i_2)} \left( {\Xi_k}^j \right)(s)  \right\vert \leq C,
	\]
	and hence, by Proposition \ref{prop:Jic} and the fact that,
	\[
		J_{(i_2)} \left( {J_{(i_1)}}^k {\Xi_k}^j \right)(0) = J_{(i_2)} \left( {J_{(i_1)}}^j \right)(0),
	\]
	\[
		\frac{d J_{(i_2)} \left( {J_{(i_1)}}^k {\Xi_k}^{j} \right)}{ds} (0) =  J_{(i_2)} \left( (\hat{\nabla}_X {J_{(i_1)}})^j  \right) (0),
	\]
	this implies,
	\[
		\left\vert J_{(4)} \left( {J_{(A)}}^B \right)(s') - {\delta_A}^B r(s') \right\vert \leq C(1+C_1),
	\]
	for $A,B= 1,2$,
	\[
		\left\vert J_{(4)} \left( {J_{(4)}}^4 \right) (s') - r(s') \right\vert \leq C(1+C_1),
	\]
	and
	\[
		\left\vert J_{(i_2)} \left( {J_{(i_1)}}^j \right) (s') \right\vert \leq C(1+C_1),
	\]
	for $i_1,i_2 = 1,\ldots,6, j = 1,\ldots,4$ otherwise, for all $s'\in[s,0]$, where Proposition \ref{prop:matrixA} has also been used.  Similarly,
	\[
		\left\vert J_{(4)} \left( {J_{(A)}}^{4+B} \right)(s') - {\delta_A}^B p^4(0) \right\vert \leq \frac{C(1+C_1)p^4(s')}{r(s')},
	\]
	for $A,B= 1,2$,
	\[
		\left\vert J_{(i_2)} \left( {J_{(i_1)}}^{4+A} \right)(s') \right\vert \leq \frac{C(1+C_1)p^4(s')}{r(s')},
	\]
	for $A=1,2$, $i_1,i_2 = 1,\ldots,6$ otherwise, and
	\[
		\left\vert J_{(i_2)} \left( {J_{(i_1)}}^{7} \right)(s') \right\vert \leq C(1+C_1)p^4(s'),
	\]
	for all $i_1,i_2 = 1,\ldots,6$, for all $s'\in [s,0]$.
	
	The remainder of the proof proceeds exactly as that of Proposition \ref{prop:Jbounds}, recalling that $[X,J_{(i_2)}] = 0$.  By Proposition \ref{prop:schematicpreserved} and Proposition \ref{prop:JmatrixA} one has the same bounds for,
	\[
		J_{(i_2)} \left( (\hat{R}(X,J_{(i_1)})X)^k {\Xi_k}^j \right),
	\]
	the right hand side of the commuted Jacobi equation, as for the uncommitted equation since the bootstrap assumptions of Section \ref{section:ba} and the Sobolev inequalities give pointwise bounds for $\mathfrak{D} \psi, \mathfrak{D} \mathcal{T}, \mathfrak{D} \Gamma, \mathfrak{D} \Gammaslash, \mathfrak{D} r \nablaslash b$.
\end{proof}

\subsection{$L^2$ Estimates for Components of Jacobi Fields at Higher Orders} \label{subsec:higherjacobi}
To estimate $J_{(i_3)} J_{(i_2)} ( J_{(i_1)}^j)$ and $J_{(i_4)}J_{(i_3)} J_{(i_2)} ( J_{(i_1)}^j)$, the Jacobi equation needs to commuted three and four times respectively.  This will generate terms involving two and three derivatives of Ricci coefficients, Weyl curvature components and energy momentum tensor components.  The higher order derivatives of the components of the Jacobi fields must therefore be estimated in $L^2$.  They will additionally only be estimated after integrating in momentum space, i.\@e.\@ after integrating over $P_x$.

Given $(x,p) \in \supp(f) \cap \{ (x,p)\in P \mid v(x)>v_0 \}$ and $v_0 \leq v'\leq v(x)$, define $s_{v'}(x,p)$ to be the parameter time such that $\pi(\exp_{s_{v'}}(x,p)) \in \{ v = v'\}$, where $\pi : P\to \mathcal{M}$ is the natural projection map.  In this notation,
\[
	-s_*(x,p) = s_{v_0}(x,p),
\]
where $s_*(x,p)$ is defined in Section \ref{subsec:fjacobi}.

The goal of this section is to show that, for all $i_1,i_2,i_3,i_4 = 1,\ldots,6$, $j = 1,\ldots,7$, the quantities
\begin{equation} \label{eq:j34bounds}
	\mathcal{T} \left[ J_{(i_3)} J_{(i_2)} ( J_{(i_1)}^j) (s_{v'}) \right],
	\qquad
	\mathcal{T} \left[ J_{(i_4)}J_{(i_3)} J_{(i_2)} ( J_{(i_1)}^j) (s_{v'}) \right],
\end{equation}
for each $\mathcal{T} = \Tslash_{44}, \Tslash_4, \Tslash, \Tslash_{34}, \Tslash_3, \Tslash_{33}$, can be controlled, for all $v' \in [v_0, v(x)]$, by up to two and three derivatives of Ricci coefficients, curvature components and energy momentum tensor components respectively.  It will then be possible to estimate the quantities \eqref{eq:j34bounds} after taking appropriate weighted square integrals.

The case where two derivatives of the components of the Jacobi fields are taken will first be considered.  Mimicking again the proof of the zeroth order estimates, second order derivatives of the matrices $\Xi$ and $\Xi^{-1}$ are first estimated, followed by estimates for second order derivatives of the initial conditions for the Jacobi equation in Proposition \ref{prop:JJmatrixA} and Proposition \ref{prop:JJic} respectively.  The following Proposition should therefore be compared to Proposition \ref{prop:matrixA} and Proposition \ref{prop:JmatrixA}.

\begin{proposition} \label{prop:JJmatrixA}
	If $v_0$ is sufficiently large, $J_{(i_2)},J_{(i_3)} = J_{(1)},\ldots,J_{(6)}$, then, for all $\mathcal{T}_p$,
	\begin{align*}
		\left\vert \mathcal{T}_p \left[ J_{(i_3)} J_{(i_2)} ( {\Xi_k}^j) (s_{v'}) \right] \right\vert
		\leq
		C \left( \frac{1}{r(x)^p v'} + H_{\Xi,2}(v') \right),
	\end{align*}
	for $j,k = 1,\ldots,7$, and
	\begin{align*}
		\left\vert \mathcal{T}_p \left[ (p^4)^{-1} J_{(i_3)} J_{(i_2)} ( {\Xi_k}^j) (s_{v'}) \right] \right\vert
		\leq
		C \left( \frac{1}{r(x)^p v'} + H_{\Xi,2}(v') \right),
	\end{align*}
	for $j=1,\ldots,4$, $k = 5,6,7$ or vice versa, for all $v_0 \leq v' \leq v(x)$.  Similarly for $\Xi^{-1}$,
	\begin{align*}
		\left\vert \mathcal{T}_p \left[ J_{(i_3)} J_{(i_2)} ( {{\Xi^{-1}}_k}^j) (s_{v'}) \right] \right\vert
		\leq
		C \left( \frac{1}{r(x)^p v'} + H_{\Xi,2}(v') \right),
	\end{align*}
	for $j,k = 1,\ldots,7$, and
	\begin{align*}
		\left\vert \mathcal{T}_p \left[ (p^4)^{-1} J_{(i_3)} J_{(i_2)} ( {{\Xi^{-1}}_k}^j) (s_{v'}) \right] \right\vert
		\leq
		C \left( \frac{1}{r(x)^p v'} + H_{\Xi,2}(v') \right),
	\end{align*}
	for $j=1,\ldots,4$, $k = 5,6,7$ or vice versa.  Here $C$ is a constant which is independent of the point $(x,p)$ (but depends on $\overline{C}$) and
	\begin{align*}
		H_{\Xi,2}(v')
		=
		&
		\frac{1}{\sqrt{v'}} \sum_{\Gamma_q} \mathcal{T}_p \left[ \left( \int_{v'}^{v(x)} r(s_{v''})^{2q-2} \vert \mathfrak{D}^2 \Gamma_q (s_{v''}) \vert^2 d v'' \right)^{\frac{1}{2}} \right]
		\\
		&
		+ \frac{1}{\sqrt{v'}} \mathcal{T}_p \left[ \left( \int_{v'}^{v(x)} \vert \mathfrak{D}^2 \Gammaslash (s_{v''}) \vert^2 d v'' \right)^{\frac{1}{2}} \right]
		+ \frac{1}{v'} \mathcal{T}_p \left[ \left( \int_{v'}^{v(x)} \vert \mathfrak{D}^2 r\nablaslash b (s_{v''}) \vert^2 d v'' \right)^{\frac{1}{2}} \right]
		\\
		&
		+ \frac{1}{v'} \mathcal{T}_p \left[ \left( \int_{v'}^{v(x)} \sum_{\psi_q} r(s_{v''})^{2q-2} \vert \mathfrak{D}^2 \psi_q (s_{v''}) \vert^2 
		+ \sum_{\mathcal{T}_q} r(s_{v''})^{2q-2} \vert \mathfrak{D}^2 \mathcal{T}_q (s_{v''}) \vert^2 d v'' \right)^{\frac{1}{2}} \right],
	\end{align*}
\end{proposition}

\begin{remark}
	It should be noted that the $L^2$ norm on incoming null hypersurfaces of the quantities
	\[
		r(x)^p \sqrt{v'} \sum_{\Gamma_q} \mathcal{T}_p \left[ \left( \int_{v'}^{v(x)} r(s_{\tilde{v}})^{2q-2} \vert \mathfrak{D}^2 \Gamma_q (s_{\tilde{v}}) \vert^2 d \tilde{v} \right)^{\frac{1}{2}} \right],
		\qquad
		r(x)^p \sqrt{v'} \mathcal{T}_p \left[ \left( \int_{v'}^{v(x)} \vert \mathfrak{D}^2 \Gammaslash (s_{\tilde{v}}) \vert^2 d \tilde{v} \right)^{\frac{1}{2}} \right],
	\]
	will be shown to be uniformly bounded.  Direct comparison with the $\frac{1}{r(s)}$ behaviour of ${\Xi_k}^j(s) - {\delta_k}^j$ and $J_{(i_2)}({\Xi_k}^j)(s)$ from Proposition \ref{prop:matrixA} and Proposition \ref{prop:JmatrixA} respectively can therefore be made.  The terms involving curvature components, energy momentum tensor components and $b$ can similarly be controlled after taking their weighted $L^2$ norms on incoming null hypersurfaces.  See Section \ref{subsec:proofemmain2} below.
\end{remark}

\begin{proof}[Proof of Proposition \ref{prop:JJmatrixA}]
The proof proceeds by a bootstrap argument.  Suppose $v' \in [v_0,v(x)]$ is such that, for each $\mathcal{T}_p$, for $j,k = 1,\ldots,7$,
\begin{align} \label{eq:JJmatrixAba}
	\left\vert \mathcal{T}_p \left[ J_{(i_3)} J_{(i_2)} ( {\Xi_k}^j) (s_{\tilde{v}}) \right] \right\vert
	\leq
	C_1 \left( \frac{1}{r(x)^p \tilde{v}} + H_{\Xi,2}(v') \right),
\end{align} 
for all $v'\leq \tilde{v} \leq v(x)$, where $C_1$ is a large constant which will be chosen later.  Note that
\[
	J_{(i_3)} J_{(i_2)} ( {\Xi_k}^j) \big\vert_{s_{v'} = 0} = 0,
\]
so this is clearly true for $s_{v'} = 0$.

Now,
\begin{align} \label{eq:JJmatrixApfmain}
\begin{split}
	&
	\frac{d}{dv'} \left( \mathcal{T}_p \left[ J_{(i_3)} J_{(i_2)} ( {\Xi_k}^j) (s_{v'}) \right] \right)
	=
	\mathcal{T}_p \left[ \frac{1}{p^4} J_{(i_3)} J_{(i_2)} \left( (\hat{\nabla}_X E_l)^j {\Xi_k}^l \right) (s_{v'}) \right]
	\\
	&
	\qquad
	=
	\mathcal{T}_p \bigg[ \frac{1}{p^4} \bigg( {\Xi_k}^l J_{(i_3)} J_{(i_2)} \left( (\hat{\nabla}_X E_l)^j \right) 
	+
	(\hat{\nabla}_X E_l)^j J_{(i_3)} J_{(i_2)} \left( {\Xi_k}^l \right)
	\\
	&
	\qquad
	+
	J_{(i_2)} \left( {\Xi_k}^l \right) J_{(i_3)} \left( (\hat{\nabla}_X E_l)^j \right)
	+
	J_{(i_3)} \left( {\Xi_k}^l \right) J_{(i_2)} \left( (\hat{\nabla}_X E_l)^j \right)
	\bigg) (s_{v'}) \bigg],
\end{split}
\end{align}
since $\frac{d}{dv'} = \frac{ds}{dv'} \frac{d}{ds}$, $\frac{dv'}{ds} = X(v') = p^4$, and
\[
	\left[ \frac{d}{ds},J \right] = [X,J] = 0,
\]
for $J = J_{(i_2)},J_{(i_3)}$.

By Proposition \ref{prop:schematicpreserved}, Proposition \ref{prop:Jbounds} and Proposition \ref{prop:framederivatives},
\[
	\left\vert J_{(i_2)} \left( \left( \hat{\nabla}_X E_l \right)^j \right) (s_{\tilde{v}}) \right\vert \leq \frac{C p^4}{r(s_{\tilde{v}})^2},
\]
for all $j,l = 1,\ldots,7$.  Also, by Proposition \ref{prop:JmatrixA},
\[
	\left\vert J_{(i_3)} \left( {\Xi_k}^l \right) (s_{\tilde{v}}) \right\vert \leq \frac{C}{r(s_{\tilde{v}})},
\]
for all $k,l = 1,\ldots,7$.  Hence,
\[
	\left\vert \mathcal{T}_p \left[ \frac{1}{p^4} J_{(i_3)} \left( {\Xi_k}^l \right) J_{(i_2)} \left( (\hat{\nabla}_X E_l)^j \right) (s_{\tilde{v}}) \right] \right\vert
	\leq
	\left\vert \mathcal{T}_p \left[ \frac{Cp^4}{r(s_{\tilde{v}})^3} \right] \right\vert \leq \frac{C}{r(x)^p \tilde{v}^3},
\]
recalling from the proof of Proposition \ref{prop:tmain} that,
\[
	\left\vert \mathcal{T}_p \left[ p^4 \mathbbm{1}_{\supp(f)} \right] \right\vert \leq \frac{C}{r(x)^p}.
\]
Similarly,
\[
	\left\vert \mathcal{T}_p \left[ \frac{1}{p^4} J_{(i_2)} \left( {\Xi_k}^l \right) J_{(i_3)} \left( (\hat{\nabla}_X E_l)^j \right) (s_{\tilde{v}}) \right] \right\vert
	\leq
	\frac{C}{r(x)^p \tilde{v}^3}.
\]

Using the bootstrap assumptions \eqref{eq:JJmatrixAba} and the pointwise bounds,
\[
	\left\vert \left( \hat{\nabla}_X E_l \right)^j (s_{\tilde{v}}) \right\vert 
	\leq 
	\frac{Cp^4}{r(s_{\tilde{v}})^2}
	\leq
	\frac{Cp^4}{\tilde{v}^2},
\]
from Proposition \ref{prop:framederivatives}, clearly have,
\begin{align*}
	\left\vert \mathcal{T}_p \left[ \frac{1}{p^4} (\hat{\nabla}_X E_l)^j J_{(i_3)} J_{(i_2)} \left( {\Xi_k}^l \right) (s_{\tilde{v}}) \right] \right\vert
	\leq
	\frac{C C_1}{\tilde{v}^2} \left( \frac{1}{r(x)^p \tilde{v}} + H_{\Xi,2}(v') \right).
\end{align*}
Now, using the schematic form of $(\hat{\nabla}_X E_l)^j$ from Proposition \ref{prop:framederivatives}, recalling the better decay for the terms involving Weyl curvature and energy momentum tensor components, the bounds for ${\Xi_j}^k$ from Proposition \ref{prop:matrixA}, Proposition \ref{prop:schematicpreserved} and Propositions \ref{prop:Jbounds}, \ref{prop:JJbounds},
\begin{align*}
	&
	\left\vert \mathcal{T}_p \left[ \frac{1}{p^4} {\Xi_k}^l J_{(i_3)} J_{(i_2)} \left( (\hat{\nabla}_X E_l)^j \right) (s_{\tilde{v}}) \right] \right\vert
	\\
	&
	\qquad
	\leq
	C \mathcal{T}_p \bigg[ \frac{1}{r(s_{\tilde{v}})^2} + \sum_{\Gamma_q} r(s_{\tilde{v}})^{q-2} \vert \mathfrak{D}^2 \Gamma_q (s_{\tilde{v}}) \vert + r(s_{\tilde{v}})^{-2} \vert \mathfrak{D}^2 r\nablaslash b (s_{\tilde{v}}) \vert
	\\
	&
	\qquad \qquad
	+ r(s_{\tilde{v}})^{-1} \vert \mathfrak{D}^2 \Gammaslash (s_{\tilde{v}}) \vert + \sum_{\psi_q} r(s_{\tilde{v}})^{q-\frac{5}{2}} \vert \mathfrak{D}^2 \psi_q (s_{\tilde{v}}) \vert + \sum_{\mathcal{T}_q} r(s_{\tilde{v}})^{q - \frac{5}{2}} \vert \mathfrak{D}^2 \mathcal{T}_q (s_{\tilde{v}}) \vert
	\bigg].
\end{align*}
Note that,
\begin{align*}
	&
	\int_{v'}^{v(x)}  \frac{C C_1}{\tilde{v}^2} \frac{1}{\sqrt{\tilde{v}}} \sum_{\Gamma_q} \mathcal{T}_p \left[ \left( \int_{\tilde{v}}^{v(x)}  r(s_{v''})^{2q-2} \vert \mathfrak{D}^2 \Gamma_q (s_{v''}) \vert^2 d v'' \right)^{\frac{1}{2}} \right] d \tilde{v}
	\\
	&
	\qquad \qquad
	\leq
	C C_1 \sum_{\Gamma_q} \mathcal{T}_p \left[ \left( \int_{v'}^{v(x)}  r(s_{v''})^{2q-2} \vert \mathfrak{D}^2 \Gamma_q (s_{v''}) \vert^2 d v'' \right)^{\frac{1}{2}} \right]  \int_{v'}^{v(x)} \frac{1}{\tilde{v}^{\frac{5}{2}}} d \tilde{v}
	\\
	&
	\qquad \qquad
	\leq
	\frac{C C_1}{v'} \frac{1}{\sqrt{v'}} \sum_{\Gamma_q} \mathcal{T}_p \left[ \left( \int_{v'}^{v(x)}  r(s_{v''})^{2q-2} \vert \mathfrak{D}^2 \Gamma_q (s_{v''}) \vert^2 d v''\right)^{\frac{1}{2}} \right] ,
\end{align*}
and similarly,
\begin{align*}
	&
	\int_{v'}^{v(x)} \frac{C C_1}{\tilde{v}^2} \Bigg( 
	\frac{1}{\sqrt{\tilde{v}}} \mathcal{T}_p \left[ \left( \int_{\tilde{v}}^{v(x)} \vert \mathfrak{D}^2 \Gammaslash (s_{v''}) \vert^2 d v'' \right)^{\frac{1}{2}} \right]
	+ \frac{1}{\tilde{v}} \mathcal{T}_p \left[ \left( \int_{\tilde{v}}^{v(x)} \vert \mathfrak{D}^2 r\nablaslash b (s_{v''}) \vert^2 d v'' \right)^{\frac{1}{2}} \right]
	\\
	&
	\qquad
	+ \frac{1}{\tilde{v}} \mathcal{T}_p \left[ \left( \int_{\tilde{v}}^{v(x)} \sum_{\psi_q} r(s_{v''})^{2q-2} \vert \mathfrak{D}^2 \psi_q (s_{v''}) \vert^2 
	+ \sum_{\mathcal{T}_q} r(s_{v''})^{2q-2} \vert \mathfrak{D}^2 \mathcal{T}_q (s_{v''}) \vert^2 d v'' \right)^{\frac{1}{2}} \right]
	\Bigg) d\tilde{v}
	\\
	&
	\leq
	\frac{C C_1}{v'} \Bigg( 
	\frac{1}{\sqrt{v'}} \mathcal{T}_p \left[ \left( \int_{v'}^{v(x)} \vert \mathfrak{D}^2 \Gammaslash (s_{v''}) \vert^2 d v'' \right)^{\frac{1}{2}} \right]
	+ \frac{1}{v'} \mathcal{T}_p \left[ \left( \int_{v'}^{v(x)} \vert \mathfrak{D}^2 r\nablaslash b (s_{v''}) \vert^2 d v'' \right)^{\frac{1}{2}} \right]
	\\
	&
	\qquad
	+ \frac{1}{v'} \mathcal{T}_p \left[ \left( \int_{v'}^{v(x)} \sum_{\psi_q} r(s_{v''})^{2q-2} \vert \mathfrak{D}^2 \psi_q (s_{v''}) \vert^2 
	+ \sum_{\mathcal{T}_q} r(s_{v''})^{2q-2} \vert \mathfrak{D}^2 \mathcal{T}_q (s_{v''}) \vert^2 d v'' \right)^{\frac{1}{2}} \right]
	\Bigg).
\end{align*}
Also,
\begin{align*}
	&
	\int_{v'}^{v(x)} \mathcal{T}_p \left[ r(s_{\tilde{v}})^{q-2} \vert \mathfrak{D}^2 \Gamma_q (s_{\tilde{v}}) \vert \right] d \tilde{v}
	\\
	&
	\qquad
	\leq
	\mathcal{T}_p \left[ \left( \int_{v'}^{v(x)} \frac{1}{r(s_{\tilde{v}})^{2}} d \tilde{v} \right)^{\frac{1}{2}} \left( \int_{v'}^{v(x)} r(s_{\tilde{v}})^{2q-2} \vert \mathfrak{D}^2 \Gamma_q(s_{\tilde{v}}) \vert^2 d \tilde{v} \right)^{\frac{1}{2}} \right]
	\\
	&
	\qquad
	\leq
	\frac{C}{\sqrt{v'}} \mathcal{T}_p \left[ \left( \int_{v'}^{v(x)} r(s_{\tilde{v}})^{2q-2} \vert \mathfrak{D}^2 \Gamma_q(s_{\tilde{v}}) \vert^2 d \tilde{v} \right)^{\frac{1}{2}} \right],
\end{align*}
and similarly,
\begin{align*}
	\int_{v'}^{v(x)} \mathcal{T}_p \left[ r(s_{\tilde{v}})^{-1} \vert \mathfrak{D}^2 \Gammaslash (s_{\tilde{v}}) \vert \right] d \tilde{v}
	\leq
	&
	\frac{1}{ \sqrt{v'}} \mathcal{T}_p \left[ \left( \int_{v'}^{v(x)} \vert \mathfrak{D}^2 \Gammaslash (s_{\tilde{v}}) \vert^2 d \tilde{v} \right)^{\frac{1}{2}} \right],
	\\
	\int_{v'}^{v(x)} \mathcal{T}_p \left[ r(s_{\tilde{v}})^{-2} \vert \mathfrak{D}^2 r \nablaslash b (s_{\tilde{v}}) \vert \right] d \tilde{v}
	\leq
	&
	\frac{1}{ v'} \mathcal{T}_p \left[ \left( \int_{v'}^{v(x)} \vert \mathfrak{D}^2 r\nablaslash b (s_{\tilde{v}}) \vert^2 d \tilde{v} \right)^{\frac{1}{2}} \right],
\end{align*}
and
\begin{multline*}
	\int_{v'}^{v(x)} \mathcal{T}_p \left[ r(s_{\tilde{v}})^{q-\frac{5}{2}} \left( \vert \mathfrak{D}^2 \psi_q (s_{\tilde{v}})\vert + \vert \mathfrak{D}^2 \mathcal{T}_q (s_{\tilde{v}})\vert \right)\right] d \tilde{v}
	\\
	\leq
	\frac{1}{v'} \mathcal{T}_p \left[ \left( \int_{v'}^{v(x)} r(s_{\tilde{v}})^{2q-2} \left( \vert \mathfrak{D}^2 \psi_q(s_{\tilde{v}}) \vert^2 + \vert \mathfrak{D}^2 \mathcal{T}_q(s_{\tilde{v}}) \vert^2 \right) d \tilde{v} \right)^{\frac{1}{2}} \right].
\end{multline*}
Hence, integrating equation \eqref{eq:JJmatrixApfmain} from $v'$ to $v(x)$ and using the fact that,
\[
	J_{(i_3)} J_{(i_2)} \left( {\Xi_k}^j \right)\big\vert_{s=0} = 0,
\]
it follows that
\begin{align*}
		\left\vert \mathcal{T}_p \left[ J_{(i_3)} J_{(i_2)} ( {\Xi_k}^j) (s_{v'}) \right] \right\vert
		\leq
		C \left( 1 + \frac{C_1}{v'} \right) \left( \frac{1}{r(x)^p v'} + H_{\Xi,2}(v') \right).
\end{align*}
Now choose $C_1$ large so that $C_1 >4C$, where $C$ is the constant appearing in the above inequality, and $v_0$ large so that $\frac{C_1}{v'} \leq 1$.  Then,
\[
	C\left( 1 + \frac{C_1}{v'} \right) \leq \frac{C_1}{4} \cdot 2 = \frac{C_1}{2},
\]
and the bootstrap assumption \eqref{eq:JJmatrixAba} has been recovered with a better constant.  Hence the set of $v'\in [v_0,v(x)]$ such that the bootstrap assumption \eqref{eq:JJmatrixAba} holds for all $v'\leq \tilde{v} \leq v(x)$ is non-empty, open and closed and hence equal to $[v_0,v(x)]$.

The proof of the second part follows by making the bootstrap assumption,
\begin{align}
	\left\vert \mathcal{T}_p \left[ (p^4)^{-1} J_{(i_3)} J_{(i_2)} ( {\Xi_k}^j) (s_{\tilde{v}}) \right] \right\vert
	\leq
	C_1 \left( \frac{1}{r(x)^p \tilde{v}} + H_{\Xi,2}(v') \right),
\end{align}
for all $v'\leq \tilde{v} \leq v(x)$, for $j = 1,\ldots,4, k = 5,6,7$ or vice versa.  The proof proceeds as before, now using the fact that
\begin{align*}
	\left\vert \left( \hat{\nabla}_X E_l \right)^j (s_{\tilde{v}}) \right\vert 
	&
	\leq \frac{C(p^4)^2}{r(s_{\tilde{v}})^2},
	\\
	\left\vert J_{(i_2)} \left( \left( \hat{\nabla}_X E_l \right)^j \right) (s_{\tilde{v}}) \right\vert 
	&
	\leq \frac{C(p^4)^2}{r(s_{\tilde{v}})^2},
	\\
	\left\vert {\Xi_l}^j (s_{\tilde{v}}) \right\vert 
	&
	\leq \frac{Cp^4}{r(s_{\tilde{v}})},
	\\
	\left\vert J_{(i_2)} \left( {\Xi_l}^j \right) (s_{\tilde{v}}) \right\vert 
	&
	\leq \frac{Cp^4}{r(s_{\tilde{v}})},
\end{align*}
and
\begin{align*}
	&
	\left\vert \mathcal{T}_p \left[ (p^4)^{-2} J_{(i_3)} J_{(i_2)} \left( (\hat{\nabla}_X E_l)^j \right) (s_{\tilde{v}}) \right] \right\vert
	\\
	&
	\qquad
	\leq
	C \mathcal{T}_p \bigg[ \frac{1}{r(s_{\tilde{v}})^2} + \sum_{\Gamma_q} r(s_{\tilde{v}})^{q-2} \vert \mathfrak{D}^2 \Gamma_q (s_{\tilde{v}}) \vert + r(s_{\tilde{v}})^{-2} \vert \mathfrak{D}^2 r\nablaslash b (s_{\tilde{v}}) \vert
	\\
	&
	\qquad \qquad
	+ r(s_{\tilde{v}})^{-1} \vert \mathfrak{D}^2 \Gammaslash (s_{\tilde{v}}) \vert + \sum_{\psi_q} r(s_{\tilde{v}})^{q-\frac{5}{2}} \vert \mathfrak{D}^2 \psi_q (s_{\tilde{v}}) \vert + \sum_{\mathcal{T}_q} r(s_{\tilde{v}})^{q - \frac{5}{2}} \vert \mathfrak{D}^2 \mathcal{T}_q (s_{\tilde{v}}) \vert
	\bigg],
\end{align*}
for $l=1,\ldots,4$, $j = 5,6,7$ or vice versa.

The proof for ${{\Xi^{-1}}_k}^j$ is identical.
\end{proof}

The next proposition gives estimates for the initial conditions for the commuted Jacobi equation.  Again, the leading order terms of some of the components have to be subtracted first.

\begin{proposition} \label{prop:JJic}
	At time $s_{v'}=0$,
	\begin{align*}
		\left\vert J_{(4)} J_{(4)} \left( {J_{(A)}}^B \right)\Big\vert_{s_{v'} = 0} - {\delta_A}^B r \right\vert
		&
		\leq
		C \left( 1 + \sum_{\Gamma_q} r^q \vert \mathfrak{D}^2 \Gamma_q \vert + \sum_{\psi_q} r^q \vert \mathfrak{D}^2 \psi_q \vert + \sum_{\mathcal{T}_q} r^q \vert \mathfrak{D}^2 \mathcal{T}_q \vert \right),
		\\
		\left\vert J_{(4)} J_{(4)} \left( (\hat{\nabla}_X J_{(A)})^B \right)\Big\vert_{s_{v'} = 0} - {\delta_A}^B p^4 \right\vert
		&
		\leq
		\frac{Cp^4}{r} \left( 1 + \sum_{\Gamma_q} r^q \vert \mathfrak{D}^2 \Gamma_q \vert + \sum_{\psi_q} r^q \vert \mathfrak{D}^2 \psi_q \vert + \sum_{\mathcal{T}_q} r^q \vert \mathfrak{D}^2 \mathcal{T}_q \vert \right),
	\end{align*}
	for $A,B = 1,2$,
	\begin{align*}
		\left\vert J_{(4)} J_{(4)} \left( {J_{(4)}}^4 \right)\Big\vert_{s_{v'} = 0} - r \right\vert
		&
		\leq
		C \left( 1 + \sum_{\Gamma_q} r^q \vert \mathfrak{D}^2 \Gamma_q \vert + \sum_{\psi_q} r^q \vert \mathfrak{D}^2 \psi_q \vert + \sum_{\mathcal{T}_q} r^q \vert \mathfrak{D}^2 \mathcal{T}_q \vert \right),
		\\
		\left\vert J_{(4)} J_{(4)} \left( (\hat{\nabla}_X J_{(4)})^4 \right)\Big\vert_{s_{v'} = 0} - p^4 \right\vert
		&
		\leq
		\frac{Cp^4}{r} \left( 1 + \sum_{\Gamma_q} r^q \vert \mathfrak{D}^2 \Gamma_q \vert + \sum_{\psi_q} r^q \vert \mathfrak{D}^2 \psi_q \vert + \sum_{\mathcal{T}_q} r^q \vert \mathfrak{D}^2 \mathcal{T}_q \vert \right),
	\end{align*} 
	and,
	\begin{align*}
		\left\vert J_{(i_3)} J_{(i_2)} \left( {J_{(i_1)}}^j \right)\Big\vert_{s_{v'} = 0} \right\vert
		&
		\leq
		C \left( 1 + \sum_{\Gamma_q} r^q \vert \mathfrak{D}^2 \Gamma_q \vert + \sum_{\psi_q} r^q \vert \mathfrak{D}^2 \psi_q \vert + \sum_{\mathcal{T}_q} r^q \vert \mathfrak{D}^2 \mathcal{T}_q \vert \right),
		\\
		\left\vert J_{(i_3)} J_{(i_2)} \left( (\hat{\nabla}_X J_{(i_1)})^j \right)\Big\vert_{s_{v'} = 0} \right\vert
		&
		\leq
		\frac{Cp^4}{r} \left( 1 + \sum_{\Gamma_q} r^q \vert \mathfrak{D}^2 \Gamma_q \vert + \sum_{\psi_q} r^q \vert \mathfrak{D}^2 \psi_q \vert + \sum_{\mathcal{T}_q} r^q \vert \mathfrak{D}^2 \mathcal{T}_q \vert \right),
	\end{align*}
	for $i_1,i_2,i_3 = 1,\ldots,6$, $j = 1,\ldots,4$ otherwise,
	\[
		\left\vert J_{(4)} J_{(4)} \left( {J_{(A)}}^{4+B} \right)\Big\vert_{s_{v'} = 0} - {\delta_A}^B p^4 \right\vert
		\leq
		\frac{Cp^4}{r} \left( 1 + \sum_{\Gamma_q} r^q \vert \mathfrak{D}^2 \Gamma_q \vert + \sum_{\psi_q} r^q \vert \mathfrak{D}^2 \psi_q \vert + \sum_{\mathcal{T}_q} r^q \vert \mathfrak{D}^2 \mathcal{T}_q \vert \right),
	\]
	for $A,B = 1,2$,
	\[
		\left\vert J_{(i_3)} J_{(i_2)} \left( {J_{(i_1)}}^{A} \right)\Big\vert_{s_{v'} = 0} \right\vert
		\leq
		\frac{Cp^4}{r} \left( 1 + \sum_{\Gamma_q} r^q \vert \mathfrak{D}^2 \Gamma_q \vert + \sum_{\psi_q} r^q \vert \mathfrak{D}^2 \psi_q \vert + \sum_{\mathcal{T}_q} r^q \vert \mathfrak{D}^2 \mathcal{T}_q \vert \right),
	\]
	for $i_1,i_2,i_3 = 1,\ldots,6$ otherwise,
	\[
		\left\vert J_{(i_3)} J_{(i_2)} \left( (\hat{\nabla}_X J_{(i_1)})^{4+A} \right)\Big\vert_{s_{v'} = 0} \right\vert
		\leq
		\frac{Cp^4}{r^2} \left( 1 + \sum_{\Gamma_q} r^q \vert \mathfrak{D}^2 \Gamma_q \vert + \sum_{\psi_q} r^q \vert \mathfrak{D}^2 \psi_q \vert + \sum_{\mathcal{T}_q} r^q \vert \mathfrak{D}^2 \mathcal{T}_q \vert \right),
	\]
	for all $i_1,i_2,i_3 = 1,\ldots,6$, and
	\begin{align*}
		\left\vert J_{(i_3)} J_{(i_2)} \left( {J_{(i_1)}}^{7} \right)\Big\vert_{s_{v'} = 0} \right\vert
		&
		\leq
		Cp^4 \left( 1 + \sum_{\Gamma_q} r^q \vert \mathfrak{D}^2 \Gamma_q \vert + \sum_{\psi_q} r^q \vert \mathfrak{D}^2 \psi_q \vert + \sum_{\mathcal{T}_q} r^q \vert \mathfrak{D}^2 \mathcal{T}_q \vert \right),
		\\
		\left\vert J_{(i_3)} J_{(i_2)} \left( (\hat{\nabla}_X J_{(i_1)})^{7} \right)\Big\vert_{s_{v'} = 0} \right\vert
		&
		\leq
		\frac{Cp^4}{r} \left( 1 + \sum_{\Gamma_q} r^q \vert \mathfrak{D}^2 \Gamma_q \vert + \sum_{\psi_q} r^q \vert \mathfrak{D}^2 \psi_q \vert + \sum_{\mathcal{T}_q} r^q \vert \mathfrak{D}^2 \mathcal{T}_q \vert \right),
	\end{align*}
	for all $i_1,i_2,i_3 = 1,\ldots,6$.
\end{proposition}

\begin{proof}
	Again follows from considering expressions for $V_{(i_1)}, \hat{\nabla}_X V_{(i_1)}$, differentiating the components and using the fact that,
	\[
		V_{(4)} (r) = r, \qquad V_{(4)} (p^4) = p^4,
	\]
	\[
		\vert V_{(i)}(r) \vert \leq C, \qquad \vert V_{(i)}(p^4) \vert \leq \frac{Cp^4}{r},
	\]
	for $i\neq 4$, as in Proposition \ref{prop:Jic}.
\end{proof}

\begin{proposition} \label{prop:JJJbounds}
	If $v_0$ is sufficiently large, $v'\in[v_0,v(x)]$, then, for each $\mathcal{T}_p$,
	\[
		\left\vert \mathcal{T}_p \left[ J_{(4)} J_{(4)} \left( {J_{(A)}}^B \right) (s_{v'}) - r(s_{v'}) {\delta_A}^B \right] \right\vert
		\leq
		C \left( \frac{1}{r(x)^p} + H_{\mathcal{T}_p,2}(v') \right),
	\]
	for $A,B = 1,2$,
	\[
		\left\vert \mathcal{T}_p \left[ J_{(4)} J_{(4)} \left( {J_{(4)}}^4 \right) (s_{v'}) - r(s_{v'}) \right] \right\vert
		\leq
		C \left( \frac{1}{r(x)^p} + H_{\mathcal{T}_p,2}(v') \right),
	\]
	and
	\[
		\left\vert \mathcal{T}_p \left[ J_{(i_3)} J_{(i_2)} \left( {J_{(i_1)}}^j \right) (s_{v'}) \right] \right\vert
		\leq
		C \left( \frac{1}{r(x)^p} + H_{\mathcal{T}_p,2}(v') \right),
	\]
	for $i_1,i_2,i_3 = 1,\ldots,6$, $j = 1,\ldots,4$ otherwise,
	\[
		\left\vert \mathcal{T}_p \left[ (p^4)^{-1} J_{(4)} J_{(4)} \left( {J_{(A)}}^{4+B} \right) (s_{v'}) - {\delta_A}^B \right] \right\vert
		\leq
		\frac{C}{v'} \left( \frac{1}{r(x)^p} + H_{\mathcal{T}_p,2}(v') \right),
	\]
	for $A,B = 1,2$,
	\[
		\left\vert \mathcal{T}_p \left[ (p^4)^{-1} J_{(i_3)} J_{(i_2)} \left( {J_{(i_1)}}^{4+A} \right) (s_{v'}) \right] \right\vert
		\leq
		\frac{C}{v'} \left( \frac{1}{r(x)^p} + H_{\mathcal{T}_p,2}(v') \right),
	\]
	for $i_1,i_2,i_3 = 1,\ldots,6$, $A = 1,2$ otherwise, and,
	\[
		\left\vert \mathcal{T}_p \left[ (p^4)^{-1} J_{(i_3)} J_{(i_2)} \left( {J_{(i_1)}}^{7} \right) (s_{v'}) \right] \right\vert
		\leq
		C \left( \frac{1}{r(x)^p} + H_{\mathcal{T}_p,2}(v') \right),
	\]
	where,
	\begin{align*}
		H_{\mathcal{T}_p,2}(v')
		=
		&
		\mathcal{T}_p \left[ \left( \int_{v'}^{v(x)} r(s_{\tilde{v}})^6 \vert \mathfrak{D}^2 \beta \vert^2 (s_{\tilde{v}}) d\tilde{v} \right)^{\frac{1}{2}} \right]
		\\
		&
		+
		\mathcal{T}_p \left[ \left( \int_{v'}^{v(x)} \sum_{\psi_q \neq \beta} r(s_{\tilde{v}})^{2q-2} \vert \mathfrak{D}^2 \psi_q \vert^2 (s_{\tilde{v}}) d\tilde{v} \right)^{\frac{1}{2}} \right]
		\\
		&
		+
		\mathcal{T}_p \left[ \left( \int_{v'}^{v(x)} \sum_{\mathcal{T}_q} r(s_{\tilde{v}})^{2q-2} \vert \mathfrak{D}^2 \mathcal{T}_q \vert^2 (s_{\tilde{v}}) d\tilde{v} \right)^{\frac{1}{2}} \right]
		\\
		&
		+
		\mathcal{T}_p \left[ \left( \int_{v'}^{v(x)} \sum_{\Gamma_q} r(s_{\tilde{v}})^{2q-\frac{3}{2}} \vert \mathfrak{D}^2 \Gamma_q \vert^2 (s_{\tilde{v}}) d\tilde{v} \right)^{\frac{1}{2}} \right]
		\\
		&
		+
		\mathcal{T}_p \left[ \left( \int_{v'}^{v(x)} r(s_{\tilde{v}})^{\frac{1}{2}} \vert \mathfrak{D}^2 \Gammaslash \vert^2 (s_{\tilde{v}}) d\tilde{v} \right)^{\frac{1}{2}} \right]
		+ \mathcal{T}_p \left[ \left( \int_{v'}^{v(x)} \vert \mathfrak{D}^2 r \nablaslash b \vert^2 (s_{\tilde{v}}) d\tilde{v} \right)^{\frac{1}{2}} \right]
		\\
		&
		+
		\frac{1}{r(x)^p} \bigg(
		\sum_{\psi_q} r(x)^q \vert \mathfrak{D}^2 \psi_q \vert (x)
		+
		\sum_{\mathcal{T}_q} r(x)^q \vert \mathfrak{D}^2 \mathcal{T}_q \vert (x)
		\\
		&
		+
		\sum_{\Gamma_q} r(x)^q \vert \mathfrak{D}^2 \Gamma_q \vert (x)
		+
		\vert \mathfrak{D}^2 r \nablaslash b \vert (x)
		+
		r(x) \vert \mathfrak{D}^2 \Gammaslash \vert (x)
		\bigg).
	\end{align*}
\end{proposition}

\begin{proof}
	Suppose $v'\in [v_0,v(x)]$ is such that the following bootstrap assumptions hold for all $\tilde{v} \in [v',v(x)]$,
	\begin{align*}
		&
		\Bigg\vert \mathcal{T}_p \left[ J_{(i_3)} J_{(i_2)} \left( {J_{(i_1)}}^k {\Xi_k}^j \right) (s_{\tilde{v}})\right]
		- 
		\mathcal{T}_p \left[ J_{(i_3)} J_{(i_2)} \left( {J_{(i_1)}}^j \right)\Big\vert_{s_{v'} = 0} \right] 
		\\
		&
		+ \left( v(x) - \tilde{v} \right) \mathcal{T}_p \left[ \frac{1}{p^4} \frac{ d J_{(i_3)} J_{(i_2)} \left( {J_{(i_1)}}^k {\Xi_k}^j \right)}{ds} \bigg\vert_{s_{v'} = 0} \right]
		\Bigg\vert
		\\
		&
		\qquad \qquad
		\leq
		\frac{C_1}{r(x)^p} \left( 1 + H_{\mathcal{T}_p,2}(\tilde{v}) \right),
	\end{align*}
	for all $i_1,i_2,i_3 = 1,\ldots,6$, $j = 1,\ldots,4$,
	\begin{align*}
		&
		\Bigg\vert \mathcal{T}_p \left[ (p^4)^{-1} J_{(i_3)} J_{(i_2)} \left( {J_{(i_1)}}^k {\Xi_k}^{4+A} \right) (s_{\tilde{v}})\right]
		- 
		\mathcal{T}_p \left[ (p^4)^{-1} J_{(i_3)} J_{(i_2)} \left( {J_{(i_1)}}^{4+A} \right)\Big\vert_{s_{v'} = 0} \right] 
		\\
		&
		+ \left( v(x) - \tilde{v} \right) \mathcal{T}_p \left[ \frac{1}{(p^4)^2} \frac{ d J_{(i_3)} J_{(i_2)} \left( {J_{(i_1)}}^k {\Xi_k}^{4+A} \right)}{ds} \bigg\vert_{s_{v'} = 0} \right]
		\Bigg\vert
		\\
		&
		\qquad \qquad
		\leq
		\frac{C_1 }{r(x)^p r(s_{\tilde{v}}) } \left( 1 + H_{\mathcal{T}_p,2}(\tilde{v}) \right),
	\end{align*}
	for $A = 1,2$, and
	\begin{align*}
		&
		\Bigg\vert \mathcal{T}_p \left[ (p^4)^{-1} J_{(i_3)} J_{(i_2)} \left( {J_{(i_1)}}^k {\Xi_k}^{7} \right) (s_{\tilde{v}})\right]
		- 
		\mathcal{T}_p \left[ (p^4)^{-1} J_{(i_3)} J_{(i_2)} \left( {J_{(i_1)}}^{7} \right)\Big\vert_{s_{v'} = 0} \right] 
		\\
		&
		+ \left( v(x) - \tilde{v} \right) \mathcal{T}_p \left[ \frac{1}{(p^4)^2} \frac{ d J_{(i_3)} J_{(i_2)} \left( {J_{(i_1)}}^k {\Xi_k}^{7} \right)}{ds} \bigg\vert_{s_{v'} = 0} \right]
		\Bigg\vert
		\\
		&
		\qquad \qquad
		\leq
		\frac{C_1}{r(x)^p} \left( 1 + H_{\mathcal{T}_p,2}(\tilde{v}) \right).
	\end{align*}
	Here $C_1$ is a large constant which will be chosen later.
	
	Note that, for each $p\in P_x$,
	\[
		\left\vert (v(x) -  \tilde{v}) + s_{\tilde{v}} p^4(0) \right\vert \leq \frac{C}{r(s_{\tilde{v}})} \leq \frac{C}{\tilde{v}}.
	\]
	The proof of this fact is identical to that of Lemma \ref{lem:sstar}, using the fact that $X(\tilde{v}) = p^4(s_{\tilde{v}})$.  Using this fact along with Proposition \ref{prop:JJic} and Lemma \ref{lem:sstar}, the bootstrap assumptions immediately give,
	\[
		\left\vert \mathcal{T}_p \left[ J_{(4)} J_{(4)} \left( {J_{(A)}}^B \right) (s_{v'}) - r(s_{v'}) {\delta_A}^B \right] \right\vert
		\leq
		\frac{C(1 + C_1)}{r(x)^p} \left( 1 + H_{\mathcal{T}_p,2}(v') \right),
	\]
	for $A,B = 1,2$,
	\[
		\left\vert \mathcal{T}_p \left[ J_{(4)} J_{(4)} \left( {J_{(4)}}^4 \right) (s_{v'}) - r(s_{v'}) \right] \right\vert
		\leq
		\frac{C(1 + C_1)}{r(x)^p} \left( 1 + H_{\mathcal{T}_p,2}(v') \right),
	\]
	and
	\[
		\left\vert \mathcal{T}_p \left[ J_{(i_3)} J_{(i_2)} \left( {J_{(i_1)}}^j \right) (s_{v'}) \right] \right\vert
		\leq
		\frac{C(1 + C_1)}{r(x)^p} \left( 1 + H_{\mathcal{T}_p,2}(v') \right),
	\]
	for $i_1,i_2,i_3 = 1,\ldots,6$, $j = 1,\ldots,4$ otherwise,
	\[
		\left\vert \mathcal{T}_p \left[ (p^4)^{-1} J_{(4)} J_{(4)} \left( {J_{(A)}}^{4+B} \right) (s_{v'}) - {\delta_A}^B \right] \right\vert
		\leq
		\frac{C(1 + C_1)}{r(x)^p v'} \left( 1 + H_{\mathcal{T}_p,2}(v') \right),
	\]
	for $A,B = 1,2$,
	\[
		\left\vert \mathcal{T}_p \left[ (p^4)^{-1} J_{(i_3)} J_{(i_2)} \left( {J_{(i_1)}}^{4+A} \right) (s_{v'}) \right] \right\vert
		\leq
		\frac{C(1 + C_1)}{r(x)^p v'} \left( 1 + H_{\mathcal{T}_p,2}(v') \right),
	\]
	for $i_1,i_2,i_3 = 1,\ldots,6$, $A = 1,2$ otherwise, and,
	\[
		\left\vert \mathcal{T}_p \left[ (p^4)^{-1} J_{(i_3)} J_{(i_2)} \left( {J_{(i_1)}}^{7} \right) (s_{v'}) \right] \right\vert
		\leq
		\frac{C(1 + C_1)}{r(x)^p} \left( 1 + H_{\mathcal{T}_p,2}(v') \right).
	\]
	
	It remains to recover the bootstrap assumptions with better constants.  This again uses the twice commuted Jacobi equation in components, which takes the form,
	\[
		\frac{d^2 J_{(i_3)} J_{(i_2)} \left( {J_{(i_1)}}^k {\Xi_k}^j \right)}{ds^2} = J_{(i_3)} J_{(i_2)} \left( (\hat{R}(X, {J_{(i_1)}})X )^k {\Xi_k}^j \right),
	\]
	By Propositions \ref{prop:Jbounds}, \ref{prop:JJbounds}, \ref{prop:JmatrixA}, \ref{prop:JJmatrixA}, \ref{prop:matrixA}, \ref{prop:schematicpreserved}, the expression \eqref{eq:curvatureuseful} for $\hat{R}$, the schematic expression \eqref{eq:schematiccurvature}\footnote{recall that the summation can begin at 3 except for terms involving $\beta$.}, the fact that\footnote{recall $\lambda$ runs over $1,2,4$ and $\tilde{\lambda}(1) = 5, \tilde{\lambda}(2) = 6, \tilde{\lambda}(4) = 7$.}
	\[
		\hat{\nabla}_X \hor_{(\gamma,\dot{\gamma})} \left( R(\dot{\gamma},J^v_{(i_1)}) \dot{\gamma} \right)
		=
		X \left( \left( R(\dot{\gamma},J^v_{(i_1)}) \dot{\gamma} \right)^{\mu} \right) E_{\mu}
		+
		\left( R(\dot{\gamma},J^v_{(i_1)}) \dot{\gamma} \right)^{\mu} \hat{\nabla}_X E_{\mu},
	\]
	\[
		\hat{\nabla}_X {}^T \ver_{(\gamma,\dot{\gamma})} \left( R(\dot{\gamma},J^h_{(i_1)}) \dot{\gamma} \right)
		=
		X \left( \left( R(\dot{\gamma},J^h_{(i_1)}) \dot{\gamma} \right)^{\lambda} \right) E_{\tilde{\lambda}(\lambda)}
		+
		\left( R(\dot{\gamma},J^h_{(i_1)}) \dot{\gamma} \right)^{\lambda} \hat{\nabla}_X E_{\tilde{\lambda}(\lambda)},
	\]
	Proposition \ref{prop:framederivatives}, and the fact that $[X, J_{(i_3)}] = [X, J_{(i_2)}] = 0$, the above bounds on,
	\[
		\mathcal{T}_p \left[ J_{(i_3)} J_{(i_2)} \left( {J_{(i_1)}}^j \right) \right],
	\]
	imply, for $j = 1,\ldots,4$,
	\begin{align} \label{eq:JJJboundsmain}
	\begin{split}
		&
		\mathcal{T}_p \left[ \frac{1}{(p^4)^2} \frac{d^2 J_{(i_3)} J_{(i_2)} \left( {J_{(i_1)}}^k {\Xi_k}^j \right)}{ds^2} (s_{\tilde{v}}) \right]
		\leq
		C \mathcal{T}_p \left[ \frac{1}{(p^4)^2} X \left( J_{(i_3)} J_{(i_2)} \left( R(\dot{\gamma},J^v_{(i_1)}) \dot{\gamma} \right)^k {\Xi_k}^j \right) \right]
		\\
		&
		\qquad \qquad
		+ \frac{C(1+C_1) }{r(x)^p \tilde{v}^{\frac{5}{2}}}  \left( 1 + H_{\mathcal{T}_p,2}(\tilde{v}) \right)
		+ C \bigg( \mathcal{T}_p \left[ r(s_{\tilde{v}}) \vert \mathfrak{D}^2 \beta \vert (s_{\tilde{v}}) \right]
		\\
		&
		\qquad \qquad
		+ \sum_{\psi_q\neq \beta} \mathcal{T}_p \left[ r(s_{\tilde{v}})^{q-3} \vert \mathfrak{D}^2 \psi_q \vert (s_{\tilde{v}}) \right]
		+ \sum_{\mathcal{T}_q} \mathcal{T}_p \left[ r(s_{\tilde{v}})^{q-3} \vert \mathfrak{D}^2 \mathcal{T}_q \vert (s_{\tilde{v}}) \right]
		\\
		&
		\qquad \qquad
		+ \sum_{\Gamma_q} \mathcal{T}_p \left[ r(s_{\tilde{v}})^{q-\frac{5}{2}} \vert \mathfrak{D}^2 \Gamma_q \vert (s_{\tilde{v}}) \right]
		+ \mathcal{T}_p \left[ r(s_{\tilde{v}})^{-\frac{5}{2}} \vert \mathfrak{D}^2 r \nablaslash b \vert (s_{\tilde{v}}) \right]
		\\
		&
		\qquad \qquad
		+ \mathcal{T}_p \left[ r(s_{\tilde{v}})^{-\frac{3}{2}} \vert \mathfrak{D}^2 \Gammaslash \vert (s_{\tilde{v}}) \right]
		\bigg),
	\end{split}
	\end{align}
	where the fact that,
	\[
		\left\vert \mathcal{T}_p \left[ p^4 \mathbbm{1}_{\supp(f)} \right]  \right\vert \leq \frac{C}{r(x)^p},
	\]
	has also been used.  Now,
	\[
		\frac{d}{dv'} \mathcal{T}_p \left[ J_{(i_3)} J_{(i_2)} \left( \frac{d{J_{(i_1)}}^k {\Xi_k}^j}{ds} \right) (s_{\tilde{v}}) \right] 
		=
		\mathcal{T}_p \left[ \frac{1}{p^4} J_{(i_3)} J_{(i_2)} \left( \frac{d^2{J_{(i_1)}}^k {\Xi_k}^j}{ds^2} \right) (s_{\tilde{v}}) \right], 
	\]
	so $\mathcal{T}_p \left[ J_{(i_3)} J_{(i_2)} \left( \frac{d{J_{(i_1)}}^k {\Xi_k}^j}{ds} \right) (s_{\tilde{v}}) \right]$ is estimated by integrating \eqref{eq:JJJboundsmain} from $\tilde{v}$ to $v(x)$.  Consider the first term on the right hand side of \eqref{eq:JJJboundsmain},
	\begin{align*}
		&
		\mathcal{T}_p \left[ \frac{1}{p^4} X \left( J_{(i_3)} J_{(i_2)} \left( R(\dot{\gamma},J^v_{(i_1)}) \dot{\gamma} \right)^k {\Xi_k}^j \right) (s_{v''}) \right]
		\\
		&
		\qquad \qquad
		=
		\mathcal{T}_p \left[ \frac{1}{p^4} \frac{d J_{(i_3)} J_{(i_2)} \left( R(\dot{\gamma},J^v_{(i_1)}) \dot{\gamma} \right)^k {\Xi_k}^j}{ds} (s_{v''}) \right]
		\\
		&
		\qquad \qquad
		= \frac{d}{dv'} \mathcal{T}_p \left[ J_{(i_3)} J_{(i_2)} \left( R(\dot{\gamma},J^v_{(i_1)}) \dot{\gamma} \right)^k {\Xi_k}^j (s_{v''}) \right],
	\end{align*}
	so,
	\begin{align*}
		&
		\int_{\tilde{v}}^{v(x)} \mathcal{T}_p \left[ \frac{1}{(p^4)^2} X \left( J_{(i_3)} J_{(i_2)} \left( R(\dot{\gamma},J^v_{(i_1)}) \dot{\gamma} \right)^k {\Xi_k}^j \right) (s_{v''}) \right] dv''
		\\
		&
		\qquad
		\leq
		\left\vert \mathcal{T}_p \left[ (p^4)^{-1} J_{(i_3)} J_{(i_2)} \left( R(\dot{\gamma},J^v_{(i_1)}) \dot{\gamma} \right)^k {\Xi_k}^j (0) \right] \right\vert
		\\
		&
		\qquad \quad
		+
		\left\vert \mathcal{T}_p \left[ (p^4)^{-1} J_{(i_3)} J_{(i_2)} \left( R(\dot{\gamma},J^v_{(i_1)}) \dot{\gamma} \right)^k {\Xi_k}^j (s_{\tilde{v}}) \right] \right\vert
		\\
		&
		\qquad
		\leq
		\frac{C(1+C_1)}{r(x)^p \tilde{v}^{\frac{5}{2}}}  \left( 1 + H_{\mathcal{T}_p,2}(\tilde{v}) \right)
		+
		C \bigg( 
		\mathcal{T}_p \left[ r(s_{\tilde{v}}) \vert \mathfrak{D}^2 \beta \vert (s_{\tilde{v}}) \right]
		\\
		&
		\qquad \quad
		+ \sum_{\psi_q\neq \beta} \mathcal{T}_p \left[ r(s_{\tilde{v}})^{q-3} \vert \mathfrak{D}^2 \psi_q \vert (s_{\tilde{v}}) \right]
		+ \sum_{\mathcal{T}_q} \mathcal{T}_p \left[ r(s_{\tilde{v}})^{q-3} \vert \mathfrak{D}^2 \mathcal{T}_q \vert (s_{\tilde{v}}) \right]
		\\
		&
		\qquad \quad
		+ \sum_{\Gamma_q} \mathcal{T}_p \left[ r(s_{\tilde{v}})^{q-\frac{5}{2}} \vert \mathfrak{D}^2 \Gamma_q \vert (s_{\tilde{v}}) \right]
		+ \mathcal{T}_p \left[ r(s_{\tilde{v}})^{-\frac{5}{2}} \vert \mathfrak{D}^2 r \nablaslash b \vert (s_{\tilde{v}}) \right]
		\\
		&
		\qquad \quad
		+ \mathcal{T}_p \left[ r(s_{\tilde{v}})^{-\frac{3}{2}} \vert \mathfrak{D}^2 \Gammaslash \vert (s_{\tilde{v}}) \right]
		\bigg),
	\end{align*}
	where the terms arising from $\mathcal{T}_p \left[ (p^4)^{-1} J_{(i_3)} J_{(i_2)} \left( R(\dot{\gamma},J^v_{(i_1)}) \dot{\gamma} \right)^k {\Xi_k}^j (0) \right]$ are contained in 
	\[
		\frac{C(1+C_1)}{r(x)^p \tilde{v}^{\frac{5}{2}}}  \left( 1 + H_{\mathcal{T}_p,2}(v') \right).
	\]
	For the second term on the right hand side of \eqref{eq:JJJboundsmain},
	\begin{align*}
		&
		\int_{\tilde{v}}^{v(x)} \frac{C(1+C_1) }{r(x)^p v''^{\frac{5}{2}}}  \left( 1 + H_{\mathcal{T}_p,2}(v'') \right) dv''
		\\
		&
		\qquad \qquad
		\leq
		\frac{C(1+C_1)}{r(x)^p}  \left( 1 + H_{\mathcal{T}_p,2}(\tilde{v}) \right) \int_{\tilde{v}}^{v(x)} \frac{1}{v''^{\frac{5}{2}}} dv''
		\\
		&
		\qquad \qquad
		\leq
		\frac{C(1+C_1)}{r(x)^p \tilde{v}^{\frac{3}{2}}}  \left( 1 + H_{\mathcal{T}_p,2}(\tilde{v}) \right).
	\end{align*}
	For the final terms,
	\begin{align*}
		&
		\int_{\tilde{v}}^{v(x)} \mathcal{T}_p \left[ r(s_{v''}) \vert \mathfrak{D}^2 \beta(s_{v''}) \vert d v'' \right]
		\\
		&
		\qquad \qquad
		\leq
		\mathcal{T}_p \left[ \left( \int_{\tilde{v}}^{v(x)} \frac{1}{r(s_{v''})^4} d v'' \right)^{\frac{1}{2}} \left( \int_{\tilde{v}}^{v(x)} r(s_{v''})^6 \vert \mathfrak{D}^2 \beta (s_{v''}) \vert^2 d v'' \right)^{\frac{1}{2}} \right]
		\\
		&
		\qquad \qquad
		\leq
		\frac{C}{\tilde{v}^{\frac{3}{2}}} \mathcal{T}_p \left[ \left( \int_{\tilde{v}}^{v(x)} r(s_{v''})^6 \vert \mathfrak{D}^2 \beta (s_{v''}) \vert^2 d v'' \right)^{\frac{1}{2}} \right],
	\end{align*}
	and similarly,
	\begin{align*}
		&
		\int_{\tilde{v}}^{v(x)} \sum_{\psi_q \neq \beta} \mathcal{T}_p \left[ r(s_{v''})^{q-3} \vert \mathfrak{D}^2 \psi_q(s_{v''}) \vert \right] d v''
		\\
		&
		\qquad \qquad
		\leq
		\frac{C}{\tilde{v}^{\frac{3}{2}}} \mathcal{T}_p \left[ \left( \int_{\tilde{v}}^{v(x)} \sum_{\psi_q \neq \beta}r(s_{v''})^{2q-2} \vert \mathfrak{D}^2 \psi_q (s_{v''}) \vert^2 d v'' \right)^{\frac{1}{2}} \right],
	\end{align*}
	\begin{align*}
		&
		\int_{\tilde{v}}^{v(x)} \sum_{\mathcal{T}_q} \mathcal{T}_p \left[ r(s_{v''})^{q-3} \vert \mathfrak{D}^2 \mathcal{T}_q(s_{v''}) \right] \vert d v''
		\\
		&
		\qquad \qquad
		\leq
		\frac{C}{\tilde{v}^{\frac{3}{2}}} \mathcal{T}_p \left[ \left( \int_{\tilde{v}}^{v(x)} \sum_{\mathcal{T}_q}r(s_{v''})^{2q-2} \vert \mathfrak{D}^2 \mathcal{T}_q (s_{v''}) \vert^2 d v'' \right)^{\frac{1}{2}} \right],
	\end{align*}
	\begin{align*}
		&
		\int_{\tilde{v}}^{v(x)} \sum_{\Gamma_q} \mathcal{T}_p \left[ r(s_{v''})^{q-\frac{5}{2}} \vert \mathfrak{D}^2 \Gamma_q(s_{v''}) \vert \right] d v''
		\\
		&
		\qquad \qquad
		\leq
		\frac{C}{\tilde{v}^{\frac{5}{4}}} \mathcal{T}_p \left[ \left( \int_{\tilde{v}}^{v(x)} \sum_{\Gamma_q}r(s_{v''})^{2q-\frac{3}{2}} \vert \mathfrak{D}^2 \Gamma_q (s_{v''}) \vert^2 d v'' \right)^{\frac{1}{2}} \right],
	\end{align*}
	\begin{align*}
		&
		\int_{\tilde{v}}^{v(x)} \mathcal{T}_p \left[ r(s_{v''})^{-\frac{5}{2}} \vert \mathfrak{D}^2 r \nablaslash b (s_{v''}) \vert \right] d v''
		\\
		&
		\qquad \qquad
		\leq
		\frac{C}{\tilde{v}^{2}} \mathcal{T}_p \left[ \left( \int_{\tilde{v}}^{v(x)} \vert \mathfrak{D}^2 r \nablaslash b (s_{v''}) \vert^2 d v'' \right)^{\frac{1}{2}} \right],
	\end{align*}
	and
	\begin{align*}
		&
		\int_{\tilde{v}}^{v(x)} \mathcal{T}_p \left[ r(s_{v''})^{-\frac{3}{2}} \vert \mathfrak{D}^2 \Gammaslash (s_{v''}) \vert \right] d v''
		\\
		&
		\qquad \qquad
		\leq
		\frac{C}{\tilde{v}^{\frac{5}{4}}} \mathcal{T}_p \left[ \left( \int_{\tilde{v}}^{v(x)} r(s_{v''})^{\frac{1}{2}} \vert \mathfrak{D}^2 \Gammaslash (s_{v''}) \vert^2 d v'' \right)^{\frac{1}{2}} \right].
	\end{align*}
	Hence,
	\begin{align} \label{eq:JJJboundsmain2}
		\begin{split}
			&
			\left\vert \mathcal{T}_p \left[ (p^4)^{-1} J_{(i_3)} J_{(i_2)} \left( \frac{d{J_{(i_1)}}^k {\Xi_k}^j}{ds} \right) (s_{\tilde{v}}) \right] 
			- 
			\mathcal{T}_p \left[ (p^4)^{-1} J_{(i_3)} J_{(i_2)} \left( \frac{d{J_{(i_1)}}^k {\Xi_k}^j}{ds} \right) \bigg\vert_{s_{v'} = 0} \right] \right\vert
			\\
			&
			\qquad
			\leq
			\frac{C(1+C_1)}{r(x)^p \tilde{v}^{\frac{5}{2}}}  \left( 1 + H_{\mathcal{T}_p,2}(\tilde{v}) \right)
			+
			C \bigg( 
			\mathcal{T}_p \left[ r(s_{\tilde{v}}) \vert \mathfrak{D}^2 \beta \vert (s_{\tilde{v}}) \right]
			+ \sum_{\psi_q\neq \beta} \mathcal{T}_p \left[ r(s_{\tilde{v}})^{q-3} \vert \mathfrak{D}^2 \psi_q \vert (s_{\tilde{v}}) \right]
			\\
			&
			\qquad \quad
			+ \sum_{\mathcal{T}_q} \mathcal{T}_p \left[ r(s_{\tilde{v}})^{q-3} \vert \mathfrak{D}^2 \mathcal{T}_q \vert (s_{\tilde{v}}) \right]
			+ \sum_{\Gamma_q} r(s_{\tilde{v}})^{q-\frac{5}{2}} \mathcal{T}_p \left[ \vert \mathfrak{D}^2 \Gamma_q \vert (s_{\tilde{v}}) \right]
			+ \mathcal{T}_p \left[ r(s_{\tilde{v}})^{-\frac{5}{2}} \vert \mathfrak{D}^2 r \nablaslash b \vert (s_{\tilde{v}}) \right]
			\\
			&
			\qquad \quad
			+ \mathcal{T}_p \left[ r(s_{\tilde{v}})^{-\frac{3}{2}} \vert \mathfrak{D}^2 \Gammaslash \vert (s_{\tilde{v}}) \right]
		\bigg)
		+ \frac{C}{\tilde{v}^{\frac{5}{4}}}
		\Bigg(
		\mathcal{T}_p \left[ \left( \int_{\tilde{v}}^{v(x)} r(s_{v''})^6 \vert \mathfrak{D}^2 \beta (s_{v''}) \vert^2 d v'' \right)^{\frac{1}{2}} \right]
		\\
		&
		\qquad \quad
		+ \mathcal{T}_p \left[ \left( \int_{\tilde{v}}^{v(x)} \sum_{\psi_q \neq \beta}r(s_{v''})^{2q-2} \vert \mathfrak{D}^2 \psi_q (s_{v''}) \vert^2 d v'' \right)^{\frac{1}{2}} \right]
		\\
		&
		\qquad \quad
		+
		\mathcal{T}_p \left[ \left( \int_{\tilde{v}}^{v(x)} \sum_{\mathcal{T}_q}r(s_{v''})^{2q-2} \vert \mathfrak{D}^2 \mathcal{T}_q (s_{v''}) \vert^2 d v'' \right)^{\frac{1}{2}} \right]
		+
		\mathcal{T}_p \left[ \left( \int_{\tilde{v}}^{v(x)} \vert \mathfrak{D}^2 r \nablaslash b (s_{v''}) \vert^2 d v'' \right)^{\frac{1}{2}} \right]
		\\
		&
		\qquad \quad
		+
		\mathcal{T}_p \left[ \left( \int_{\tilde{v}}^{v(x)} \sum_{\Gamma_q}r(s_{v''})^{2q-\frac{3}{2}} \vert \mathfrak{D}^2 \Gamma_q (s_{v''}) \vert^2 d v'' \right)^{\frac{1}{2}} \right]
		+
		\mathcal{T}_p \left[ \left( \int_{\tilde{v}}^{v(x)} r(s_{v''})^{\frac{1}{2}} \vert \mathfrak{D}^2 \Gammaslash (s_{v''}) \vert^2 d v'' \right)^{\frac{1}{2}} \right]
		\Bigg).
		\end{split}
	\end{align}
	Now,
	\[
		\frac{d}{dv'} \mathcal{T}_p \left[ J_{(i_3)} J_{(i_2)} \left( {J_{(i_1)}}^k {\Xi_k}^j \right) (s_{\tilde{v}}) \right] 
		=
		\mathcal{T}_p \left[ \frac{1}{p^4} J_{(i_3)} J_{(i_2)} \frac{d{J_{(i_1)}}^k {\Xi_k}^j}{ds} (s_{\tilde{v}}) \right], 
	\]
	and hence,
	\begin{align*}
		&
		\Bigg\vert \mathcal{T}_p \left[ J_{(i_3)} J_{(i_2)} \left( {J_{(i_1)}}^k {\Xi_k}^j \right) (s_{\tilde{v}})\right]
		- 
		\mathcal{T}_p \left[ J_{(i_3)} J_{(i_2)} \left( {J_{(i_1)}}^j \right)\Big\vert_{s_{v'} = 0} \right] 
		\\
		&
		+ \left( v(x) - \tilde{v} \right) \mathcal{T}_p \left[ \frac{1}{p^4} \frac{ d J_{(i_3)} J_{(i_2)} \left( {J_{(i_1)}}^k {\Xi_k}^j \right)}{ds} \bigg\vert_{s_{v'} = 0} \right]
		\Bigg\vert
		\\
		&
		\qquad \qquad
		\leq
		\int_{v'}^{v(x)} \Bigg\vert
		\mathcal{T}_p \left[ \frac{1}{p^4} J_{(i_3)} J_{(i_2)} \left( \frac{d{J_{(i_1)}}^k {\Xi_k}^j}{ds} \right) (s_{\tilde{v}}) \right]
		\\
		&
		\qquad \qquad \quad
		- \mathcal{T}_p \left[ \frac{1}{p^4} J_{(i_3)} J_{(i_2)} \left( \frac{d{J_{(i_1)}}^k {\Xi_k}^j}{ds} \right) \bigg\vert_{s_{v'} = 0} \right]
		\Bigg\vert d \tilde{v}.
	\end{align*}
	Integrating each term on the right hand side of \eqref{eq:JJJboundsmain2} then gives,
	\begin{align*}
		&
		\Bigg\vert \mathcal{T}_p \left[ J_{(i_3)} J_{(i_2)} \left( {J_{(i_1)}}^k {\Xi_k}^j \right) (s_{\tilde{v}})\right]
		- 
		\mathcal{T}_p \left[ J_{(i_3)} J_{(i_2)} \left( {J_{(i_1)}}^j \right)\Big\vert_{s_{v'} = 0} \right] 
		\\
		&
		+ \left( v(x) - \tilde{v} \right) \mathcal{T}_p \left[ \frac{1}{p^4} \frac{ d J_{(i_3)} J_{(i_2)} \left( {J_{(i_1)}}^k {\Xi_k}^j \right)}{ds} \bigg\vert_{s_{v'} = 0} \right]
		\Bigg\vert
		\\
		&
		\qquad \qquad
		\leq
		\frac{C(1+C_1)}{r(x)^p v'^{\frac{1}{4}}} \left( 1 + H_{\mathcal{T}_p,2}(v') \right).
	\end{align*}
	Taking $v_0$ large so that $\frac{C(1+C_1)}{v_0} \leq \frac{C_1}{2}$, this then recovers the first bootstrap assumption with a better constant.  The other bootstrap assumptions can be recovered similarly.  Hence the set of $v'\in[v_0,v(x)]$ where they hold is non-empty, open and closed and hence equal to $[v_0,v(x)]$.
\end{proof}

Finally, at the very top order, we have the following.
\begin{proposition} \label{prop:JJJJbounds}
	If $v_0$ is sufficiently large, $v'\in[v_0,v(x)]$, then, for each $\mathcal{T}_p$,
	\[
		\left\vert \mathcal{T}_p \left[ J_{(4)} J_{(4)} J_{(4)} \left( {J_{(A)}}^B \right) (s_{v'}) - r(s_{v'}) {\delta_A}^B \right] \right\vert
		\leq
		C \left( \frac{1}{r(x)^p} + H_{\mathcal{T}_p,3}(v') \right),
	\]
	for $A,B = 1,2$,
	\[
		\left\vert \mathcal{T}_p \left[ J_{(4)} J_{(4)} J_{(4)} \left( {J_{(4)}}^4 \right) (s_{v'}) - r(s_{v'}) \right] \right\vert
		\leq
		C \left( \frac{1}{r(x)^p} + H_{\mathcal{T}_p,3}(v') \right),
	\]
	and
	\[
		\left\vert \mathcal{T}_p \left[ J_{(i_4)} J_{(i_3)} J_{(i_2)} \left( {J_{(i_1)}}^j \right) (s_{v'}) \right] \right\vert
		\leq
		C \left( \frac{1}{r(x)^p} + H_{\mathcal{T}_p,3}(v') \right),
	\]
	for $i_1,i_2,i_3,i_4 = 1,\ldots,6$, $j = 1,\ldots,4$ otherwise,
	\[
		\left\vert \mathcal{T}_p \left[ (p^4)^{-1} J_{(4)} J_{(4)} J_{(4)} \left( {J_{(A)}}^{4+B} \right) (s_{v'}) - {\delta_A}^B \right] \right\vert
		\leq
		\frac{C}{v'} \left( \frac{1}{r(x)^p} + H_{\mathcal{T}_p,3}(v') \right),
	\]
	for $A,B = 1,2$,
	\[
		\left\vert \mathcal{T}_p \left[ (p^4)^{-1} J_{(i_4)}J_{(i_3)} J_{(i_2)} \left( {J_{(i_1)}}^{4+A} \right) (s_{v'}) \right] \right\vert
		\leq
		\frac{C}{v'} \left( \frac{1}{r(x)^p} + H_{\mathcal{T}_p,3}(v') \right),
	\]
	for $i_1,i_2,i_3,i_4 = 1,\ldots,6$, $A = 1,2$ otherwise, and,
	\[
		\left\vert \mathcal{T}_p \left[ (p^4)^{-1} J_{(i_4)} J_{(i_3)} J_{(i_2)} \left( {J_{(i_1)}}^{7} \right) (s_{v'}) \right] \right\vert
		\leq
		C \left( \frac{1}{r(x)^p} + H_{\mathcal{T}_p,3}(v') \right),
	\]
	where,
	\begin{align*}
		H_{\mathcal{T}_p,3}(v')
		=
		&
		\sum_{k = 2,3} \Bigg(
		\mathcal{T}_p \left[ \left( \int_{v'}^{v(x)} r(s_{\tilde{v}})^6 \vert \mathfrak{D}^k \beta \vert^2 (s_{\tilde{v}}) d\tilde{v} \right)^{\frac{1}{2}} \right]
		\\
		&
		+
		\mathcal{T}_p \left[ \left( \int_{v'}^{v(x)} \sum_{\psi_q \neq \beta} r(s_{\tilde{v}})^{2q-2} \vert \mathfrak{D}^k \psi_q \vert^2 (s_{\tilde{v}}) d\tilde{v} \right)^{\frac{1}{2}} \right]
		\\
		&
		+
		\mathcal{T}_p \left[ \left( \int_{v'}^{v(x)} \sum_{\mathcal{T}_q} r(s_{\tilde{v}})^{2q-2} \vert \mathfrak{D}^k \mathcal{T}_q \vert^2 (s_{\tilde{v}}) d\tilde{v} \right)^{\frac{1}{2}} \right]
		\\
		&
		+
		\mathcal{T}_p \left[ \left( \int_{v'}^{v(x)} \sum_{\Gamma_q} r(s_{\tilde{v}})^{2q-\frac{3}{2}} \vert \mathfrak{D}^k \Gamma_q \vert^2 (s_{\tilde{v}}) d\tilde{v} \right)^{\frac{1}{2}} \right]
		\\
		&
		+
		\mathcal{T}_p \left[ \left( \int_{v'}^{v(x)} r(s_{\tilde{v}})^{\frac{1}{2}} \vert \mathfrak{D}^k \Gammaslash \vert^2 (s_{\tilde{v}}) d\tilde{v} \right)^{\frac{1}{2}} \right]
		+ \mathcal{T}_p \left[ \left( \int_{v'}^{v(x)} \vert \mathfrak{D}^k r \nablaslash b \vert^2 (s_{\tilde{v}}) d\tilde{v} \right)^{\frac{1}{2}} \right]
		\\
		&
		+
		\frac{1}{r(x)^p} \bigg(
		\sum_{\psi_q} r(x)^q \vert \mathfrak{D}^k \psi_q \vert (x)
		+
		\sum_{\mathcal{T}_q} r(x)^q \vert \mathfrak{D}^k \mathcal{T}_q \vert (x)
		\\
		&
		+
		\sum_{\Gamma_q} r(x)^q \vert \mathfrak{D}^k \Gamma_q \vert (x)
		+
		\vert \mathfrak{D}^k r \nablaslash b \vert (x)
		+
		r(x) \vert \mathfrak{D}^k \Gammaslash \vert (x)
		\bigg)
		\Bigg).
	\end{align*}
\end{proposition}

\begin{proof}
	The proof is identical to that of Proposition \ref{prop:JJJbounds}, using appropriate versions of Propositions \ref{prop:JJmatrixA} and \ref{prop:JJic}.
\end{proof}

\subsection{Proof of Proposition \ref{prop:emmain2}} \label{subsec:proofemmain2}
The proof of Proposition \ref{prop:emmain2} follows from Propositions \ref{prop:Jbounds}, \ref{prop:JJbounds}, \ref{prop:JJJbounds}, \ref{prop:JJJJbounds}.

\begin{proof}[Proof of Proposition \ref{prop:emmain2}]
	Recall the frame $\tilde{E}_1,\ldots,\tilde{E}_7$ from Section \ref{section:ba} defined by,
	\[
		\tilde{E}_i = E_i \text{ for } i = 1,2,3,4, \qquad \tilde{E}_i = p^4 E_i \text{ for } i=5,6,7.
	\]
	Recall also from Section \ref{subsec:fjacobi} that,
	\[
		V_{(i_1)} f \vert_{(x,p)} = J_{(i_1)} f \vert_{\exp_{-s_*}(x,p)} = {J_{(i_1)}}^j E_j f \vert_{\exp_{-s_*}(x,p)},
	\]
	for $i_1 = 1,\ldots,6$.  By assumption,
	\[
		\sum_{j=1}^7 \sup_{P\vert_{\{v=v_0\}}} \vert \tilde{E}_j f \vert < \bc,
	\]
	and so Proposition \ref{prop:Jbounds}, which gives,
	\[
		\vert {J_{(i_1)}}^j  (-s_*) \vert \leq C \text{ for } j = 1,2,3,4, \qquad \vert {J_{(i_1)}}^j  (-s_*) \vert \leq Cp^4 \text{ for } j = 5,6,7,
	\]
	implies that,
	\[
		\big\vert V_{(i_1)} f \vert_{(x,p)} \big\vert \leq C\bc.
	\]
	Hence,
	\[
		\vert \mathcal{T}_p [ V_{(i_1)} f ] \vert 
		\leq 
		C \bc \vert \mathcal{T}_p [ \mathbbm{1}_{\supp(f\vert_{P_x})} ] \vert
		\leq 
		\frac{C \bc}{r^p}.
	\]
	Similarly,
	\begin{align*}
		V_{(i_2)} V_{(i_1)} f \vert_{(x,p)} 
		&
		= 
		J_{(i_2)} J_{(i_1)} f \vert_{\exp_{-s_*}(x,p)} 
		\\
		&
		= 
		J_{(i_2)} \left( {J_{(i_1)}}^{j_1} \right) E_{j_1} f \vert_{\exp_{-s_*}(x,p)} + {J_{(i_2)}}^{j_2} {J_{(i_1)}}^{j_1} E_{j_2} E_{j_1} f \vert_{\exp_{-s_*}(x,p)}.
	\end{align*}
	Again, by assumption,
	\[
		\sum_{j_1 j_2 =1}^7 \sup_{P\vert_{\{v=v_0\}}} \vert \tilde{E}_{j_2} \tilde{E}_{j_1} f \vert < \bc,
	\]
	hence Proposition \ref{prop:Jbounds} and Proposition \ref{prop:JJbounds}, which in particular gives,
	\begin{align*}
		\left\vert J_{(i_2)} \left( {J_{(i_1)}}^{j} \right) (-s_*) \right\vert 
		&
		\leq C \text{ for } j=1,2,3,4,
		\\
		\left\vert J_{(i_2)} \left( {J_{(i_1)}}^{j} \right) (-s_*) \right\vert 
		&
		\leq C p^4 \text{ for } j=5,6,7,
	\end{align*}
	imply that,
	\[
		\left\vert V_{(i_2)} V_{(i_1)} f \vert_{(x,p)} \right\vert \leq C \bc,
	\]
	and hence,
	\[
		\vert \mathcal{T}_p [ V_{(i_2)} V_{(i_1)} f ] \vert 
		\leq 
		C \bc \vert \mathcal{T}_p [ \mathbbm{1}_{\supp(f\vert_{P_x})} ] \vert
		\leq 
		\frac{C \bc}{r^p}.
	\]
	
	For the third order derivatives recall that,
	\begin{align} \label{eq:VVVf}
	\begin{split}
		V_{(i_3)} V_{(i_2)} V_{(i_1)} f =
		&
		J_{(i_3)} J_{(i_2)} \left( {J_{(i_1)}}^{j_1} \right) E_{j_1} f\vert_{\exp_{-s_*}(x,p)}
		\\
		&
		+
		\bigg[ J_{(i_2)} \left( {J_{(i_1)}}^{j_1} \right) {J_{(i_3)}}^{j_2} + J_{(i_3)} \left( {J_{(i_1)}}^{j_1} \right) {J_{(i_2)}}^{j_2}
		\\
		&
		\qquad + J_{(i_3)} \left( {J_{(i_2)}}^{j_2} \right) {J_{(i_1)}}^{j_1} \bigg] E_{j_2} E_{j_1} f \vert_{\exp_{-s_*}(x,p)}
		\\
		&
		+
		{J_{(i_3)}}^{j_3} {J_{(i_2)}}^{j_2} {J_{(i_1)}}^{j_1} E_{j_3} E_{j_2} E_{j_1} f \vert_{\exp_{-s_*}(x,p)}.
	\end{split}
	\end{align}
	The second terms can be estimated pointwise as before, as can the final terms using the assumption,
	\[
		\sum_{j_1,j_2,j_3 = 1}^7 \sup_{P\vert_{\{v=v_0\}}} \vert \tilde{E}_{j_3} \tilde{E}_{j_2} \tilde{E}_{j_1} f \vert < \bc.
	\]
	Consider the estimate for $\mathcal{T}_p \left[ V_{(i_3)} V_{(i_2)} V_{(i_1)} f \right]$ on the incoming, $v =$ constant, hypersurface.  The above pointwise bounds clearly give,
	\begin{align*}
		&
		\int_{u_0}^u \int_{S_{u',v}} r^{2p-2} \bigg\vert \mathcal{T}_p \bigg[
		\bigg( J_{(i_2)} \left( {J_{(i_1)}}^{j_1} \right) {J_{(i_3)}}^{j_2} + J_{(i_3)} \left( {J_{(i_1)}}^{j_1} \right) {J_{(i_2)}}^{j_2}
		\\
		&
		\qquad \qquad
		+ J_{(i_3)} \left( {J_{(i_2)}}^{j_2} \right) {J_{(i_1)}}^{j_1} \bigg) E_{j_2} E_{j_1} f \vert_{\exp_{-s_*}(x,p)}
		\\
		&
		\qquad \qquad
		+
		{J_{(i_3)}}^{j_3} {J_{(i_2)}}^{j_2} {J_{(i_1)}}^{j_1} E_{j_3} E_{j_2} E_{j_1} f \vert_{\exp_{-s_*}(x,p)}
		\bigg] \bigg\vert^2 d \mu_{S_{u',v}} du'
		\leq C \bc^2,
	\end{align*}
	so it remains to estimate the first term in \eqref{eq:VVVf}.  Proposition \ref{prop:JJJbounds} gives,
	\[
		\left\vert \mathcal{T}_p \bigg[ J_{(i_3)} J_{(i_2)} \left( {J_{(i_1)}}^{j} \right) \vert_{\exp_{-s_*}(x,p)} \bigg] \right\vert
		\leq
		C \left( \frac{1}{r^p} + H_{\mathcal{T}_p,2}(v_0) \right),
	\]
	for $j = 1,2,3,4$, and,
	\[
		\left\vert \mathcal{T}_p \bigg[ \frac{1}{p^4} J_{(i_3)} J_{(i_2)} \left( {J_{(i_1)}}^{j} \right) \vert_{\exp_{-s_*}(x,p)} \bigg] \right\vert
		\leq
		C \left( \frac{1}{r^p} + H_{\mathcal{T}_p,2}(v_0) \right),
	\]
	for $j = 5,6,7$, where $H_{\mathcal{T}_p,2}$ is defined in Proposition \ref{prop:JJJbounds}.  Consider first the final terms in $H_{\mathcal{T}_p,2}$.  Clearly,
	\begin{align*}
		\int_{u_0}^u \int_{S_{u',v}} r^{2p-2} \frac{1}{r^{2p}} \bigg(
		\sum_{\psi_q \neq \alpha} r^{2p} \vert \mathfrak{D}^2 \psi_q \vert^2
		+
		\sum_{\mathcal{T}_q} r^{2p} \vert \mathfrak{D}^2 \mathcal{T}_q \vert^2
		+
		\sum_{\Gamma_q} r^{2p} \vert \mathfrak{D}^2 \Gamma_q \vert^2
		&
		\\
		+
		\vert \mathfrak{D}^2 r\nablaslash b \vert^2
		+
		r^2 \vert \mathfrak{D}^2 \Gammaslash \vert^2
		\bigg) d\mu_{S_{u',v}} dv'
		&
		\leq
		C,
	\end{align*}
	by the bootstrap assumptions of Section \ref{section:ba}.  For $\psi_q = \alpha$, $2q-2 = 6$ and,
	\begin{align*}
		\partial_v \int_{u_0}^u \int_{S_{u',v}} r^6 \vert \mathfrak{D}^2 \alpha \vert^2 d \mu_{S_{u',v}} d u'
		= \int_{u_0}^u \int_{S_{u',v}} r^5 \vert \mathfrak{D}^2 \alpha \vert^2 + 2 r^6 \mathfrak{D}^2 \alpha \cdot \nablaslash_4 \mathfrak{D}^2 \alpha + \tr \chi r^6 \vert \mathfrak{D}^2 \alpha \vert^2 d \mu_{S_{u',v}} d u',
	\end{align*}
	and hence,
	\begin{align*}
		&
		\int_{u_0}^u \int_{S_{u',v}} r^6 \vert \mathfrak{D}^2 \alpha \vert^2 d \mu_{S_{u',v}} d u'
		-
		\int_{u_0}^u \int_{S_{u',v_0}} r^6 \vert \mathfrak{D}^2 \alpha \vert^2 d \mu_{S_{u',v_0}} d u'
		\\
		&
		\qquad \qquad
		\leq
		C \int_{v_0}^v \int_{u_0}^u \int_{S_{u',v'}} r^5 \vert \mathfrak{D}^2 \alpha \vert^2 + r^5 \vert \mathfrak{D}^2 \alpha \vert \vert \mathfrak{D}^3 \alpha \vert  d \mu_{S_{u',v'}} d u'
		\\
		&
		\qquad \qquad
		\leq
		C \int_{v_0}^v \int_{u_0}^u \int_{S_{u',v'}} r^5 \vert \mathfrak{D}^2 \alpha \vert^2 + r^5 \vert \mathfrak{D}^3 \alpha \vert^2  d \mu_{S_{u',v'}} d u'
		\\
		&
		\qquad \qquad
		\leq
		C \int_{u_0}^u F^1_{v_0,v}(u') du'
		\\
		&
		\qquad \qquad
		\leq
		C,
	\end{align*}
	by the bootstrap assumptions for the weighted $L^2$ integral of $\mathfrak{D}^2 \alpha, \mathfrak{D}^3 \alpha$ on the outgoing null hypersurfaces.  This, together with the assumption on the initial data gives,
	\[
		\int_{u_0}^u \int_{S_{u',v}} r^6 \vert \mathfrak{D}^2 \alpha \vert^2 d \mu_{S_{u',v}} d u' \leq C.
	\]
	Consider now the first term involving $\beta$ in $H_{\mathcal{T}_p,2}(v_0)$.  By the Cauchy--Schwarz inequality,
	\begin{align*}
		&
		\left\vert \mathcal{T}_p \left[ \mathbbm{1}_{\supp(f)} \left( \int_{v_0}^v r(s_{\tilde{v}})^6 \vert \mathfrak{D}^2 \beta (s_{\tilde{v}}) \vert^2 d \tilde{v} \right)^{\frac{1}{2}} \right] \right\vert^2
		\\
		&
		\qquad
		\leq
		\left\vert \mathcal{T}_p \left[ \mathbbm{1}_{\supp(f)} \right] \right\vert 
		\left\vert \mathcal{T}_p \left[ \int_{v_0}^v r(s_{\tilde{v}})^6 \vert \mathfrak{D}^2 \beta (s_{\tilde{v}}) \vert^2 d \tilde{v} \right] \right\vert
		\\
		&
		\qquad
		\leq
		\frac{C}{r^p}
		\left\vert \mathcal{T}_p \left[ \int_{v_0}^v r(s_{\tilde{v}})^6 \vert \mathfrak{D}^2 \beta (s_{\tilde{v}}) \vert^2 d \tilde{v} \right] \right\vert.
	\end{align*}
	Hence, using the fact that $c \leq \frac{\sqrt{\det \gslash}}{r^2} \leq C$ and that $c \leq \frac{v}{r} \leq C$,
	\begin{align*}
		&
		\int_{u_0}^u r^{2p-2} \int_{S_{u',v}} \left\vert \mathcal{T}_p \left[ \mathbbm{1}_{\supp(f)} \left( \int_{v_0}^v r(s_{\tilde{v}})^6 \vert \mathfrak{D}^2 \beta (s_{\tilde{v}}) \vert^2 d \tilde{v} \right)^{\frac{1}{2}} \right] \right\vert^2 d \mu_{S_{u',v}} du'
		\\
		&
		\quad
		\leq
		C \int_{u_0}^u r^{p-2} \int_{S_{u',v}} r^{4-p} \int_0^C \int_{\vert p^1 \vert, \vert p^2 \vert \leq \frac{C}{r^2}} \int_{v_0}^v r(s_{\tilde{v}})^6 \vert \mathfrak{D}^2 \beta (s_{\tilde{v}}) \vert^2 d \tilde{v} dp^1 dp^2 dp^4 d \mu_{S_{u',v}} du'
		\\
		&
		\quad
		\leq
		C \sum_{U_1,U_2} \int_{u_0}^u r^{4} \int_0^C \int_{\vert p^1 \vert, \vert p^2 \vert \leq \frac{C}{r^2}} \int_{v_0}^v \int_{\theta_1,\theta_2} r(s_{\tilde{v}})^6 \vert \mathfrak{D}^2 \beta (s_{\tilde{v}}) \vert^2 d \theta^1 d \theta^2 d \tilde{v} dp^1 dp^2 dp^4 du'
		\\
		&
		\quad
		\leq
		C \sum_{U_1,U_2} \int_0^C \int_{\vert p^1 \vert, \vert p^2 \vert \leq \frac{C}{v^2}} v^4 \int_{v_0}^v \int_{u_0}^u \int_{\theta_1,\theta_2} r(s_{\tilde{v}})^4 \vert \mathfrak{D}^2 \beta (s_{\tilde{v}}) \vert^2 \sqrt{\det \gslash (s_{\tilde{v}})} d \theta^1 d \theta^2 du' d \tilde{v} dp^1 dp^2 dp^4,
	\end{align*}
	where $U_1,U_2$ are the two spherical charts.  We now perform the change of coordinates,
	\[
		(u',\tilde{v},\theta^1,\theta^2,p^1,p^2,p^4) \mapsto (\hat{u},\tilde{v},\hat{\theta}^1,\hat{\theta}^2,p^1,p^2,p^4),
	\]
	where 
	\[
		\hat{u} := u(\exp_{s_{\tilde{v}}}(x,p)), \qquad \hat{\theta}^A := \theta^A(\exp_{s_{\tilde{v}}}(x,p)) \text{ for } A=1,2,
	\]
	with $(x,p) = (u',v,\theta^1,\theta^2,p^1,p^2,p^4)$.  The determinant of the Jacobian of this transformation is equal to the determinant of,
	\begin{equation*}
  		\begin{pmatrix}
    			\frac{\partial \hat{\theta}^1}{\partial \theta^1} & \frac{\partial \hat{\theta}^1}{\partial \theta^2} & \frac{\partial \hat{\theta}^1}{\partial u}
    			\\
    			\frac{\partial \hat{\theta}^2}{\partial \theta^1} & \frac{\partial \hat{\theta}^2}{\partial \theta^2} & \frac{\partial \hat{\theta}^2}{\partial u}
    			\\
    			\frac{\partial \hat{u}}{\partial \theta^1} & \frac{\partial \hat{u}}{\partial \theta^2} & \frac{\partial \hat{u}}{\partial u}
  		\end{pmatrix}.
	\end{equation*}
	Note that,
	\begin{align*}
		\partial_{\theta^A} = 
		&
		\ \hor_{(x,p)}(e_A) + \frac{p^4}{r} \partial_{\overline{p}^A} + \left[ \frac{p^4}{2} \left( \tr \chi - \frac{2}{r} \right) - \frac{p^3}{r} + \frac{p^3}{2} \left( \tr \chibar + \frac{2}{r} \right) \right] \partial_{\overline{p}^A}
		\\
		&
		+ \left[ p^4 {\hat{\chi}_A}^B + p^3 {\hat{\chibar}_A}^B + p^C \Gammaslash_{AC}^B \right] \partial_{\overline{p}^B}
		\\
		&
		+ \left[ p^B \hat{\chibar}_{AB} + \frac{p^B}{4} \left( \tr \chibar + \frac{2}{r} \right) \gslash_{AB} - \frac{p^B}{2r} \gslash_{AB} + p^4 \etabar_A \right] \partial_{\overline{p}^4}
		\\
		=
		&
		\ V_{(A)} + r^2 \left[ \frac{1}{2} \left( \tr \chi - \frac{2}{r} \right) - \frac{p^3}{rp^4} + \frac{p^3}{2p^4} \left( \tr \chibar + \frac{2}{r} \right) \right] V_{(4+A)}
		\\
		&
		+ r^2 \left[ {\hat{\chi}_A}^B + \frac{p^3}{p^4} {\hat{\chibar}_A}^B + \frac{p^C}{p^4} \Gammaslash_{AC}^B \right] V_{(4+B)}
		\\
		&
		+ \left[ \frac{p^B}{p^4} \hat{\chibar}_{AB} + \frac{p^B}{4p^4} \left( \tr \chibar + \frac{2}{r} \right) \gslash_{AB} - \frac{p^B}{2rp^4} \gslash_{AB} + \etabar_A \right] p^4 \partial_{\overline{p}^4},
	\end{align*}
	and
	\begin{align*}
		\partial_u 
		= 
		&
		\ \Omega^2 \left( \hor_{(x,p)} (e_3) + (p^B {\hat{\chibar}_B}^A + \tr \chibar p^A + 2p^4 \eta^A) \partial_{\overline{p}^A} \right)
		\\
		=
		&
		\ \Omega^2 \left( V_{(3)} + r^2 \left( \frac{p^B}{p^4} {\hat{\chibar}_B}^A + \frac{p^A}{p^4} \tr \chibar + 2 \eta^A \right) V_{(4+A)} \right).
	\end{align*}
	Hence,
	\begin{align*}
		\frac{\partial \hat{\theta}^B}{\partial \theta^A}
		=
		&
		\ d\hat{\theta}^B\vert_{(x,p)} \partial_{\theta^A}
		\\
		=
		&
		\ d \theta^B \vert_{\exp_{s_{\tilde{v}}}(x,p)} \cdot d \exp_{s_{\tilde{v}}} \vert_{(x,p)} \partial_{\theta_A}
		\\
		=
		&
		\frac{{J_{(A)}}^B(s_{\tilde{v}})}{r(s_{\tilde{v}})} + r(x)^2 \left[ \frac{1}{2} \left( \tr \chi - \frac{2}{r} \right) - \frac{p^3}{rp^4} + \frac{p^3}{2p^4} \left( \tr \chibar + \frac{2}{r} \right) \right] \frac{{J_{(4+A)}}^B(s_{\tilde{v}})}{r(s_{\tilde{v}})}
		\\
		&
		+ r(x)^2 \left[ {\hat{\chi}_A}^D + \frac{p^3}{p^4} {\hat{\chibar}_A}^D + \frac{p^C}{p^4} \Gammaslash_{AC}^D \right] \frac{{J_{(4+D)}}^B(s_{\tilde{v}})}{r(s_{\tilde{v}})}
		\\
		&
		+ \left[ \frac{p^C}{p^4} \hat{\chibar}_{AC} + \frac{p^C}{4p^4} \left( \tr \chibar + \frac{2}{r} \right) \gslash_{AC} - \frac{p^C}{2rp^4} \gslash_{AC} + \etabar_A \right] \frac{1}{r(s_{\tilde{v}})} \left( d \exp_{s_{\tilde{v}}} p^4 \partial_{\overline{p}^4}\right)^B,
	\end{align*}
	Proposition \ref{prop:Jbounds} and the bootstrap assumptions for the Ricci coefficients therefore imply that\footnote{The proof of Proposition \ref{prop:Jbounds} can easily be adapted to show that 
	\[
		\left\vert \left( d \exp_{s_{\tilde{v}}} p^4 \partial_{\overline{p}^4}\right)^j - r(x) {\delta_4}^j \right\vert \leq C,
	\]
	for $v_0 \leq \tilde{v} \leq v$.}
	\[
		\left\vert \frac{\partial \hat{\theta}^B}{\partial \theta^A}(s_{\tilde{v}}) - {\delta_A}^B \right\vert \leq \frac{C}{r(s_{\tilde{v}})}.
	\]
	Similarly,
	\begin{align*}
		\left\vert \frac{\partial \hat{u}}{\partial u}(s_{\tilde{v}}) - 1 \right\vert \leq \frac{C}{r(s_{\tilde{v}})},
	\end{align*}
	and
	\[
		\left\vert \frac{\partial \hat{\theta}^1}{\partial u}(s_{\tilde{v}}) \right\vert, \left\vert \frac{\partial \hat{\theta}^2}{\partial u}(s_{\tilde{v}}) \right\vert \leq \frac{C}{r(s_{\tilde{v}})},
		\qquad
		\left\vert \frac{\partial \hat{u}}{\partial \theta^1}(s_{\tilde{v}}) \right\vert, \left\vert \frac{\partial \hat{u}}{\partial \theta^2}(s_{\tilde{v}}) \right\vert \leq C.
	\]
	Hence, if $v_0$ is taken suitably large,
	\[
		c \leq \det
		\begin{pmatrix}
    			\frac{\partial \hat{\theta}^1}{\partial \theta^1} & \frac{\partial \hat{\theta}^1}{\partial \theta^2} & \frac{\partial \hat{\theta}^1}{\partial u}
    			\\
    			\frac{\partial \hat{\theta}^2}{\partial \theta^1} & \frac{\partial \hat{\theta}^2}{\partial \theta^2} & \frac{\partial \hat{\theta}^2}{\partial u}
    			\\
    			\frac{\partial \hat{u}}{\partial \theta^1} & \frac{\partial \hat{u}}{\partial \theta^2} & \frac{\partial \hat{u}}{\partial u}
  		\end{pmatrix}
		\leq C,
	\]
	for some constants $C,c>0$ independent of $(x,p)$.  The determinant of the Jacobian of the transformation is therefore controlled from above and below independent of $r$, hence,
	\begin{align*}
		&
		\sum_{U_1,U_2} \int_{v_0}^v \int_{u_0}^u \int_{\theta^1,\theta^2} r(u(\exp_{s_{\tilde{v}}}(x,p)),\tilde{v})^4 \vert \mathfrak{D}^2 \beta (s_{\tilde{v}}) \vert^2 \sqrt{\det \gslash (s_{\tilde{v}})} d \theta^1 d \theta^2 du' d \tilde{v}
		\\
		&
		\qquad
		\leq
		C \int_{v_0}^v \int_{\hat{u}(u_0)}^{\hat{u}(u)} \int_{S_{\hat{u},\tilde{v}}} r(\hat{u},\tilde{v})^4 \vert \mathfrak{D}^2 \beta \vert^2 d \mu_{S_{\hat{u},\tilde{v}}} d\hat{u} d \tilde{v}
		\\
		&
		\qquad
		\leq
		C \int_{\hat{u}(u_0)}^{\hat{u}(u)}F^1_{v_0,v}(\hat{u}) d\hat{u}
		\\
		&
		\qquad
		\leq
		C,
	\end{align*}
	and
	\[
		\int_{u_0}^u r^{2p-2} \int_{S_{u',v}} \left\vert \mathcal{T}_p \left[ \left( \int_{v_0}^v r(s_{\tilde{v}})^6 \vert \mathfrak{D}^2 \beta (s_{\tilde{v}}) \vert^2 d \tilde{v} \right)^{\frac{1}{2}} \right] \right\vert^2 d \mu_{S_{u',v}} du'
		\leq
		C \int_0^C \int_{\vert p^1 \vert, \vert p^2 \vert \leq \frac{C}{v^2}} v^4 dp^1 dp^2 dp^4
		\leq C.
	\]
	Similarly, for the remaining terms in $H_{\mathcal{T}_p,2}(v_0)$,
	\begin{align*}
		\int_{u_0}^u r^{2p-2} \int_{S_{u',v}} \left\vert \mathcal{T}_p \left[ \left( \int_{v_0}^v \sum_{\psi_q \neq \beta} r(s_{\tilde{v}})^{2q-2} \vert \mathfrak{D}^2 \psi_q (s_{\tilde{v}}) \vert^2 d \tilde{v} \right)^{\frac{1}{2}} \right] \right\vert^2 d \mu_{S_{u',v}} du'
		&
		\leq C,
		\\
		\int_{u_0}^u r^{2p-2} \int_{S_{u',v}} \left\vert \mathcal{T}_p \left[ \left( \int_{v_0}^v \sum_{\mathcal{T}_q} r(s_{\tilde{v}})^{2q-2} \vert \mathfrak{D}^2 \mathcal{T}_q (s_{\tilde{v}}) \vert^2 d \tilde{v} \right)^{\frac{1}{2}} \right] \right\vert^2 d \mu_{S_{u',v}} du'
		&
		\leq C,
		\\
		\int_{u_0}^u r^{2p-2} \int_{S_{u',v}} \left\vert \mathcal{T}_p \left[ \left( \int_{v_0}^v \sum_{\Gamma_q} r(s_{\tilde{v}})^{2q-\frac{3}{2}} \vert \mathfrak{D}^2 \Gamma_q (s_{\tilde{v}}) \vert^2 d \tilde{v} \right)^{\frac{1}{2}} \right] \right\vert^2 d \mu_{S_{u',v}} du'
		&
		\leq C,
		\\
		\int_{u_0}^u r^{2p-2} \int_{S_{u',v}} \left\vert \mathcal{T}_p \left[ \left( \int_{v_0}^v r(s_{\tilde{v}})^{\frac{1}{2}} \vert \mathfrak{D}^2 \Gammaslash (s_{\tilde{v}}) \vert^2 d \tilde{v} \right)^{\frac{1}{2}} \right] \right\vert^2 d \mu_{S_{u',v}} du'
		&
		\leq C,
		\\
		\int_{u_0}^u r^{2p-2} \int_{S_{u',v}} \left\vert \mathcal{T}_p \left[ \left( \int_{v_0}^v \vert \mathfrak{D}^2 r \nablaslash b (s_{\tilde{v}}) \vert^2 d \tilde{v} \right)^{\frac{1}{2}} \right] \right\vert^2 d \mu_{S_{u',v}} du'
		&
		\leq C.
	\end{align*}
	Hence,
	\[
		\int_{u_0}^u r^{2p-2} \int_{S_{u',v}} \left\vert \mathcal{T}_p \left[ V_{(i_3)} V_{(i_2)} V_{(i_1)} f \right] \right\vert^2 d \mu_{S_{u',v}} du'
		\leq C \bc^2.
	\]
	
	Consider now the fourth order derivatives of $f$.  For $i_1,i_2,i_3,i_4 = 1,\ldots,6$,
	\[
		V_{(i_4)} V_{(i_3)} V_{(i_2)} V_{(i_1)} f \vert_{(x,p)},
	\]
	can be written as a sum of,
	\[
		J_{(i_4)} J_{(i_3)} J_{(i_2)} \left( {J_{(i_1)}}^j \right) E_j f \vert_{\exp_{-s_*}(x,p)} 
		+
		{J_{(i_4)}}^{j_4} {J_{(i_3)}}^{j_3} {J_{(i_2)}}^{j_2} {J_{(i_1)}}^{j_1} E_{j_4} E_{j_3} E_{j_2} E_{j_1} f \vert_{\exp_{-s_*}(x,p)},
	\]
	and terms which involve lower order derivatives and can be treated as before.  Clearly the second term can also be treated as before using Proposition \ref{prop:Jbounds} and the assumption,
	\[
		\sum_{j_1,j_2,j_3,j_4 = 1}^7 \sup_{P\vert_{\{v=v_0\}}} \left\vert \tilde{E}_{j_4} \tilde{E}_{j_3} \tilde{E}_{j_2} \tilde{E}_{j_1} f \right\vert
		\leq \bc,
	\]
	so consider just the first term.  By Proposition \ref{prop:JJJJbounds},
	\[
		\left\vert \mathcal{T}_p \left[ J_{(i_4)} J_{(i_3)} J_{(i_2)} \left( {J_{(i_1)}}^j \right) \vert_{\exp_{-s_*}(x,p)} \right] \right\vert
		\leq
		C \left( \frac{1}{r^p} + H_{\mathcal{T}_p,3} (v_0) \right),
	\]
	where $H_{\mathcal{T}_p,3}$ is defined in Proposition \ref{prop:JJJJbounds}.  Using the same argument as for $H_{\mathcal{T}_p,2}$ (except for the $r^{4-p}(\vert \mathfrak{D}^3 \alpha \vert + \vert \mathfrak{D}^2 \alpha \vert)$ terms\footnote{It is because of these terms $\mathfrak{D}^4 \mathcal{T}_p$ is only estimated in spacetime, rather than on null hypersurfaces.}),
	\begin{align*}
		&
		\int_{v_0}^v \int_{u_0}^u r^{2p-4} \int_{S_{u,v}} \vert H_{\mathcal{T}_p,3} (v_0) \vert^2 d \mu_{S_{u',v'}} du' dv'
		\\
		&
		\qquad \qquad
		\leq
		C \left( \int_{v_0}^v \frac{1}{v'^2} dv' + \int_{v_0}^v \int_{u_0}^u \int_{S_{u,v}} r^4 (\vert \mathfrak{D}^3 \alpha \vert + \vert \mathfrak{D}^2 \alpha \vert) d \mu_{S_{u',v'}} du' dv' \right)
		\\
		&
		\qquad \qquad
		\leq
		C \left( 1 + \int_{u_0}^u F^1_{v_0,v}(u') du' \right)
		\\
		&
		\qquad \qquad
		\leq
		C.
	\end{align*}
	Hence,
	\[
		\int_{v_0}^v \int_{u_0}^u r^{2p-4} \int_{S_{u,v}} \left\vert \mathcal{T}_p \left[
		V_{(i_4)} V_{(i_3)} V_{(i_2)} V_{(i_1)} f 
		\right] \right\vert^2 d \mu_{S_{u',v'}} du' dv'
		\leq
		C \bc^2.
	\]
	The proof then follows from the considerations of Section \ref{subsec:overview}.
\end{proof}

\section{Estimates for Weyl Curvature Components} \label{section:curvature}
The Weyl curvature components $\psi$ are estimated in $L^2$ on null hypersurfaces through weighted energy estimates for the Bianchi equations.  The main proposition of this section, Proposition \ref{prop:curvmain}, will show that, at any point $x\in \mathcal{A}$ (see Theorem \ref{thm:main4}), the bootstrap assumptions for curvature \eqref{eq:curvatureba} can be retrieved with better constants.

Each Bianchi pair is assigned a weight $q$,
\begin{equation} \label{eq:Bianchiweights}
	q(\alpha, \beta) = 5, \qquad q(\beta , (\rho,\sigma)) = 4, \qquad q((\rho,\sigma), \betabar) = 2, \qquad q(\betabar, \alphabar) = 0.
\end{equation}
The energy estimates will be derived by integrating the following identities over a spacetime region.
\begin{lemma}
	The following identities hold for any $k$,
	\begin{align} \label{eq:enid1}
	\begin{split}
	 & \mathrm{Div} \left( r^5 \vert \mathfrak{D}^k\alpha \vert^2 e_3\right) + 2 \mathrm{Div} \left( r^5 \vert \mathfrak{D}^k \beta \vert^2 e_4 \right) - 4 \frac{1}{\Omega^2} \divslash \left( \Omega^2 r^5  \mathfrak{D}^k \alpha \cdot \mathfrak{D}^k \beta \right)
	 \\
	 & \qquad \qquad \qquad = r^5 \big( \rp_1 \vert \mathfrak{D}^k \alpha \vert^2 + \rp_1 \vert \mathfrak{D}^k \beta \vert^2 - 4(\eta + \etabar)\cdot \mathfrak{D}^k \alpha \cdot \mathfrak{D}^k \beta 
	 \\
	 & \qquad \qquad \qquad \quad + 2 \mathfrak{D}^k \alpha \cdot E_3[\mathfrak{D}^k \alpha] + 4 \mathfrak{D}^k \beta \cdot E_4[\mathfrak{D}^k \beta] \big),
	 \end{split}
	\end{align}
	\begin{align} \label{eq:enid2}
	\begin{split}
 	& \mathrm{Div} \left( r^4 \vert \mathfrak{D}^k \beta \vert^2 e_3 \right) + \mathrm{Div} \left( r^4 \vert \mathfrak{D}^k \rho \vert^2 e_4 \right) + \mathrm{Div} \left( r^4 \vert \mathfrak{D}^k \sigma \vert^2 e_4 \right)
	\\
	& - 2 \frac{1}{\Omega^2} \divslash \left( \Omega^2 r^4 \mathfrak{D}^k \beta \cdot (\mathfrak{D}^k \rho \gslash - \mathfrak{D}^k \sigma \epsslash ) \right)
	\\
	& \qquad \qquad \qquad= r^4 \big( \rp_1 \vert \mathfrak{D}^k \beta \vert^2 - 2(\eta + \etabar)\cdot \mathfrak{D}^k \beta \cdot (\mathfrak{D}^k \rho \gslash - \mathfrak{D}^k \sigma \epsslash )
	\\
	& \qquad \qquad \qquad \quad + 2 \mathfrak{D}^k \beta \cdot E_3[\mathfrak{D}^k \beta] + 2 \mathfrak{D}^k \rho \cdot E_4[\mathfrak{D}^k \rho] + 2 \mathfrak{D}^k \sigma \cdot E_4[\mathfrak{D}^k \sigma] \big),
	\end{split}
	\end{align}
	\begin{align} \label{eq:enid3}
	\begin{split}
 	& \mathrm{Div} \left( r^2 \vert \mathfrak{D}^k \rho \vert^2 e_3 \right) + \mathrm{Div} \left( r^2 \vert \mathfrak{D}^k \sigma \vert^2 e_3 \right) + \mathrm{Div} \left( r^2 \vert \mathfrak{D}^k \betabar \vert^2 e_4 \right)
	\\
	& + 2 \frac{1}{\Omega^2} \divslash \left( \Omega^2 r^2 \mathfrak{D}^k \betabar \cdot (\mathfrak{D}^k \rho \gslash - \mathfrak{D}^k \sigma \epsslash ) \right)
	\\
	& \qquad \qquad \qquad = r^2 \big( \rp_1 \vert \mathfrak{D}^k \rho \vert^2 + \rp_1 \vert \mathfrak{D}^k \sigma \vert^2 + 2(\eta + \etabar)\cdot \mathfrak{D}^k \betabar \cdot (\mathfrak{D}^k \rho \gslash - \mathfrak{D}^k \sigma \epsslash )
	\\
	& \qquad \qquad \qquad \quad + 2 \mathfrak{D}^k \rho \cdot E_3[\mathfrak{D}^k \rho] + 2 \mathfrak{D}^k \sigma \cdot E_3[\mathfrak{D}^k \sigma] + 2 \mathfrak{D}^k \betabar \cdot E_4[\mathfrak{D}^k \betabar] \big),
	\end{split}
	\end{align}
	\begin{multline} \label{eq:enid4}
 	2 \mathrm{Div} \left( \vert \mathfrak{D}^k \betabar \vert^2 e_3 \right) + \mathrm{Div} \left( \vert \mathfrak{D}^k \alphabar \vert^2 e_4 \right) + 4 \frac{1}{\Omega^2} \divslash \left( \Omega^2 \mathfrak{D}^k \alphabar \cdot \mathfrak{D}^k \betabar \right) \\ = \rp_1 \vert \mathfrak{D}^k \betabar \vert^2 + 4(\eta + \etabar)\cdot \mathfrak{D}^k \betabar \cdot \mathfrak{D}^k \alphabar + 4 \mathfrak{D}^k \betabar \cdot E_3[\mathfrak{D}^k \betabar] + 2 \mathfrak{D}^k \alphabar \cdot E_4[\mathfrak{D}^k \alphabar],
	\end{multline}
	where $\mathrm{Div}$ denotes the spacetime divergence.
\end{lemma}

\begin{proof}
	The proof follows by applying the product rule to each term on the left hand side of each identity.  For the first terms each Bianchi equation contracted with its corresponding weighted curvature component is used, i.\@e.\@ equation \eqref{eq:commBianchi3} contracted with $r^{q(\uppsi_p,\uppsi_{p'}')} \uppsi_p$, and equation \eqref{eq:commBianchi4} contracted with $r^{q(\uppsi_p,\uppsi_{p'}')} \uppsi_{p'}'$.  Then use the fact that
	\[
		\mathrm{Div} (e_3)  = \tr \chibar, \qquad \mathrm{Div} (e_4) = \tr \chi + \omega,
	\]
	and
	\[
		e_3 (r^n) = - \frac{n}{r} r^n \frac{1}{\Omega^2}, \qquad e_4 (r^n) = \frac{n}{2} r^n \tr \chi_{\circ}.
	\]
	For the final term on the left hand side of each identity, use the fact that
	\[
		\frac{\nablaslash \Omega}{\Omega} = \nablaslash (\log \Omega) = \frac{1}{2} \left( \eta + \etabar \right).
	\]
	
	The proof of \eqref{eq:enid3} is presented to illustrate a cancellation which occurs in \eqref{eq:enid2}, \eqref{eq:enid3}, \eqref{eq:enid4}.  Suppose, to reduce notation, that $k=0$.  Clearly,
	\begin{align*}
		\mathrm{Div} \left( r^2 \vert \rho \vert^2 e_3 \right)
		&
		=
		- \frac{2}{\Omega} r \vert \rho \vert^2 + 2 \rho \cdot \nablaslash_3 \rho + \vert \rho \vert^2 \tr \chibar
		\\
		&
		=
		- \frac{2}{\Omega^2} r \vert \rho \vert^2
		- 2 r^2 \vert \rho \vert^2 \tr \chibar - 2 r^2 \rho \divslash \betabar
		+ 2 r^2 \rho \cdot E_3[\rho]
		\\
		\mathrm{Div} \left( r^2 \vert \sigma \vert^2 e_3 \right)
		&
		=
		- \frac{2}{\Omega^2} r \vert \sigma \vert^2
		- 2 r^2 \vert \sigma \vert^2 \tr \chibar - 2 r^2 \sigma \curlslash \betabar
		+ 2 r^2 \sigma \cdot E_3[\sigma]
		\\
		 \mathrm{Div} \left( r^2 \vert \betabar \vert^2 e_4 \right)
		 &
		 =
		 2r \vert \betabar \vert^2 + 2 r^2 \betabar \cdot \nablaslash_4 \betabar + r^2 \vert \betabar \vert^2 \left( \tr \chi + \omega \right)
		 \\
		 &
		 =
		 r^2 \left( - \left( \tr \chi - \tr \chi_{\circ} \right) \vert \betabar \vert^2 + \omega \vert \betabar \vert^2
		 + 
		 2 \betabar \cdot \left( - \nablaslash \rho + {}^* \nablaslash \sigma \right) + 2 \betabar \cdot E_4[\betabar]
		 \right).
	\end{align*}
	Note that in the expression for $ \mathrm{Div} \left( r^2 \vert \betabar \vert^2 e_4 \right)$ the term generated by $\nablaslash_4$ acting on $r^2$ exactly cancels the $\tr \chi \vert \betabar \vert^2$ term to leave $r^2 \left( \tr \chi - \tr \chi_{\circ} \right) \vert \betabar \vert^2$.  This cancellation occurs precisely because the weight $q=2$ was chosen for the Bianchi pair $((\rho,\sigma), \betabar)$.  This resulting term, and most of the others, have the same form as the error terms and so can be absorbed to give,
	\begin{align*}
		\mathrm{Div} \left( r^2 \vert \rho \vert^2 e_3 \right)
		&
		=
		- 2 r^2 \rho \divslash \betabar
		+ 2 r^2 \rho \cdot E_3[\rho]
		\\
		\mathrm{Div} \left( r^2 \vert \sigma \vert^2 e_3 \right)
		&
		=
		- 2 r^2 \sigma \curlslash \betabar
		+ 2 r^2 \sigma \cdot E_3[\sigma]
		\\
		 \mathrm{Div} \left( r^2 \vert \betabar \vert^2 e_4 \right)
		 &
		 =
		 r^2 \left(
		 2 \betabar \cdot \left( - \nablaslash \rho + {}^* \nablaslash \sigma \right) + 2 \betabar \cdot E_4[\betabar]
		 \right).
	\end{align*}
	Terms of the form $r^q \rp_1 \vert \betabar \vert^2$, which would appear if any weight other than $q=2$ had been chosen, would not have the correct form to be absorbed by the error term in the expression for $\mathrm{Div} \left( r^2 \vert \betabar \vert^2 e_4 \right)$.
	The proof follows by computing $\frac{1}{\Omega^2} \divslash \left( \Omega^2 r^2 \betabar \cdot (\rho \gslash - \sigma \epsslash ) \right)$.
\end{proof}

\begin{remark} \label{rmk:peeling}
	The weights \eqref{eq:Bianchiweights} were chosen carefully so that a cancellation would occur in the above identities, as illustrated in the proof.  This cancellation does not occur in the identity for the Bianchi pair $(\alpha,\beta)$.  It would if the weight $q(\alpha,\beta) = 6$ had been chosen.  This would however lead one to impose a faster rate of decay for $\alpha,\beta$ along $\{u=u_0\}$, consistent the decay required for a the spacetime to admit a conformal compactification.  The estimates will close without imposing this stronger decay.
\end{remark}

\begin{proposition} \label{prop:curvmain}
	If $x\in \mathcal{A}$ and $u = u(x)$, $v = v(x)$, then
	\[
		F^1_{v_0,v}(u) + F^2_{u_0,u}(v) \leq C \left( \bc + \frac{1}{v_0} \right),
	\]
	for some constant $C$.
\end{proposition}

\begin{proof}
	Integrating the identity \eqref{eq:enid1} over the spacetime characteristic rectangle $u_0 \leq u' \leq u, v_0 \leq v' \leq v$ for a fixed $0\leq k\leq 3$ gives\footnote{Note that the final term on the left hand side of \eqref{eq:enid1} is a spherical divergence and hence vanished when integrated over the spheres.}
	\begin{align*}
		& \int_{v_0}^v \int_{S_{u,v'}} r^5 \vert \mathfrak{D}^k \alpha \vert^2 d \mu_{S_{u,v'}} d v' + \int_{u_0}^u \int_{S_{u',v}} r^5 \vert \mathfrak{D}^k \beta \vert^2 d \mu_{S_{u',v}} d u' \\
		& \qquad \qquad = \int_{v_0}^v \int_{S_{u_0,v'}} r^5 \vert \mathfrak{D}^k \alpha \vert^2 d \mu_{S_{u_0,v'}} d v' + \int_{u_0}^u \int_{S_{u',v_0}} r^5 \vert \mathfrak{D}^k \beta \vert^2 d \mu_{S_{u',v_0}} d u' \\
		& \qquad \qquad \quad + \int_{u_0}^u \int_{v_0}^v \int_{S_{u',v'}} r^5 \big( \rp_1 \vert \mathfrak{D}^k \alpha \vert^2 + \rp_1 \vert \mathfrak{D}^k \beta \vert^2 - 4(\eta + \etabar)\cdot \mathfrak{D}^k \alpha \cdot \mathfrak{D}^k \beta  \\
		& \qquad \qquad \quad + 2 \mathfrak{D}^k \alpha \cdot E_3[\mathfrak{D}^k \alpha] + 4 \mathfrak{D}^k \beta \cdot E_4[\mathfrak{D}^k \beta] \big) d \mu_{S_{u',v'}} \Omega^2 d v' d u'.
	\end{align*}
	Clearly
	\[
		\int_{u_0}^u \int_{v_0}^v \int_{S_{u',v'}} r^5 \left( \rp_1 \vert \mathfrak{D}^k \alpha \vert^2 + \rp_1 \vert \mathfrak{D}^k \beta \vert^2 \right) d \mu_{S_{u',v'}} d v' d u' \leq \int_{u_0}^u F^1_{v_0,v}(u') d u',
	\]
	and
	\begin{align*}
		& \int_{u_0}^u \int_{v_0}^v \int_{S_{u',v'}} r^5 (\eta + \etabar)\cdot \mathfrak{D}^k \alpha \cdot \mathfrak{D}^k \beta \ d \mu_{S_{u',v'}} \Omega^2 d v' d u' 
		\\
		& \qquad \qquad \qquad \qquad \leq C \int_{u_0}^u \int_{v_0}^v \int_{S_{u',v'}} r^4 \vert \mathfrak{D}^k \alpha \vert \vert \mathfrak{D}^k \beta \vert \ d \mu_{S_{u',v'}} d v' d u'
		\\
		& \qquad \qquad \qquad \qquad \leq C \int_{u_0}^u \int_{v_0}^v \int_{S_{u',v'}} r^4 \vert \mathfrak{D}^k \alpha \vert^2 + r^4 \vert \mathfrak{D}^k \beta \vert^2 \ d \mu_{S_{u',v'}} d v' d u'
		\\
		& \qquad \qquad \qquad \qquad \leq C \int_{u_0}^u F^1_{v_0,v}(u') d u',
	\end{align*}
	using the bootstrap assumptions for $\eta, \etabar$ and the upper bound for $\Omega$.  In Lemma \ref{lem:curvatureerrors} below it will be shown that
	\begin{multline*}
		\int_{u_0}^u \int_{v_0}^v \int_{S_{u',v'}} r^{q(\uppsi_{p},\uppsi_{p'}')} \big( \mathfrak{D}^k \uppsi_{p} \cdot E_3[\mathfrak{D}^k \uppsi_{p}] + \mathfrak{D}^k \uppsi_{p'}' \cdot E_4[\mathfrak{D}^k \uppsi_{p'}'] \big) d \mu_{S_{u',v'}} \Omega^2 d v' d u' \\
		\leq C \left( \int_{u_0}^u F^1_{v_0,v}(u') d u' + \frac{1}{v_0} + \bc \right),
	\end{multline*}
	for each Bianchi pair $(\uppsi_{p},\uppsi_{p'}')$.  Hence
	\begin{multline*}
		\int_{v_0}^v \int_{S_{u,v'}} r^5 \vert \mathfrak{D}^k \alpha \vert^2 d \mu_{S_{u,v'}} d v' + \int_{u_0}^u \int_{S_{u',v}} r^5 \vert \mathfrak{D}^k \beta \vert^2 d \mu_{S_{u',v}} d u' \\
		\leq C \left( \int_{u_0}^u F^1_{v_0,v}(u') d u' + \frac{1}{v_0} + F^1_{v_0,v}(u_0) + F^2_{u_0,u}(v_0) + \bc \right).
	\end{multline*}
	
	Repeating this for each of the identities \eqref{eq:enid2},\eqref{eq:enid3},\eqref{eq:enid4} for $k = 0,1,\ldots ,s$ and summing then gives
	\begin{equation} \label{eq:curvpregron1}
		F^1_{v_0,v}(u) \leq C \left( \int_{u_0}^u F^1_{v_0,v}(u') d u' + \frac{1}{v_0} + F^1_{v_0,v}(u_0) + F^2_{u_0,u}(v_0) + \bc \right),
	\end{equation}
	and
	\begin{equation} \label{eq:curvpregron2}
		F^2_{u_0,u}(v) \leq C \left( \int_{u_0}^u F^1_{v_0,v}(u') d u' + \frac{1}{v_0} + F^1_{v_0,v}(u_0) + F^2_{u_0,u}(v_0) + \bc \right).
	\end{equation}
	Note that $\alphabar$ doesn't appear in $F^1_{v_0,v}(u_0)$ so that the term involving $(\eta + \etabar)\cdot \betabar \cdot \alphabar$ from the identity \eqref{eq:enid4} is estimated slightly differently:
	\begin{align*}
		& \int_{u_0}^u \int_{v_0}^v \int_{S_{u',v'}} (\eta + \etabar)\cdot \betabar \cdot \alphabar \ d \mu_{S_{u',v'}} \Omega^2 d v' d u' \\
		& \qquad \qquad \qquad \leq C \int_{u_0}^u \int_{v_0}^v \int_{S_{u',v'}} \frac{1}{r} \vert \betabar \vert \vert \alphabar \vert d \mu_{S_{u',v'}} \Omega^2 d v' d u' \\
		& \qquad \qquad \qquad \leq C \int_{u_0}^u \int_{v_0}^v \int_{S_{u',v'}} \vert \betabar \vert^2 + \frac{\vert \alphabar \vert^2}{r^2} d \mu_{S_{u',v'}} \Omega^2 d v' d u' \\
		& \qquad \qquad \qquad \leq C \left( \int_{u_0}^u F^1_{v_0,v}(u') d u' + \frac{1}{v_0} \right),
	\end{align*}
	where the last line follows from the inequality
	\begin{align*}
		\int_{u_0}^u \int_{v_0}^v \int_{S_{u',v'}} \frac{\vert \alphabar \vert^2}{r^2} d \mu_{S_{u',v'}} \Omega^2 d v' d u' & \leq \int_{v_0}^v \frac{1}{r^2} F^2_{u_0,u}(v') d v' \\
		& \leq C \int_{v_0}^v \frac{1}{v^2} d v' \\
		& \leq \frac{C}{v_0},
	\end{align*}
	using the bootstrap assumption for $F^2_{u_0,u}(v')$ and the fact that $r \sim v$ in the ``wave zone''.  Similarly for the terms involving $(\eta + \etabar)\cdot \mathfrak{D}^k \betabar \cdot \mathfrak{D}^k \alphabar$.
	
	Applying the Gr\"{o}nwall inequality to equation \eqref{eq:curvpregron1} and using the fact that $u \leq u_f$ gives
	\[
		F^1_{v_0,v}(u) \leq C \left( \frac{1}{v_0} + \bc + F^1_{v_0,v}(u_0) + F^2_{u_0,u}(v_0) \right).
	\]
	Inserting this in equation \eqref{eq:curvpregron2} gives
	\[
		F^2_{u_0,u}(v) \leq C \left( \frac{1}{v_0} + \bc + F^1_{v_0,v}(u_0) + F^2_{u_0,u}(v_0) \right).
	\]
	
\end{proof}

It remains to prove the following lemma which provides control over the error terms.
\begin{lemma} \label{lem:curvatureerrors}
	Under the assumptions of Proposition \ref{prop:curvmain}, for each Bianchi pair $(\uppsi_{p},\uppsi_{p'}')$,
	\begin{multline*}
		\int_{u_0}^u \int_{v_0}^v \int_{S_{u',v'}} r^{q(\uppsi_{p},\uppsi_{p'}')} \big( \mathfrak{D}^k \uppsi_{p} \cdot E_3[\mathfrak{D}^k \uppsi_{p}] + \mathfrak{D}^k \uppsi_{p'}' \cdot E_4[\mathfrak{D}^k \uppsi_{p'}'] \big) d \mu_{S_{u',v'}} \Omega^2 d v' d u' \\
		\leq C \left( \int_{u_0}^u F^1_{v_0,v}(u') d u' + \frac{1}{v_0} + \bc \right),
	\end{multline*}
\end{lemma}

\begin{proof}
	For the sake of brevity, unless specified otherwise $\int$ will denote the integral 
	\[
		\int_{u_0}^u \int_{v_0}^v \int_{S_{u',v'}} d \mu_{S_{u',v'}} d v' d u'.
	\]
	
	Consider first the errors in the $\nablaslash_3$ Bianchi equations.  Recall from Proposition \ref{prop:Bianchi} and Proposition \ref{prop:commutedBianchi} that
	\[
		E_3[\mathfrak{D}^k \uppsi_p] =  \mathfrak{D}(E_3[\mathfrak{D}^{k-1}\uppsi_{p}]) + \Lambda_1 ( \mathfrak{D}^k \uppsi_p + \mathfrak{D}^k\uppsi_{p'}') + \Lambda_1( \mathfrak{D}^{k-1}\uppsi_p + \mathfrak{D}^{k-1}\uppsi_{p'}') ,
	\]
	for $1 \leq k \leq 3$, and
	\begin{equation} \label{eq:Bianchierror3}
		E_3[\uppsi_p] = \rp_1 \uppsi_p + \sum_{p_1+p_2 \geq p} \Gamma_{p_1} \cdot \psi_{p_2} + \sum_{p_1+p_2\geq p} \rp_{p_1} \mathfrak{D} \mathcal{T}_{p_2} + \sum_{p_1+p_2\geq p} \Gamma_{p_1} \cdot \mathcal{T}_{p_2}.
	\end{equation}
	
	The first term in $E_3[\uppsi_p]$ will contribute terms of the form $\rp_1 \mathfrak{D}^{k'} \uppsi_p$ to the error where $0\leq k'\leq k$ (recall that $\mathfrak{D} \rp_1 = \rp_1$) and these can be dealt with easily
	\begin{align*}
		& \int_{u_0}^u \int_{v_0}^v \int_{S_{u',v'}} r^{q(\uppsi_{p},\uppsi_{p'}')} \rp_1 \mathfrak{D}^k \uppsi_{p} \cdot \mathfrak{D}^{k'} \uppsi_p d \mu_{S_{u',v'}} d v' d u' \\
		& \qquad \qquad \qquad \leq C \int_{u_0}^u \int_{v_0}^v \int_{S_{u',v'}} r^{q(\uppsi_{p},\uppsi_{p'}')} \left( \vert \mathfrak{D}^k \uppsi_{p} \vert^2 + \vert \mathfrak{D}^{k'} \uppsi_p \vert^2 \right) d \mu_{S_{u',v'}} d v' d u' \\
		& \qquad \qquad \qquad \leq C \int_{u_0}^u F^1_{v_0,v}(u') d u'
	\end{align*}
	
	The second term in $E_3[\uppsi_p]$ will contribute terms of the form $\mathfrak{D}^{k_1} \Gamma_{p_1} \cdot \mathfrak{D}^{k_2} \psi_{p_2}$ where $p_1 + p_2 \geq p$ and $0 \leq k_1,k_2 \leq k$.  Note also that, since $k_1 + k_2 = k$, at most one of $k_1$ or $k_2$ can be greater than 1.  Assume first that $k_1 \leq 1$.
	
	Suppose $\uppsi_p \neq \alpha, \beta$, then $q(\uppsi_{p},\uppsi_{p'}') = 2p-4$ and
	\begin{align*}
		\int r^q \mathfrak{D}^{k_1} \Gamma_{p_1} \cdot \mathfrak{D}^{k_2} \psi_{p_2} \cdot \mathfrak{D}^k \uppsi_p & \leq \sup_{u',v'} \left( r^{p_1} \Vert \mathfrak{D}^{k_1} \Gamma_{p_1} \Vert_{L^{\infty}} \right) \int r^{q + p_2 - p} \vert \mathfrak{D}^{k_2} \psi_{p_2} \vert \vert \mathfrak{D}^k \uppsi_p \vert \\
		& \leq C \int \left( r^{2p_2 - 4} \vert \mathfrak{D}^{k_2} \psi_{p_2} \vert^2 + r^{2q + 4 - 2p} \vert \mathfrak{D}^k \uppsi_p \vert^2 \right) \\
		& \leq C \left( \int_{u_0}^u F^1_{v_0,v}(u') d u' + \frac{1}{v_0} \right),
	\end{align*}
	where the first line follows from the Sobolev inequality \eqref{eq:Riccisob} (and the fact that $k_1 \leq 1$) and uses $p_1 + p_2 \geq p$.  The third line uses the fact that $q = 2p-4$ (and recall that $\int \frac{\vert \alphabar \vert}{r^2} \leq \frac{C}{v_0}$).
	
	If $\uppsi_p = \alpha$ or $\beta$ then $q(\uppsi_{p},\uppsi_{p'}') = 2p-3$ and the second term in the second line above would be $\int r^{q +1} \vert \mathfrak{D}^k \uppsi_p \vert^2$ which can't be controlled by the last line.  The sum in the error \eqref{eq:Bianchierror3}, however, begins at $p + \frac{1}{2}$ for $\alpha$ and $\beta$, and so in the first line in the above would have $p_1 + p_2 \geq p + \frac{1}{2}$.  Using this fact these terms can be controlled.
	
	If $k_1 > 1$ then it must be the case that $k_2 \leq 1$.  The above steps can then be repeated but using the Sobolev inequality \eqref{eq:curvsob} for $\mathfrak{D}^{k_2} \psi_{p_2}$.  For $\uppsi_p \neq \alpha, \beta$ then get
	\begin{align*}
		\int r^q \mathfrak{D}^{k_1} \Gamma_{p_1} \cdot \mathfrak{D}^{k_2} \psi_{p_2} \cdot \mathfrak{D}^k \uppsi_p
		& \leq C \int r^{q + p_1 - p} \vert \mathfrak{D}^{k_1} \Gamma_{p_1} \vert \vert \mathfrak{D}^k \uppsi_p \vert \\
		& \leq C \int \left( r^{2p_1 - 4} \vert \mathfrak{D}^{k_1} \Gamma_{p_1} \vert^2 + r^{2q + 4 - 2p} \vert \mathfrak{D}^k \uppsi_p \vert^2 \right) \\
		& \leq C \left( \int_{u_0}^u F^1_{v_0,v}(u') d u' + \frac{1}{v_0} \right),
	\end{align*}
	where the last line now uses the fact that
	\begin{equation} \label{eq:gammabound}
		\int r^{2p_1 - 4} \vert \mathfrak{D}^{k_1} \Gamma_{p_1} \vert^2 \leq C(u_f - u_0) \int_{v_0}^v \frac{1}{r^2} d v' \leq \frac{C}{v_0},
	\end{equation}
	by the bootstrap assumption \eqref{eq:Ricciba} and the fact that $v\sim r$ in the ``wave zone''.  Similarly for $\uppsi_p = \alpha, \beta$.
	
	The third term in $E_3[\uppsi_p]$ will contribute terms of the form $\rp_{p_1} \mathfrak{D}^{k'} \mathcal{T}_{p_2}$ to $E_3[ \mathfrak{D}^k \uppsi_p]$ where $1 \leq k' \leq k+1$ and $p_1 + p_2 \geq p$.  Recall that, if $\psi_p = \alpha$ or $\beta$ then actually $p_1 + p_2 \geq p + \frac{1}{2}$.  If $\psi_p \neq \alpha, \beta$ then $q = 2p-4$ and,
	\begin{align*}
		\int r^q \rp_{p_1} \mathfrak{D}^{k'} \mathcal{T}_{p_2} \cdot \mathfrak{D}^k \uppsi_p 
		&
		\leq C \int r^{p+p_2 - 4} \mathfrak{D}^{k'} \mathcal{T}_{p_2} \cdot \mathfrak{D}^k \uppsi_p
		\\
		&
		\leq C \int r^{2p_2 - 4} \vert \mathfrak{D}^{k'} \mathcal{T}_{p_2} \vert^2 + r^{2p-4} \vert \mathfrak{D}^k \uppsi_p \vert^2
		\\
		&
		\leq C \left( \int_{u_0}^u F^1_{v_0,v}(u') du' + \frac{1}{v_0} + \bc \right),
	\end{align*}
	by Proposition \ref{prop:tmain}.  Similarly, if $\uppsi_p = \alpha$ or $\beta$, then $q = 2p-3$ and $p_1 \geq p - p_2 + \frac{1}{2}$, so $r^q \leq r^{p+p_2 - \frac{7}{2}}$ and,
	\begin{align*}
		\int r^q \rp_{p_1} \mathfrak{D}^{k'} \mathcal{T}_{p_2} \cdot \mathfrak{D}^k \uppsi_p
		&
		\leq C \int r^{2p_2 - 4} \vert \mathfrak{D}^{k'} \mathcal{T}_{p_2} \vert^2 + r^{2p-3} \vert \mathfrak{D}^k \uppsi_p \vert^2
		\\
		&
		\leq C \left( \int_{u_0}^u F^1_{v_0,v}(u') du' + \frac{1}{v_0} + \bc \right).
	\end{align*}
	The final term in $E_3[\uppsi_p]$ contributes terms of the form $ \mathfrak{D}^{k_1} \Gamma_{p_1} \cdot \mathfrak{D}^{k_2} \mathcal{T}_{p_2}$ with $0\leq k_1,k_2 \leq k$, $k_1 + k_2 = k$ and $p_1 + p_2 \geq p$, or $p_1 + p_2 \geq p + \frac{1}{2}$ if $\uppsi_p = \alpha$ or $\beta$.  These terms can be dealt with as before using the fact that either $k_1 \leq 1$ or $k_2 \leq 1$, and the pointwise bounds for $\mathcal{T}_p, \mathfrak{D} \mathcal{T}_p, \mathfrak{D}^2 \mathcal{T}_p$ from Proposition \ref{prop:tmain}.
	
	The final terms in $E_3[\mathfrak{D}^k \uppsi_p]$, i.\@e.\@ the terms of the form $\mathfrak{D}^{k_1} \Lambda_1 \cdot \mathfrak{D}^{k_2} \uppsi_p$ etc.\@ can be dealt with similarly (since all of the terms in $\Lambda_1$ are zeroth order and the numerology for the Sobolev inequalities still work out).
	
	The errors in the $\nablaslash_4$ Bianchi equations can be dealt with in a similar manner.  First recall that
	\begin{align*}
		E_4[\mathfrak{D}^k \uppsi_{p'}'] = & \ \mathfrak{D}(E_4[\mathfrak{D}^{k-1}\uppsi_{p'}']) + \rp_0 E_4[\uppsi_{p'}'] + \Lambda_1\mathfrak{D}^k \uppsi_p + \Lambda_2\mathfrak{D}^k\uppsi_{p'}'  \\
	 & + \Lambda_1 \mathfrak{D}^{k-1} \uppsi_p + \Lambda_2' \mathfrak{D}^{k-1}\uppsi_{p'}',
	\end{align*}
	and
	\begin{equation} \label{eq:Bianchierror4}
		E_4[\uppsi_{p'}'] = \sum_{p_1+p_2 \geq p' + \frac{3}{2}} \Gamma_{p_1} \cdot \psi_{p_2} 
		+ \sum_{p_1 + p_2 \geq p+2} \rp_{p_1} \mathfrak{D} \mathcal{T}_{p_2} 
		+ \sum_{p_1 + p_2 \geq p+2} \Gamma_{p_1} \cdot \mathcal{T}_{p_2}.
	\end{equation}
	Recall also that the first summation in the error \eqref{eq:Bianchierror4} always begins at $p'+2$, except for the term $\eta^{\#} \cdot \alpha$ appearing in $E_4[\beta]$.
	
	Assume first then that $\uppsi_{p'}' \neq \beta$.  Terms in the first sum \eqref{eq:Bianchierror4} will then contribute terms of the form $\mathfrak{D}^{k_1} \Gamma_{p_1} \cdot \mathfrak{D}^{k_2} \psi_{p_2}$ to the error $E_4[\mathfrak{D}^k \uppsi_{p'}']$, where $p_1+p_2\geq p'+2$, $0\leq k_1,k_2\leq k$ and at most one of $k_1,k_2$ is bigger than 1.  Again, suppose first that $k_1\leq 1$.  If $\uppsi_{p'}' \neq \beta$ then $2 q(\uppsi_{p},\uppsi_{p'}') - 2p' = q'(\uppsi_{p'}')$\footnote{Here $q'(\uppsi_{p'}')$ denotes the power of $r$ multiplying $\vert \uppsi_{p'}'\vert^2$ in $F^1_{v_0,v}(u')$.  So, for example, $q'(\betabar) = 0$, whilst $q((\rho,\sigma),\betabar) = 2$.  Set $q'(\alphabar) = -2$.} and so
	\begin{align*}
		\int r^q \mathfrak{D}^{k_1} \Gamma_{p_1} \cdot \mathfrak{D}^{k_2} \psi_{p_2} \cdot \mathfrak{D}^k \uppsi_{p'}' & \leq \sup_{u',v'} \left( r^{p_1} \Vert \mathfrak{D}^{k_1} \Gamma_{p_1} \Vert_{L^{\infty}} \right) \int r^{q + p_2 - p' - 2} \vert \mathfrak{D}^{k_2} \psi_{p_2} \vert \vert \mathfrak{D}^k \uppsi_{p'}' \vert \\
		& \leq C \int \left( r^{2p_2 - 4} \vert \mathfrak{D}^{k_2} \psi_{p_2} \vert^2 + r^{2q - 2p'} \vert \mathfrak{D}^k \uppsi_{p'}' \vert^2 \right) \\
		& \leq C \left( \int_{u_0}^u F^1_{v_0,v}(u') d u' + \frac{1}{v_0} \right).
	\end{align*}
	If $\uppsi_{p'}' = \beta$ then will have terms of the form $\mathfrak{D}^{k_1} \Gamma_{p_1} \cdot \mathfrak{D}^{k_2} \psi_{p_2}$ with $p_1+p_2\geq p'+\frac{3}{2}$, however $2 q(\alpha,\beta) - 2p' = 3 = q'(\beta) - 1$.  Hence
	\begin{align*}
		\int r^q \mathfrak{D}^{k_1} \Gamma_{p_1} \cdot \mathfrak{D}^{k_2} \psi_{p_2} \cdot \mathfrak{D}^k \uppsi_{p'}' & \leq \sup_{u',v'} \left( r^{p_1} \Vert \mathfrak{D}^{k_1} \Gamma_{p_1} \Vert_{L^{\infty}} \right) \int r^{q + p_2 - p' - \frac{3}{2}} \vert \mathfrak{D}^{k_2} \psi_{p_2} \vert \vert \mathfrak{D}^k \uppsi_{p'}' \vert \\
		& \leq C \int \left( r^{2p_2 - 4} \vert \mathfrak{D}^{k_2} \psi_{p_2} \vert^2 + r^{2q - 2p' + 1} \vert \mathfrak{D}^k \uppsi_{p'}' \vert^2 \right) \\
		& \leq C \left( \int_{u_0}^u F^1_{v_0,v}(u') d u' + \frac{1}{v_0} \right).
	\end{align*}
	
	The second summation in $E_4[\uppsi_{p'}']$, \eqref{eq:Bianchierror4}, will contribute terms of the form $\rp_{p_1} \mathfrak{D}^{k'} \mathcal{T}_{p_2}$ to $E_4[\mathfrak{D}^k \uppsi_{p'}']$, where $1\leq k' \leq k+1$ and $p_1 + p_2 \geq p+2$.  Since $2q(\uppsi_p, \uppsi_{p'}') - 2p' \leq q'(\uppsi_{p'}')$,
	\begin{align*}
		\int r^q \rp_{p_1} \mathfrak{D}^{k'} \mathcal{T}_{p_2} \cdot \mathfrak{D}^k \uppsi_{p'}'
		&
		\leq C \int r^{q-p + p_2 - 2} \vert \mathfrak{D}^{k'} \mathcal{T}_{p_2} \vert  \vert \mathfrak{D}^k \uppsi_{p'}' \vert
		\\
		&
		\leq C \int r^{2p_2 - 4} \vert \mathfrak{D}^{k'} \mathcal{T}_{p_2} \vert^2 + r^{q'} \vert \mathfrak{D}^k \uppsi_{p'}' \vert^2
		\\
		&
		\leq C \left( \int_{u_0}^u F^1_{v_0,v}(u') du' + \frac{1}{v_0} + \bc \right).
	\end{align*}
	The final summation in \eqref{eq:Bianchierror4} contributes terms of the form $\mathfrak{D}^{k_1} \Gamma_{p_1} \cdot \mathfrak{D}^{k_2} \mathcal{T}_{p_2}$ to $E_4[\mathfrak{D}^k \uppsi_{p'}']$, where $0 \leq k_1,k_2 \leq k$, $k_1 + k_2 = k$ and $p_1+p_2 \geq p' + 2$.  These terms can be treated similarly using the fact that either $k_1\leq 1$ or $ k_2 \leq 1$, and the pointwise bounds for $\mathcal{T}_{p_2}, \mathfrak{D} \mathcal{T}_{p_2}$ from Proposition \ref{prop:tmain}.
	
	The remaining terms in $E_4[\mathfrak{D}^k \uppsi_{p'}']$ can again be dealt with similarly.  It is important to note that $\Lambda_1$ and $\Lambda_2$ both contain zero-th order derivatives only of $\Gamma$ and $\psi$.  Whilst $\Lambda_2'$ does contain first order derivatives of the form $\mathfrak{D} \Gamma$, they only appears in $E_4[\mathfrak{D}^{k'} \uppsi_{p'}']$ multiplying $\mathfrak{D}^{k'-1} \psi_p$.  Hence, when these terms (for $k'\leq k$) appear in $E_4[\mathfrak{D}^k \uppsi_{p'}']$, it will always be possible to control one of the terms in the product pointwise via the Sobolev inequality.
\end{proof}

\section{Transport Estimates for Ricci Coefficients} \label{section:Ricci}
In this section the Ricci coefficients are estimated in $L^2$ on each of the spheres $S_{u,v}$ through transport estimates for the null structure equations.  This is done by using the identities, which hold for any scalar function $h$,
\begin{equation} \label{eq:transportid4}
	\partial_v \left( \int_{S_{u,v}} h d \mu_{S_{u,v}} \right) = \int_{S_{u,v}} \nablaslash_4 h + h \tr \chi d \mu_{S_{u,v}},
\end{equation}
and
\begin{equation} \label{eq:transportid3}
	\partial_u \left( \int_{S_{u,v}} h d \mu_{S_{u,v}} \right) = \int_{S_{u,v}} \left( \nablaslash_3 h + h \tr \chibar \right) \Omega^2 d \mu_{S_{u,v}},
\end{equation}
with $h = r^{2p-2} \vert \mathfrak{D}^k \Gamma_p \vert^2$.

The quantities $\Riccit$ and $\Riccif$ are treated separately.  Recall the set $\mathcal{A}$ from Theorem \ref{thm:main4}.

\subsection{Null Structure Equations in the Outgoing Direction}
Consider first the $\Riccif$ quantities, which satisfy null structure equations in the outgoing direction.
\begin{proposition} \label{prop:Ricci4main}
	If $x \in \mathcal{A}$ and $u = u(x)$, $v = v(x)$ then, for each $\Riccif$ and each $k = 0,1,2,3$,
	\[
		r^{2p-2} \int_{S_{u,v}} \vert \mathfrak{D}^k \Riccif \vert^2 d \mu_{S_{u,v}} 
		\leq 
		C \left( \bc + \frac{1}{v_0} \right),
	\]
	for some constant $C$.
\end{proposition}

\begin{proof}
Recall from Proposition \ref{prop:commutednullstructure} that the null structure equations in the $4$ direction take the form
\[
	\nablaslash_4 (\mathfrak{D}^k \overset{(4)}{\Gamma_p}) + \frac{p}{2} \ tr \chi \ \mathfrak{D}^k \overset{(4)}{\Gamma_p} = E_4[\mathfrak{D}^k \overset{(4)}{\Gamma_p}].
\]
Using the renormalisation of Remark \ref{rem:renorm} and the fact that $e_4(r^{-2}) = - r^{-2} \tr \chi_{\circ}$, the identity \eqref{eq:transportid4} with $h = r^{2p-2} \vert \mathfrak{D}^k \Gamma_p \vert^2$ implies that
\begin{align*}
	&
	\partial_v \left( r^{2p-2} \int_{S_{u,v}} \vert \mathfrak{D}^k \Riccif \vert^2 d \mu_{S_{u,v}} \right)
	\\
	& \quad
	= \int_{S_{u,v}} 2 r^{p-2} \Riccif \cdot \nablaslash_4 (r^p \mathfrak{D}^k \Riccif) + r^{2p} \vert \mathfrak{D}^k \Riccif \vert^2 e_4(r^{-2}) + r^{2p-2} \vert \mathfrak{D}^k \Riccif \vert^2 \tr \chi \ d \mu_{S_{u,v}}
	\\
	& \quad
	= \int_{S_{u,v}} 2 r^{p-2} \mathfrak{D}^k \Riccif \cdot \left( r^p E_4[\mathfrak{D}^k \overset{(4)}{\Gamma_p}] \right) + r^{2p-2} \vert \mathfrak{D}^k \Riccif \vert^2 \left( \tr \chi - \tr \chi_{\circ} \right) d \mu_{S_{u,v}}
	\\
	& \quad
	= 2 r^{2p-2} \int_{S_{u,v}} \mathfrak{D}^k \Riccif \cdot E_4[\mathfrak{D}^k \Riccif] \ d \mu_{S_{u,v}},
\end{align*}
where in the last line the term $\vert \mathfrak{D}^k \Riccif \vert^2 \left( \tr \chi - \tr \chi_{\circ} \right) = \mathfrak{D}^k \Riccif \cdot \mathfrak{D}^k \Gamma_p \cdot \Gamma_2$ has been absorbed into the error $\mathfrak{D}^k \Riccif \cdot E_4[\mathfrak{D}^k \overset{(4)}{\Gamma_p}]$.

Note that a precise cancellation occurs here.  If one were to apply \eqref{eq:transportid4} with $h = r^q \vert \mathfrak{D}^k \Gamma_p \vert^2$ for any $q \neq 2p-2$, there would be an additional term of the form $\rp_1 \vert \mathfrak{D}^k \Riccif \vert^2$ in the integral in the last line above.  It would not be possible to deal with this term as the terms in $\mathfrak{D}^k \Riccif \cdot E_4[\mathfrak{D}^k \overset{(4)}{\Gamma_p}]$ are dealt with below.

Integrating gives
\begin{align*}
	r^{2p-2} \int_{S_{u,v}} \vert \mathfrak{D}^k \Riccif \vert^2 d \mu_{S_{u,v}} \leq \
	&
	r(u,v_0)^{2p-2} \int_{S_{u,v_0}} \vert \mathfrak{D}^k \Riccif \vert^2 d \mu_{S_{u,v_0}}
	\\
	&
	+ 2\int_{v_0}^v r^{2p-2} \int_{S_{u,v'}} \mathfrak{D}^k \Riccif \cdot E_4[\mathfrak{D}^k \Riccif] \ d \mu_{S_{u,v'}} d v',
\end{align*}
so that it remains to bound the error terms.

Recall that
\[
	E_4[\mathfrak{D}^k \Riccif] =  \mathfrak{D} (E_4[\mathfrak{D}^{k-1} \Riccif]) + \Lambda_2 \cdot \mathfrak{D}^{k} \Riccif + \Lambda_2 ' \cdot\mathfrak{D}^{k-1} \Riccif ,
\]
for $k = 1,2,3$ and
\begin{equation} \label{eq:Riccierror4}
	E_4[\Riccif] = \psi_{p+2} + \sum_{p_1+p_2 \geq p + 2} \rp_{p_1} \cdot \Gamma_{p_2} + \sum_{p_1+p_2 \geq p + 2} \Gamma_{p_1} \cdot \Gamma_{p_2} + \mathcal{T}_{p+2}.
\end{equation}
The first term in \eqref{eq:Riccierror4} will contribute a term of the form $\mathfrak{D}^{k} \uppsi_{p+2}$ to the error $E_4[\mathfrak{D}^k \Riccif]$.  This term can be easily dealt with as follows.  Here $\int$ will be used to denote the integral
\[
	\int_{v_0}^v \int_{S_{u,v'}} d \mu_{S_{u,v'}} dv',
\]
(instead of the full spacetime integral in Section \ref{section:curvature}).  By the Cauchy--Schwarz inequality,
\[
	\int r^{2p-2} \mathfrak{D}^k \Riccif \cdot \mathfrak{D}^{k} \psi_{p+2} \leq \int r^{2p-4} \vert \mathfrak{D}^k \Riccif \vert^2 + \int r^{2p} \vert \mathfrak{D}^{k} \psi_{p+2} \vert^2.
\] 
The first term is clearly bounded by $\frac{C}{v_0}$ as in \eqref{eq:gammabound}.  Using the fact that the only curvature components appearing in the $\nablaslash_4 \Riccif$ equations are $\alpha, \beta$ (so that $\psi_{p+2} \in \{\alpha ,\beta \}$), one can explicitly check that the second term can be controlled by $\frac{1}{v_0} F^1_{v_0,v} (u)$ and hence, by the bootstrap assumption \eqref{eq:curvatureba},
\[
	\int r^{2p-2} \mathfrak{D}^k \Riccif \cdot \mathfrak{D}^{k} \psi_{p+2} \leq \frac{C}{v_0}.
\]

Consider now the terms in $E_4[\mathfrak{D}^k \Riccif]$ arising from the first sum in \eqref{eq:Riccierror4}.  These will all be of the form $\rp_{p_1} \mathfrak{D}^{k'} \Gamma_{p_2}$ where $0\leq k'\leq k$ and $p_1 + p_2 \leq p + 2$ and so
\begin{align*}
	\int r^{2p-2} \rp_{p_1} \mathfrak{D}^{k'} \Gamma_{p_2} \cdot \mathfrak{D}^k \Riccif \leq
	&
	\ C\int r^{p+p_2-4} \vert \mathfrak{D}^{k'} \Gamma_{p_2} \vert \vert \mathfrak{D}^k \Riccif \vert
	\\
	\leq
	&
	\ C \left( \int r^{2p_2 - 4} \vert \mathfrak{D}^{k'} \Gamma_{p_2} \vert^2 + \int r^{2p-4} \vert \mathfrak{D}^k \Riccif \vert^2 \right)
	\\
	\leq
	&
	\ \frac{C}{v_0}.
\end{align*}

The terms arising from the second sum will have the form $\mathfrak{D}^{k_1} \Gamma_{p_1} \cdot \mathfrak{D}^{k_2} \Gamma_{p_2}$ where $p_1+p_2 \geq p+2$, $k_1+k_2 = k$ and, since $k \leq 3$, interchanging $k_1$ and $k_2$ if necessary, $k_1\leq 1$.  These terms can be dealt with exactly as the previous terms by using the Sobolev inequality \eqref{eq:Riccisob} on $\mathfrak{D}^{k_1} \Gamma_{p_1}$.

Similarly, for the $\mathfrak{D}^k \mathcal{T}_{p+2}$ term in $E_4 [ \mathfrak{D}^k \Riccif]$,
\[
	\int r^{2p-2} \mathfrak{D}^k \Riccif \cdot \mathfrak{D}^k \mathcal{T}_{p+2}
	\leq
	\int r^{2p-4} \vert \mathfrak{D}^k \Riccif \vert^2 + r^{2p} \vert \mathfrak{D}^k \mathcal{T}_{p+2} \vert^2.
\]
Setting $q = p+2$, the second term is of the form,
\[
	\int_{v_0}^v \int_{S_{u,v'}} r^{2	q-4} \vert \mathfrak{D}^k \mathcal{T}_{q} \vert^2 d\mu_{S_{u,v'}} d v',
\]
and hence, since $k\leq 3$, Proposition \ref{prop:tmain} implies that,
\[
	\int r^{2p-2} \mathfrak{D}^k \Riccif \cdot \mathfrak{D}^k \mathcal{T}_{p+2}
	\leq
	C \left( \frac{1}{v_0} + \bc \right).
\]

The remaining terms in $E_4[\mathfrak{D}^k \Riccif]$ can be dealt with in exactly the same way using the fact that $\Lambda_2$ contains only zeroth-order derivatives, and $\Lambda_2'$ contains only first order derivatives of Ricci coefficients (see the end of the proof of Lemma \ref{lem:curvatureerrors}).

\end{proof}

\subsection{Null Structure Equations in the Incoming Direction}

The $\Riccit$ quantities are estimated in roughly the same way as the $\Riccif$ quantities.  Since the $u$ coordinate is bounded above by $u_f$ however, the term
\[
	C\int_{u_0}^u r^{2p-2} \int_{S_{u',v}} \vert \Riccit \vert^2 d \mu_{S_{u',v}} d u'
\]
can appear on the right hand side of the estimates and be dealt with by the Gr\"{o}nwall inequality.  The estimates will also rely on the results of Proposition \ref{prop:curvmain} and Proposition \ref{prop:Ricci4main}.  It is also worth noting that we do not rely on any cancellation occurring when applying the identity \eqref{eq:transportid3}, as was the case for the $\Riccif$ quantities.

\begin{proposition} \label{prop:Ricci3main}
	If $x \in \mathcal{A}$ and $u = u(x)$, $v = v(x)$ then, for each $\Riccit$ and each $k = 0,1,\ldots ,3$,
	\begin{align*}
	r^{2p-2} \int_{S_{u,v}} \vert \mathfrak{D}^k \Riccit \vert^2 d \mu_{S_{u,v}} \leq
	&
	\ C \Bigg( \bc + \frac{1}{v_0} \Bigg).
\end{align*}
	for some constant $C$.
\end{proposition}

\begin{proof}
Recall the upper bound on $\Omega$.

For fixed $0 \leq k\leq 3$, setting $h = r^{2p-2} \vert \mathfrak{D}^k \Riccit \vert^2$ in the identity \eqref{eq:transportid3} and using the commuted equations,
\[
	\nablaslash_3 (\mathfrak{D}^k \Riccit) = E_3[\mathfrak{D}^k \Riccit],
\]
one obtains
\begin{align*}
	\partial_u \left( r^{2p-2} \int_{S_{u,v}} \vert \mathfrak{D}^k \Riccit \vert^2 d_{\mu_{S_{u,v}}} \right)
	&
	\leq C \int_{S_{u,v}} r^{2p-2} \mathfrak{D}^k \Riccit \cdot \nablaslash_3 (\mathfrak{D}^k \Riccit )
	\\
	&
	\qquad \qquad + \left\vert e_3(r^{2p-2}) + r^{2p-2} \tr \chibar \right\vert \vert \mathfrak{D}^k \Riccit \vert^2 d\mu_{S_{u,v}}
	\\
	&
	\leq C \int_{S_{u,v}} r^{2p-2} \mathfrak{D}^k \Riccit \cdot E_3[\mathfrak{D}^k \Riccit]
	+ \rp_1 r^{2p-2}\vert \mathfrak{D}^k \Riccit \vert^2 d\mu_{S_{u,v}}.
\end{align*}
The last line is obtained by recalling that $e_3(r^{2p-2}) = \frac{-1}{\Omega^2} \frac{(2p-2)}{r} r^{2p-2}$, using the lower bound for $\Omega$, rewriting $\tr \chibar = (\tr \chibar - \tr \chibar_{\circ}) + \tr \chibar_{\circ}$ and absorbing the term $\vert \mathfrak{D}^k \Riccit \vert^2 (\tr \chibar - \tr \chibar_{\circ})$ into the error $\mathfrak{D}^k \Riccit \cdot E_3[\mathfrak{D}^k \Riccit]$.

Integrating from $u_0$ gives
\begin{align*}
	r^{2p-2} \int_{S_{u,v}} \vert \mathfrak{D}^k \Riccit \vert^2 d_{\mu_{S_{u,v}}}
	\leq C \int_{u_0}^u \int_{S_{u',v}} r^{2p-2} \mathfrak{D}^k \Riccit \cdot E_3[\mathfrak{D}^k \Riccit]
	+ \rp_1 r^{2p-2}\vert \mathfrak{D}^k \Riccit \vert^2 d\mu_{S_{u',v}} d u'.
\end{align*}
The final term will be dealt with by the Gr\"{o}nwall inequality, so it remains to bound the integrals of the error terms.  Here $\int$ will denote the integral
\[
	\int_{u_0}^u \int_{S_{u',v}} d\mu_{S_{u',v}} du'.
\]
Recall that
\[
	E_3[\mathfrak{D}^k \Riccit] = \mathfrak{D} (E_3[\mathfrak{D}^{k-1} \Riccit]) + \Lambda_1 ( \mathfrak{D}^{k} \Riccit + \mathfrak{D}^{k-1} \Riccit)
\]
for $k = 1,\ldots ,3$, and
\begin{equation} \label{eq:Riccierror3}
	E_3[\Riccit] = \psi_p + \sum_{p_1+p_2 \geq p} \rp_{p_1} \cdot \Gamma_{p_2} + \sum_{p_1+p_2 \geq p} \Gamma_{p_1} \cdot \Gamma_{p_2} + \mathcal{T}_p.
\end{equation}

The curvature term in \eqref{eq:Riccierror3} will contribute a term of the form $\mathfrak{D}^k \psi_p$ to $E_3[\mathfrak{D}^k \Riccit]$, and
\begin{align*}
	\int r^{2p-2} \mathfrak{D}^k \Riccit \cdot \mathfrak{D}^k \psi_p \leq \int r^{2p-2} \vert \mathfrak{D}^k \Riccit \vert^2 + \int r^{2p-2} \vert \mathfrak{D}^k \psi_p \vert^2.
\end{align*}
The Gr\"{o}nwall inequality will be used on the first term.  For the second term note that, for $\psi_p \neq \alpha$, the $r$ weight of $\mathfrak{D}^k \psi_p$ which appears in $F^2_{u_0,u}(v)$ is $r^{2p-2}$.  Hence, since $\alpha$ doesn't appear in any $\nablaslash_3 \Riccit$ equations, the second term can be controlled by $F^2_{u_0,u}(v)$ and, by Proposition \ref{prop:curvmain},
\[
	\int r^{2p-2} \vert \mathfrak{D}^k \psi_p \vert^2 \leq C \left( F^1_{v_0,v}(u_0) + F^2_{u_0,u}(v_0) + \frac{1}{v_0} \right).
\]

Similarly the energy momentum tensor term in \eqref{eq:Riccierror3} will contribute a term of the form $\mathfrak{D}^k \mathcal{T}_p$ to $E_3 [ \mathfrak{D}^k \Riccit ]$ and,

\begin{align*}
	\int r^{2p-2} \mathfrak{D}^k \Riccit \cdot \mathfrak{D}^k \mathcal{T}_p
	&
	\leq \int r^{2p-2} \vert \mathfrak{D}^k \Riccit \vert^2 + r^{2p-2} \vert \mathfrak{D}^k \mathcal{T}_p \vert^2
	\\
	&
	\leq \int r^{2p-2} \vert \mathfrak{D}^k \Riccit \vert^2 + C \left( \frac{1}{v_0} + \bc \right),
\end{align*}
by Proposition \ref{prop:tmain}.

Consider now the terms in \eqref{eq:Riccierror3} of the form\footnote{The ``borderline terms'' in \eqref{eq:Riccierror3}, i.\@e.\@ the terms $\rp_{p_1} \cdot \Gamma_{p_2}$ and $\Gamma_{p_1} \cdot \Gamma_{p_2}$ for which $p_1 + p_2 = p$ are slightly more problematic and will be dealt with separately.}
\[
	\sum_{p_1+p_2 \geq p + 1} \rp_{p_1} \cdot \Gamma_{p_2} + \sum_{p_1+p_2 \geq p + 1} \Gamma_{p_1} \cdot \Gamma_{p_2}.
\]
The first sum contributes terms of the form $\rp_{p_1} \mathfrak{D}^{k'} \Gamma_{p_2}$ to the error $E_3[\mathfrak{D}^k \Riccit]$ where $0\leq k'\leq k$ and $p_1 + p_2 \geq p + 1$, so that
\begin{align*}
	\int r^{2p-2} \rp_{p_1} \mathfrak{D}^{k} \Riccit \cdot \mathfrak{D}^{k'} \Gamma_{p_2}
	&
	\leq C \int r^{p-1}r^{p_2 - 2} \vert \mathfrak{D}^{k} \Riccit \vert \vert \mathfrak{D}^{k'} \Gamma_{p_2}\vert
	\\
	&
	\leq C \left( \int r^{2p-2} \vert \mathfrak{D}^{k} \Riccit \vert^2 + \int r^{2p_2-4} \vert \mathfrak{D}^{k'} \Gamma_{p_2}\vert^2 \right)
	\\
	&
	\leq C \left( \int r^{2p-2} \vert \mathfrak{D}^{k} \Riccit \vert^2 + \frac{1}{v_0} \right),
\end{align*}
where the last inequality follows from the fact that
\[
	\int r^{2p_2-4} \vert \mathfrak{D}^{k'} \Gamma_{p_2}\vert^2 \leq C \int_{u_0}^u \frac{1}{r(u',v)^2} d u' \leq C(u_f - u_0) \frac{1}{v_0},
\]
since $r \sim v$ in the ``wave zone'' and $v_0$ is large.

The terms arising from the second summation are dealt with similarly using the Sobolev inequality, as are the terms $\Lambda_1 ( \mathfrak{D}^{k} \Riccit + \mathfrak{D}^{k-1} \Riccit)$ and the similar terms arising from lower order errors.

Note that $E_3[\hat{\chibar}]$ contains no ``borderline terms''\footnote{Terms of the form $\rp_{p_1} \cdot \Gamma_{p_2}$ or $\Gamma_{p_1} \cdot \Gamma_{p_2}$ for which $p_1 + p_2 = p$} and so in the above it has been shown, for $k = 0,1,\ldots ,3$,
\begin{multline*}
	r^{2p-2} \int_{S_{u,v}}  \vert \mathfrak{D}^k \hat{\chibar} \vert^2 d \mu_{S_{u,v}} \leq
	C \Big( \int_{u_0}^u \int_{S_{u',v}} r^{2p-2}\vert \mathfrak{D}^k \hat{\chibar} \vert^2 d \mu_{S_{u',v}} d u'  + F^1_{v_0,v}(u_0)
	\\
	 + F^2_{u_0,u}(v_0) + r(u_0,v)^{2p-2} \int_{S_{u_0,v}} \vert \mathfrak{D}^k \hat{\chibar} \vert^2 d \mu_{S_{u_0,v}} + \frac{1}{v_0} \Big),
\end{multline*}
and hence, by the Gr\"{o}nwall inequality
\begin{align*}
	r^{2p-2} \int_{S_{u,v}}  \vert \mathfrak{D}^k \hat{\chibar} \vert^2 d \mu_{S_{u,v}} \leq
	C \Big( r(u_0,v)^{2p-2} \int_{S_{u_0,v}} \vert \mathfrak{D}^k \hat{\chibar} \vert^2 d \mu_{S_{u_0,v}} + F^1_{v_0,v}(u_0)
	 + F^2_{u_0,u}(v_0) + \frac{1}{v_0} \Big).
\end{align*}
This proves the proposition for $\hat{\chibar}$.

The error $E_3[\etabar]$ contains two borderline terms $\rp_1\eta$ and $\hat{\chibar} \cdot \eta$.  The idea is that these terms can be dealt with since the proposition has already been proved for $\hat{\chibar}$ and $\eta$ was controlled in Proposition \ref{prop:Ricci4main}.  Consider first the term $\hat{\chibar} \cdot \eta$.  This will contribute terms of the form $\mathfrak{D}^{k_1} \hat{\chibar} \cdot \mathfrak{D}^{k_2} \eta$ to $E_3[\mathfrak{D}^k \etabar]$, where $k_1 + k_2 = k$.  Assume $k_1 \leq 1$, then
\begin{align*}
	\int r^2 \mathfrak{D}^{k_1} \hat{\chibar} \cdot \mathfrak{D}^{k_2} \eta \cdot \mathfrak{D}^{k} \etabar 
	\leq C\int r \vert \mathfrak{D}^{k_2} \eta \vert \vert \mathfrak{D}^{k} \etabar \vert
	\leq C\left( \int \vert \mathfrak{D}^{k_2} \eta \vert^2 + \int r^2 \vert \mathfrak{D}^{k} \etabar \vert^2 \right),
\end{align*}
and similarly if $k_1 > 1$ then it must be the case that $k_2 \leq 1$ and so 
\[
	\int r^2 \mathfrak{D}^{k_1} \hat{\chibar} \cdot \mathfrak{D}^{k_2} \eta \cdot \mathfrak{D}^{k} \etabar 
	\leq C\left( \int \vert \mathfrak{D}^{k_2} \hat{\chibar} \vert^2 + \int r^2 \vert \mathfrak{D}^{k} \etabar \vert^2 \right).
\]
Repeating this for the terms arising from $\rp_1\eta$, using the bounds already obtained for $\int \vert \mathfrak{D}^{k_2} \hat{\chibar} \vert^2$, Proposition \ref{prop:Ricci4main} and the Gr\"{o}nwall inequality this gives,
\begin{align*}
	r^{2} \int_{S_{u,v}}  \vert \mathfrak{D}^k \etabar \vert^2 d \mu_{S_{u,v}} \leq
	&
	\ C \Bigg( \sum_{k'=0}^k \sum_{\Riccit} r(u_0,v)^{2p-2} \int_{S_{u_0,v}} \vert \mathfrak{D}^{k'} \Riccit \vert^2 d \mu_{S_{u_0,v}}
	\\
	&
	\qquad \qquad + \sum_{k'=0}^k \sum_{\Riccif} r(u,v_0)^{2p-2} \int_{S_{u,v_0}} \vert \mathfrak{D}^{k'} \Riccif \vert^2 d \mu_{S_{u,v_0}}
	\\
	&
	\qquad \qquad + F^1_{v_0,v}(u_0) + F^2_{u_0,u}(v_0) + \frac{1}{v_0} \Bigg).
\end{align*}

The only borderline term in $E_3[\tr \chi - \tr \chi_{\circ}]$ is $(\eta,\etabar)$, the only borderline term in $E_3[\omega]$ is $(\eta,\etabar)$, the only borderline terms in $E_3[\gslash - \gslash^{\circ}]$ are $\left( \frac{1}{\Omega^2} - 1\right)\rp_1$ and $\hat{\chibar}$, and the only borderline term in $E_3[b]$ is $\eta$.  Since either the proposition has already been proved for each of these terms, or they were controlled in Proposition \ref{prop:Ricci4main}, they can be dealt with exactly as before.

\end{proof}

\subsection{Estimates for $\Gammaslash$}

In order to estimate $\Gammaslash - \Gammaslash^{\circ}$, it is first necessary to derive equations which they satisfy.

\begin{proposition}
	The spherical Christoffel symbols satisfy the following propagation equations,
	\begin{align} \label{eq:Gammaslash4}
		\nablaslash_4 \left( \Gammaslash - \Gammaslash^{\circ} \right)^C_{AB} =
		& \
		\nablaslash_A {\chi_B}^C + \nablaslash_B {\chi_A}^C - \nablaslash^C \chi_{BA} 
		\\
		& \
		- \left( {\chi_A}^D - \nablaslash_A b^D + b^E \Gammaslash_{AE}^D \right) \left( \Gammaslash - \Gammaslash^{\circ} \right)_{DB}^C
		\nonumber \\
		& \
		- \left( {\chi_B}^D - \nablaslash_B b^D + b^E \Gammaslash_{BE}^D \right) \left( \Gammaslash - \Gammaslash^{\circ} \right)_{AD}^C
		\nonumber \\
		& \  
		+ \left( {\chi_D}^C - \nablaslash_D b^C + b^E \Gammaslash_{DE}^C \right) \left( \Gammaslash - \Gammaslash^{\circ} \right)_{AB}^D,
		\nonumber \\
		\nablaslash_3 \left( \Gammaslash  - \Gammaslash^{\circ} \right)^C_{AB} =
		& \
		\nablaslash_A {\chibar_B}^C + \nablaslash_B {\chibar_A}^C - \nablaslash^C \chibar_{BA} 
		- {\chibar_A}^D \left( \Gammaslash  - \Gammaslash^{\circ} \right)_{DB}^C 
		\label{eq:Gammaslash3}
		\\
		& \
		- {\chibar_B}^D \left( \Gammaslash  - \Gammaslash^{\circ} \right)_{AD}^C + {\chibar_D}^C \left( \Gammaslash - \Gammaslash^{\circ} \right)_{AB}^D.
		\nonumber
	\end{align}
\end{proposition}

\begin{proof}
	Recall
	\begin{align*}
		\nablaslash_3 \left( \Gammaslash  - \Gammaslash^{\circ} \right)^C_{AB} =
		& \
		e_3 \left( \Gammaslash^C_{AB}  - {\Gammaslash^{\circ}}^C_{AB} \right) 
		- {\chibar_A}^D \left( \Gammaslash  - \Gammaslash^{\circ} \right)_{DB}^C 
		\\
		& \
		- {\chibar_B}^D \left( \Gammaslash  - \Gammaslash^{\circ} \right)_{AD}^C + {\chibar_D}^C \left( \Gammaslash - \Gammaslash^{\circ} \right)_{AB}^D.
	\end{align*}
	The equation in the $e_3$ direction follows from the fact that,
	\[
		e_3 \left({\Gammaslash^{\circ}}^C_{AB} \right) = 0,
	\]
	and,
	\[
		e_3 \left( \Gammaslash^C_{AB} \right) =
		\nablaslash_A {\chibar_B}^C + \nablaslash_B {\chibar_A}^C - \nablaslash^C \chibar_{BA}.
	\]
	See Lemma 4.1 of \cite{Ch}.
	
	The equation in the $e_4$ direction can similarly be derived using the fact that,
	\[
		e_4 \left( \Gammaslash^C_{AB} \right) =
		\nablaslash_A {\chi_B}^C + \nablaslash_B {\chi_A}^C - \nablaslash^C \chi_{BA}.
	\]
\end{proof}

\begin{proposition} \label{prop:Gammaslashmain}
	If $x \in \mathcal{A}$ and $u = u(x)$, $v = v(x)$ then, for $k = 0,1,2,3$, if $\mathfrak{D}^k$ contains $\mathfrak{D} = r\nablaslash$ at most 2 times, then
	\[
		\int_{S_{u,v}} \left\vert \mathfrak{D}^k \left( \Gammaslash - \Gammaslash^{\circ} \right) \right\vert^2 d \mu_{S_{u,v}} \leq C \left( \data + \frac{1}{v_0} \right).
	\]
\end{proposition}
\begin{proof}
	Equation \eqref{eq:Gammaslash3} takes the schematic form,
	\[
		\nablaslash_3 \left( \Gammaslash - \Gammaslash^{\circ} \right) 
		= 
		\sum_{p_1 + p_2 \geq 2} \rp_{p_1} \cdot \mathfrak{D} \Gamma_{p_2} \cdot \left( 1 + \gslash \right)
		+ 
		\sum_{p_1 + p_2 \geq 2} \rp_{p_1} \cdot ( \rp_{p_2} + \Gamma_{p_2}) \cdot r \left( \Gammaslash - \Gammaslash^{\circ} \right).
	\]
	The estimates for $\mathfrak{D}^k \left( \Gammaslash - \Gammaslash^{\circ} \right)$ with $k \leq 2$ then follow exactly as in Proposition \ref{prop:Ricci3main} (in fact these are even easier since there are no borderline terms).  The estimates for $\mathfrak{D}^2 \nablaslash_3 \left( \Gammaslash - \Gammaslash^{\circ} \right)$ follow from applying $\mathfrak{D}^2$ to equation \eqref{eq:Gammaslash3}, and the estimates for $\mathfrak{D}^2 r \nablaslash_4 \left( \Gammaslash - \Gammaslash^{\circ} \right)$ follow from multiplying equation \eqref{eq:Gammaslash4} by $r$ and applying $\mathfrak{D}^2$.
\end{proof}

This recovers the bootstrap assumptions \eqref{eq:Gammaslashba} and the \eqref{eq:Gammaslashbanull} for when $\mathfrak{D}^3\neq (r\nablaslash)^3$.  This remaining case will be recovered in the next section.

\section{Ricci Coefficients at the Top Order} \label{section:Riccitop}
The goal of this section is to estimate $\mathfrak{D}^3 r\nablaslash b$ and $(r\nablaslash)^3 \left( \Gammaslash - \Gammaslash^{\circ} \right)$.  This will recover all of the bootstrap assumptions of Section \ref{section:ba}.  In order to do this, $\mathfrak{D}^3 r\nablaslash \Gamma_p$ must be estimated for most of the other Ricci coefficients $\Gamma_p$.  Recall the set $\mathcal{A}$ from Theorem \ref{thm:main4}.

\subsection{Propagation Equations for Auxiliary $\Theta$ Variables}

Propagation equations are first derived for certain auxiliary quantities.

\begin{proposition}
	The angular derivatives of the null expansions satisfy the following propagation equations.
	\begin{align*}
		\nablaslash_3 r \nablaslash \left( \tr \chibar + \frac{2}{r} \right)
		=
		&
		- \frac{3}{2} \left( \tr \chibar + \frac{2}{r} \right) r\nablaslash \left( \tr \chibar + \frac{2}{r} \right) + \frac{2}{r} r\nablaslash \left( \tr \chibar + \frac{2}{r} \right)
		\\
		&
		- 2 \hat{\chibar} \cdot (r\nablaslash)\hat{\chibar} - \hat{\chibar} \cdot (r \nablaslash) \left( \tr \chibar + \frac{2}{r} \right) - \frac{2}{r} ( \eta + \etabar)
		\\
		&
		- \frac{r}{2} (\eta + \etabar) \left( \tr \chibar + \frac{2}{r} \right)^2 + 2 (\eta + \etabar) \left( \tr \chibar + \frac{2}{r} \right)
		\\
		&
		- r (\eta + \etabar) \vert \hat{\chibar} \vert^2 - r \nablaslash \Tslash_{33} - r (\eta + \etabar) \Tslash_{33},
	\end{align*}
	and
	\begin{align*}
		&
		\nablaslash_4 r \nablaslash \left( \tr \chi - \frac{2}{r} \right) + \tr \chi r \nablaslash \left( \tr \chi - \frac{2}{r} \right)
		=
		\\
		& \qquad \qquad
		- \frac{1}{2} \left( \tr \chi - \frac{2}{r} \right) r \nablaslash \left( \tr \chi - \frac{2}{r} \right) - \hat{\chi} \cdot r \nablaslash \left( \tr \chi - \frac{2}{r} \right) + \omega r \nablaslash \left( \tr \chi - \frac{2}{r} \right)
		\\
		& \qquad \qquad
		+ \left( \tr \chi - \frac{2}{r} \right) r \nablaslash \omega + \frac{2}{r} r \nablaslash \omega - 2 \hat{\chi} \cdot (r\nablaslash) \hat{\chi} -r \nablaslash \Tslash_{44}.
	\end{align*}
\end{proposition}

\begin{proof}
	The proof follows by using Lemma \ref{lemma:commutation} to commute the propagation equations for $\tr \chibar + \frac{2}{r}$ and $\tr \chi - \frac{2}{r}$.  When computing $r\nablaslash \left( 1 - \frac{1}{\Omega^2} \right)$, which arises in the expression for $ r \nablaslash \nablaslash_3 \left( \tr \chibar + \frac{2}{r} \right)$, the fact that $2\nablaslash \log \Omega = (\eta + \etabar)$, and hence
	\[
		\nablaslash \left( 1 - \frac{1}{\Omega^2} \right) = 2 \frac{\nablaslash \Omega}{\Omega^3} = (\eta + \etabar) - \left( 1 - \frac{1}{\Omega^2} \right) (\eta + \etabar),
	\]
	is used.  This means that $\left( 1 - \frac{1}{\Omega^2} \right)$ doesn't appear in the propagation equations as a principal term.
\end{proof}

Define the mass aspect functions,
\[
	\mu = \frac{1}{2}\hat{\chi} \cdot \hat{\chibar} - \rho - \divslash \eta, \qquad \mubar = \frac{1}{2}\hat{\chi} \cdot \hat{\chibar} - \rho - \divslash \etabar,
\]
and the $S_{u,v}$ 1-form,
\[
	\kappa = \nablaslash \omega + {}^* \nablaslash \omegadag + \beta.
\]
Here $\omegadag$ is defined to be the solution to
\[
	\nablaslash_3 \omegadag = - \rho,
\]
with zero initial data on $\{ u = u_0\}$.

\begin{proposition}
	The mass aspect functions and $\kappa$ satisfy the following propagation equations,
	\begin{align*}
		\nablaslash_4 \mu
		=
		&
		- \tr \chi \mu + \frac{1}{2} \mubar + (\eta - \etabar) \cdot \nablaslash \tr \chi + 2 \hat{\chibar} \cdot \nablaslash \eta
		\\
		&
		- \frac{1}{2} \tr \chi \hat{\chi} \cdot \hat{\chibar} + \tr \chi \rho - \tr \chi \eta \cdot \etabar - \frac{1}{4} \tr \chi \vert \hat{\chi} \vert^2 + 2 \hat{\chi} \cdot (\eta \otimes \etabar) - 2 \beta \cdot \eta + \frac{1}{2} \tr \chi \vert \etabar \vert^2
		\\
		&
		+ \frac{1}{2} \Tslash_4 \cdot \eta - \frac{1}{2} \Tslash_4 \cdot \etabar + \frac{1}{2} \hat{\chi} \cdot \Tslash + \frac{1}{4} \omega \Tslash_{34} + \frac{1}{2} \divslash \Tslash_4 + \frac{1}{4} \nablaslash_3 \Tslash_{44} - \frac{1}{4} \nablaslash_4 \Tslash_{34},
		\\
		\nablaslash_3 \mubar
		=
		&
		- \tr \chibar \mubar + \frac{1}{2} \tr \chibar \mu + (\etabar - \eta) \cdot \nablaslash \tr \chibar + 2 \hat{\chibar} \cdot \nablaslash \etabar
		+ 2 \etabar \cdot \betabar + \tr \chibar \rho
		\\
		&
		- \tr \chibar \hat{\chi} \cdot \hat{\chibar}  + \tr \chibar ( \vert \etabar \vert^2 - \frac{1}{2} \vert \eta \vert^2 - \eta \cdot \etabar) + 2 \hat{\chibar} \cdot (\eta \otimes \etabar) - \frac{1}{4} \tr \chi \hat{\chi} \cdot \hat{\chibar} 
		\\
		&
		+ \frac{1}{2} \hat{\chibar} \cdot \Tslash + \frac{1}{2} \omega \Tslash_{33} + \frac{1}{2} \etabar \cdot \Tslash_3 + \frac{1}{2} \eta \cdot \Tslash_3 + \frac{1}{4} \nablaslash_4 \Tslash_{33} - \frac{1}{4} \nablaslash_3 \Tslash_{34} + \frac{1}{2} \divslash \Tslash_3,
		\\
		\nablaslash_3 \kappa
		=
		&
		\ 2 \nablaslash (\eta \cdot \etabar) - \nablaslash( \vert \etabar \vert^2)  - \frac{1}{2} \tr \chibar \kappa  - \hat{\chibar} \cdot \nablaslash \omega - {}^* \hat{\chibar} \cdot \nablaslash \omegadag + \frac{1}{2} \tr \chibar \beta
		\\
		&
		+ ( 2 \eta \cdot \etabar - \vert \etabar \vert^2)(\eta + \etabar) + 2 \hat{\chi} \cdot \betabar - \rho \etabar + 2\rho \eta - {}^* \etabar \sigma + 2 \sigma {}^* \eta - \frac{1}{2} \hat{\chibar}\cdot \Tslash_4
		\\
		&
		 - \frac{1}{4} \tr \chibar \Tslash_4 - \frac{1}{2} \hat{\chi} \cdot \Tslash_3 - \frac{1}{4} \tr \chi \Tslash_3 - \frac{1}{2} \Tslash_{34} \eta - \Tslash_{34} \etabar - \frac{1}{2} \omega \Tslash_3 + \etabar \cdot \Tslash - \frac{1}{2} \nablaslash_4 \Tslash_3.
	\end{align*}
\end{proposition}

\begin{proof}
  	From the definition of $\mu$,
	\[
		\nablaslash_4 \mu 
		= 
		\frac{1}{2} \left( \nablaslash_4 \hat{\chi} \right) \cdot \hat{\chibar} + \frac{1}{2} \hat{\chi} \cdot \left( \nablaslash_4 \hat{\chibar} \right)
		-
		\divslash \nablaslash_4 \eta
		-
		[\nablaslash_4, \divslash ] \eta
		-
		\nablaslash_4 \rho. 
	\]
	The equation is obtained by substituting on the right hand side the null structure equations for $\nablaslash_4 \hat{\chi}$ and $\nablaslash_4 \eta$, equation \eqref{eq:chibarhat4}, the Bianchi equation for $\nablaslash_4 \rho$ and using Lemma \ref{lemma:commutation} to compute the commutator term.  The Codazzi equation \eqref{eq:Codazziout} is also used to replace the $\divslash \hat{\chi}$ term arising from $\divslash \nablaslash_4 \eta$.
	
	The equation for $\nablaslash_3 \mubar$ is obtained similarly using the null structure equations for $\nablaslash_3 \hat{\chibar}$ and $\nablaslash_3 \etabar$, equation \eqref{eq:chihat3}, the Bianchi equation for $\nablaslash_3 \rho$ and the Codazzi equation \eqref{eq:Codazziin}.
	
	Finally,
	\[
		\nablaslash_3 \kappa 
		= 
		\nablaslash \nablaslash_3 \omega 
		+ 
		[\nablaslash_3,\nablaslash] \omega 
		+ 
		{}^* \nablaslash \nablaslash_3 \omegadag
		+
		{}^* \left( [\nablaslash_3,\nablaslash] \omegadag \right)
		+
		\nablaslash_3 \beta,
	\]
	and the equation for $\kappa$ can be computed similarly.
\end{proof}

\begin{proposition}
	If $x \in \mathcal{A}$ and $u = u(x)$, $v = v(x)$ then, for $k=0,\ldots,4$, $\omegadag$ satisfies,
	\[
		r^4 \int_{S_{u,v}} \vert \mathfrak{D}^k \omegadag \vert^2 d \mu_{S_{u,v}} \leq C\left( \bc + \frac{1}{v_0} \right).
	\]
\end{proposition}

\begin{proof}
	Since $\omegadag$ satisfies an equation of the form
	\[
		\nablaslash_3 \omegadag = \psi_3,
	\]
	with zero initial data, this can be proved in exactly the same way as the estimates for $\mathfrak{D}^k \Riccit$ in Proposition \ref{prop:Ricci3main}.
\end{proof}

Let $\Theta$ schematically denote the following quantities,
\[
	\Theta = r\nablaslash \left( \tr \chi - \frac{2}{r} \right), r\nablaslash \left( \tr \chibar + \frac{2}{r} \right), \mu, \mubar, \kappa ,
\]
and further decompose as
\begin{equation*}
\begin{array}{ll}
	\Thetat_2 = r\nablaslash \left( \tr \chibar + \frac{2}{r} \right),
	&
	\Thetaf_2 = r\nablaslash \left( \tr \chi - \frac{2}{r} \right), \mu ,
	\\
	\Thetat_3 = \mubar,
	&
	\Thetat_4 = \kappa.
\end{array}
\end{equation*}

As with the $\Gamma_p, \psi_p, \mathcal{T}_p$, the subscript $p$ indicates that $\Theta_p$ should decay like $\frac{1}{r^p}$.  Similarly, the $(3)$ indicates that $\Thetat$ satisfies an equation in the 3 direction, and the $(4)$ indicates that $\Thetaf$ satisfies an equation in the 4 direction.

The propagation equations for the $\Thetaf$ variables take the following schematic form,
\[
	\nablaslash_4 \Thetaf_p + \frac{p}{2} \tr \chi \Thetaf_p = E_4 \left[ \Thetaf_p \right]
\]
and the for the $\Thetat$ variables take the form,
\[
	\nablaslash_3 \Thetat_p = E_3 \left[ \Thetat_p \right],
\]
where,
\[
	E_4 \left[ r\nablaslash \left( \tr \chi - \frac{2}{r} \right) \right] = \frac{2}{r} r \nablaslash \omega + \left( \tr \chi - \frac{2}{r} \right) r \nablaslash \omega + \sum_{p_1 + p_2 \geq 4} \Gamma_{p_1} \cdot r \nablaslash \Gamma_{p_2} + \sum_{p_1 \geq 6} \mathfrak{D} \mathcal{T}_{p_1},
\]
\begin{align*}
	E_4[\mu] =
	& 
	\frac{1}{r} \mubar + \frac{1}{2} \left( \tr \chi - \frac{2}{r} \right) \mubar + \sum_{p_1+p_2+p_3 \geq 4} \rp_{p_1} \cdot \Gamma_{p_2} \cdot r \nablaslash \Gamma_{p_3}
	\\
	&
	+ \sum_{p_1+p_2+p_3 \geq 4} (\rp_{p_1} + \Gamma_{p_1}) \cdot \Gamma_{p_2} \cdot \Gamma_{p_3} + \sum_{p_1+p_2\geq 4} (\rp_{p_1} + \Gamma_{p_1}) \cdot \psi_{p_2}
	\\
	&
	+ \sum_{p_1+p_2 \geq 6} \Gamma_{p_1} \cdot \mathcal{T}_{p_2} + \sum_{p_1+p_2\geq 5} \rp_{p_1} \cdot \mathfrak{D} \mathcal{T}_{p_2},
\end{align*}
and
\begin{align*}
	E_3 \left[ r\nablaslash \left( \tr \chibar + \frac{2}{r} \right) \right] =
	&
	\frac{2}{r} r\nablaslash \left( \tr \chibar + \frac{2}{r} \right) + \sum_{p_1+p_2\geq 2} \Gamma_{p_1} \cdot r \nablaslash \Gamma_{p_2} + \sum_{p_1+p_2\geq 2} \rp_{p_1} \cdot \Gamma_{p_2}
	\\
	&
	+ \sum_{p_1+p_2+p_3+p_4 \geq 2} \rp_{p_1} \cdot \Gamma_{p_2} \cdot \Gamma_{p_3} \cdot \Gamma_{p_4} 
	\\
	&
	+ \sum_{p_1+p_2+p_3 \geq 2} \rp_{p_1} \cdot \Gamma_{p_2} \cdot \left( \Gamma_{p_3} + \mathcal{T}_{p_3} \right) + \sum_{p_1 \geq 2} \mathfrak{D} \mathcal{T}_{p_1},
\end{align*}
\begin{align*}
	E_3 [\mubar]
	=
	&
	\frac{2}{r} \mubar - \frac{1}{r} \mu - \left( \tr \chibar + \frac{2}{r} \right) \mubar + \frac{1}{2} \left( \tr \chi - \frac{2}{r} \right) \mu
	\\
	&
	+ \sum_{p_1+p_2+p_3\geq 4} (\rp_{p_1} + \Gamma_{p_1} ) \cdot \Gamma_{p_2} \cdot \Gamma_{p_3} + \sum_{p_1+p_2\geq 4} (\rp_{p_1} + \Gamma_{p_1}) \cdot \psi_{p_2}
	\\
	&
	+ \sum_{p_1+p_2\geq 4} \Gamma_{p_1} \cdot \mathcal{T}_{p_2} + \sum_{p_1+p_2\geq 3} \rp_{p_1} \cdot \mathfrak{D} \mathcal{T}_{p_2},
\end{align*}
\begin{align*}
	E_3[\kappa]
	=
	&
	\frac{1}{r} \kappa - \frac{1}{2} \left( \tr \chibar + \frac{2}{r} \right) \kappa - \hat{\chibar} \cdot \nablaslash \omega - {}^* \hat{\chibar} \cdot \nablaslash \omegadag + \sum_{p_1+p_2+p_3 \geq 4} \rp_{p_1} \cdot \Gamma_{p_2} \cdot r \nablaslash \Gamma_{p_3}
	\\
	&
	+ \sum_{p_1+p_2+p_3\geq 4} \Gamma_{p_1} \cdot \Gamma_{p_2} \cdot \Gamma_{p_3} + \sum_{p_1+p_2\geq 4} (\rp_{p_1} + \Gamma_{p_1}) \cdot (\psi_{p_2} + \mathcal{T}_{p_2} )
	\\
	&
	+ \sum_{p_1+p_2\geq 4} \rp_{p_1} \cdot \mathfrak{D} \mathcal{T}_{p_2}.
\end{align*}
All of the $\Gamma$ appearing in the $r \nablaslash \Gamma$ terms in the errors, unless explicitly stated, are $\hat{\chi}, \hat{\chibar}, \eta, \etabar, \tr \chi - \frac{2}{r}, \tr \chibar + \frac{2}{r}$ and hence the bootstrap assumptions of Section \ref{section:ba} give an estimate for $\mathfrak{D}^2 r \nablaslash \Gamma$ in $L^2$ on the spheres.  It is the \emph{linear} principal terms which will require the most care below.  When such terms appear, they have been written first in the errors above.  Linear here means linear in $\Gamma, \psi, \mathcal{T}, \Theta$, so one example of a linear term is $\frac{2}{r} r \nablaslash \omega$ appearing in $E_4 \left[ r\nablaslash \left( \tr \chi - \frac{2}{r} \right) \right]$.  Principal means of the form $r \nablaslash \Gamma$ or $\Theta$, since the $\Theta$ variables live at one degree of differentiability greater than $\Gamma$.  The principal energy momentum tensor terms, $\mathfrak{D} \mathcal{T}$, will not be problematic as they have all already been estimated at the top order.  Note that there are no principal curvature terms, i.\@e.\@ terms of the form $\mathfrak{D} \psi$, appearing in the errors.  Finally, notice that the propagation equations have the same structure as the propagation equations for the Ricci coefficients highlighted in Section \ref{section:equations}, i.\@e.\@ the error terms $E_3 \left[ \Thetat_p \right]$ should decay like $\frac{1}{r^p}$ and the error terms $E_4 \left[ \Thetaf_p \right]$ should decay like $\frac{1}{r^{p+2}}$.  The next proposition, akin to Proposition \ref{prop:commutednullstructure}, says this structure is preserved under commutation by $\mathcal{D}$.  Unlike Proposition \ref{prop:commutednullstructure}, we here keep track of the principal terms.

\begin{proposition}
	The commuted propagation equations for the $\Theta$ variables, for $k=1,2, \ldots$, take the form,
	\[
			\nablaslash_3 \mathfrak{D}^k \Thetat_p
			= E_3 \left[ \mathfrak{D}^k \Thetat_p \right],
	\]
	and
	\[
			\nablaslash_4 \mathfrak{D}^k \Thetaf_p + \frac{p}{2} \tr \chi \mathfrak{D}^k \Thetaf_p
			= E_4 \left[ \mathfrak{D}^k \Thetaf_p \right],
	\]
	where
	\begin{align*}
		E_3 \left[ \mathfrak{D}^k \Thetat_p \right] 
		= 
		&
		\mathfrak{D} E_3 \left[ \mathfrak{D}^{k-1} \Thetat_p \right] + \left( \rp_1 + \sum_{p_1+p_2 \geq 1} \rp_{p_1} \cdot \Gamma_{p_2}  \right) \cdot \mathfrak{D}^k \Thetat_p
		\\
		&
		+ \Lambda_1\cdot \mathfrak{D}^{k-1} \Thetat_p,
	\end{align*}
	and
	\[
		E_4 \left[ \mathfrak{D}^k \Thetaf_p \right] = \mathfrak{D} E_4 \left[ \mathfrak{D}^{k-1} \Thetaf_p \right] + \left( \sum_{p_1+p_2 \geq 2} \rp_{p_1} \cdot \Gamma_{p_2} \right) \cdot \mathfrak{D}^k \Thetaf_p + \Lambda_2' \cdot \mathfrak{D}^{k-1} \Thetaf_p.
	\]
	Moreover,
	\begin{align*}
		E_3 \left[ \mathfrak{D}^{k-1} r \nablaslash \Thetat_p \right] 
		= 
		&
		\mathfrak{D} E_3 \left[ \mathfrak{D}^{k-2} r\nablaslash \Thetat_p \right] + \left( \rp_1 + \sum_{p_1+p_2 \geq 1} \rp_{p_1} \cdot \Gamma_{p_2}  \right) \cdot \mathfrak{D}^{k-1} r \nablaslash \Thetat_p 
		\\
		&
		+ \Lambda_1\cdot \mathfrak{D}^{k-2} r \nablaslash \Thetat_p,
	\end{align*}
	and
	\begin{align*}
		E_4 \left[ \mathfrak{D}^{k-1} r \nablaslash \Thetaf_p \right]
		=
		&
		\mathfrak{D} E_4 \left[ \mathfrak{D}^{k-2} r \nablaslash \Thetaf_p \right] + \left( \sum_{p_1+p_2 \geq 2} \rp_{p_1} \cdot \Gamma_{p_2} \right) \cdot \mathfrak{D}^{k-1} r \nablaslash \Thetaf_p 
		\\
		&
		+ \Lambda_2' \cdot \mathfrak{D}^{k-2} r \nablaslash \Thetaf_p.
	\end{align*}
	Recall that $\Lambda_1, \Lambda_2'$ from Proposition \ref{prop:commutednullstructure}, where the $'$ stresses that $\Lambda_2'$ contains terms of the form $\mathfrak{D} \Gamma$, involving one derivative of $\Gamma$.
\end{proposition}

\begin{proof}
	The proof is identical to that of Proposition \ref{prop:commutednullstructure}, except we keep track of the principal terms.
\end{proof}

Note that the ``moreover'' part of the Proposition says that commuting the propagation equations with $r \nablaslash$ only produces principal error terms involving an $r \nablaslash$ derivative, unlike commuting with $\nablaslash_3$ and $r\nablaslash_4$ which can produce principal error terms involving $\nablaslash_3, r \nablaslash_4$ and $r \nablaslash$ derivatives.  This is important when estimating $\kappa$ since we only estimate $\mathfrak{D}^2 r \nablaslash \kappa$ rather than $\mathfrak{D}^3 \kappa$.

\subsection{Additional Bootstrap Assumptions}
The results of this section will be shown using an additional bootstrap argument.  Let $\mathcal{A}'' \subset \mathcal{A}$ denote the set of $x \in \mathcal{A}$ such that the following additional bootstrap assumptions hold for all $y \in \mathcal{A}$ with $u(y) \leq u(x)$, $v(y) \leq v(x)$,
\begin{align}
	\int_{u_0}^u \int_{S_{u',v}} r^2 \left\vert \mathfrak{D}^k r\nablaslash \left( \tr \chi - \frac{2}{r} \right) \right\vert^2 d \mu_{S_{u',v}} d u'
	&
	\leq \overline{C}, \label{eq:thetabcfirst}
	\\
	\int_{v_0}^v \int_{S_{u,v'}} \left\vert \mathfrak{D}^k r\nablaslash \left( \tr \chibar + \frac{2}{r} \right) \right\vert^2 d \mu_{S_{u,v'}} d v'
	&
	\leq \overline{C},
	\\
	\int_{u_0}^u \int_{S_{u',v}} r^2 \left\vert \mathfrak{D}^k \mu \right\vert^2 d \mu_{S_{u',v}} d u'
	&
	\leq \overline{C},
	\\
	\int_{v_0}^v \int_{S_{u,v'}} r^2 \left\vert \mathfrak{D}^k \mubar \right\vert^2 d \mu_{S_{u,v'}} d v'
	&
	\leq \overline{C},
	\\
	\int_{u_0}^u \int_{S_{u',v}} r^2 \left\vert \mathfrak{D}^k r\nablaslash \hat{\chi} \right\vert^2 d \mu_{S_{u',v}} d u'
	&
	\leq \overline{C},
	\\
	\int_{v_0}^v \int_{S_{u,v'}} r^{-2} \left\vert \mathfrak{D}^k r\nablaslash \hat{\chibar} \right\vert^2 d \mu_{S_{u,v'}} d v'
	&
	\leq \overline{C},
	\\
	\int_{u_0}^u \int_{S_{u',v}} \left\vert \mathfrak{D}^k r\nablaslash \eta \right\vert^2 d \mu_{S_{u',v}} d u'
	&
	\leq \overline{C},
	\\
	\int_{v_0}^v \int_{S_{u,v'}} \left\vert \mathfrak{D}^k r\nablaslash \etabar \right\vert^2 d \mu_{S_{u,v'}} d v'
	&
	\leq \overline{C},
	\\
\end{align}
for $k = 0,1,2,3$, and,
\begin{align}
	\int_{v_0}^v \int_{S_{u,v'}} r^4 \left\vert \mathfrak{D}^k r \nablaslash \kappa \right\vert^2 d \mu_{S_{u,v'}} d v'
	&
	\leq \overline{C},
	\\
	\int_{v_0}^v \int_{S_{u,v'}} r^2 \left\vert \mathfrak{D}^k (r\nablaslash)^2 \omega \right\vert^2 d \mu_{S_{u,v'}} d v'
	&
	\leq \overline{C},
	\\
	\int_{v_0}^v \int_{S_{u,v'}} r^2 \left\vert \mathfrak{D}^k \nablaslash_3 r\nablaslash \omega \right\vert^2 d \mu_{S_{u,v'}} d v'
	&
	\leq \overline{C},
	\\
	\int_{v_0}^v \int_{S_{u,v'}} r^2 \left\vert \mathfrak{D}^k (r\nablaslash)^2 \omegadag \right\vert^2 d \mu_{S_{u,v'}} d v'
	&
	\leq \overline{C}, \label{eq:thetabclast}
\end{align}
for $k=0,1,2$, where $u = u(y)$, $v = v(y)$.

\subsection{Estimates for Auxiliary $\Theta$ Variables}
The bootstrap assumptions \eqref{eq:thetabcfirst}--\eqref{eq:thetabclast} can now be used to obtain estimates for the $\Theta$ variables.

\begin{proposition} \label{prop:Thetamain}
	For any $x \in \mathcal{A}''$, if $u = u(x)$, $v = v(y)$ then, for $k = 0,1,2,3$, for all $\Thetaf$ and $\Thetat \neq \kappa$,
	\begin{equation} \label{eq:Thetatbound}
		\int_{v_0}^v \int_{S_{u,v'}} r^{2p-4} \vert \mathfrak{D}^k \Thetat_p \vert^2 d\mu_{S_{u,v'}} dv' \leq C \left( \data + \frac{1}{v_0} \right),
	\end{equation}
	\begin{equation} \label{eq:Thetafbound}
		\int_{u_0}^u \int_{S_{u',v}} r^{2p-2} \vert \mathfrak{D}^k \Thetaf_p \vert^2 d\mu_{S_{u',v}} du' \leq C \left( \data + \frac{1}{\sqrt{v_0}} \right),
	\end{equation}
	and, for $k=0,1,2$,
	\begin{equation} \label{eq:kappabound}
		\int_{v_0}^v \int_{S_{u,v'}} r^{4} \vert \mathfrak{D}^k r \nablaslash \kappa \vert^2 d\mu_{S_{u,v'}} dv' \leq C \left( \data + \frac{1}{v_0} \right).
	\end{equation}
\end{proposition}

\begin{proof}
	For $k \leq 2$, bounds for
	\[
		\int_{S_{u,v}} r^{2p-2} \left\vert \Thetat_p \right\vert^2 d \mu_{S_{u,v}}, \qquad \text{and} \qquad \int_{S_{u,v}} r^{2p-2} \left\vert \Thetaf_p \right\vert^2 d \mu_{S_{u,v}},
	\]
	can be obtained exactly as in Propositions \ref{prop:Ricci3main} and \ref{prop:Ricci4main}, then integrated to give \eqref{eq:Thetatbound} and \eqref{eq:Thetafbound}.  For $k\leq 1$, \eqref{eq:kappabound} can be obtained similarly.
	
	The new difficulties are at the top order, so assume now $k=3$ and consider $\Thetat \neq \kappa$.  Note that the bootstrap assumptions \eqref{eq:thetabcfirst}--\eqref{eq:thetabclast} together with the Sobolev inequalities of Section \ref{section:Sobolev} give the pointwise bounds\footnote{For $\Theta_p = r \nablaslash \left( \tr \chi - \frac{2}{r} \right), \mu$, actually have the pointwise bounds $r^{p}\vert \mathfrak{D} \Theta_p \vert \leq C$.  There is a loss of $r^{\frac{1}{2}}$ for the other variables since the Sobolev inequality on the outgoing null hypersurfaces have to be used.}
	\[
		r^p\vert \Theta_p \vert, r^{p-\frac{1}{2}}\vert \mathfrak{D} \Theta_p \vert \leq C,
	\]
	for $\Theta_p \neq \kappa$, and
	\[
		r^4\vert \kappa \vert, r^{\frac{7}{2}} \vert r \nablaslash \kappa \vert \leq C.
	\]
	Equation \eqref{eq:transportid3} with $h = r^{2p-4} \Thetat_p$ gives,
	\begin{align*}
		&
		\int_{S_{u,v'}} r^{2p-4} \left\vert \mathfrak{D}^3 \Thetat_p \right\vert^2 d \mu_{S_{u,v'}}
		=
		\int_{S_{u_0,v'}} r(u_0,v')^{2p-4} \left\vert \mathfrak{D}^3 \Thetat_p \right\vert^2 d \mu_{S_{u_0,v'}}
		\\
		& \qquad \qquad
		+ \int_{u_0}^u \int_{S_{u',v'}} r^{2p-4} \mathfrak{D}^3 \Thetat_p \cdot E_3 \left[ \mathfrak{D}^3 \Thetat_p \right] d \mu_{S_{u',v'}} du',
	\end{align*}
	where $\tr \chibar \mathfrak{D}^3 \Thetat_p = ( \Gamma_2 + \rp_1) \cdot \mathfrak{D}^3 \Thetat_p$ has been absorbed into the error $E_3 \left[ \mathfrak{D}^3 \Thetat_p \right]$.  Integrating in $v'$, this gives,
	\begin{align*}
		&
		\int_{v_0}^v \int_{S_{u,v'}} r^{2p-4} \left\vert \mathfrak{D}^3 \Thetat_p \right\vert^2 d \mu_{S_{u,v'}} dv'
		=
		\int_{v_0}^v \int_{S_{u_0,v'}} r(u_0,v')^{2p-4} \left\vert \mathfrak{D}^3 \Thetat_p \right\vert^2 d \mu_{S_{u_0,v'}} dv'
		\\
		& \qquad \qquad
		+ \int_{v_0}^v \int_{u_0}^u \int_{S_{u',v'}} r^{2p-4} \mathfrak{D}^3 \Thetat_p \cdot E_3 \left[ \mathfrak{D}^3 \Thetat_p \right] d \mu_{S_{u',v'}} du' dv'.
	\end{align*}
	Recall,
	\[
		\int_{v_0}^v \int_{S_{u_0,v'}} r(u_0,v')^{2p-4} \left\vert \mathfrak{D}^3 \Thetat_p \right\vert^2 d \mu_{S_{u_0,v'}} dv' \leq \data.
	\]
	
	It remains to estimate the error terms.  Consider first the quadratic terms
	\[
		\sum_{p_1+p_2 \geq p} \Gamma_{p_1} \cdot r \nablaslash \Gamma_{p_2},
	\]
	in $E_3 \left[ \Thetat_p \right]$.  They will give rise to terms in $E_3 \left[ \mathfrak{D}^3 \Thetat_p \right]$ of the form\footnote{There will also be cubic and higher order terms but, since they have the same $r$ decay and there will always be at most one factor which is not lower order and hence cannot be estimated in $L^{\infty}$, they are treated in exactly the same manner as the quadratic terms.}
	\[
		\mathfrak{D}^{k_1} \Gamma_{p_1} \cdot \mathfrak{D}^{k_2} r \nablaslash \Gamma_{p_2},
	\]
	where $p_1 + p_2 \geq p$ and $k_1+k_2 = 3$.  It must be the case that either $k_1\leq 1$ or $k_2 \leq 1$.  Assume first that $k_1 \leq 1$.  Then, by Propositions \ref{prop:Ricci4main} and \ref{prop:Ricci3main} and the Sobolev inequalities of Section \ref{section:Sobolev},
	\begin{align*}
		r^{2p-4} \left\vert \mathfrak{D}^{k_1} \Gamma_{p_1} \cdot \mathfrak{D}^{k_2} r \nablaslash \Gamma_{p_2} \right\vert^2
		\leq 
		&
		C r^{2p-4} r^{-2p_1} \left( \data + \frac{1}{v_0} \right) \left\vert \mathfrak{D}^{k_2} r \nablaslash \Gamma_{p_2} \right\vert^2
		\\
		\leq
		&
		C \left( \data + \frac{1}{v_0} \right) r^{2p_2 - 4} \left\vert \mathfrak{D}^{k_2} r \nablaslash \Gamma_{p_2} \right\vert^2,
	\end{align*}
	since $p_1+p_2 \geq p$.  Recall that such terms only occur for $\Gamma_{p_2} = \tr\chi - \frac{2}{r}, \tr \chibar + \frac{2}{r}, \hat{\chi}, \hat{\chibar}, \eta, \etabar$.  If $\Gamma_{p_2} = \tr \chibar + \frac{2}{r}, \hat{\chibar}, \etabar$ then, using the bootstrap assumptions \eqref{eq:thetabcfirst}--\eqref{eq:thetabclast},
	\begin{align*}
		&
		\int_{v_0}^v \int_{u_0}^u \int_{S_{u,v'}} r^{2p-4} \mathfrak{D}^3 \Thetat_p \cdot \mathfrak{D}^{k_1} \Gamma_{p_1} \cdot \mathfrak{D}^{k_2} r \nablaslash \Gamma_{p_2} d \mu_{S_{u',v'}} du' dv'
		\\
		& \qquad
		\leq
		\int_{v_0}^v \int_{u_0}^u \int_{S_{u,v'}} r^{2p-4} \left( \left\vert \mathfrak{D}^3 \Thetat_p \right\vert^2 + \left\vert \mathfrak{D}^{k_1} \Gamma_{p_1} \cdot \mathfrak{D}^{k_2} r \nablaslash \Gamma_{p_2} \right\vert^2 \right) d \mu_{S_{u',v'}} du' dv'
		\\
		& \qquad
		\leq
		\int_{v_0}^v \int_{u_0}^u \int_{S_{u,v'}} r^{2p-4} \left\vert \mathfrak{D}^3 \Thetat_p \right\vert^2 d \mu_{S_{u',v'}} du' dv'
		\\
		& \qquad \qquad
		+
		C\left( \data + \frac{1}{v_0} \right) \int_{v_0}^v \int_{u_0}^u \int_{S_{u,v'}} r^{2p_2-4} \left\vert \mathfrak{D}^{k_2} r \nablaslash \Gamma_{p_2} \right\vert^2 d \mu_{S_{u',v'}} du' dv'
		\\
		& \qquad
		\leq
		\int_{v_0}^v \int_{u_0}^u \int_{S_{u,v'}} r^{2p-4} \left\vert \mathfrak{D}^3 \Thetat_p \right\vert^2 d \mu_{S_{u',v'}} du' dv' +
		C\left( \data + \frac{1}{v_0} \right) \int_{u_0}^u  du'
		\\
		& \qquad
		\leq
		\int_{v_0}^v \int_{u_0}^u \int_{S_{u,v'}} r^{2p-4} \left\vert \mathfrak{D}^3 \Thetat_p \right\vert^2 d \mu_{S_{u',v'}} du' dv' +
		C\left( \data + \frac{1}{v_0} \right).
	\end{align*}
	Similarly, if $\Gamma_{p_2} = \tr\chi - \frac{2}{r}, \hat{\chi}, \eta$, then
	\begin{align*}
		&
		\int_{v_0}^v \int_{u_0}^u \int_{S_{u,v'}} r^{2p-4} \mathfrak{D}^3 \Thetat_p \cdot \mathfrak{D}^{k_1} \Gamma_{p_1} \cdot \mathfrak{D}^{k_2} r \nablaslash \Gamma_{p_2} d \mu_{S_{u',v'}} du' dv'
		\\
		& \qquad
		\leq
		\int_{v_0}^v \int_{u_0}^u \int_{S_{u,v'}} r^{2p-4} \left\vert \mathfrak{D}^3 \Thetat_p \right\vert^2 d \mu_{S_{u',v'}} du' dv'
		\\
		& \qquad \qquad
		+
		C\left( \data + \frac{1}{v_0} \right) \int_{v_0}^v \int_{u_0}^u \int_{S_{u,v'}} r^{2p_2-4} \left\vert \mathfrak{D}^{k_2} r \nablaslash \Gamma_{p_2} \right\vert^2 d \mu_{S_{u',v'}} du' dv'
		\\
		& \qquad
		\leq
		\int_{v_0}^v \int_{u_0}^u \int_{S_{u,v'}} r^{2p-4} \left\vert \mathfrak{D}^3 \Thetat_p \right\vert^2 d \mu_{S_{u',v'}} du' dv' +
		C\left( \data + \frac{1}{v_0} \right) \int_{v_0}^v \frac{1}{(v')^2} dv'
		\\
		& \qquad
		\leq
		\int_{v_0}^v \int_{u_0}^u \int_{S_{u,v'}} r^{2p-4} \left\vert \mathfrak{D}^3 \Thetat_p \right\vert^2 d \mu_{S_{u',v'}} du' dv' +
		C\left( \data + \frac{1}{v_0} \right).
	\end{align*}
	Suppose now that $k_1 \geq 2$.  Then $k_2 \leq 1$ and so, since $\Gamma_{p_2} = \tr\chi - \frac{2}{r}, \tr \chibar + \frac{2}{r}, \hat{\chi}, \hat{\chibar}, \eta, \etabar$, the bootstrap assumptions \eqref{eq:thetabcfirst}--\eqref{eq:thetabclast} and the Sobolev inequality imply that\footnote{Similarly to the $\Theta$ variables, for $\Gamma_{p_2} = \hat{\chibar}, \tr \chibar + \frac{2}{r}, \etabar$ actually get $r^{p_2} \left\vert \mathfrak{D}^{k_2} r \nablaslash \Gamma_{p_2} \right\vert \leq C$ but for $\hat{\chi}, \tr \chi - \frac{2}{r}, \eta$, have to use the Sobolev inequality on the outgoing null hypersurfaces and hence lose the power of $r^{\frac{1}{2}}$.}
	\[
		r^{p_2 - \frac{1}{2}} \left\vert \mathfrak{D}^{k_2} r \nablaslash \Gamma_{p_2} \right\vert \leq C.
	\]
	Hence,
	\begin{align*}
		&
		\int_{v_0}^v \int_{u_0}^u \int_{S_{u,v'}} r^{2p-4} \mathfrak{D}^3 \Thetat_p \cdot \mathfrak{D}^{k_1} \Gamma_{p_1} \cdot \mathfrak{D}^{k_2} r \nablaslash \Gamma_{p_2} d \mu_{S_{u',v'}} du' dv'
		\\
		& \qquad
		\leq
		\int_{v_0}^v \int_{u_0}^u \int_{S_{u,v'}} r^{2p-4} \left( \left\vert \mathfrak{D}^3 \Thetat_p \right\vert^2 + \left\vert \mathfrak{D}^{k_1} \Gamma_{p_1} \cdot \mathfrak{D}^{k_2} r \nablaslash \Gamma_{p_2} \right\vert^2 \right) d \mu_{S_{u',v'}} du' dv'
		\\
		& \qquad
		\leq
		\int_{v_0}^v \int_{u_0}^u \int_{S_{u,v'}} r^{2p-4} \left\vert \mathfrak{D}^3 \Thetat_p \right\vert^2 d \mu_{S_{u',v'}} du' dv'
		\\
		& \qquad \qquad
		+
		C \int_{v_0}^v \int_{u_0}^u \int_{S_{u,v'}} r^{2p_1-\frac{7}{2}} \left\vert \mathfrak{D}^{k_1} \Gamma_{p_1} \right\vert^2 d \mu_{S_{u',v'}} du' dv'
		\\
		& \qquad
		\leq
		\int_{v_0}^v \int_{u_0}^u \int_{S_{u,v'}} r^{2p-4} \left\vert \mathfrak{D}^3 \Thetat_p \right\vert^2 d \mu_{S_{u',v'}} du' dv' +
		C\left( \data + \frac{1}{v_0} \right) \int_{v_0}^v \int_{u_0}^u \frac{1}{r^{\frac{3}{2}}} du' dv'
		\\
		& \qquad
		\leq
		\int_{v_0}^v \int_{u_0}^u \int_{S_{u,v'}} r^{2p-4} \left\vert \mathfrak{D}^3 \Thetat_p \right\vert^2 d \mu_{S_{u',v'}} du' dv' +
		C\left( \data + \frac{1}{v_0} \right),
	\end{align*}
	by Propositions \ref{prop:Ricci4main}, \ref{prop:Ricci3main}, since $k_1\leq 3$.
	
	The quadratic terms arising from $\left( \tr \chibar + \frac{2}{r} \right) \mubar = \Gamma_2 \Theta_3$ and $\left( \tr \chi - \frac{2}{r} \right) \mu = \Gamma_2 \Theta_2$ in $E_3 [ \mubar ]$ can be estimated similarly.
	
	The terms 
	\[
		\sum_{p_1+p_2 \geq p} \Gamma_{p_1} \cdot \psi_{p_2},
	\]
	in $E_3 \left[ \Thetat_p \right]$ give rise to quadratic terms of the form
	\[
		\mathfrak{D}^{k_1} \Gamma_{p_1} \cdot \mathfrak{D}^{k_2} \psi_{p_2},
	\]
	with $p_1+p_2 \geq p$ and $k_1 + k_2 \leq 3$.  These terms can be treated similarly since
	\[
		\int_{v_0}^v \int_{u_0}^u \int_{S_{u,v'}} r^{2p_2-4} \left\vert \mathfrak{D}^{k_2} \psi_{p_2} \right\vert^2 d \mu_{S_{u',v'}} du' dv'
		\leq
		C \left( \data + \frac{1}{v_0} \right),
	\]
	by Proposition \ref{prop:curvmain}.
	
	The quadratic terms arising from
	\[
		\sum_{p_1+p_2 \geq p} \Gamma_{p_1} \cdot \mathcal{T}_{p_2},
	\]
	are also similar using Proposition \ref{prop:tmain}.
	
	The terms
	\[
		\sum_{p_1+p_2 \geq p} \rp_{p_1} \cdot \mathfrak{D} \mathcal{T}_{p_2},
	\]
	in $E_3 \left[ \Thetat_p \right]$ give rise to terms of the form
	\[
		\rp_{p_1} \cdot \mathfrak{D}^{k'} \mathcal{T}_{p_2},
	\]
	in $E_3 \left[ \mathfrak{D}^3 \Thetat_p \right]$ with $p_1+p_2 \geq p$, $k'\leq 4$.  For these,
	\begin{align*}
		&
		\int_{v_0}^v \int_{u_0}^u \int_{S_{u,v'}} r^{2p-4} \mathfrak{D}^3 \Thetat_p \cdot \rp_{p_1} \cdot \mathfrak{D}^{k'} \mathcal{T}_{p_2} d \mu_{S_{u',v'}} du' dv'
		\\
		& \qquad
		\leq
		\int_{v_0}^v \int_{u_0}^u \int_{S_{u,v'}} r^{2p-4} \left( \left\vert \mathfrak{D}^3 \Thetat_p \right\vert^2 + \left\vert \rp_{p_1} \cdot \mathfrak{D}^{k'} \mathcal{T}_{p_2} \right\vert^2 \right) d \mu_{S_{u',v'}} du' dv'
		\\
		& \qquad
		\leq
		\int_{v_0}^v \int_{u_0}^u \int_{S_{u,v'}} r^{2p-4} \left\vert \mathfrak{D}^3 \Thetat_p \right\vert^2 d \mu_{S_{u',v'}} du' dv'
		\\
		& \qquad \qquad
		+
		C \int_{v_0}^v \int_{u_0}^u \int_{S_{u,v'}} r^{2p_2-4} \left\vert \mathfrak{D}^{k'} \mathcal{T}_{p_2} \right\vert^2 d \mu_{S_{u',v'}} du' dv'
		\\
		& \qquad
		\leq
		\int_{v_0}^v \int_{u_0}^u \int_{S_{u,v'}} r^{2p-4} \left\vert \mathfrak{D}^3 \Thetat_p \right\vert^2 d \mu_{S_{u',v'}} du' dv' +
		C\left( \data + \frac{1}{v_0} \right),
	\end{align*}
	by Proposition \ref{prop:tmain}.
	
	Consider now the linear term $\frac{1}{r} \mathfrak{D}^3\mu$ appearing in $E_3 \left[ \mathfrak{D}^3 \mubar \right]$.  Recall that $\mubar = \Theta_3$ so $r^{2p-4} = r^2$ and,
	\begin{align*}
		&
		\int_{v_0}^v \int_{u_0}^u \int_{S_{u,v'}} r^{2} \mathfrak{D}^3 \mubar \cdot \frac{1}{r} \mathfrak{D}^{3} \mu d \mu_{S_{u',v'}} du' dv'
		\\
		& \qquad
		\leq
		\int_{v_0}^v \int_{u_0}^u \int_{S_{u,v'}} r^{2} \vert \mathfrak{D}^3 \mubar \vert^2 + \vert \mathfrak{D}^{3} \mu \vert^2 d \mu_{S_{u',v'}} du' dv'.
	\end{align*}
	Now,
	\begin{align*}
		\int_{v_0}^v \int_{u_0}^u \int_{S_{u,v'}} \vert \mathfrak{D}^{3} \mu \vert^2 d \mu_{S_{u',v'}} du' dv'
		&
		\leq
		\int_{v_0}^v \frac{1}{(v')^2} \int_{u_0}^u \int_{S_{u,v'}} r^2 \vert \mathfrak{D}^{3} \mu \vert^2 d \mu_{S_{u',v'}} du' dv'
		\\
		&
		\leq
		C \int_{v_0}^v \frac{1}{(v')^2} dv'
		\\
		&
		\leq
		\frac{C}{v_0}.
	\end{align*}
	The other linear principal terms in $E_3 \left[ \mathfrak{D}^3 \Thetat_p \right]$ arising from the other linear term in $E_3 \left[ \mubar \right]$, the linear term in $E_3 \left[ r \nablaslash \left( \tr \chibar + \frac{2}{r} \right) \right]$ and the linear terms arising from the commutation can be treated similarly.  These are actually even easier as they appear with an additional factor of $\frac{1}{r}$.
	Hence,
	\begin{align*}
		&
		\int_{v_0}^v \int_{S_{u,v'}} r^{2p-4} \left\vert \mathfrak{D}^3 \Thetat_p \right\vert^2 d \mu_{S_{u,v'}} dv'
		\leq
		\\
		& \qquad \qquad
		C \left( \int_{v_0}^v \int_{u_0}^u \int_{S_{u',v'}} r^{2p-4} \left\vert \mathfrak{D}^3 \Thetat_p \right\vert^2  d \mu_{S_{u',v'}} du' dv' + \data + \frac{1}{v_0} \right).
	\end{align*}
	This is true for all $u_0 \leq u' \leq u$ hence, by the Gr\"onwall inequality,
	\[
		\int_{v_0}^v \int_{S_{u,v'}} r^{2p-4} \left\vert \mathfrak{D}^3 \Thetat_p \right\vert^2 d \mu_{S_{u,v'}} dv'
		\leq
		C \left( \data + \frac{1}{v_0} \right).
	\]
	
	The estimate,
	\[
		\int_{v_0}^v \int_{S_{u,v'}} r^{4} \left\vert \mathfrak{D}^2 (r\nablaslash)^2 \kappa \right\vert^2 d \mu_{S_{u,v'}} dv'
		\leq
		C \left( \data + \frac{1}{v_0} \right),
	\]
	is obtained in exactly the same way.  Note that the $r\nablaslash \omega, r\nablaslash \omegadag$ terms which appear in $E_3[\kappa]$ will only give rise to terms involving $\mathfrak{D}^2(r\nablaslash)^2 \omega$ and $\mathfrak{D}^2(r\nablaslash)^2 \omegadag$ in $E_3[ \mathfrak{D}^2 r\nablaslash \kappa]$ which can be controlled by the bootstrap assumptions \eqref{eq:thetabcfirst}--\eqref{eq:thetabclast}.  At the principal order, $\omega$ and $\omegadag$ do not appear otherwise.
	
	Consider now the $\Thetaf$ variables.  Recall the renormalisation of Remark \ref{rem:renorm},
	\[
		\nablaslash_4 \left( r^p \mathfrak{D}^3 \Thetaf_p \right) = r^p E_4 \left[ \mathfrak{D}^3 \Thetaf_p \right],
	\]
	and hence that
	\[
		\nablaslash_4 \left( r^{2p-2} \left\vert \mathfrak{D}^3 \Thetaf_p \right\vert^2 \right) = 2r^{2p-2} \mathfrak{D}^3 \Thetaf_p \cdot E_4 \left[ \mathfrak{D}^3 \Thetaf_p \right] - \frac{2}{r} r^{2p-2} \left\vert \mathfrak{D}^3 \Thetaf_p \right\vert^2.
	\]
	By equation \eqref{eq:transportid4} with $h = r^{2p-2} \left\vert \mathfrak{D}^3 \Thetaf_p \right\vert^2$ then get
	\begin{multline*}
		\int_{S_{u',v}} r^{2p-2} \left\vert \mathfrak{D}^3 \Thetaf_p \right\vert^2 d \mu_{S_{u',v}} = \int_{S_{u',v_0}} r^{2p-2} \left\vert \mathfrak{D}^3 \Thetaf_p \right\vert^2 d \mu_{S_{u',v_0}}
		\\
		+ \int_{v_0}^v \int_{S_{u',v'}} 2 r^{2p-2} \mathfrak{D}^3 \Thetaf_p \cdot E_4 \left[ \mathfrak{D}^3 \Thetaf_p \right] d \mu_{S_{u',v'}} dv',
	\end{multline*}
	where $\left( \tr \chi - \frac{2}{r} \right) \cdot \mathfrak{D}^3 \Thetaf_p = \Gamma_2 \cdot \Theta_p$ has been absorbed into the error.  Integrating in $u'$ gives,
	\begin{multline*}
		\int_{u_0}^u \int_{S_{u',v}} r^{2p-2} \left\vert \mathfrak{D}^3 \Thetaf_p \right\vert^2 d \mu_{S_{u',v}} du'
		\\
		\leq C\left( \data + \int_{u_0}^u \int_{v_0}^v \int_{S_{u',v'}} 2 r^{2p-2} \mathfrak{D}^3 \Thetaf_p \cdot E_4 \left[ \mathfrak{D}^3 \Thetaf_p \right] d \mu_{S_{u',v'}} dv' du' \right).
	\end{multline*}
	Since $E_4 \left[ \mathfrak{D}^3 \Thetaf_p \right]$ is schematically like $\frac{1}{r^{p+2}}$, most of the error terms are estimated in exactly the same way as before (when the weight was $r^{2p-4}$) by
	\[
		C\left( \int_{u_0}^u \int_{v_0}^v \int_{S_{u',v'}} r^{2p-4} \left\vert \mathfrak{D}^3 \Thetaf_p \right\vert^2 d \mu_{S_{u',v'}} dv' du' + \data + \frac{1}{v_0} \right).
	\]
	Some care, however, needs to be taken with the pricipal linear terms and especially with the $\frac{1}{r} r \nablaslash \omega$ term in $E_4\left[ r \nablaslash \left( \tr \chi - \frac{2}{r} \right) \right]$.
	
	Consider first the $\frac{1}{r} \mubar$ term in $E_4[\mu]$.  Since $\mu = \Theta_2$ this will give the following error term,
	\begin{align*}
		&
		\int_{u_0}^u \int_{v_0}^v \int_{S_{u',v}} r^{2} \mathfrak{D}^3 \mu \cdot \frac{1}{r} \mathfrak{D}^3 \mubar  d \mu_{S_{u',v}} d v' du'
		\\
		& \qquad \qquad
		\leq
		\int_{u_0}^u \int_{v_0}^v \int_{S_{u',v}} r^{\frac{1}{2}} \vert \mathfrak{D}^3 \mu \vert^2 + r^{\frac{3}{2}} \vert \mathfrak{D}^3 \mubar \vert^2 d \mu_{S_{u',v}} d v' du'
		\\
		& \qquad \qquad
		\leq
		\int_{u_0}^u \int_{v_0}^v \int_{S_{u',v}} r^{\frac{1}{2}} \vert \mathfrak{D}^3 \mu \vert^2 d \mu_{S_{u',v}} d v' du'
		\\
		& \qquad \qquad \qquad
		+ \frac{1}{\sqrt{v_0}} \int_{u_0}^u \int_{v_0}^v \int_{S_{u',v}} r^{2} \vert \mathfrak{D}^3 \mubar \vert^2 d \mu_{S_{u',v}} d v' du'
		\\
		& \qquad \qquad
		\leq
		\int_{u_0}^u \int_{v_0}^v \int_{S_{u',v}} r^{\frac{1}{2}} \vert \mathfrak{D}^3 \mu \vert^2 d \mu_{S_{u',v}} d v' du' + \frac{C}{\sqrt{v_0}}.
	\end{align*}
	Hence
	\begin{align*}
		&
		\int_{u_0}^u \int_{S_{u',v}} r^2 \vert \mathfrak{D}^3 \mu \vert^2 d \mu_{S_{u',v}} d u'
		\\
		& \qquad \qquad
		\leq
		C \left( \int_{u_0}^u \int_{v_0}^v \int_{S_{u',v}} r^{\frac{1}{2}} \vert \mathfrak{D}^3 \mu \vert^2 d \mu_{S_{u',v}} d v' du' + \data + \frac{1}{\sqrt{v_0}} \right)
		\\
		& \qquad \qquad
		\leq
		C \left( \int_{v_0}^v \frac{1}{(v')^{\frac{3}{2}}} \int_{u_0}^u \int_{S_{u',v}} r^2 \vert \mathfrak{D}^3 \mu \vert^2 d \mu_{S_{u',v}} d v' du' + \data + \frac{1}{\sqrt{v_0}} \right),
	\end{align*}
	and so the Gr\"onwall inequality implies,
	\begin{align*}
		\int_{u_0}^u \int_{S_{u',v}} r^2 \vert \mathfrak{D}^3 \mu \vert^2 d \mu_{S_{u',v}} d u'
		&
		\leq C \exp \left( \int_{v_0}^v \frac{1}{(v')^{\frac{3}{2}}} dv' \right) \left( \data + \frac{1}{\sqrt{v_0}} \right)
		\\
		&
		\leq C \left( \data + \frac{1}{\sqrt{v_0}} \right).
	\end{align*}
	
	Consider now the $\frac{2}{r} r\nablaslash \omega$ term in $E_4 \left[ r \nablaslash \left( \tr \chi - \frac{2}{r} \right) \right]$.  If $\mathfrak{D}^3 = \mathfrak{D}^2 r \nablaslash$ or $\mathfrak{D}^2 \nablaslash_3$ then similarly get
	\begin{align*}
		&
		\int_{u_0}^u \int_{v_0}^v \int_{S_{u',v}} r^{2} \mathfrak{D}^3 r\nablaslash \left( \tr \chi - \frac{2}{r} \right) \cdot \frac{2}{r} \mathfrak{D}^3 r \nablaslash \omega  d \mu_{S_{u',v}} d v' du'
		\\
		& \qquad \qquad
		\leq
		\int_{u_0}^u \int_{v_0}^v \int_{S_{u',v}} r^{\frac{1}{2}} \left\vert \mathfrak{D}^3 r\nablaslash \left( \tr \chi - \frac{2}{r} \right) \right\vert^2 + r^{\frac{3}{2}} \vert \mathfrak{D}^3 r \nablaslash \omega \vert^2 d \mu_{S_{u',v}} d v' du'
		\\
		& \qquad \qquad
		\leq
		\int_{u_0}^u \int_{v_0}^v \int_{S_{u',v}} r^{\frac{1}{2}} \left\vert \mathfrak{D}^3 r\nablaslash \left( \tr \chi - \frac{2}{r} \right) \right\vert^2 d \mu_{S_{u',v}} d v' du' + \frac{C}{\sqrt{v_0}},
	\end{align*}
	using the bootstrap assumptions \eqref{eq:thetabcfirst}--\eqref{eq:thetabclast} for $\mathfrak{D}^2 (r \nablaslash)^2 \omega$ and $\mathfrak{D}^2 \nablaslash_3 r \nablaslash \omega$.  Using the Gr\"onwall inequality then again gives,
	\[
		\int_{u_0}^u \int_{S_{u',v}} r^2 \left\vert \mathfrak{D}^3 r\nablaslash \left( \tr \chi - \frac{2}{r} \right) \right\vert^2 d \mu_{S_{u',v}} d u'
		\leq C \left( \data + \frac{1}{\sqrt{v_0}} \right).
	\]
	
	Suppose now that $\mathfrak{D}^3 = \mathfrak{D}^2 r \nablaslash_4$.  This derivative of $r\nablaslash \left( \tr \chi - \frac{2}{r} \right)$ is estimated from the propagation equation directly,
	\begin{align*}
		\mathfrak{D}^2 r \nablaslash_4 r\nablaslash \left( \tr \chi - \frac{2}{r} \right) 
		= \mathfrak{D}^2 \left( - \left( \tr \chi - \frac{2}{r} \right) r\nablaslash \left( \tr \chi - \frac{2}{r} \right) - \frac{2}{r} \left( \tr \chi - \frac{2}{r} \right) + E_4\left[ r \nablaslash \left( \tr \chi - \frac{2}{r} \right) \right] \right).
	\end{align*}
	Note that there are no principal terms on the right hand side (all involve at most 3 derivatives) and, since they all decay like $\frac{1}{r^3}$, by Propositions \ref{prop:Ricci4main}, \ref{prop:Ricci3main} and \ref{prop:tmain},
	\[
		\int_{u_0}^u \int_{S_{u',v}} r^2 \left\vert \mathfrak{D}^2 r \nablaslash_4 r\nablaslash \left( \tr \chi - \frac{2}{r} \right) \right\vert^2 d \mu_{S_{u',v}} d u'
		\leq C \left( \data + \frac{1}{\sqrt{v_0}} \right).
	\]
\end{proof}

\subsection{Top Order Estimates for Ricci Coefficients}
In order to estimate the remaining Ricci coefficients at the top order, the following estimates for the Gauss curvature of the spheres is required.

\begin{proposition} \label{prop:Gauss}
	For any $x \in \mathcal{A}$, if $u = u(x)$, $v = v(x)$ then, for $k = 0,1,2$, the Gauss curvature of the sphere $S_{u,v}$ satisfies,
	\[
		r^4 \int_{S_{u,v}} \left\vert \mathfrak{D}^k \left( K - \frac{1}{r^2} \right) \right\vert^2 d \mu_{S_{u,v}} \leq C.
	\]
\end{proposition}

\begin{proof}
	Recall the Gauss equation \eqref{eq:Gauss}, which can be rewritten,
	\begin{multline} \label{eq:newGauss}
		K - \frac{1}{r^2} = \frac{1}{2} \hat{\chi} \cdot \hat{\chibar} - \frac{1}{4} \left( \tr \chi - \frac{2}{r} \right) \left( \tr \chibar + \frac{2}{r} \right) 
		+ \frac{1}{2r} \left( \tr \chi - \frac{2}{r} \right) - \frac{1}{2r} \left( \tr \chibar + \frac{2}{r} \right) - \rho + \frac{1}{2} \Tslash_{34}.
	\end{multline}
	If $k\leq 1$ then, since $\Tslash_{34} = \mathcal{T}_4$, Proposition \ref{prop:tmain} implies that
	\[
		\vert \mathfrak{D}^k \Tslash_{34} \vert \leq \frac{C}{r^4},
	\]
	hence
	\[
		r^4 \int_{S_{u,v}} \vert \mathfrak{D}^k \Tslash_{34} \vert^2 d \mu_{S_{u,v}} \leq \frac{C}{r^4} \int_{S_{u,v}} d \mu_{S_{u,v}} \leq C.
	\]
	For $k=2$, first recall that Proposition \ref{prop:tmain} implies that
	\[
		\int_{v = v'} r^6 \vert \mathfrak{D}^2 \Tslash_{34} \vert^2 d \mu_{v = v'} + \int_{v = v'} r^6 \vert \mathfrak{D}^3 \Tslash_{34} \vert^2 d \mu_{v = v'} \leq C,
	\]
	for all $v_0 \leq v' \leq \infty$.  Hence
	\begin{align*}
		&
		\left\vert \int_{S_{u,v'}} r^6 \vert \mathfrak{D}^2 \Tslash_{34} \vert^2 d \mu_{S_{u,v'}} - \int_{S_{u_0,v'}} r^6 \vert \mathfrak{D}^2 \Tslash_{34} \vert^2 d \mu_{S_{u_0,v'}} \right\vert
		\\
		& \qquad \qquad
		= \left\vert  \partial_u \int_{u_0}^u \int_{S_{u',v'}} r^6 \vert \mathfrak{D}^2 \Tslash_{34} \vert^2 d \mu_{S_{u',v'}} \right\vert
		\\
		& \qquad \qquad
		\leq C \left( \int_{v = v'} r^6 \vert \mathfrak{D}^2 \Tslash_{34} \vert^2 + r^6 \vert \mathfrak{D}^3 \Tslash_{34} \vert^2 d \mu_{v = v'} \right)
		\\
		& \qquad \qquad
		\leq C.
	\end{align*}
	Similarly, since
	\[
		\int_{v=v'} r^4 \vert \mathfrak{D}^k \rho \vert^2 d \mu_{S_{u,v}} \leq C,
	\]
	for $k = 0,1,2,3$, one obtains,
	\[
		\int_{S_{u,v'}} r^4 \vert \mathfrak{D}^k \rho \vert^2 d \mu_{S_{u,v'}} \leq C,
	\]
	for $k = 0,1,2$.
	
	Similarly for the Ricci coefficient terms on the right hand side of \eqref{eq:newGauss}, one can easily check that each has the correct decay to be controlled after being multiplied by $r^4$ and integrated on the spheres.  Moreover, once 3 derivatives have been taken, in the nonlinear terms there will be at most one factor involving $3$ derivatives and so the other terms can be estimated in $L^{\infty}$ as in Sections \ref{section:curvature} and \ref{section:Ricci}.  Hence
	\[
		r^4 \int_{S_{u,v}} \left\vert \mathfrak{D}^k \left( K - \frac{1}{r^2} \right) \right\vert^2 d \mu_{S_{u,v}} \leq C.
	\]
\end{proof}

\begin{proposition} \label{prop:divcurltr}
	Let $\xi$ be a totally symmetric $(0,j+1)$ $S_{u,v}$ tensor such that
	\[
		\divslash \xi = a, \qquad \curlslash \xi = b, \qquad \tr \xi = c,
	\]
	and assume that the bounds on the Gauss curvature of Proposition \ref{prop:Gauss} hold.  Then, for $1\leq k \leq 4$,
	\[
		\Vert (r \nablaslash)^k \xi \Vert_{L^2(S_{u,v})} 
		\leq C \Bigg[ \sum_{i=1}^{k-1} \Big( \Vert r (r\nablaslash)^i a \Vert_{L^2(S_{u,v})} + \Vert r (r\nablaslash)^i b \Vert_{L^2(S_{u,v})} 
		+ \Vert (r\nablaslash)^i c \Vert_{L^2(S_{u,v})} \Big)
		+ \Vert \xi \Vert_{L^2(S_{u,v})} \Bigg],
	\]
	where, if $j=0$, then $\tr \xi$ is defined to be $0$.
\end{proposition}

\begin{proof}
	The following identity is satisfied by $\xi,a,b,c$,
	\begin{equation} \label{eq:elliptic}
		\int_{S_{u,v}} \vert \nablaslash \xi \vert^2 + (j+1) K\vert \xi \vert^2 d \mu_{S_{u,v}} = \int_{S_{u,v}} \vert a \vert^2 + \vert b \vert^2 + j K \vert c \vert^2 d \mu_{S_{u,v}},
	\end{equation}
	where $K$ is the Gauss curvature of $S_{u,v}$.  See Chapter 7 of \cite{Ch}.
	
	By Proposition \ref{prop:Gauss} and the Sobolev inequality,
	\[
		\vert K \vert \leq \frac{C}{r^2},
	\]
	uniformly.  This immediately gives the estimate for $k=1$ after multiplying the identity \eqref{eq:elliptic} by $r^2$.
	
	For $k=2$, note that the symmetrised angular derivative of $\xi$,
	\[
		(\nablaslash \xi)^s_{BA_1\ldots A_{j+1}} = \frac{1}{j+2} \left( \nablaslash_B \xi_{A_1 \ldots A_{j+1}} + \sum_{i=1}^{j+1} \nablaslash_{A_i} \xi_{A_1 \ldots A_{i-1}A_{i+1}\ldots A_{j+1}} \right),
	\]
	satisfies
	\begin{align*}
		\divslash (\nablaslash \xi)^s 
		&
		=
		(\nablaslash a)^s - \frac{1}{j+2} ({}^* \nablaslash b)^s + (j+1) K \xi - \frac{2K}{j+1} ( g \otimes^s c)
		\\
		\curlslash (\nablaslash \xi)^s
		&
		=
		\frac{j+1}{j+2} (\nablaslash b)^s + (j+1) K({}^* \xi)^s
		\\
		\tr (\nablaslash \xi)^s
		&
		=
		\frac{2}{j+2} a + \frac{j}{j+2}(\nablaslash c)^s,
	\end{align*}
	where,
	\[
		({}^*\nablaslash b)^s_{A_1\ldots A_{j+1}} = \frac{1}{j+1} \sum_{i=1}^{j+1} {\epsslash_{A_i}}^B \nablaslash_B b_{A_1 \ldots A_{i-1}A_{i+1}\ldots A_{j+1}},
	\]
	and
	\[
		(g \otimes^s c)_{A_1 \ldots A_{j+1}} = \sum_{i<l=2,\ldots,j+1} \gslash_{A_i A_l} c_{A_1 \ldots A_{i-1}A_{i+1}\ldots A_{l-1}A_{l+1} \ldots A_{j+1}} .
	\]
	Again see Chapter 7 of \cite{Ch}.
	
	The identity \eqref{eq:elliptic} then gives,
	\begin{align*}
		\Vert \nablaslash^2 \xi \Vert_{L^2(S_{u,v})}^2
		\leq
		&
		C \bigg[ \Vert K( \vert \nablaslash \xi \vert^2 + \vert a\vert^2 + \vert \nablaslash c \vert^2 )\Vert_{L^1(S_{u,v})} + \Vert \nablaslash a \Vert_{L^2(S_{u,v})}^2
		\\
		&
		 + \Vert \nablaslash b \Vert_{L^2(S_{u,v})}^2  + \Vert K \xi\Vert_{L^2(S_{u,v})}^2 + \Vert K c \Vert_{L^2(S_{u,v})}^2 \bigg]
		 \\
		 \leq
		 &
		 C \bigg[ \frac{1}{r^2} \left( \Vert \nablaslash \xi \Vert_{L^2(S_{u,v})}^2 + \Vert a\Vert_{L^2(S_{u,v})}^2 + \Vert \nablaslash c \Vert_{L^2(S_{u,v})}^2 \right) + \Vert \nablaslash a \Vert_{L^2(S_{u,v})}^2
		\\
		&
		 + \Vert \nablaslash b \Vert_{L^2(S_{u,v})}^2  + \frac{1}{r^4} \left( \Vert \xi\Vert_{L^2(S_{u,v})}^2 + \Vert c \Vert_{L^2(S_{u,v})}^2 \right) \bigg],
	\end{align*}
	again using $\vert K\vert \leq \frac{C}{r^2}$.  Multiplying by $r^4$ and using the $k=1$ estimate then gives the estimate for $k=2$.
	
	For $k=3$ can similarly compute $(\nablaslash (\nablaslash \xi)^s )^s$ to get
	\begin{align*}
		\Vert \nablaslash^3 \xi \Vert_{L^2(S_{u,v})}^2
		\leq
		&
		C \bigg( \sum_{i=0}^2 \bigg[ \frac{1}{r^{4-2i}} \left( \Vert \nablaslash^i a \Vert_{L^2(S_{u,v})}^2 + \Vert \nablaslash^i b \Vert_{L^2(S_{u,v})}^2 \right) 
		\\
		&
		+ \frac{1}{r^{6-2i}} \left( \Vert \nablaslash^i c \Vert_{L^2(S_{u,v})}^2 + \Vert \nablaslash^i \xi \Vert_{L^2(S_{u,v})}^2 \right) \bigg]
		\\
		&
		+ \Vert (\nablaslash K ) \xi \Vert_{L^2(S_{u,v})}^2 + \Vert (\nablaslash K) c \Vert_{L^2(S_{u,v})}^2  \bigg).
	\end{align*}
	By the Sobolev inequality,
	\begin{align*}
		\Vert (\nablaslash K) \xi \Vert_{L^2(S_{u,v})}^2
		&
		\leq
		\Vert \xi \Vert_{L^{\infty}(S_{u,v})}^2 \Vert \nablaslash K \Vert_{L^2(S_{u,v})}^2
		\\
		&
		\leq
		\frac{C}{r^2} \Vert \nablaslash K \Vert_{L^2(S_{u,v})}^2 \sum_{i=0}^2 \Vert (r\nablaslash)^i \xi \Vert_{L^2(S_{u,v})}^2
		\\
		&
		\leq
		\frac{C}{r^4} \Vert r\nablaslash K \Vert_{L^2(S_{u,v})}^2 \sum_{i=0}^2 \Vert (r\nablaslash)^i \xi \Vert_{L^2(S_{u,v})}^2
		\\
		&
		\leq
		\frac{C}{r^8} \sum_{i=0}^2 \Vert (r\nablaslash)^i \xi \Vert_{L^2(S_{u,v})}^2.
	\end{align*}
	Similarly,
	\[
		\Vert (\nablaslash K) c \Vert_{L^2(S_{u,v})}^2 \leq \frac{C}{r^8} \sum_{i=0}^2 \Vert (r\nablaslash)^i c \Vert_{L^2(S_{u,v})}^2.
	\]
	Inserting this into the above inequality and multiplying by $r^6$ gives the result for $k=3$.
	
	Finally, for $k=4$, one similarly gets,
	\begin{align*}
		\Vert \nablaslash^4 \xi \Vert_{L^2(S_{u,v})}^2
		\leq
		&
		C \bigg( \sum_{i=1}^3 \bigg[ \frac{1}{r^{6-2i}} \Big( \Vert \nablaslash^i a \Vert_{L^2(S_{u,v})}^2 + \nablaslash^i b \Vert_{L^2(S_{u,v})}^2 \Big)
		\\
		&
		+ \frac{1}{r^{8-2i}} \Big( \Vert \nablaslash^i a \Vert_{L^2(S_{u,v})}^2 + \nablaslash^i b \Vert_{L^2(S_{u,v})}^2 \Big) \bigg]  
		\\
		&
		+ \Vert \nablaslash K \nablaslash \xi \Vert_{L^2(S_{u,v})}^2 + \Vert \xi \nablaslash^2 K \Vert_{L^2(S_{u,v})}^2 + \frac{1}{r^2} \Vert \xi \nablaslash K \Vert_{L^2(S_{u,v})}^2
		\\
		&
		+ \Vert \nablaslash K \nablaslash c \Vert_{L^2(S_{u,v})}^2 + \Vert c \nablaslash^2 K \Vert_{L^2(S_{u,v})}^2 + \frac{1}{r^2} \Vert c \nablaslash K \Vert_{L^2(S_{u,v})}^2 
		\\
		&
		+ \Vert a \nablaslash K \Vert_{L^2(S_{u,v})}^2\bigg).
	\end{align*}
	Using the Sobolev inequality as above get,
	\begin{align*}
		&
		\Vert \xi \nablaslash^2 K \Vert_{L^2(S_{u,v})}^2 + \frac{1}{r^2} \Vert \xi \nablaslash K \Vert_{L^2(S_{u,v})}^2
		\\
		& \qquad \qquad
		\leq
		\Vert \xi \Vert_{L^{\infty}(S_{u,v})}^2 \left( \Vert \nablaslash^2 K \Vert_{L^2(S_{u,v})}^2 + \frac{1}{r^2} \Vert \nablaslash K \Vert_{L^2(S_{u,v})}^2 \right)
		\\
		& \qquad \qquad
		\leq
		\frac{C}{r^2} \sum_{i=0}^2 \Vert (r \nablaslash)^i \xi \Vert_{L^2(S_{u,v})}^2 \left( \frac{1}{r^4} \Vert (r\nablaslash)^2 K \Vert_{L^2(S_{u,v})}^2 + \frac{1}{r^4} \Vert r\nablaslash K \Vert_{L^2(S_{u,v})}^2 \right)
		\\
		& \qquad \qquad
		\leq
		\frac{C}{r^{10}} \sum_{i=0}^2 \Vert (r \nablaslash)^i \xi \Vert_{L^2(S_{u,v})}^2,
	\end{align*}
	by Proposition \ref{prop:Gauss}.  Similarly,
	\[
		\Vert c \nablaslash^2 K \Vert_{L^2(S_{u,v})}^2 + \frac{1}{r^2} \Vert c \nablaslash K \Vert_{L^2(S_{u,v})}^2 
		\leq
		\frac{C}{r^{10}} \sum_{i=0}^2 \Vert (r \nablaslash)^i a \Vert_{L^2(S_{u,v})}^2,
	\]
	and,
	\begin{align*}
		\Vert a \nablaslash K \Vert_{L^2(S_{u,v})}^2
		&
		\leq
		\frac{C}{r^2} \sum_{i=0}^2 \Vert (r \nablaslash)^i a \Vert_{L^2(S_{u,v})}^2 \Vert \nablaslash K \Vert_{L^2(S_{u,v})}^2
		\\
		&
		\leq
		\frac{C}{r^4} \sum_{i=0}^2 \Vert (r \nablaslash)^i a \Vert_{L^2(S_{u,v})}^2 \Vert r\nablaslash K \Vert_{L^2(S_{u,v})}^2
		\\
		&
		\leq 
		\frac{C}{r^8} \sum_{i=0}^2 \Vert (r \nablaslash)^i a \Vert_{L^2(S_{u,v})}^2.
	\end{align*}
	For the remaining terms,
	\begin{align*}
		\Vert \nablaslash K \nablaslash \xi \Vert_{L^2(S_{u,v})}^2
		&
		\leq 
		\frac{1}{r^4} \Vert r \nablaslash \xi \Vert_{L^{\infty}(S_{u,v})}^2 \Vert r \nablaslash K \Vert_{L^2(S_{u,v})}^2
		\\
		&
		\leq 
		\frac{C}{r^{10}} \sum_{i=1}^3 \Vert (r\nablaslash)^i \xi \Vert_{L^2(S_{u,v})}^2,
	\end{align*}
	and
	\[
		\Vert \nablaslash K \nablaslash c \Vert_{L^2(S_{u,v})}^2
		\leq 
		\frac{C}{r^{10}} \sum_{i=1}^3 \Vert (r\nablaslash)^i c \Vert_{L^2(S_{u,v})}^2.
	\]
	Inserting into the above estimate and multiplying by $r^8$ gives the result for $k=4$.
\end{proof}

If $\xi$ is a symmetric trace free $(0,2)$ $S_{u,v}$ tensor then it suffices to know only its divergence.

\begin{proposition} \label{prop:divelliptic}
	 Let $\xi$ be a symmetric trace free $(0,2)$ $S_{u,v}$ tensor such that
	 \[
		  \divslash \xi = a,
	 \]
	 and assume that the bounds on the Gauss curvature of Proposition \ref{prop:Gauss} hold.  Then, for $1\leq k\leq 4$,
	 \[
		  \Vert (r \nablaslash)^k \xi \Vert_{L^2(S_{u,v})} 
		\leq C \Bigg[ \sum_{i=1}^{k-1} \Vert r (r\nablaslash)^i a \Vert_{L^2(S_{u,v})}
		+ \Vert \xi \Vert_{L^2(S_{u,v})} \Bigg].
	 \]
\end{proposition}

\begin{proof}
	 This follows from the previous proposition since,
	 \[
		  \curlslash \xi = {}^* a,
	 \]
	 which follows from the fact that $\xi$ is a symmetric trace free $(0,2)$ $S_{u,v}$ tensor.
\end{proof}

These elliptic estimates can be used to recover the remainder of the bootstrap assumptions \eqref{eq:thetabcfirst}--\eqref{eq:thetabclast} for Ricci coefficients at the top order.

\begin{proposition} \label{prop:Riccitop}
	For any $x \in \mathcal{A}''$, if $u = u(x)$, $v = v(x)$ then, for $k = 0,1,2,3$,
	\begin{align*}
		\int_{u_0}^u \int_{S_{u',v}} r^2 \left\vert \mathfrak{D}^k r\nablaslash \hat{\chi} \right\vert^2 d \mu_{S_{u',v}} d u'
		&
		\leq C \left( \data + \frac{1}{v_0} \right),
		\\
		\int_{v_0}^v \int_{S_{u,v'}} r^{-2} \left\vert \mathfrak{D}^k r\nablaslash \hat{\chibar} \right\vert^2 d \mu_{S_{u,v'}} d v'
		&
		\leq C \left( \data + \frac{1}{v_0} \right),
		\\
		\int_{u_0}^u \int_{S_{u',v}} \left\vert \mathfrak{D}^k r\nablaslash \eta \right\vert^2 d \mu_{S_{u',v}} d u'
		&
		\leq C \left( \data + \frac{1}{v_0} \right),
		\\
		\int_{v_0}^v \int_{S_{u,v'}} \left\vert \mathfrak{D}^k r\nablaslash \etabar \right\vert^2 d \mu_{S_{u,v'}} d v'
		&
		\leq C \left( \data + \frac{1}{v_0} \right),
		\\
	 \end{align*}
	 and for $k = 0,1,2$,
	 \begin{align*}
		\int_{v_0}^v \int_{S_{u,v'}} r^2 \left\vert \mathfrak{D}^k (r\nablaslash)^2 \omega \right\vert^2 d \mu_{S_{u,v'}} d v'
		&
		\leq C \left( \data + \frac{1}{v_0} \right),
		\\
		\int_{v_0}^v \int_{S_{u,v'}} r^2 \left\vert \mathfrak{D}^k \nablaslash_3 r\nablaslash \omega \right\vert^2 d \mu_{S_{u,v'}} d v'
		&
		\leq C \left( \data + \frac{1}{v_0} \right),
		\\
		\int_{v_0}^v \int_{S_{u,v'}} r^2 \left\vert \mathfrak{D}^k (r\nablaslash)^2 \omegadag \right\vert^2 d \mu_{S_{u,v'}} d v'
		&
		\leq C \left( \data + \frac{1}{v_0} \right).
	 \end{align*}
\end{proposition}

\begin{proof}
	 Consider first $\hat{\chi}$.  Recall the Codazzi equation \eqref{eq:Codazziout}, which can be schematically written,
	 \[
		  \divslash \hat{\chi} = \frac{1}{2r} r \nablaslash \left( \tr \chi - \frac{2}{r} \right) + \sum_{p_1+p_2\geq 3} (\rp_{p_1} + \Gamma_{p_1}) \cdot \Gamma_{p_2} + \sum_{p_1\geq \frac{7}{2}} (\psi_{p_1} + \mathcal{T}_{p_1}).
	 \]
	 Hence, if $\mathfrak{D}^k = (r\nablaslash)^k$, Proposition \ref{prop:divelliptic} immediately gives the result by Propositions \ref{prop:Ricci4main}, \ref{prop:Ricci3main}, \ref{prop:curvmain}, \ref{prop:tmain}, and \ref{prop:Thetamain}.
	 
	 If $\mathfrak{D}^k = (r\nablaslash)^{k-1} \nablaslash_3$ one has to commute the equation by $\nablaslash_3$.  By Lemma \ref{lemma:commutation} this will only generate terms (with good $r$ weights) which have already been estimated in Propositions \ref{prop:Ricci4main}, \ref{prop:Ricci3main}, \ref{prop:curvmain}, \ref{prop:tmain}.  Since $\nablaslash_3 \hat{\chi}$ is still a symmetric trace free $(0,2)$ $S_{u,v}$ tensor, can again apply Proposition \ref{prop:Gauss} to get,
	 \[
		  \int_{u_0}^u \int_{S_{u',v}} r^2 \left\vert  (r\nablaslash)^k \nablaslash_3 \hat{\chi} \right\vert^2 d \mu_{S_{u',v}} d u' 
		  \leq 
		  C \left( \data + \frac{1}{v_0} \right)
	 \]
	 One of the $r\nablaslash$ can be commuted with $\nablaslash_3$.  Again this will only generate error terms which ahve already been estimated.
	 
	 The same procedure works for general $\mathfrak{D}^k$ as commuting the Codazzi equation with $r\nablaslash$ will also only produce terms which have already been estimated.
	 
	 Consider now $\hat{\chibar}$.  This is estimated in exactly the same way using the other Codazzi equation \eqref{eq:Codazziin},
	 \[
		  \divslash \hat{\chibar} = \frac{1}{2r} r\nablaslash \left(  \tr \chibar + \frac{2}{r} \right) + \sum_{p_1+p_2\geq 3} (\rp_{p_1} + \Gamma_{p_1}) \cdot \Gamma_{p_2} + \sum_{p_1\geq 2} (\psi_{p_1} + \psi_{p_2}),
	 \]
	 Propositions \ref{prop:Ricci4main}, \ref{prop:Ricci3main}, \ref{prop:curvmain}, \ref{prop:tmain}, and the estimate for $\mathfrak{D}^3 r \nablaslash \left( \tr \chibar + \frac{2}{r} \right)$ on the outgoing null hypersurfaces from Proposition \ref{prop:Thetamain}.
	 
	 Consider now $\eta,\etabar$.  They satisfy the following $\divslash,\curlslash$ systems,
	 \[
		  \divslash \eta = \frac{1}{2} \hat{\chi} \cdot \hat{\chibar} - \rho - \mu, \qquad \curlslash \eta = \sigma - \frac{1}{2} \hat{\chi} \wedge \hat{\chibar},
	 \]
	 \[
		 \divslash \etabar = \frac{1}{2} \hat{\chi} \cdot \hat{\chibar} - \rho - \mubar, \qquad \curlslash \etabar = \frac{1}{2} \hat{\chi} \wedge \hat{\chibar} - \sigma, 
	 \]
	 and so can be estimated exactly as $\hat{\chi}, \hat{\chibar}$, now using Proposition \ref{prop:divcurltr}.  Recall that we set $\tr \xi = 0$ if $\xi$ is an $S_{u,v}$ one form.
	 
	 Finally, since
	 \[
		  \kappa = \nablaslash \omega + {}^* \nablaslash \omegadag + \beta,
	 \]
	 by definition, $\nablaslash \omega$ and $\nablaslash \omegadag$ satisfy the $\divslash, \curlslash$ systems,
	 \[
		  \divslash \nablaslash \omega = \divslash \kappa - \divslash \beta, \qquad \curlslash \nablaslash \omega = 0,
	 \]
	 \[
		  \divslash \nablaslash \omegadag = \curlslash \kappa - \curlslash \beta, \qquad \curlslash \nablaslash \omegadag = 0,
	 \]
	 and so $\mathfrak{D}^2 (r\nablaslash)^2 \omega$ and $\mathfrak{D}^2 (r\nablaslash)^2 \omegadag$ can be estimated similarly.
	 
	 Finally, for $\mathfrak{D}^2 \nablaslash_3 r\nablaslash \omega$, recall that $\omega$ satisfies the propagation equation,
	 \[
		  \nablaslash_3 \omega = 2 \eta \cdot \etabar - \vert \etabar \vert^2 - \rho - \frac{1}{2} \Tslash_{34}.
	 \]
	 Commuting with $r\nablaslash$ this gives, by Proposition \ref{lemma:commutation},
	 \begin{align*}
		  \nablaslash_3 r\nablaslash \omega
		  =
		  &
		  2 (r\nablaslash \eta) \cdot \etabar + 2 \eta \cdot(r\nablaslash \etabar) - 2\etabar \cdot (r\nablaslash \etabar) - r \nablaslash \rho - \frac{1}{2} r \nablaslash \Tslash_{34}
		  \\
		  &
		  + r \left( 2 \eta \cdot \etabar - \vert \etabar \vert^2 - \rho - \frac{1}{2} \Tslash_{34} \right)(\eta + \etabar) - \hat{\chibar} \cdot r \nablaslash \omega - \frac{1}{2} \left( \tr \chibar + \frac{2}{r} \right) r \nablaslash \omega.
	 \end{align*}
	 The estimate for $\mathfrak{D}^2 \nablaslash_3 r\nablaslash \omega$ follows by applying $\mathfrak{D}^2$ to the right hand side, multiplying by $r^2$, integrating over the constant $u$ hypersurfaces and applying Propositions \ref{prop:Ricci4main}, \ref{prop:Ricci3main}, \ref{prop:curvmain}, \ref{prop:tmain}.
\end{proof}

Now the bootstrap assumptions \eqref{eq:thetabcfirst}--\eqref{eq:thetabclast} have been recovered with better constants and hence, provided $\bc$, $v_0$ are taken suitably small then $\mathcal{A}'' \subset \mathcal{A}$ is open, closed, connected, non-empty, and hence $\mathcal{A}'' = \mathcal{A}$.  The remaining bootstrap assumptions of Section \ref{section:ba}, \eqref{eq:bba}, \eqref{eq:Gammaslashbanull} can now be recovered.

\begin{proposition} \label{prop:bmain}
	 If $x \in \mathcal{A}$, and if $u = u(x)$, $v = v(x)$, then,
	 \[
		  \int_{v_0}^v \int_{S_{u,v'}} r^{-2} \vert \mathfrak{D}^3 r \nablaslash b \vert^2 d\mu_{S_{u,v'}} dv' \leq C \left( \data + \frac{1}{v_0} \right),
	 \]
	 and
	 \[
		  \int_{u_0}^u \int_{S_{u',v}} \left\vert (r\nablaslash)^3 \left( \Gammaslash - \Gammaslash^{\circ} \right) \right\vert^2 d\mu_{S_{u',v}} du' \leq C \left( \data + \frac{1}{v_0} \right).
	 \]
\end{proposition}

\begin{proof}
	 The proof of the first estimate is identical to that of Proposition \ref{prop:Thetamain} using the propagation equation,
	 \[
		  \nablaslash_3 b = 2(\eta^{\#} - \etabar^{\#}) + \hat{\chibar}^{\#} \cdot b + \frac{1}{2} \left( \tr \chibar + \frac{2}{r} \right) b - \frac{1}{r} b.
	 \]
	 
	 For $\Gammaslash - \Gammaslash^{\circ}$, the propagation equation in the outgoing direction \eqref{eq:Gammaslash4} is used.  The commuted equation will take the form,
	 \[
		  \nablaslash_4 (r\nablaslash)^3 \left( \Gammaslash - \Gammaslash^{\circ} \right) 
		  + 
		  \frac{1}{2} \tr \chi (r\nablaslash)^3 \left( \Gammaslash - \Gammaslash^{\circ} \right) 
		  = 
		  E_4 \left[ (r\nablaslash)^3 \left( \Gammaslash - \Gammaslash^{\circ} \right) \right],
	 \]
	 where, by Propositions \ref{prop:Ricci4main}, \ref{prop:Ricci3main}, \ref{prop:curvmain}, \ref{prop:tmain}, \ref{prop:Thetamain} and \ref{prop:Riccitop}, an argument identical to that in the proof of Proposition \ref{prop:Thetamain} can be used to show that,
	 \[
		  \int_{u_0}^u \int_{v_0}^v \int_{S_{u',v'}} (r\nablaslash)^3 \left( \Gammaslash - \Gammaslash^{\circ} \right) \cdot E_4 \left[ (r\nablaslash)^3 \left( \Gammaslash - \Gammaslash^{\circ} \right) \right] d\mu_{S_{u',v'}} dv' du' 
		  \leq 
		  C\left( \data + \frac{1}{v_0} \right).
	 \]
	 It follows that,
	 \[
		  \nablaslash_4 \left( \left\vert (r\nablaslash)^3 \left( \Gammaslash - \Gammaslash^{\circ} \right) \right\vert^2 \right) 
		  = 
		  (r\nablaslash)^3 \left( \Gammaslash - \Gammaslash^{\circ} \right) \cdot \left( - \tr \chi (r\nablaslash)^3 \left( \Gammaslash - \Gammaslash^{\circ} \right)
		  + 
		  2 E_4 \left[ (r\nablaslash)^3 \left( \Gammaslash - \Gammaslash^{\circ} \right) \right] \right),
	 \]
	 and hence, using the identity \eqref{eq:transportid4},
	 \begin{multline*}
		  \int_{v_0}^v \int_{S_{u,v}} \left\vert (r\nablaslash)^3 \left( \Gammaslash - \Gammaslash^{\circ} \right) \right\vert^2 d\mu_{S_{u,v}}
		  \leq 
		  \int_{S_{u,v_0}} \left\vert (r\nablaslash)^3 \left( \Gammaslash - \Gammaslash^{\circ} \right) \right\vert^2 d\mu_{S_{u,v_0}}
		  \\
		  + \int_{v_0}^v \int_{S_{u,v'}} (r\nablaslash)^3 \left( \Gammaslash - \Gammaslash^{\circ} \right) \cdot E_4 \left[ (r\nablaslash)^3 \left( \Gammaslash - \Gammaslash^{\circ} \right) \right] d\mu_{S_{u,v'}} dv'.
	 \end{multline*}
	 The result follows by integrating in $u$.
\end{proof}

\section{The Last Slice Argument and the End of the Proof} \label{section:lastslice}
The proof of Theorem \ref{thm:main3} follows from Theorem \ref{thm:main4} together with the following two local existence theorems, whose proofs are not discussed here, via a \emph{last slice} argument.  The structure of the last slice argument is outlined below.

\begin{theorem}[Local existence for the Cauchy problem for the massless Einstein--Vlasov system \cite{ChBr71}, \cite{ChBrGe}, \cite{Ri}] \label{thm:loccauchy}
	Given a smooth initial data set $(\Sigma,g_0,k,f_0)$ for the massless Einstein--Vlasov system (satisfying constraint equations) there exists a unique smooth maximal Cauchy development satisfying the massless Einstein--Vlasov system such that $\Sigma$ is a Cauchy hypersurface with induced first and second fundamental form $g_0, k$ respectively and $f \vert_{P\vert_{\Sigma}} = f_0$.
\end{theorem}

\begin{theorem}[Local existence for the characteristic initial value problem for the massless Einstein--Vlasov system] \label{thm:locchar}
	Given smooth characteristic initial data for the massless Einstein--Vlasov system (satisfying constraint equations) on (what will become) null hypersurfaces $N_1, N_2$ intersecting transversely at a spacelike surface $S = N_1 \cap N_2$, there exists a non-empty maximal development of the data, bounded in the past by a neighbourhood of $S$ in $N_1 \cap N_2$.
\end{theorem}

The analogue of Theorem \ref{thm:locchar} for the vacuum Einstein equations is a result of Rendall \cite{Re}.  For the Einstein--Vlasov system see also \cite{ChPa}, \cite{ChBrCh}.

Suppose that $\bc$, $\frac{1}{v_0}$ and the bootstrap constant $\overline{C}$ satisfy the smallness assumption of Theorem \ref{thm:main4}.  Define the function $t:=v+u$, and the hypersurfaces $\Sigma_{t'} := \{t = t' \} \cap \{ u_0 \leq u \leq u_f, v_0 \leq v <\infty \}$.  Whenever the bootstrap assumptions on $b$ and $1 - \frac{1}{\Omega^2}$ hold (see Section \ref{section:ba}), clearly $dt$ is timelike and hence the surfaces $\Sigma_t$ are spacelike.  For a given time $t$, define the region,
\[
	\mathcal{M}_t := \{u_0 \leq u \leq u_f\} \cap \{v_0 \leq v < \infty \} \cap \bigcup_{t_0 \leq t' < t} \Sigma_{t'},
\]
where $t_0 = v_0 + u_0$.  Let $t^*$ denote the supremum over all times $t$ such that a smooth solution to the massless Einstein--Vlasov system \eqref{eq:EV}--\eqref{eq:vlas} exists in the region $\mathcal{M}_{t^*}$ attaining the given data on $\left( \{u=u_0\} \cup \{ v = v_0\} \right) \cap \bigcup_{t_0 \leq t < t^*} \Sigma_t$ and, for any $u',v'$ with $u' + v' \leq t^*$, the bootstrap assumptions \eqref{eq:Ricciba}--\eqref{eq:Gammaslashbanull} hold for all $u,v$ with $u_0 \leq u \leq u', v_0 \leq v \leq v'$.  Such a time clearly exists by Theorem \ref{thm:locchar}, provided $\bc$ is sufficiently small.\footnote{The transformation from harmonic coordinates to double null coordinates is carried out in detail in a related setting in Chapter 16.3 of \cite{Ch}.}

The aim is to show that $t^* = \infty$, so suppose for contradiction that $t^* < \infty$.  From the bounds \eqref{eq:Ricciba}--\eqref{eq:Gammaslashbanull}, which hold in for $u,v$ in the regions $u_0 \leq u \leq u', v_0 \leq v \leq v'$ uniformly for all $u',v'$ such that $u' + v' < t^*$, higher regularity bounds can be obtained from the equations via commutation, the equations being essentially linear at this stage (this is carried out in detail in a related setting in Chapter 16.2 of \cite{Ch}).  Hence the solution extends smoothly to $\Sigma_{t^*}$, providing Cauchy data for the Einstein equations on $\Sigma_{t^*}$.  Using this Cauchy data together with the characteristic data on $\{ u =u_0\}$ (and possibly the characteristic data on $\{ v = v_0\}$ if $t^* < u_f + v_0$), Theorem \ref{thm:loccauchy} and Theorem \ref{thm:locchar} imply that a smooth solution to the mixed Cauchy, characteristic initial value problem exists in the region $\mathcal{M}_{t^* + \varepsilon}$ for some small $\varepsilon > 0$.  This is depicted in Figure \ref{fig:lastslice}.

\begin{figure} 
  \centering
	\includegraphics[scale=0.25]{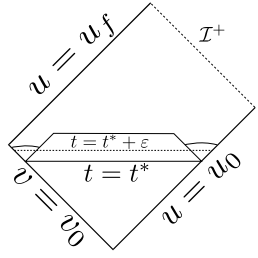}
	\caption{The solution to the mixed Cauchy, characteristic initial value problem when $t^* < u_f + v_0$.} \label{fig:lastslice}
\end{figure}

Since the bootstrap assumptions \eqref{eq:Ricciba}--\eqref{eq:Gammaslashbanull} hold in $u_0 \leq u \leq u', v_0 \leq v \leq v'$ for all $u',v'$ such that $u' + v' \leq t^*$, Theorem \ref{thm:main4} implies they in fact hold with the better constant $\frac{\overline{C}}{2}$.  Then, taking $\varepsilon$ smaller if necessary, by compactness of $\Sigma_{t^*}$ and continuity they will hold for all $u',v'$ with $u' + v' \leq t^* + \varepsilon$ (with constant $\overline{C}$).  This contradicts the maximality of $t^*$ and hence the solution exists in the entire region.

\end{document}